\documentclass[11pt]{article}
\usepackage{amsfonts,amsmath,amssymb,boxedminipage,color,url,fullpage,nccmath}
\usepackage{paralist,amsthm}
\usepackage{enumerate}
\usepackage[numbers]{natbib} 
\usepackage[pdfstartview=FitH,colorlinks,linkcolor=blue,filecolor=blue,citecolor=blue,urlcolor=blue]{hyperref} 
\usepackage[labelfont=bf]{caption}
\usepackage{aliascnt}
\usepackage{cleveref}
\usepackage{xspace}
\usepackage{tikz}
\usetikzlibrary{arrows,shapes,automata,backgrounds,petri}
\usepackage{verbatim}
\usepackage{multirow}
\usepackage{versions}
\usepackage{subcaption}
\usepackage{array}
\usepackage{mathtools}
\usepackage{xstring}
\usepackage{titlesec}

\ifdefined \excludeexample
    \excludeversion{example}
\fi

\newcommand{\remove}[1]{}
\newcommand{\Draft}[1]{\ifdefined\IsDraft\texttt{ #1} \fi}

\ifdefined\IsDraft
    \newcommand{\authnote}[2]{{\bf [{\color{red} #1's Note:} {\color{blue} #2}]}}
\else
    \newcommand{\authnote}[2]{}
\fi

\newcommand{\sdotfill}{\textcolor[rgb]{0.8,0.8,0.8}{\dotfill}} 

\newenvironment{algorithm}{\begin{algo}}{\vspace{-\topsep}\sdotfill\end{algo}}
\newenvironment{experiment}{\begin{expr}}{\vspace{-\topsep}\sdotfill\end{expr}}

\titleclass{\subsubsubsection}{straight}[\subsection]

\newcounter{subsubsubsection}[subsubsection]
\renewcommand\thesubsubsubsection{\thesubsubsection.\arabic{subsubsubsection}}

\titleformat{\subsubsubsection}
{\normalfont\normalsize\bfseries}{\thesubsubsubsection}{1em}{}
\titlespacing*{\subsubsubsection}
{0pt}{3.25ex plus 1ex minus .2ex}{1.5ex plus .2ex}

\makeatletter
\renewcommand\paragraph{\@startsection{paragraph}{5}{\z@}%
	{3.25ex \@plus1ex \@minus.2ex}%
	{-1em}%
	{\normalfont\normalsize\bfseries}}
\renewcommand\subparagraph{\@startsection{subparagraph}{6}{\parindent}%
	{3.25ex \@plus1ex \@minus .2ex}%
	{-1em}%
	{\normalfont\normalsize\bfseries}}
\def\toclevel@subsubsubsection{4}
\def\toclevel@paragraph{5}
\def\toclevel@paragraph{6}
\def\l@subsubsubsection{\@dottedtocline{4}{7em}{4em}}
\def\l@paragraph{\@dottedtocline{5}{10em}{5em}}
\def\l@subparagraph{\@dottedtocline{6}{14em}{6em}}
\makeatother

\crefname{subsubsubsection}{Section}{Sections}

\newcommand{\Ensuremath}[1]{\ensuremath{#1}\xspace}
\newcommand{\MathAlg}[1]{\mathsf{#1}}

\newcommand{\MathAlgX}[1]{\Ensuremath{\MathAlg{#1}}}

\newcommand{\resp}{resp.,\xspace}
\newcommand{\ie}  {i.e.,\xspace}
\newcommand{\eg}  {e.g.,\xspace}
\newcommand{\whp}  {with high probability\xspace}
\newcommand{\wrt} {with respect to\xspace}
\newcommand{\wlg} {without loss of generality\xspace}

\newcommand{\abs}[1]{\left\lvert #1 \right\rvert}
\newcommand{\ceil}[1]{\left\lceil #1 \right\rceil}

\newcommand{\set}[1]{\ens{#1}}

\newcommand{\paren}[1]{\left(#1\right)}

\let\originalleft\left
\let\originalright\right
\renewcommand{\left}{\mathopen{}\mathclose\bgroup\originalleft}
\renewcommand{\right}{\aftergroup\egroup\originalright}

\newcommand{\eqdef}{:=}

\newcommand{\R}{{\mathbb R}}
\newcommand{\N}{{\mathbb{N}}}
\newcommand{\Z}{{\mathbb Z}}
\newcommand{\F}{{\cal F}}

\newcommand{\zo}{\{0,1\}}
\newcommand{\zn}{{\zo^n}}

\newcommand{\condition}{\;\ifnum\currentgrouptype=16 \middle\fi|\;}

\newcommand{\eps}{\varepsilon}

\newcommand{\la}{\gets}

\newcommand{\poly}{\operatorname{poly}}

\newcommand{\Exp}{\Ex}

\newcommand{\negl}{\operatorname{neg}}

\newcommand{\Supp}{\operatorname{Supp}}

\newcommand{\class}[1]{\mathrm{#1}}
\newcommand{\BPP}{\class{BPP}}
\newcommand{\NP}{\class{NP}}

\newcommand{\PSPACE}{\class{PSPACE}}

\renewcommand{\cref}{\Cref}

\newtheorem{theorem}{Theorem}[section]

\newaliascnt{lemma}{theorem}
\newtheorem{lemma}[lemma]{Lemma}
\aliascntresetthe{lemma}
\crefname{lemma}{Lemma}{Lemmas}

\newaliascnt{claim}{theorem}
\newtheorem{claim}[claim]{Claim}
\aliascntresetthe{claim}
\crefname{claim}{Claim}{Claims}

\newaliascnt{corollary}{theorem}

\aliascntresetthe{corollary}
\crefname{corollary}{Corollary}{Corollaries}

\newaliascnt{construction}{theorem}

\aliascntresetthe{construction}
\crefname{construction}{Construction}{Constructions}

\newaliascnt{fact}{theorem}
\newtheorem{fact}[fact]{Fact}
\aliascntresetthe{fact}
\crefname{fact}{Fact}{Facts}

\newaliascnt{proposition}{theorem}
\newtheorem{proposition}[proposition]{Proposition}
\aliascntresetthe{proposition}
\crefname{proposition}{Proposition}{Propositions}

\newaliascnt{conjecture}{theorem}

\aliascntresetthe{conjecture}
\crefname{conjecture}{Conjecture}{Conjectures}

\newaliascnt{definition}{theorem}
\newtheorem{definition}[definition]{Definition}
\aliascntresetthe{definition}
\crefname{definition}{Definition}{Definitions}

\newaliascnt{notation}{theorem}
\newtheorem{notation}[notation]{Notation}
\aliascntresetthe{notation}
\crefname{notation}{Notation}{Notation}

\newaliascnt{remark}{theorem}
\newtheorem{remark}[remark]{Remark}
\aliascntresetthe{remark}
\crefname{remark}{Remark}{Remarks}

\newaliascnt{example}{theorem}
\ifdefined\excludeexample
\else
   \newenvironment{example}{\refstepcounter{example}\par\setlength{\topskip}{1em}
  \begin{list}{}{%
    \setlength{\leftmargin}{3em}%
  }\item\relax \rule{\linewidth}{0.4pt} \\ \noindent \textbf{Example~\theexample} \rmfamily}{\par \nopagebreak\rule[2mm]{\linewidth}{0.4pt}\end{list}}
\fi

\aliascntresetthe{example}
\crefname{exmaple}{Example}{Examples}

\newaliascnt{exm}{theorem}

\crefname{equation}{Equation}{Equations}

\newaliascnt{proto}{theorem}

\newtheorem{proto}[proto]{Protocol}

\aliascntresetthe{proto}
\crefname{proto}{protocol}{protocols}

\newaliascnt{algo}{theorem}
\newtheorem{algo}[algo]{Algorithm}
\aliascntresetthe{algo}
\crefname{algo}{algorithm}{algorithms}

\newaliascnt{expr}{theorem}
\newtheorem{expr}[expr]{Experiment}
\aliascntresetthe{expr}
\crefname{experiment}{experiment}{experiments}

\def\FullBox{$\Box$}
\def\qed{\ifmmode\qquad\FullBox\else{\unskip\nobreak\hfil
\penalty50\hskip1em\null\nobreak\hfil\FullBox
\parfillskip=0pt\finalhyphendemerits=0\endgraf}\fi}

\def\qedsketch{\ifmmode\Box\else{\unskip\nobreak\hfil
\penalty50\hskip1em\null\nobreak\hfil$\Box$
\parfillskip=0pt\finalhyphendemerits=0\endgraf}\fi}

\newenvironment{proofsketch}{\begin{trivlist} \item {\it
Proof sketch.}} {\qed\end{trivlist}}

\newcommand{\Tau}{\mathrm{T}}           

\newcommand{\eex}[2]{\Ex_{#1}\left[#2\right]}

\newcommand{\Ex}{{\mathrm E}}
\renewcommand{\Pr}{{\mathrm {Pr}}}
\newcommand{\pr}[1]{\Pr\left[#1\right]}
\newcommand{\ppr}[2]{\Pr_{#1}\left[#2\right]}

\newcommand{\Ac}{\mathsf{A}}

\newcommand{\Bc}{\mathsf{B}}
\newcommand{\Cc}{\mathsf{C}}
\newcommand{\sfA}{\mathsf{A}}
\newcommand{\sfB}{\mathsf{B}}

\newcommand{\sfH}{\mathsf{H}}

\newcommand{\cp}{\operatorname{CP}}

\newcommand{\concat}{\circ}

\newcommand{\ens}[1]{\left\{#1\right\}}
\newcommand{\size}[1]{\left|#1\right|}

\newcommand{\out}{\operatorname{out}}

\newcommand{\trans}{{\operatorname{trans}}}

\newcommand{\Uni}{{\mathord{\mathcal{U}}}}

\newcommand{\prob}[1]{\mathsf{\textsc{#1}}}

\newcommand{\SD}{\prob{SD}}
\newcommand{\SDP}[2]{{\SD\paren{#1,#2}}}

\newcommand{\cL}{{\cal{L}}}

\newcommand{\I}{\mathcal{I}}
\newcommand{\J}{\mathcal{J}}

\newcommand{\ppt}{{\sc ppt}\xspace}
\newcommand{\pptm}{{\sc pptm}\xspace}

\newcommand{\cA}{\mathcal{A}}
\newcommand{\cB}{\mathcal{B}}
\newcommand{\cS}{\mathcal{S}}
\newcommand{\cU}{\mathcal{U}}
\newcommand{\cV}{\mathcal{V}}

\newcommand{\cT}{\mathcal{T}}

\newcommand{\cF}{\mathcal{F}}

\newcommand{\cE}{\mathcal{E}}
\newcommand{\cC}{\mathcal{C}}
\newcommand{\cN}{\mathcal{N}}

\newcommand{\cs}{{\cal{S}}}

\newcommand{\msg}{{\mathsf{msg}}}

\newcommand{\round}{{\operatorname{round}}}

\newcommand{\conc}{\circ}

\newcommand{\MaxP}[2]{\ifx&#2&
  {\rm MaxP}_{#1}
\else
  {\rm MaxP}_{#1}\left(#2\right)
\fi}
\newcommand{\MaxIn}[2]{\ifx&#2&
  {\rm MaxIn}_{#1}
\else
  {\rm MaxIn}_{#1}\left(#2\right)
\fi}

\tikzset{
  treenode/.style = {align=center, text centered,
    font=\sffamily},
  A-node/.style = {treenode, circle, white, font=\sffamily\bfseries, draw=black,
    fill=darkgray, inner sep=1pt},
  B-node/.style = {treenode, circle, white, font=\sffamily\bfseries, draw=black,
    fill=gray, inner sep=1pt},
  leaf/.style = {treenode, rectangle, draw=black,inner sep=1pt},
  leaf-red/.style = {treenode, rectangle, draw=red,inner sep=1pt},
  A-Lnode/.style = {treenode, circle, white, font=\sffamily\bfseries, draw=black,
    fill=darkgray, inner sep=1pt},
  B-Lnode/.style = {treenode, circle, white, font=\sffamily\bfseries, draw=black,
    fill=gray, inner sep=1pt,minimum size=0.8cm},
}

\newcommand{\vect}[1]{{ \boldsymbol #1}} 

\newcommand{\Tableofcontents}{
	\ifdefined\IsLLNCS \else

	\thispagestyle{empty}
	\pagenumbering{gobble}
	\clearpage
	
	\setcounter{tocdepth}{2}
	
	\Draft{
		\setcounter{tocdepth}{4}
	}
	\tableofcontents
	\thispagestyle{empty}
	\clearpage
	\pagenumbering{arabic}
	\fi
}

\newcommand{\bigcupdot}{\bigcup}
\newcommand{\inote}[1]{\authnote{Iftach}{#1}}
\def\Inote{\inote}
\newcommand{\Bnote}[1]{\authnote{Itay}{#1}}

\newcommand{\RandomCont}{\MathAlgX{\MathAlg{BiasedCont}}\xspace}
\newcommand{\DMS}[1]{{\sf DMS}\paren{#1}}

\newcommand{\cpi}{\Pi}
\newcommand{\cpit}{{\widetilde{\cpi}}}
\newcommand{\icpi}[2]{{\cpi^{#2}_{#1}}}
\newcommand{\icpih}[2]{{\widehat{\cpi}^{#2}_{#1}}}
\newcommand{\idsm}[3]{{\submeas^{#2,#3}_{#1}}}
\newcommand{\imu}[3]{{\widehat{\mu}^{#2,#3}_{#1}}}
\newcommand{\icpib}[3]{{\paren{\widehat{\cpi_{#3}}}^{#2}_{#1}}}
\newcommand{\DMEx}[3]{{\mu^{{#1},{#2}}_{#3}}}
\newcommand{\DMExRes}[2]{\xi_{#1}^{#2}}

\newcommand{\RandomContXD}[2]{\RandomCont^{#1,#2}}

\newcommand{\HonA}{\sfA}
\newcommand{\HonB}{\sfB}
\newcommand{\HonAt}{{\widetilde{\HonA}}}
\newcommand{\HonBt}{{\widetilde{\HonB}}}

\newcommand{\someAadv}{{\HonA^\ast}}
\newcommand{\someBadv}{{\HonB^\ast}}

\newcommand{\HonC}{{\sf C}}
\newcommand{\HonD}{{\sf D}}

\newcommand{\rcc}[2]{{#1^{(#2)}}}

\newcommand{\rcA}[1]{\rcc{\Ac}{#1}}
\newcommand{\rccc}[3]{{#1^{(#2),#3}}}
\newcommand{\rcccA}[2]{\rccc{\Ac}{#1}{#2}}
\newcommand{\rcccB}[2]{\rccc{\Bc}{#1}{#2}}

\newcommand{\rcB}[1]{\rcc{\Bc}{#1}}
\newcommand{\rcC}[1]{\rcc{\Cc}{#1}}

\newcommand{\rcPro}[3]{\rcc{{#1}_{#2}}{#3}}
\newcommand{\rcAP}[2]{\rcPro{\HonA}{#1}{#2}}
\newcommand{\rcCP}[2]{\rcPro{\HonC}{#1}{#2}}
\newcommand{\rcBP}[2]{\rcPro{\HonB}{#1}{#2}}

\newcommand{\HonAHonB}{{(\HonA,\HonB)}}

\newcommand{\ProArc}[1]{{\rcA{#1},\HonB}}
\newcommand{\ProArcP}[1]{{\paren{\rcA{#1},\HonB}}}
\newcommand{\ProArcb}[2]{{\rcc{\HonA_{\cpi_#2}}{#1},\HonB_{\cpi_#2}}}

\newcommand{\CondPro}[2]{{#1|\neg #2}}
\newcommand{\CondProP}[2]{{#1|\neg \left(#2\right)}}

\newcommand{\Tree}{{\cT}}

\newcommand{\Vertices}{{\cV}}
\newcommand{\Edges}{{\cE}}
\newcommand{\Leaves}{{\mathcal{L}}}
\newcommand{\Root}{\mathsf{root}}

\newcommand{\Val}{{\mathsf{val}}}
\newcommand{\rnd}{{m}}
\newcommand{\dmsi}{{z}}

\newcommand{\LeafValue}{{\chi}}
\newcommand{\Color}{{\chi}}
\newcommand{\EdgeDist}{e}
\newcommand{\VerticesDist}{\MathAlg{v}}

\newcommand{\LDist}[1]{{\left\langle#1\right\rangle}}

\newcommand{\OPT}{\mathsf{OPT}}

\newcommand{\BEST}[2]{\OPT_{#1}\left(#2\right)}

\newcommand{\BESTB}{{\BEST{\HonB}{\HonA,\HonB}}}

\newcommand{\BestB}[1]{{\BEST{\HonB}{#1}}}
\newcommand{\BestA}[1]{{\BEST{\HonA}{#1}}}

\newcommand{\Meas}{{M}}
\newcommand{\FMeas}{{L}}
\newcommand{\submeas}{{\widehat{\Meas}}}

\newcommand{\DomMeas}[2]{{\Meas_{#1}^{#2}}}
\newcommand{\ADomMeas}[1]{{\DomMeas{#1}{\HonA}}}
\newcommand{\BDomMeas}[1]{{\DomMeas{#1}{\HonB}}}
\newcommand{\Restrict}[2]{{\left(#1 \right)}_{#2}}

\newcommand{\FinalMeas}[3]{{\widehat{\FMeas}^{#2,#3}_{#1}}}
\newcommand{\CombineMeas}[3]{{\FMeas^{#2,#3}_{#1}}}

\newcommand{\Aadv}{\someAadv}
\newcommand{\Badv}{\someBadv}

\newcommand{\domain}{{\cal{D}}}
\newcommand{\range}{{\cal{R}}}
\newcommand{\Inv}{\ensuremath{\MathAlg{Inv}}\xspace}

\newcommand{\EmptyString}{{\lambda}}

\newcommand{\unbal}[4]{{\Unbal_{#1}^{#3}}}
\newcommand{\UnBal}[2]{{\Unbal_{#1}^{#2}}}

\newcommand{\Smalls}{\mathcal{S}\mathsf{mall}}

\newcommand{\Larges}{\mathcal{L}\mathsf{arge}}

\newcommand{\low}[2]{{\Smalls_{#1}^{#2}}}
\newcommand{\high}[2]{{\Larges_{#1}^{#2}}}

\newcommand{\veta}[1]{{\vetas_{\bf #1}}}
\newcommand{\vetas}{{{\boldsymbol \eta}}}

\newcommand{\phiBal}{\phiBalE{}}
\newcommand{\phiBalE}[1]{{\phi^{\mathsf{Bal}}_{#1}}}

\newcommand{\phiPru}{\phiPruE{}}
\newcommand{\phiPruE}[1]{{\phi^{\mathsf{It}}_{#1}}}

\newcommand{\Smaller}[2]{\mathsf{Smaller}_{#1}\left(#2\right)}

\newcommand{\desc}{\mathsf{desc}}
\newcommand{\descP}[1]{{\desc\paren{#1}}}

\newcommand{\propDesc}{\overline{\desc}}
\newcommand{\propDescP}[1]{\propDesc\paren{#1}}
\newcommand{\ctrl}[2]{\mathcal{C}\mathsf{trl}_{#1}^{#2}}
\newcommand{\ctrls}{\mathsf{cntrl}}

\newcommand{\frnt}[1]{\mathsf{frnt}\paren{#1}}
\newcommand{\pru}[2]{{#1}^{[#2]}}
\newcommand{\pruparty}[3]{{#1}^{[#2]}_{#3}}
\newcommand{\RC}[2]{#1^{(#2)}}

\newcommand{\pruAttack}[2]{\pruAttGen{\HonA}{#1}{#2}}
\newcommand{\finalAttack}[2]{\pruAttGen{\widehat{\HonA}}{#1}{#2}}
\newcommand{\pruAttGen}[3]{{#1}^{(#2)}_{#3}}

\newcommand{\nCc}{\overline{\HonC}}

\newcommand{\pred}{\mathsf{pred}}
\newcommand{\succe}{\mathsf{succ}}
\newcommand{\SecParam}{n}

\newcommand{\Unbal}{{\mathcal{U}\mathsf{nBal}}}

\newcommand{\rHonContStateless}{{\MathAlg{HonContSL}}}
\newcommand{\HonCont}{{\MathAlg{HonCont}}}

\newcommand{\ConsisDist}{\mathrm{Consis}}

\title{Coin Flipping of \emph{Any} Constant Bias Implies One-Way
  Functions\footnote{This is the final draft of this paper. The full
    version was published in the Journal of the ACM
    \cite{BermanHT18-jacm}. An extended abstract of this work appeared in the
    proceedings of STOC 2014 \cite{BermanHT14}.}  \Draft{\\{\small \sc Working
      Draft: Please Do Not Distribute}} \ifdefined\excludeexample \\{\small
    \color{purple} Examples Are Excluded} \fi }

\author{Itay Berman\thanks{MIT Computer Science and Artificial Intelligence
    Laboratory. E-mail: \texttt{itayberm@mit.edu}. Most of this work was done
    while the author was in the School of Computer Science, Tel Aviv
    University. Research supported in part by NSF Grants CNS-1413920 and
    CNS-1350619, and by the Defense Advanced Research Projects Agency (DARPA)
    and the U.S. Army Research Office under contracts W911NF-15-C-0226 and
    W911NF-15-C-0236.}~\textsuperscript{\S} \and Iftach Haitner\thanks{School of
    Computer Science, Tel Aviv University. E-mail:
    \texttt{iftachh@cs.tau.ac.il}.}~\textsuperscript{\S} \and Aris
  Tentes\thanks{E-mail: \texttt{tentes@cims.nyu.edu}. Most of this work was done
    while the author was in the Department of Computer Science, New York
    University.}  \thanks{Research supported by ISF grant 1076/11, the Israeli
    Centers of Research Excellence (I-CORE) program (Center No. 4/11), US-Israel
    BSF grant 2010196 and Check Point Institute for Information Security. Preparation of the journal version was supported by  by ERC starting grant 638121.}}

\begin{document}
\sloppy
\maketitle
\begin{abstract}
  We show that the existence of a coin-flipping protocol safe against \emph{any}
  non-trivial constant bias (\eg $.499$) implies the existence of one-way
  functions. This improves upon a recent result of \citeauthor{HaitnerOmri11}
  [FOCS '11], who proved this implication for protocols with bias
  $\frac{\sqrt2 -1}2 - o(1) \approx .207$. Unlike the result of
  \citeauthor{HaitnerOmri11}, our result also holds for \emph{weak}
  coin-flipping protocols.
\end{abstract}

\Tableofcontents


\newcommand{\transf}{\mathsf{trans}}

\section{Introduction}\label{section:intro}
A central focus of modern cryptography has been to investigate the weakest possible assumptions under which various cryptographic primitives exist. This direction of research has been quite fruitful,
and minimal assumptions are known for a wide variety of primitives. In particular, it has been shown that one-way functions (\ie easy to compute but hard to invert) imply pseudorandom
generators, pseudorandom functions, symmetric-key encryption/message authentication, commitment schemes, and digital signatures
\cite{GoldreichGoMi85,GoldreichGoMi86,HastadImLeLu99,HaitnerNgOnReVa09,Naor91,NaorYu89,GoldreichL89,Rompel90}, where one-way functions were also shown to be implied by each of these primitives
\cite{ImpagliazzoLu89}.

An important exception to the above successful characterization is that of coin-flipping (-tossing) protocols. A coin-flipping protocol \cite{Blum81} allows the honest parties to jointly flip an unbiased coin, where even a cheating (efficient) party cannot bias the outcome of the protocol by very much. Specifically, a coin-flipping protocol is $\delta$-bias if no efficient cheating party can make the common output to be 1, or to be 0, with probability greater than $\frac12 + \delta$. While one-way functions are known to imply negligible-bias coin-flipping protocols \cite{Blum81,Naor91,HastadImLeLu99}, the other direction is less clear. \citet{ImpagliazzoLu89} showed that $\Theta\paren{1/\sqrt{\rnd}}$-bias coin-flipping protocols imply one-way functions, where $\rnd$ is the number of rounds in the protocol.\footnote{In \cite{ImpagliazzoLu89}, only $\negl(\rnd)$-bias was stated. Proving the same implication for $\Theta\paren{1/\sqrt{\rnd}}$-bias follows from the proof outlined in \cite{ImpagliazzoLu89} and the result by \citet{CleveI93}.}
Recently, \citet*{Maji10} extended the above for $(\frac12 - 1/\poly(n))$-bias \emph{constant-round} protocols, where $n$ is the security parameter. More recently, \citet{HaitnerOmri11} showed that the above implication holds for $(\frac{\sqrt{2} -1}2 -o(1) \approx 0.207)$-bias coin-flipping protocols (of arbitrary round complexity). No such implications were known for any other choice of parameters, and in particular for protocols with bias greater than $\frac{\sqrt{2} -1}2$ with super-constant round complexity.

\subsection{Our Result}\label{sec:intro:ourResult}
In this work, we make progress towards answering the question of whether coin-flipping protocols also imply one-way functions. We show that (even weak) coin-flipping protocols, safe against any non-trivial bias (\eg 0.4999), do in fact imply such functions.
We note that unlike \cite{HaitnerOmri11}, but like \cite{ImpagliazzoLu89,Maji10}, our result also applies to the so-called \emph{weak coin-flipping protocols} (see \cref{sec:CFprotocols} for the formal definition of strong and weak coin-flipping protocols).
Specifically, we prove the following theorem.
\begin{theorem}[informal]\label{thm:mainInf}
For any $c > 0$, the existence of a $(\frac 12 - c)$-bias coin-flipping protocol (of any round complexity) implies the existence of one-way functions.
\end{theorem}

Note that $\frac12$-bias coin-flipping protocol requires no assumption (\ie one party flips a coin and announces the result to the other party). So our result is tight as long as constant biases (\ie independent of the security parameter) are involved.


To prove \cref{thm:mainInf}, we observe a connection between the success probability of the best (valid) attacks in a two-party game (\eg tic-tac-toe) and the success of the biased-continuation attack of \cite{HaitnerOmri11} in winning this game (see more in \cref{sec:intro:Technique}). The implications of this interesting connection seem to extend beyond the question at the focus of this paper.

\subsection{Related Results}\label{sec:intro:relatedResult}
As mentioned above, \citet{ImpagliazzoLu89} showed that negligible-bias coin-flipping protocols imply one-way functions. \citet{Maji10} proved the same for $(\frac12 - o(1))$-bias yet constant-round protocols. Finally, \citet{HaitnerOmri11} showed that the above implication holds for $\frac{\sqrt{2} -1}2 -o(1) \approx 0.207)$-bias (strong) coin-flipping protocols (of arbitrary round complexity). Results of weaker complexity implications are also known.

\citet{Zachos86} has shown that non-trivial (\ie ($\frac12 -o(1)$)-bias), constant-round coin-flipping protocols imply that $\NP \nsubseteq \BPP$, where \citet{Maji10} proved the same implication for $(\frac14 - o(1))$-bias coin-flipping protocols of arbitrary round complexity. Finally, it is well known that the existence of non-trivial coin-flipping protocols implies that $\PSPACE \nsubseteq \BPP$. Apart from \cite{HaitnerOmri11}, all the above results extend to weak coin-flipping protocols. See \cref{fig:summeryOfResult} for a summary.

\begin{table}[h]
\begin{center}
\begin{tabular}{|c|c|c|}
 \hline
 \textit{Implication} & \textit{Protocol type} & \textit{Paper} \\
 \hline
 Existence of OWFs & $(\frac12 - c)$-bias, for some $c> 0$& \textbf{This work}\\
 \hline
 Existence of OWFs & $(\frac{\sqrt{2} -1}2 -o(1))$-bias & \citet{HaitnerOmri11}\footnotemark\\
 \hline

 Existence of OWFs & $(\frac12 - o(1))$-bias, \emph{constant round} & \citet{Maji10} \\
\hline

Existence of OWFs & Negligible bias & \citet{ImpagliazzoLu89} \\
 \hline \hline

 $\NP \nsubseteq \BPP$ & $(\frac14 - o(1))$-bias & \citet{Maji10} \\
 \hline
 $\NP \nsubseteq \BPP$ & $(\frac12 - o(1))$-bias, \emph{constant round} &\citet{Zachos86} \\
 \hline

 $\PSPACE \nsubseteq \BPP$ & Non-trivial & Common knowledge\\
\hline
\end{tabular}
\end{center}
\caption{\label{fig:summeryOfResult} Results summary.}
\end{table}
\footnotetext{Only holds for \emph{strong} coin-flipping protocols.}

\emph{Information theoretic} coin-flipping protocols (\ie whose security holds against all-powerful attackers) were shown to exist in the quantum world; \citet{Moc07} presented an $\eps$-bias
quantum weak coin-flipping protocol for any $\eps>0$. \citet{ChaillouxKerenidis09} presented a $\left(\frac{\sqrt{2}-1}{2}-\eps\right)$-bias quantum strong coin-flipping protocol for any $\eps>0$ (this bias was shown in \cite{Kit03} to be tight). A key step in \cite{ChaillouxKerenidis09} is a reduction from strong to weak coin-flipping protocols, which holds also in the classical world.

A related line of work considers \emph{fair} coin-flipping protocols. In this setting the honest party is required to always output a bit, whatever the other party does. In particular, a cheating
party might bias the output coin just by aborting. We know that one-way functions imply fair $(1/\sqrt \rnd)$-bias coin-flipping protocols \cite{ABCGM85,Cleve86}, where $\rnd$ is the round complexity of
the protocol, and this quantity is known to be tight for $o(\rnd/\log \rnd)$-round protocols with fully black-box reductions \cite{Dachman11}. Oblivious transfer, on the other hand, implies fair $1/\rnd$-bias protocols \cite{MoranNS09,BeimelOO10} (this bias was shown in \cite{Cleve86} to be tight).

\subsection{Our Techniques}\label{sec:intro:Technique}
The following is a rather elaborate, high-level description of the ideas underlying our proof.

That the existence of a given (cryptographic) primitive implies the existence of one-way functions is typically proven by looking at the \textit{primitive core function} --- an efficiently computable function (not necessarily unique) whose inversion on uniformly chosen outputs implies breaking the security of the primitive.\footnote{For the sake of this informal discussion, inverting a function on a given value means returning a \emph{uniformly} chosen preimage of this value.} For private-key encryption, for instance, a possible core function is the mapping from the inputs of the encryption algorithm (\ie message, secret key, and randomness) into the ciphertexts. Assuming that one has defined such a core function for a given primitive, then, by definition, this function should be one-way. So it all boils down to finding, or proving the existence of, such a core function for the primitive under consideration. For a \emph{non-interactive} primitive, finding such a core function is typically easy. In contrast, for an \emph{interactive} primitive, finding such a core function is, at least in many settings, a much more involved task. The reason is that in order to break an interactive primitive, the attacker typically needs, for a given function, pre-images for many different outputs, where these outputs are chosen \emph{adaptively} by the attacker, after seeing the pre-images to the previous outputs. As a result, it is challenging to find a single function, or even finitely many functions, whose output distributions (on uniformly chosen input) match the distribution of the pre-images the attacker needs.\footnote{If the attacker makes a \emph{constant} number of queries, one can overcome the above difficulty by defining a set of core functions $f_1,\ldots,f_k$, where $f_1$ is the function defined by the primitive, $f_2$ is the function defined by the attacker after making the first inversion call, and so on. Since the evaluation time of $f_{i+1}$ is polynomial in the evaluation time of $f_i$ (since evaluating $f_{i+1}$ requires a call to an inverter of $f_i$), this approach fails miserably for attackers of super-constant query complexity.}

The only plausible  candidate to serve as a core function of a coin-flipping protocol would seem to be its \textit{transcript function}: the function that maps the parties' randomness into the resulting protocol transcript (\ie the transcript produced by executing the protocol with this randomness). In order to bias the output of an $\rnd$-round coin-flipping protocol by more than $O(\frac1{\sqrt \rnd})$, a super-constant number of adaptive inversions of the transcript function seems necessary. Yet we managed to prove that the transcript function is a core function of  any (constant-bias) coin-flipping protocol. This is done by designing an adaptive attacker for any such protocol whose query distribution is ``not too far" from the output distribution of the transcript function (when invoked on uniform inputs). Since our attacker, described below, is not only adaptive, but also defined in a recursive manner, proving that it possesses the aforementioned property was one of the major challenges we faced.

In what follows, we give a high-level overview of our attacker that ignores computational issues (\ie assumes it has a perfect inverter for any function). We then explain how to adjust this attacker to work with the inverter of the protocol's transcript function.

\subsubsection{Optimal Valid Attacks and The Biased-Continuation Attack}
The crux of our approach lies in an interesting connection between the optimal attack on a coin-flipping protocol and the more feasible, \textit{recursive biased-continuation} attack. The latter attack recursively applies the biased-continuation attack used by \citet{HaitnerOmri11} to achieve their constant-bias attack (called there, the \textit{random-continuation} attack) and is the basis of our efficient attack (assuming one-way functions do not exist) on coin-flipping protocols. The results outlining the aforementioned connection, informally stated in this section and formally stated and proven in \cref{sec:IdealAttacker}, hold for any two-player full information game with binary common outcome.

Let $\cpi = \HonAHonB$ be a coin-flipping protocol (\ie the common output of the honest parties is a uniformly chosen bit). In this discussion we restrict ourselves to analyzing attacks that, when carried out by the left-hand party, \ie $\HonA$, are used to bias the outcome towards one, and when carried out by the right-hand party, \ie $\HonB$, are used to bias the outcome towards zero. Analogous statements hold for opposite attacks (\ie attacks carried out by $\HonA$ and used to bias towards zero, and attacks carried out by $\HonB$ and used to bias towards one). The optimal valid attacker $\cA$ carries out the \emph{best} attack $\Ac$ can employ (using unbounded power) to bias the protocol towards \emph{one}, while sending \emph{valid} messages --- ones that could have been sent by the honest party. The optimal valid attacker $\cB$, carrying out the best attack $\HonB$ can employ to bias the protocol towards \emph{zero}, is analogously defined. Since, \wlg, the optimal valid attackers are deterministic, the expected outcome of $(\cA,\cB)$ is either zero or one. As a first step, we give a lower bound on the success probability of the recursive biased-continuation attack carried out by the party winning the aforementioned game. As this lower bound might not be sufficient for our goal (it might be less than constant) --- and this is a crucial point in the description below --- our analysis takes additional steps to give an arbitrarily-close-to-one lower bound on the success probability of the recursive biased-continuation attack carried out by \emph{some} party, which may or may not be the same party winning the aforementioned game.\footnote{That the identity of the winner in $(\cA,\cB)$ cannot be determined by the recursive biased-continuation attack is crucial. Since we show that the latter attack can be efficiently approximated assuming one-way functions do not exist, the consequences of revealing this identity would be profound. It would mean that we can estimate the outcome of the optimal attack (which is implemented in $\PSPACE$) using only the assumption that one-way functions do not exist.}

Assume that $\cA$ is the winning party when playing against $\cB$. Since $\cA$ sends only valid messages, it follows that the expected outcome of $(\Ac,\cB)$, \ie honest $\Ac$ against the optimal attacker for $\Bc$, is larger than zero (since $\Ac$ might send the optimal messages ``by mistake''). Let $\BestA{\cpi}$ be the expected outcome of the protocol $(\cA,\HonB)$ and let $\BestB{\cpi}$  be $1$ minus the expected outcome of the protocol $(\HonA,\cB)$. The above observation yields that $\BestA{\cpi}=1$, while $\BestB{\cpi}= 1- \alpha<1$. This gives rise to the following question: \emph{what does give $\cA$ an advantage over $\cB$?}

We show that if $\BestB{\cpi}= 1- \alpha$, then there exists a set $\cs^\Ac$ of 1-transcripts, full transcripts in which the parties' common output is $1$,\footnote{Throughout, we assume \wlg that the protocol's transcript determines the common output of the parties.} that is $\alpha$-dense (meaning that the chance that a random full transcript of the protocol is in the set is $\alpha$) and is ``dominated by $\HonA$''. The $\HonA$-dominated set has an important property --- its density is ``immune'' to any action $\HonB$ might take, even if $\HonB$ is employing its optimal attack; specifically, the following holds:
\begin{align}\label{eq:intro1}
\ppr{\LDist{\HonA,\HonB}}{\cs^\Ac} = \ppr{\LDist{\HonA,\cB}}{\cs^\Ac} = \alpha,
\end{align}
where $\LDist{\cpi'}$ samples a random full transcript of protocol $\cpi'$. It is easy to see that the above holds if $\HonA$ controls the root of the tree and has a $1$-transcript as a direct descendant; see \cref{fig:simpleTree} for a concrete example. The proof of the general case can be found in \cref{sec:IdealAttacker}. Since the $\HonA$-dominated set is $\HonB$-immune, a possible attack for $\cA$ is to go towards this set. Hence, what seems like a feasible adversarial attack for $\HonA$ is to mimic $\cA$'s attack by hitting the $\HonA$-dominated set with high probability. It turns out that the biased-continuation attack of \cite{HaitnerOmri11} does exactly that.

The biased-continuation attacker $\rcA{1}$, taking the role of $\HonA$ in $\cpi$ and trying to bias the output of $\cpi$ towards one, is defined as follows: given that the partial transcript is $\transf$, algorithm $\rcA{1}$ samples a pair of random coins $(r_\HonA,r_\HonB)$ that is consistent with $\transf$ and leads to a $1$-transcript, and then acts as the honest $\HonA$ on the random coins $r_\HonA$, given the transcript $\transf$. In other words, $\rcA{1}$ takes the first step of a random continuation of $\HonAHonB$ leading to a $1$-transcript. (The attacker $\rcB{1}$, taking the role of $\HonB$ and trying to bias the outcome towards zero, is analogously defined.) \citet{HaitnerOmri11} showed that for any coin-flipping protocol, if either $\HonA$ or $\HonB$ carries out the biased-continuation attack towards one, the outcome of the protocol will be biased towards one by $\frac{\sqrt{2} -1}2$ (when interacting with the honest party).\footnote{They show that the same holds for the analogous attackers carrying out the biased-continuation attack towards zero.} Our basic attack employs the above biased-continuation attack recursively. Specifically, for $i>1$ we consider the attacker $\rcA{i}$ that takes the first step of a random continuation of $(\rcA{i-1},\HonB)$ leading to a $1$-transcript, letting $\rcA{0} \equiv \HonA$. The attacker $\rcB{i}$ is analogously defined. Our analysis takes a different route from that of \cite{HaitnerOmri11}, whose approach is only applicable for handling bias up to $\frac{\sqrt{2} -1}2$ and cannot be applied to weak coin-flipping protocols.\footnote{A key step in the analysis of \citet{HaitnerOmri11} is to consider the ``all-cheating protocol" $(\rcccA{1}{1},\rcccB{1}{1})$, where $\rcccA{1}{1}$ and $\rcccB{1}{1}$ taking the roles of $\HonA$ and $\HonB$ respectively, and they both carry out the biased-continuation attack trying to bias the outcome towards one (as opposed to having the attacker taking the role of $\HonB$ trying to bias the outcome towards zero, as in the discussion so far). Since, and this is easy to verify, the expected outcome of $(\rcccA{1}{1},\rcccB{1}{1})$ is one, using symmetry one can show that the expected outcome of either $(\rcccA{1}{1},\Bc)$ or $(\Ac,\rcccB{1}{1})$ is at least $\frac1{\sqrt{2}}$, yielding a bias of $\frac1{\sqrt{2}} - \frac12$. As mentioned in \cite{HaitnerOmri11}, symmetry cannot be used to prove a bias larger than $\frac1{\sqrt{2}} - \frac12$.}
Instead, we analyze the probability of the biased-continuation attacker to hit the dominated set we introduced above.

Let $\transf$ be a $1$-transcript of $\cpi$ in which all messages are sent by $\HonA$. Since $\rcA{1}$ picks a random $1$-transcript, and $\HonB$ cannot force $\rcA{1}$ to diverge from this transcript, the probability to produce $\transf$ under an execution of $(\rcA{1},\HonB)$ is \emph{doubled} \wrt this probability under an execution of $(\HonA,\HonB)$ (assuming the expected outcome of $(\HonA,\HonB)$ is $1/2$). The above property, that $\HonB$ cannot force $\rcA{1}$ to diverge from a transcript, is in fact the $\HonB$-immune property of the $\HonA$-dominated set. A key step we take is to generalize the above argument to show that for the $\alpha$-dense $\HonA$-dominated set $\cs^\Ac$ (which exists assuming that $\BestB{\cpi}= 1- \alpha <1$), it holds that:
\begin{align}\label{eq:intro2}
\ppr{\LDist{\rcA{1},\HonB}}{\cs^\Ac}\geq \frac{\alpha}{ \Val(\cpi)},
\end{align}
where $\Val(\cpi')$ is the expected outcome of $\cpi'$. Namely, in $(\rcA{1},\HonB)$ the probability of hitting the set $\cS^\HonA$ of $1$-transcripts is larger by a factor of at least $\frac{1}{\Val(\cpi)}$ than the probability of hitting this set in the original protocol $\cpi$. Again, it is easy to see that the above holds if $\HonA$ controls the root of the tree and has a $1$-transcript as a direct descendant; see \cref{fig:simpleTree} for a concrete example. The proof of the general case can be found in \cref{sec:IdealAttacker}.

Consider now the protocol $(\rcA{1},\HonB)$. In this protocol, the probability of hitting the set $\cs^\Ac$ is at least $\frac{\alpha}{ \Val(\cpi)}$, and clearly the set $\cs^\Ac$ remains $\HonB$-immune. Hence, we can apply \cref{eq:intro2} again, to deduce that
\begin{align}\label{eq:intro3}
\ppr{\LDist{\rcA{2},\Bc}}{\cs^\Ac} = \ppr{\LDist{(\rcA{1})^{(1)},\HonB}}{\cs^\Ac} \geq \frac{\ppr{\LDist{\rcA{1},\Bc}}{\cs^\Ac}}{\Val(\rcA{1},\Bc)}  \geq \frac{\alpha}{ \Val(\cpi) \cdot \Val(\rcA{1},\Bc)}.
\end{align}
Continuing it for $\kappa$ iterations yields that
\begin{align}\label{eq:intro4}
\Val(\rcA{\kappa},\Bc) \geq \ppr{\LDist{\rcA{\kappa},\Bc}}{\cs^\Ac} \geq \frac{\alpha}{\prod_{i=0}^{\kappa-1} \Val(\rcA{i},\Bc)}.
\end{align}
So, modulo some cheating,\footnote{The actual argument is somewhat more complicated than the one given above. To ensure the above argument holds we need to consider measures over the $1$-transcripts (and not sets). In addition, while (the measure variant of) \cref{eq:intro3} is correct, deriving it from \cref{eq:intro2} takes some additional steps.}
 it seems that we are in good shape.
Taking, for example, $\kappa = \log(\frac1\alpha) / \log(\frac1{0.9})$, \cref{eq:intro4} yields that $\Val(\rcA{\kappa},\Bc) > 0.9$. Namely, if we assume that $\cA$ has an advantage over $\cB$, then by recursively applying the biased-continuation attack for $\HonA$ enough times, we arbitrarily bias the expected output of the protocol towards one. Unfortunately, if this advantage (\ie $\alpha=(1-\BestB{\cpi})$) is very small, which is the case in typical examples, the number of recursions required might be linear in the protocol depth (or even larger). Given the recursive nature of the above attack, the running time of the described attacker is \emph{exponential}. To overcome this obstacle, we consider not only the dominated set, but additional sets that are ``close to" being dominated. Informally, we can say that a $1$-transcript belongs to the $\HonA$-dominated set if it can be generated by an execution of $(\cA,\HonB)$. In other words, the probability, over $\HonB$'s coins, that a transcript generated by a random execution of $(\cA,\HonB)$ belongs to the $\HonA$-dominated set is one. We define a set of $1$-transcripts that does not belong to the $\HonA$-dominated set to be ``close to'' $\HonA$-dominated if there is an (unbounded) attacker $\widehat{\cA}$, such that the probability, over $\HonB$'s coins, that a transcript generated by a random execution of $(\widehat{\cA},\HonB)$ belongs to the set is close to one. These sets are formally defined via the notion of conditional protocols, discussed next.

\begin{figure}
\centering
\begin{tikzpicture}[->,>=stealth',level/.style={sibling distance = 2cm, level distance = 1.5cm},every label/.style={draw,fill=none,shape=ellipse,inner sep=1pt}]
\begin{scope}
\node [A-node] (A) {$\HonA$}
 child{ node [leaf] (A1) {$1$} edge from parent node[above left] {$\alpha_1$}}
 child{ node [B-node] (B) {$\HonB$}
  child{ node [leaf] (L00) {$0$} edge from parent node[above left] {$\beta_1$}}
  child{ node [A-node] (A2) {$\HonA$}
   child{ node [leaf] (L011) {$1$} edge from parent node[above left] {$\alpha_2$}}
   child{ node [leaf] (L010) {$0$} edge from parent node[above right] {$1-\alpha_2$}}
   edge from parent node[above right] {$1-\beta_1$}
		}
  edge from parent node[above right] {$1-\alpha_1$}
 }

;
\end{scope}
\end{tikzpicture}
\caption{Coin-flipping protocol $\cpi$. The label of an internal node (\ie partial transcript) denotes the name of the party controlling it (\ie the party that sends the next message given this partial transcript), and that of a leaf (\ie full transcript) denotes its value --- the parties' common output once reaching this leaf. Finally, the label on an edge leaving a node $u$ to node $u'$ denotes the probability that a random execution of $\cpi$ visits $u'$ once in $u$.~\\
Note that $\BestA{\cpi} = 1$ and $\BestB{\cpi} = 1 - \alpha_1$. The $\HonA$-dominated set $\cs^\Ac$ in this case consists of the single $1$-leaf to the left of the root. The conditional protocol $\cpi'$ is the protocol rooted in the node to the right of the root (of $\cpi$), and the $\HonB'$-dominated set $\cs^\Bc$ consists of the single $0$-leaf to the left of the root of $\cpi'$.
}
\label{fig:simpleTree}
\end{figure}
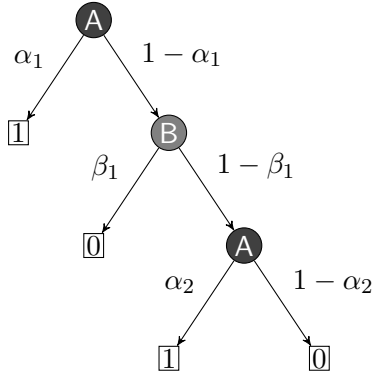

\paragraph{Conditional Protocols.}
Let $\cpi=\HonAHonB$ be a coin-flipping protocol in which there exists an $\HonA$-dominated set $\cs^\Ac$ of density $\alpha >0$. Consider the ``conditional" protocol $\cpi' = (\Ac',\Bc')$, resulting from conditioning on not hitting the set $\cs_\Ac$. Namely, the message distribution of $\cpi'$ is that induced by a random execution of $\cpi$ that does not generate transcripts in $\cs_\Ac$.\footnote{More formally, the conditional protocol $\cpi'$ is defined as follows. Let $\transf$ be a partial transcript, and let $p$ be the probability, in $\cpi$, that the message following $\transf$ is $0$. Let $\alpha$ be the probability of generating a transcript in $\cs^\Ac$ for which $\transf$ is a prefix and similarly let $\alpha_0$ be the probability of generating a transcript in $\cs^\Ac$ for which $\transf\concat 0$ is a prefix ($\transf\concat 0$ is the transcript $\transf$ followed by the message $0$).
Then, the probability that the message following $\transf$ is $0$ in $\cpi'$ is $p\cdot(1-\alpha_0)/(1-\alpha)$.} See \cref{fig:simpleTree} for a concrete example. We note that the protocol $\cpi'$ might not be efficiently computable (even if $\cpi$ is), but this does not bother us, since we only use it as a thought experiment.

We have effectively removed all the $1$-transcripts dominated by $\Ac$ (the set $\cs^\Ac$ must contain all such transcripts; otherwise $\BestB{\cpi}$ would be smaller than $1-\alpha$). Thus, the expected outcome of $(\cA',\cB')$ is zero, where $\cA'$ and $\cB'$ are the optimal valid attackers of the parties in the conditional protocol $\cpi'$. Therefore, $\BEST{\Bc'}{\cpi'} = 1$ and $\BEST{\Ac'}{\cpi'} = 1 - \beta <1$. It follows from this crucial observation that there exists a $\HonB'$-dominated $\cS^\HonB$ of density $\beta$, over the $0$-transcripts of $\cpi'$. Applying a similar argument to that used for \cref{eq:intro4} yields that for large enough $\kappa$, the biased-continuation attacker $\RC{\Bc'}{\kappa}$, playing the role of $\Bc'$, succeeds in biasing the outcome of $\cpi'$ toward zero, where $\kappa$ is proportional to $\log(\frac1\beta)$. Moreover, if $\alpha$ is small, the above yields that $\rcB{\kappa}$ does almost equally well in the original protocol $\cpi$. If $\beta$ is also small, we can now consider the conditional protocol $\cpi''$, obtained by conditioning $\cpi'$ on not hitting the $\HonB'$-dominated set, and so on.

By iterating the above process enough times, the $\HonA$-dominated sets cover all the $1$-transcripts, and the $\Bc$-dominated sets cover all the $0$-transcripts.\footnote{When considering measures and not sets, as done in the actual proof, this covering property is not trivial.} Assume that in the above iterated process, the density of the $\Ac$-dominated sets is the first to go beyond $\eps>0$. It can be shown --- and this a key technical contribution of this paper --- that it is almost as good as if the density of the \emph{initial} set $\cs_\Ac$ was $\eps$.\footnote{More accurately, let $\widetilde{\cs}^{\Ac}$ be the union of these $1$-transcript sets and let $\widetilde{\alpha}$ be the density of $\widetilde{\cs}^{\Ac}$ in $\cpi$. Then $\Val(\rcA{\kappa},\Bc) \geq \ppr{\LDist{\rcA{\kappa},\Bc}}{\widetilde{\cs}^\Ac} \geq \frac{\widetilde{\alpha}}{\prod_{i=0}^{\kappa-1} \Val(\rcA{i},\Bc)}$.} We can now apply the above analysis and conclude that for any constant $\eps >0$, there exists a constant $\kappa=\kappa(\eps)$ such that $\Val(\rcA{\kappa},\Bc) > 1- \eps$.\footnote{The assumption that the density of the $\Ac$-dominated sets is the first to go beyond $\eps>0$ is independent of the assumption that $\cA$ wins in the zero-sum game $(\cA,\cB)$. Specifically, the fact that $\rcA{\kappa}$ succeeds in biasing the protocol does not guarantee that $\cA$, which we only know how to implemented in $\PSPACE$, is the winner of $(\cA,\cB)$.}

\subsubsection{Using the Transcript Inverter}\label{intro:sec:tech:real}
We have seen above that for any constant $\eps$, by recursively applying the biased-continuation attack for constantly many times, we get an attack that biases the outcome of the protocol by $\frac12 - \eps$. The next thing is to implement the above attack \emph{efficiently}, under the assumption that one-way functions do not exist. Given a partial transcript $u$ of protocol $\cpi$, we wish to return a uniformly chosen full transcript of $\cpi$ that is consistent with $u$ and the common outcome it induces is one. Biased continuation can be reduced to the task of finding \emph{honest continuation}: returning a uniformly chosen full transcript of $\cpi$ that is consistent with $u$. Assuming honest continuation can be found for the protocol, biased-continuation can also be found by calling the honest continuation many times, until a transcript whose output is one is obtained. The latter can be done efficiently, as long as the value of the partial transcript $u$ --- the expected outcome of the protocol conditioned on $u$, is not too low. (If it is too low, too much time might pass before a full transcript leading to one is obtained.) Ignoring this low value problem, and noting that honest continuation of a protocol can be reduced to inverting the protocol's transcript function, all we need to do to implement $\rcA{i}$ is to invert the transcript functions of the protocols $(\Ac,\Bc),(\rcA{1},\Bc),\ldots,(\rcA{i-1},\Bc)$. Furthermore, noting that the attackers $\rcA{1},\ldots,\rcA{i-1}$ are \emph{stateless}, it suffices to have the ability to invert \emph{only} the transcript function of $(\Ac,\Bc)$.

So attacking a coin-flipping protocol $\cpi$ boils down to inverting the transcript function $f_\cpi$ of $\cpi$, and making sure we are not doing that on low value transcripts. Assuming one-way functions do not exist, there exists an efficient inverter $\Inv$ for $f_\cpi$ that is guaranteed to work well when invoked on random outputs of $f_\cpi$ (\ie when $f_\cpi$ is invoked on the uniform distribution; nothing is guaranteed for distributions far from uniform). By the above discussion, algorithm $\Inv$ implies an efficient approximation of $\rcA{i}$, as long as the partial transcripts attacked by $\rcA{i}$ are neither \emph{low-value} nor \emph{unbalanced} (by low-value transcript we mean that the expected outcome of the protocol conditioned on the transcript is low; by unbalanced transcript we mean that its density \wrt $(\rcA{i},\Bc)$ is not too far from its density \wrt $(\Ac,\Bc)$). Whereas  the authors of  \cite{HaitnerOmri11} proved that the queries of $\rcA{1}$ obey the two conditions with sufficiently high probability, we were unable to prove this (and believe it is untrue) for the queries of $\rcA{i}$, for $i>1$. Thus, we simply cannot argue that $\rcA{i}$ has an efficient approximation, assuming one-way functions do not exist. Fortunately, we managed to prove the above for the ``pruned" variant of $\rcA{i}$, defined below.

\paragraph{Unbalanced and low-value transcripts.}
Before defining our final attacker, we relate the problem of unbalanced transcripts to that of low-value transcripts. We say that a (partial) transcript $u$ is \emph{$\gamma$-unbalanced} if the probability that $u$ is visited \wrt a random execution of $(\rcA{1},\Bc)$ is at least $\gamma$ times larger than \wrt a random execution of $(\Ac,\Bc)$. Furthermore, we say that a (partial) transcript $u$ is \emph{$\delta$-small} if the expected outcome of $(\Ac,\Bc)$, conditioned on visiting $u$, is at most $\delta$.
We prove (a variant of) the following statement. For any $\delta>0$ and $\gamma >1$, there exists $c$ that depends on $\delta$, such that
\begin{align}\label{eq:intro:2}
\ppr{\ell \la \LDist{\rcA{1},\Bc}}{ \ell \mbox{ has a {\sf $\gamma$-unbalanced} prefix but no {\sf $\delta$-small} prefix}}\leq \frac {1}{\gamma^c}.
\end{align}

Namely, as long as $(\rcA{1},\Bc)$ does not visit low-value transcript, it is only at low risk to significantly deviate (in a multiplicative sense) from the distribution induced by $(\Ac,\Bc)$. \cref{eq:intro:2} naturally extends to recursive biased-continuation attacks. It also has an equivalent form for the attacker $\rcB{1}$, trying to bias the protocol towards zero, \wrt $\delta$-high transcripts --- the expected outcome of $\cpi$, conditioned on visiting the transcript, is at least $1-\delta$.

\paragraph{The pruning attacker.}
At last we are ready to define our final attacker. To this end, for protocol $\cpi = (\Ac,\Bc)$ we define its \emph{$\delta$-pruned variant} $\cpi_\delta = (\Ac_\delta,\Bc_\delta)$, where $\delta \in (0,\frac12)$, as follows. As long as the execution does not visit a $\delta$-low or $\delta$-high transcript, the parties act as in $\cpi$. Once a $\delta$-low transcript is visited, only the party $\Bc$ sends messages, and it does so according to the distribution induced by $\cpi$. If a $\delta$-high transcript is visited (and has no $\delta$-low prefix), only the party $\Ac$ sends messages, and again it does so according to the distribution induced by $\cpi$.

Since the transcript distribution induced by $\cpi_\delta$ is the same as of $\cpi$, protocol $\cpi_\delta$ is also a coin-flipping protocol. We also note that $\cpi_\delta$ can be implemented efficiently assuming one-way functions do not exist (simply use the inverter of $\cpi$'s transcript function to estimate the value of a given transcript). Finally, by \cref{eq:intro:2}, $\RC{\Ac_\delta}{i}$ (\ie recursive biased-continuation attacks for $\cpi_\delta$) can be efficiently implemented, since there are \emph{no} low-value transcripts where $\HonA$ needs to send the next message. (Similarly, $\RC{\Bc_\delta}{i}$ can be efficiently implemented since there are no high-value transcripts where $\HonB$ needs to send the next message.)

It follows that for any constant $\eps>0$, there exists constant $\kappa$ such that either the expected outcome of $(\RC{\Ac_\delta}{\kappa},\Bc_\delta)$ is a least $1- \eps$, or the expected outcome of $(\Ac_\delta,\RC{\Bc_\delta}{\kappa})$ is at most $\eps$. Assume for concreteness that it is the former case. We define our pruning attacker $\rcA{\kappa,\delta}$ as follows. When playing against $\Bc$, the attacker $\rcA{\kappa,\delta}$ acts like $\RC{\Ac_\delta}{\kappa}$ would when playing against $\Bc_\delta$. Namely, the attacker pretends that it is in the $\delta$-pruned protocol $\cpi_\delta$. But once a low- or high-value transcript is reached, $\rcA{\kappa,\delta}$ acts \emph{honestly} in the rest of the execution (like $\Ac$ would).

It follows that until a low- or high-value transcript has been reached for the first time, the distribution of $(\rcA{\kappa,\delta},\Bc)$ is the same as that of $(\RC{\Ac_\delta}{\kappa},\Bc_\delta)$. Once a $\delta$-low transcript is reached, the expected outcome of both $(\rcA{\kappa,\delta},\Bc)$ and $(\RC{\Ac_\delta}{\kappa},\Bc_\delta)$ is $\delta$, but when a $\delta$-high transcript is reached, the expected outcome of $(\rcA{\kappa,\delta},\Bc)$ is $(1-\delta)$ (since it plays like $\Ac$ would), where the expected outcome of $(\RC{\Ac_\delta}{\kappa},\Bc_\delta)$ is at most one. All in all, the expected outcome of $(\rcA{\kappa,\delta},\Bc)$ is $\delta$-close to that of $(\RC{\Ac_\delta}{\kappa},\Bc_\delta)$, and thus the expected outcome of $(\rcA{\kappa,\delta},\Bc)$ is at least $1- \eps-\delta$. Since $\eps$ and $\delta$ are arbitrary constants, we have established an efficient attacker to bias the outcome of $\cpi$ by a value that is an arbitrary constant close to one.

\subsection{Open Questions}\label{section:discussion}
\emph{Does the existence of any non-trivial coin-flipping protocol (\ie bias $\frac12 - \frac1{\poly(n)}$) imply the existence of one-way functions?} This is the main question left open. Answering it would fully resolve the computational complexity of coin-flipping protocols.

\subsection*{Paper Organization}
General notations and definitions used throughout the paper are given in \cref{section:Preliminaries}. Our ideal attacker (which has access to a perfect sampler) to bias any coin-flipping protocol is presented and analyzed in \cref{sec:IdealAttacker}, while in \cref{sec:RealAttacker} we show how to modify the above attacker to be useful when the perfect sampler is replaced with a one-way function inverter.

\subsection*{Acknowledgment}
We are very grateful to Hemanta Maji, Yishay Mansour, Eran Omri and Alex Samorodnitsky for useful discussions.


\section{Preliminaries}\label{section:Preliminaries}

\subsection{Notations}\label{sec:notations}
We use lowercase letters for values, uppercase for random variables, uppercase calligraphic letters (\eg $\cU$) to denote sets, boldface for vectors, and uppercase sans-serif (\eg $\HonA$) for algorithms (\ie Turing
Machines). All logarithms considered here are in base two. Let $\N$ denote the set of natural numbers, where $0$ is considered as a natural number, \ie $\N=\set{0,1,2,3,\ldots}$. For $n\in\N$, let $(n)=\set{0,\ldots,n}$ and if $n$ is positive let $[n]=\set{1,\cdots,n}$, where $[0]=\emptyset$. For $a\in \R$ and $b\geq 0$, let $[a\pm b]$ stand for the interval $[a-b,a+b]$, $(a\pm b]$ for $(a-b,a+b]$ etc. We let $\conc$ denote string concatenation. For a non-empty string $t\in\zo^\ast$ and $i\in[\size{t}]$, let $t_i$ be the $i$'th bit of $t$, and for $i,j\in[\size{t}]$ such that $i<j$, let $t_{i,\ldots,j}=t_i\concat t_{i+1}\concat\ldots\concat t_j$. The empty string is denoted by $\EmptyString$, and for a non-empty string, let $t_{1,\ldots,0}=\EmptyString$. We let $\poly$ denote the set all polynomials and let \pptm denote a probabilistic algorithm that runs in \emph{strictly} polynomial time. Given a \pptm algorithm $\Ac$, we let $\Ac(u;r)$ be an execution of $\Ac$ on input $u$ given randomness $r$. A function $\nu \colon \N \to [0,1]$ is \textit{negligible}, denoted $\nu(n) = \negl(n)$, if $\nu(n)<1/p(n)$ for every $p\in\poly$ and large enough $n$.

Given a random variable $X$, we write $x\gets X$ to indicate that $x$ is selected according to $X$. Similarly, given a finite
set $\cs$, we let $s\la \cs$ denote that $s$ is selected according to the uniform distribution on $\cs$. We adopt the convention that when the same random variable occurs several times in an expression, all occurrences refer to a single sample. For example, $\Pr[f(X)=X]$ is defined to be the probability that when $x\gets X$, we have $f(x)=x$. We write $U_n$ to denote the random variable distributed uniformly over $\zn$. The support of a distribution $D$ over a finite set $\Uni$, denoted $\Supp(D)$, is defined as $\set{u\in\Uni: D(u)>0}$. The \emph{statistical distance} of two distributions $P$ and $Q$ over a finite set $\Uni$, denoted as $\SD(P,Q)$, is defined as $\max_{\cs\subseteq \Uni} \size{P(\cs)-Q(\cs)} = \frac{1}{2} \sum_{u\in \Uni}\size{P(u)-Q(u)}$.

A \emph{measure} is a function  $\Meas \colon\Omega\to [0,1]$. The
support of $\Meas$ over a set $\Omega$, denoted $\Supp(\Meas)$, is defined as $\set{\omega\in \Omega\colon \Meas(\omega) > 0}$. A measure $\Meas$ over $\Omega$ is the \emph{zero measure} if $\Supp(\Meas)=\emptyset$.

\subsection{Two-Party Protocols}\label{def:protocols}
The following discussion is restricted to no-input (possibly randomized), two-party protocols, where each message consists of a \emph{single} bit. We do not assume, however, that the parties play in turns (\ie the same party  might send two consecutive messages), but only that the protocol's transcript uniquely determines which party is playing next (\ie the protocol is well defined). In an $\rnd$-round protocol, the parties exchange exactly $\rnd$ messages (\ie bits). The tuple of the messages sent so far in any partial execution of a protocol is called the \emph{(communication) transcript} of this execution.

We write that a protocol $\cpi$ is equal to $\HonAHonB$, when $\HonA$ and $\HonB$ are the interactive Turing Machines that control the left- and right-hand party respectively, of the interaction according to $\cpi$. For a party $\HonC$ interacting according to $\cpi$, let $\nCc_\cpi$ be the other party in $\cpi$, where if $\cpi$ is clear from the context, we simply write $\nCc$.

If $\HonA$ and $\HonB$ are deterministic, then $\trans(\HonA,\HonB)$ denotes the uniquely defined transcript of the protocol $(\HonA,\HonB)$. If  $\HonA$ and $\HonB$ are randomized, we let $\rho_\HonA$ and $\rho_\HonB$ be the (maximal) number of random bits used by $\HonA$ and $\HonB$ respectively. For $r_\HonA\in\zo^{\rho_\HonA}$, $\HonA(\cdot;r_\HonA)$ stands for the variant of $\HonA$ when $r_\HonA$ are set as its random coins, and $\HonA(u;r_\HonA)$ is the message sent by $\HonA(\cdot;r_\HonA)$ when given a partial transcript $u$, for which the party $\HonA$ sends the next message. The above notations naturally extend for the party $\HonB$ as well. The transcript of the protocol $\paren{\HonA(\cdot;r_\HonA),\HonB(\cdot;r_\HonB)}$ is denoted by $\trans\paren{\HonA(\cdot;r_\HonA),\HonB(\cdot;r_\HonB)}$.
For a (partial) transcript $u$ of a protocol $\cpi=(\HonA,\HonB)$, let $\ConsisDist_\cpi(u)$ be the distribution of choosing $(r_\HonA,r_\HonB)\la \zo^{\rho_\HonA}\times \zo^{\rho_\HonB}$ conditioned on $\trans\paren{\HonA(\cdot;r_\HonA),\HonB(\cdot;r_{\HonB})}_{1,\ldots,\size{u}} = u$. 


\subsubsection{Binary Trees}
\begin{definition}[binary trees]\label{def:BinaryTrees}
For $\rnd\in \N$, let $\Tree^\rnd$ be the complete directed binary tree of height $\rnd$. We naturally identify the vertices of $\Tree^\rnd$ with binary strings: the root is denoted by the empty string $\EmptyString$, and the left- and right-hand children of a non-leaf node $u$ are denoted by $u0$ and $u1$ respectively.
\begin{itemize}

  \item Let $\Vertices(\Tree^\rnd)$, $\Edges(\Tree^\rnd)$, $\Root(\Tree^\rnd)$ and  $\Leaves(\Tree^\rnd)$ denote the vertices, edges, root and leaves of $\Tree^\rnd$ respectively.

  \item For $u\in \Vertices(\Tree^\rnd) \setminus \Leaves(\Tree^\rnd)$, let $\Tree^\rnd_u$ be the sub-tree of $\Tree^\rnd$ rooted at $u$.

  \item For $u\in \Vertices(\Tree^\rnd)$, let $\desc_\rnd(u)$ [\resp $\propDesc_\rnd(u)$] be the {\sf descendants} of $u$ in $\Tree^\rnd$ including $u$ [\resp excluding $u$], and for $\Uni\subseteq \Vertices(\Tree^\rnd)$ let  $\desc_\rnd(\Uni) = \bigcup_{u\in \Uni} \desc_\rnd(u)$ and $\propDesc_\rnd(\Uni) = \bigcup_{u\in \Uni} \propDesc_\rnd(u)$.

   \item The {\sf frontier} of a set $\cU\subseteq\Vertices(\Tree^\rnd)$, denoted by $\frnt{\cU}$, is defined as $\cU \setminus \propDesc_\rnd(\cU)$.\footnote{This is the set of all ``maximal'' transcripts in $\cU$ under the partial order subsequence relation.}
\end{itemize}
\end{definition}
When $\rnd$ is clear from the context, it is typically omitted from the above notation.
We will make use of the following simple observations.
\begin{proposition}\label{prop:UnBalLowValueG}
For  subsets $\cA$ and $\cB$ of $\Vertices(\Tree)$, it holds that  $\descP{\cA} \subseteq \descP{\cA\setminus \propDescP{\cB}} \cup \descP{\cB\setminus \cA}$.
\end{proposition}
\begin{proof}
Let $u\in \descP{\cA}$ and let $v\in\frnt{\cA}$ be such that $u\in\descP{v}$. We show that $v\in\descP{\cA\setminus \propDescP{\cB}} \cup \descP{\cB\setminus \cA}$. Clearly, if $v\notin\propDescP{\cB}$ we are done. Assume that $v\in \propDescP{\cB}$, namely, that there exists $w\in \cB$ such that $v\in\propDescP{w}$. Since $v$ is in the frontier of $\cA$ it follows that $w\notin\cA$. Hence, $v\in\descP{\cB\setminus\cA}$, and proof follows.
\end{proof}

\begin{proposition}\label{prop:UnBal1}
For subsets $\cA$, $\cB$ and $\cC$ of $\Vertices(\Tree)$, it holds that $\descP{\cA}\subseteq \descP{\paren{\cA\cup\cB}\setminus \descP{\cC}} \cup \descP{\cC\setminus\propDescP{\cB}}$.
\end{proposition}
\begin{proof}
Let $u\in \descP{\cA}$ and let $v\in\frnt{\cA}$ be such that $u\in\descP{v}$. We show that $v\in\descP{\paren{\cA\cup\cB}\setminus \descP{\cC}} \cup \descP{\cC\setminus\propDescP{\cB}}$. Clearly, if $v\notin\descP{\cC}$ we are done. Assume that $v\in \descP{\cC}$, and let $w\in \frnt{\cC}$ such that $v\in\descP{w}$. If $w\notin\propDescP{\cB}$, then $w\in\cC\setminus \descP{\cB}$, thus $v\in\descP{\cC\setminus \descP{\cB}}$ and we are done. 
Otherwise, if $w\in\propDescP{\cB}$, then since $w$ is on the frontier of $\cC$ it follows that $w\in\descP{\cB\setminus\descP{\cC}}$ and thus also $v\in\descP{\cB\setminus\descP{\cC}}$. The proof follows.  
\end{proof}

\subsubsection{Protocol Trees}
We naturally identify a (possibly partial)  transcript of an $\rnd$-round, single-bit message  protocol with a rooted path in $\Tree^\rnd$. That is, the transcript $t\in \zo^m$ is identified with the path $\EmptyString,t_1,t_{1,2},\dots,t$.

\begin{definition}[tree representation of a protocol]\label{def:ProtocolTree}
We make use of the following definitions \wrt an $\rnd$-round protocol $\cpi = \HonAHonB$, and $\HonC\in \set{\Ac,\Bc}$.
\begin{itemize}

  \item Let $\round(\cpi) = \rnd$, let $\Tree(\cpi) = \Tree^\rnd$, and for $X\in \set{\Vertices,\Edges,\Root,\Leaves}$  let $X(\cpi) = X(\Tree(\cpi))$.

  \item The {\sf edge distribution} induced by a protocol $\cpi$ is the function $\EdgeDist_\cpi\colon \Edges(\cpi) \to[0,1]$, defined as $\EdgeDist_\cpi(u,v)$ being the probability that the transcript of a random execution of $\cpi$  visits $v$, conditioned that it visits $u$.
  \item For $u\in \Vertices(\cpi)$, let $\VerticesDist_\cpi(u) = \EdgeDist_\cpi(\EmptyString,u_1)\cdot\EdgeDist_\cpi(u_1,u_{1,2}) \ldots\cdot \EdgeDist_\cpi(u_{1,\dots,\size{u}-1},u)$, and let the {\sf leaf distribution} induced by $\cpi$ be the distribution $\LDist\cpi$ over $\Leaves(\cpi)$, defined by  $\LDist\cpi(u) = \VerticesDist_\cpi(u)$.

  \item The party that sends the next message on transcript $u$ is said to {\sf control}  $u$, and we denote this party by $\ctrls_\cpi(u)$. We call $\ctrls_\cpi\colon \Vertices(\cpi)\to \set{\HonA,\HonB}$ the {\sf control scheme} of $\cpi$. Let $\ctrl{\cpi}{\HonC} = \set{u\in \Vertices(\cpi)\colon \ctrls_\cpi(u) = \HonC}$.
 \end{itemize}
\end{definition}

For $\cS\subseteq \Vertices(\cpi)$, let $\ppr{\LDist{\cpi}}{\cS}$ be abbreviation for $\ppr{\ell\la\LDist{\cpi}}{\ell\in\cS}$. Note that every function $\EdgeDist\colon \Edges(\Tree^\rnd) \to[0,1]$ with $\EdgeDist(u,u0) + \EdgeDist(u,u1) =1$ for every $u \in \Vertices(\Tree^\rnd) \setminus \Leaves(\Tree^\rnd)$ with $\VerticesDist(u)>0$,  along with a control scheme (active in each node), defines a two party, $\rnd$-round, single-bit message protocol (the resulting protocol might be inefficient). The analysis in \cref{sec:IdealAttacker} naturally gives rise to functions over binary  trees that do not correspond to any two-party execution. We identify the ``protocols" induced  by such functions by the special symbol $\perp$. We let $\eex{\LDist\perp}{f} = 0$, for any real-value function $f$.

The view of a protocol as an edge-distribution function allows us to consider protocols induced  by sub-trees of $\Tree(\cpi)$.

\begin{definition}[sub-protocols]\label{def:subProtocolTree}
Let $\cpi$ be a protocol and let $u\in \Vertices(\cpi)$. Let $\Restrict{\cpi}{u}$ denote the protocol induced by the function $\EdgeDist_\cpi$  on the sub-tree of $\Tree(\cpi)$ rooted at $u$, if $\VerticesDist_\cpi(u)>0$, and let $\Restrict{\cpi}{u} = \perp$ otherwise.
\end{definition}

Namely, the protocol $\Restrict{\cpi}{u}$ is the protocol $\cpi$ conditioned on $u$ being the transcript of the first $\size{u}$ rounds. When convenient, we remove the parentheses from notation, and simply write $\cpi_u$. Two sub-protocols of interest are $\cpi_0$ and $\cpi_1$, induced by $\EdgeDist_\cpi$ and the trees rooted at the left- and right-hand descendants of $\Root(\Tree)$. For a measure $\Meas\colon \Leaves(\cpi)\to [0,1]$ and $u\in\Vertices(\cpi)$, let $\Restrict{\Meas}{u}\colon \Leaves(\cpi_u)\to [0,1]$ be the restricted measure induced by $\Meas$ on the sub-protocol $\cpi_u$. Namely, for any $\ell\in\Leaves(\cpi_u)$, $\Restrict{\Meas}{u}(\ell)=\Meas(\ell)$.

\subsubsection{Tree Value}
\begin{definition}[tree value]\label{def:ProtocolValue}
Let $\cpi$ be a two-party protocol that at the end of any of its executions, the parties output the {\sf same} real value. Let $\Color_\cpi\colon \Leaves(\cpi) \to \R$ be the {\sf common output} function of $\cpi$ ---  $\Color_\cpi(\ell)$ is the {\sf common} output of the parties in an execution ending in $\ell$.\footnote{Conditioned that  an execution of the protocol generates a transcript $\ell$, the parties' coins are in a product distribution. Hence, if the parties always have the same output, then the protocol's output is indeed a (deterministic) function of its transcript.} Let $\Val(\cpi) = \Ex_{\LDist\cpi}[\Color_\cpi]$, and for $x\in \R$ let $\Leaves_x(\cpi) = \set{\ell\in\Leaves(\cpi) \colon \Color_\cpi(\ell)=x}$.
\end{definition}

Throughout this paper we restrict ourselves to protocols whose common output is either one or zero, \ie the image of $\Color_{\cpi}$ is the set $\zo$. 
The following immediate fact states that the expected value of a measure, whose support is a subset of the 1-leaves of some protocol, is always smaller than the value of that protocol.
\begin{fact}\label{fact:ValueExp}
Let $\cpi$ be a protocol and let $\Meas$ be a measure over $\Leaves_1(\cpi)$. Then $\eex{\LDist\cpi}{\Meas} \leq \Val(\cpi)$.
\end{fact}

We will also make use of the following proposition, showing that if two protocols are close and there exists a set of nodes whose value (the probability that the common output is one conditioned on reaching these nodes) is large in one protocol but small in the other, then the probability of reaching this set is small.

\begin{proposition}\label{prop:CloseProDiffValues}
Let $\cpi=\paren{\HonA,\HonB}$ and $\cpi'=\paren{\HonC,\HonD}$ be two $m$-round protocols with $\Color_\cpi \equiv \Color_{\cpi'}$, and let $\cF\subseteq \Vertices(\cpi)$ be a frontier. Assume that $\SDP{\LDist{\cpi}}{\LDist{\cpi'}}\leq \eps$ , that $\ppr{\LDist{\cpi}}{\Leaves_1(\cpi)\mid \descP{\cF}} \leq \alpha$, and that $\ppr{\LDist{\cpi'}}{\Leaves_1(\cpi)\mid \descP{\cF}} \geq \beta$, for some $\eps > 0$ and $0 \leq \alpha < \beta \leq 1$. Then, $\ppr{\LDist{\cpi}}{\descP{\cF}}\leq \eps \cdot \frac{1+\beta}{\beta-\alpha}$.
\end{proposition}
Note that since both $\cpi$ and $\cpi'$ have $m$-rounds, it holds that  $\Vertices(\cpi)= \Vertices(\cpi')$ and $\Leaves(\cpi)=\Leaves(\cpi')$. Moreover, since $\Color_\cpi \equiv \Color_{\cpi'}$, it also holds that $\Leaves_1(\cpi)$, the set of $1$-leaves in $\cpi$, is identical to $\Leaves_1(\cpi')$, the set of $1$-leaves in $\cpi'$.
\begin{proof}
Let $\mu=\ppr{\LDist{\cpi}}{\descP{\cF}}$, $\mu'=\ppr{\LDist{\cpi'}}{\descP{\cF}}$ and $\cS=\Leaves_1(\cpi)\cap \descP{\cF}$. It follows that
\begin{align}
\ppr{\LDist{\cpi}}{\cS} &= \ppr{\LDist{\cpi}}{\descP{\cF}}\cdot \ppr{\LDist{\cpi}}{\Leaves_1(\cpi) \condition \descP{\cF}} \leq \mu \cdot \alpha
\end{align}
and that
\begin{align}
\ppr{\LDist{\cpi'}}{\cS} &= \ppr{\LDist{\cpi'}}{\descP{\cF}}\cdot \ppr{\LDist{\cpi'}}{\Leaves_1(\cpi) \condition \descP{\cF}} \geq \mu' \cdot \beta.
\end{align}
Moreover, since $\SDP{\LDist{\cpi}}{\LDist{\cpi'}}\leq \eps$, it follows that $\mu'\geq\mu-\eps$ and that $\ppr{\LDist{\cpi'}}{\cS}-\ppr{\LDist{\cpi}}{\cS}\leq\eps$. Putting it all together, we get
\begin{align*}
\eps & \geq \ppr{\LDist{\cpi'}}{\cS}-\ppr{\LDist{\cpi}}{\cS} \\
&\geq \mu'\cdot\beta-\mu\cdot \alpha \\
&\geq (\mu-\eps)\cdot \beta-\mu\cdot \alpha \\
& =   (\beta - \alpha)\cdot \mu-\beta\cdot \eps,
\end{align*}
which implies the proposition.
\end{proof}

\subsubsection{Protocol with Common Inputs}
We sometimes would like to apply the above terminology  to a protocol $\cpi = \HonAHonB$ whose parties get a common security parameter $1^\SecParam$. This is formally done by considering the protocol $\cpi_n = (\HonA_n,\HonB_n)$, where $\HonC_n$ is the algorithm derived by ``hardwiring" $1^\SecParam$ into the code of $\HonC$.

\subsection{Coin-Flipping Protocols}\label{sec:CFprotocols}
In a coin-flipping protocol two parties interact and in the end have a common output bit. Ideally, this bit should be random and no cheating party should be able to bias its outcome to either direction (if the other party remains honest). For interactive, probabilistic  algorithms $\HonA$ and $\HonB$, and $x\in \zo^\ast$,  let $\out(\HonA,\HonB)(x)$ denote the parties' output, on common input $x$.
\begin{definition}[(strong) coin-flipping]\label{def:strongCF}
A \ppt protocol $(\HonA,\HonB)$ is a {\sf $\delta$-bias coin-flipping protocol} if the following holds.
\begin{itemize}
  \item[Correctness:]\label{item:HonSngCF} $\Pr[\out(\HonA,\HonB)(1^n) =0] = \Pr[\out(\HonA,\HonB)(1^n) =1] = \frac12$.

  \item[Security:] \label{item:CheatSngCF} $\Pr[\out(\someAadv,\HonB)(1^n) =c],\Pr[\out(\HonA,\someBadv)(1^n)=c] \leq \frac12 +  \delta(n)$, for any \pptm's  $\someAadv$ and $\someBadv$, bit $c\in \zo$ and large enough $n$.

\end{itemize}
\end{definition}

Sometimes, \eg if the parties have (a priori known) opposite preferences, an even weaker definition of coin-flipping protocols is of interest.
\begin{definition}[weak coin-flipping]\label{def:weakCF}
A \ppt protocol $\HonAHonB$ is a {\sf weak $\delta$-bias coin-flipping protocol} if the following holds.
\begin{itemize}
  \item[Correctness:] Same as in \cref{def:strongCF}.

  \item[Security:] There exist bits $c_\Ac \neq c_\Bc \in \zo$ such that 
  $$\Pr[\out(\Aadv,\Bc)(1^n) =c_\Ac], \Pr[\out(\Ac,\Badv)(1^n) =c_\Bc] \leq \frac12 +  \delta(n)$$
  for any \pptm's  $\Aadv$ and $\Badv$, and large enough $n$.
\end{itemize}
\end{definition}

\begin{remark}
	Our result still holds when  the allowing the common bit  in a random honest execution of the protocol to be  an arbitrary  constant in $(0,1)$.  In contrast, our proof critically relies on the assumption that the honest parties are \emph{always} in agreement. 
\end{remark}

In the rest of the paper we restrict our attention to $\rnd$-round single-bit message coin-flipping protocols, where $\rnd=\rnd(n)$ is a function of the protocol's security parameter. Given such a protocol $\cpi = (\HonA,\HonB)$, we assume that its common output (\ie the coin) is efficiently computable from a (full) transcript of the protocol. (It is easy to see that these assumptions are \wlg .)

\subsection{One-Way Functions and Distributional One-Way Functions}
A one-way function (OWF) is an efficiently computable function whose inverse cannot be computed on average by any \pptm.
\begin{definition}\label{def:owf}
A polynomial-time  computable function $f\colon\zn\to\zo^{\ell(n)}$ is {\sf one-way} if
$$\ppr{x\la \zn; y =f(x)}{\Ac(1^n,y) \in f^{-1}(y)}=\negl(n)$$
for any \pptm $\Ac$.
\end{definition}

A seemingly weaker definition is that of a distributional OWF. Such a function is easy to compute, but it is hard to compute uniformly random preimages of random images.

 \begin{definition}\label{def:dowf}
 A polynomial-time  computable  $f\colon\zn\to\zo^{\ell(n)}$ is  {\sf distributional one-way}, if $\exists p\in\poly$ such that
 $$\SD\left((x,f(x))_{x\la \zn},(\Ac(f(x)),f(x))_{x\la \zn}\right)\geq \frac{1}{p(n)}$$
  for any \pptm $\Ac$ and large enough $n$.
 \end{definition}

Clearly, any one-way function is also a distributional one-way function. While the other implication is not necessarily always true, \citet{ImpagliazzoLu89} showed that the existence of distributional one-way functions  implies that of (standard) one-way functions. In particular, the authors of \cite{ImpagliazzoLu89} proved that if one-way functions do not exist, then any efficiently computable function has an inverter of the following form.
\begin{definition}[$\xi$-inverter]\label{sef:gammaInverter}
 An algorithm $\Inv$ is an {\sf $\xi$-inverter} of $f\colon\domain\to \range$ if the following holds.
\begin{align*}
\ppr{x\la \domain; y = f(x)}{\SD\left((x')_{x'\la f^{-1}(y)},(\Inv(y))\right) > \xi} \leq \xi.
\end{align*} 
\end{definition}
\begin{lemma}[{\cite[Lemma 1]{ImpagliazzoLu89}}]\label{lemma:NoDistOWF}
Assume one-way functions do not exist. Then for any polynomial-time computable function $f\colon\zn\to \zo^{\ell(n)}$ and $p\in \poly$, there exists a \pptm algorithm $\Inv$ such that the following holds for infinitely many $n$'s.  On security parameter $1^n$, algorithm $\Inv$ is a $1/p(n)$-inverter of $f_n$  (\ie $f$ is restricted to $\zn$).
\end{lemma}

\citet{ImpagliazzoLu89} only gave a proof sketch for the above lemma. The full proof can be found in \cite[Theorem 4.2.2]{ImpagliazzoPHD}.

\begin{remark}[Definition of inverter]
In their original definition, \citet{ImpagliazzoLu89} defined a $\xi$-inverter as an algorithm $\Inv$ for which it holds that
\begin{align*}
\SDP{(x,f(x))_{x\la \zn}}{(\Inv(f(x)),f(x))_{x\la \zn}} < \xi.
\end{align*}
They also proved \cref{lemma:NoDistOWF} \wrt this definition. By taking, for example, $\xi'=\xi^2$ and applying their proof with $\xi'$, it is easy to see how our version of \cref{lemma:NoDistOWF} follows \wrt the above definition of a $\xi$-inverter. 
\end{remark}

Note that nothing is guaranteed when invoking  a good inverter (\ie a  $\gamma$-inverter for some small $\gamma$) on an \emph{arbitrary distribution}. Yet the following lemma yields that if the distribution in consideration  is ``not too different" from  the output distribution of $f$, then such good inverters are useful.

\begin{lemma}\label{lemma:SDmQueriesAlg}
Let $f$ and $g$ be two randomized functions  over the same domain $\domain\cup\set{\perp}$ such that $f(\perp)\equiv g(\perp)$, and let $\set{P_i}_{i\in [k]}$ be a set of distributions over $\domain\cup\set{\perp}$ such that for some $a \geq 0$ it holds that $\Ex_{q\la P_i}[\SD(f(q),g(q))]\leq a$ for every $i\in [k]$. Let $\Ac$ be a $k$-query oracle-aided algorithm that only makes queries in $\domain$. Let $Q=(Q_1,\ldots,Q_k)$ be the random variable of the queries of $\Ac^f$ in such a random execution, setting $Q_i=\perp$ if $\Ac$ makes less than $i$ queries.

Assume that $\ppr{(q_1,\ldots,q_k) \la Q}{\exists i\in [k] \colon  q_i\neq \perp \land \; Q_i(q_i) > \lambda\cdot P_i(q_i)} \leq b$ for some $\lambda,b \geq 0$. Then
$\SD\left(\Ac^{f},\Ac^{g}\right)\le b +  ka\lambda $.
\end{lemma} 
To prove \cref{lemma:SDmQueriesAlg}, we use the following proposition.
\begin{proposition}\label{fact:coupling}
For every two distributions $P$ and $Q$ over a set $\domain$, there exists a distribution $R_{P,Q}$ over $\domain \times \domain$, such that the following hold:
\begin{enumerate}
  \item $(R_{P,Q})_1 \equiv P$ and $(R_{P,Q})_2 \equiv Q$, where $(R_{P,Q})_b$ is the projection of $R_{P,Q}$ into its $b$'th coordinate.
  \item $\ppr{(x_1,x_2) \la R_{P,Q}}{x_1 \neq x_2} = \SD(P,Q)$.
\end{enumerate}
\end{proposition}
\begin{proof}
For every $x\in \domain $, let $M(x) = \min\set{P(x),Q(x)}$, let $M_P(x) = P(x) - M(x)$ and  $M_Q(x) = Q(x) - M(x)$. The distribution $R_{P,Q}$ is defined by the following procedure. With probability $\mu = \sum_{x\in \domain}M(x)$, sample an element $x$ according to $M$ (\ie $x$ is returned with probability $\frac{M(x)}\mu$), and return $(x,x)$; otherwise return $(x_P,x_Q)$ where $x_P$ is sampled according to $M_P$ and $x_Q$ is sampled according to $M_Q$. It is clear that $\ppr{(x_1,x_2) \la R_{P,Q}}{x_1 \neq x_2} = \SD(P,Q)$. It also holds that
\begin{align*}
(R_{P,Q})_1(x) &= \mu\cdot \frac{M(x)}\mu + (1- \mu) \cdot \frac{M_P(x)}{\mu_P}\\
& =M(x) + M_P(x)\\
& = P(x),
\end{align*}
where $\mu_P \eqdef  \sum_{x\in\domain}{M_P} = (1- \mu)$. Namely, $(R_{P,Q})_1 \equiv P$. The proof that $(R_{P,Q})_2 \equiv Q$ is  analogous.

\remove{
Using standard argument, it holds that $\SD\left(\Ac^{f},\Ac^{g}\right)\le \frac{\gamma\cdot(q+1)}{\delta}$ is at most the probability that the following experiment aborts.
\begin{experiment}~
\begin{enumerate}
  \item Start emulating a random execution of $\Ac$.
  \item Do until $\Ac$ halts.
  \begin{enumerate}
    \item Let $q$ be the next query of $\Ac$.

    \item Let $w = f(q)$.
    \item  With probability $\frac{\pr{g(q) = w}}{\pr{f(q) = w}}$, give  $w$ to $\Ac$ as the oracle answer.

    Otherwise, abort.
  \end{enumerate}
  \end{enumerate}
\end{experiment}
We prove by induction, and thus complete the proof of the proposition, that probability of aborting in the first $i$ rounds of the above experiment is bounded by $\frac{\gamma\cdot(i)}{\delta}$.  Assume for $i-1$, and let $q_i$ and $w_i$ be the values of $q$ and $w$ respectively, in the $i$'th round of a random execution the above experiment (we assume \wlg that the $i$'th round is reachable, as otherwise the proof is immediate).
}

\end{proof}

\begin{proof}[Proof of \cref{lemma:SDmQueriesAlg}]
Using \cref{fact:coupling} and standard argument, it holds that $\SD\left(\Ac^{f},\Ac^{g}\right)$ is at most the probability that the following experiment aborts.
\begin{experiment}~
\begin{enumerate}
  \item Start emulating a random execution of $\Ac$.
  \item Do until $\Ac$ halts:
  \begin{enumerate}
    \item Let $q$ be the next query of $\Ac$.
   \item Sample $(a_1,a_2) \la R_{f(q),g(q)}$.
    \item  If $a_1 = a_2$,  give $a_1$ to $\Ac$ as the oracle answer.

    Otherwise, abort.
  \end{enumerate}
  \end{enumerate}
\end{experiment}
By setting $\cS_i=\set{q:q\in\Supp(Q_i)\land Q_i(q)\leq\lambda \cdot P_i(q)}$ for $i\in[k]$ and recalling that by assumption $f(\perp)\equiv g(\perp)$ (thus, when sampling $(a_1,a_2) \la R_{f(\perp),g(\perp)}$, $a_1$ always equals $a_2$), we conclude that
 \begin{align*}
\SD\left(\Ac^{f},\Ac^{g}\right)
&\leq \ppr{(q_1,\ldots,q_k) \la Q}{\exists i\in [k] \colon  q_i\notin  \cS_i\cup\set{\perp}}\\
&\quad+\ppr{(q_1,\ldots,q_k)\la Q}{\exists i\in[k] \colon a_1\ne a_2 \text{~where~}(a_1,a_2)\la R_{f(q_i),g(q_i)} \land q_i\in \cS_i} \\
 &\leq b+\sum_{i\in[k]}\sum_{q\in \cS_i } Q_i(q)\cdot\pr{a_1\ne a_2 \text{~where~}(a_1,a_2)\la R_{f(q),g(q)}}\\
 &\overset{(1)}{\leq}  b+\sum_{i\in [k]} \sum_{q\in \cS_i } Q_i(q)\cdot\SD(f(q),g(q))\\
 &\overset{(2)}{\leq}  b+\sum_{i\in [k]} \sum_{q\in \Supp(P_i) }\lambda\cdot P_i(q)\cdot\SD(f(q),g(q))\\
 & \leq  b + \lambda \sum_{i\in[k]}\Ex_{q\la P_i}[\SD(f(q),g(q))] \\
 & \leq  b +  k a \lambda,
 \end{align*}
where (1) follows from \cref{fact:coupling} and (2) from the definition of the sets $\set{S_i}_{i\in[k]}$.
\end{proof}

\subsection{Two Inequalities}
We make use of following technical lemmas, whose  proofs are given in \cref{sec:MissProof}.
\newcommand{\CalcLemmaOne}[1]{%
Let $ x,y \in [0,1]$, let $k\geq 1$ be an integer and let $a_1,\ldots,a_k,b_1,\ldots,b_k \in (0,1]$. Then for any $p_0, p_1 \geq 0$ with $p_0+p_1=1$, it holds that
\begin{align}\IfEqCase{#1}{{1}{\label{eq:CalcLemma}}{2}{}}
p_0\cdot \frac{x^{k+1}}{\prod_{i=1}^k a_i} + p_1\cdot \frac{y^{k+1}}{\prod_{i=1}^k b_i} \geq \frac{(p_0x+p_1y)^{k+1}}{\prod_{i=1}^k (p_0a_i + p_1b_i)}.%
\IfEqCase{#1}{{1}{}{2}{\nonumber}}
\end{align}
}%

\begin{lemma}\label{lemma:calculus1}
\CalcLemmaOne{2}
\end{lemma}

\newcommand{\CalcLemmaTwo}[1]{%
For every $\delta\in(0,\frac12]$, there exists  $\alpha=\alpha(\delta)\in(0,1]$ such that
\begin{align}\IfEqCase{#1}{{1}{\label{eq:calculus2}}{2}{}}
\lambda\cdot a_1^{1+\alpha}\cdot(2-a_1\cdot x )+ a_2^{1+\alpha}\cdot(2-a_2\cdot x)\leq(1+\lambda)\cdot(2-x),%
\IfEqCase{#1}{{1}{}{2}{\nonumber}}
\end{align}
for every $x\geq\delta$ and $\lambda,y \geq0$ with $\lambda y \leq 1$, for $a_1=1+y$ and $a_2=1-\lambda y$.
}%

\begin{lemma}\label{lemma:calculus2}
\CalcLemmaTwo{2}
\end{lemma}

\section{The Biased-Continuation Attack}\label{sec:IdealAttacker}
In this section we describe an attack to bias any coin-flipping protocol. The described attack, however, might be impossible to implement efficiently (even when assuming one-way functions do not exist). Specifically, we assume access to an ideal sampling algorithm to sample a \emph{uniform} preimage of \emph{any} output of the functions under consideration. Our actual attack, the subject of \cref{sec:RealAttacker}, tries to mimic the behavior of this attack while being efficiently implemented (assuming one-way functions do not exist).

The following discussion is restricted to (coin-flipping) protocols whose parties always output the same bit as their common output, and this bit is determined by the protocol's transcript. In all protocols considered in this section, the messages are bits. In addition, the protocols under consideration have no inputs (neither private nor common), and in particular no security parameter is involved.\footnote{In \cref{sec:RealAttacker}, we make use of these input-less  protocols by ``hardwiring" the security parameter of the protocols under consideration.}  Recall that $\perp$ stands for a canonical invalid/undefined protocol, and that $\Ex_{\LDist\perp}[f] = 0$, for any real value function $f$. (We refer the reader to \cref{section:Preliminaries} for a discussion of the conventions and assumptions used above.) Although the focus of this paper is coin-flipping protocols, all the results in this section hold true for any two-party protocol meeting the above assumptions. Specifically, we do not assume that an honest execution of the protocol produces a uniformly random bit, nor do we assume that the parties executing the protocol can be implemented by a polynomial time probabilistic Turing machine. For this reason we omit the term ``coin-flipping" in this section.

Throughout the section we prove statements \wrt attackers that, when playing the role of the left-hand party of the protocol (\ie $\HonA$), are trying to bias the common output of the protocol towards one, and, when playing the role of the right-hand party of the protocol (\ie $\HonB$), are trying to bias the common output of the protocol towards zero.
All statements have analogues ones \wrt the opposite attack goals.

Let $\cpi = \HonAHonB$ be a protocol. The \textit{recursive biased-continuation attack} described below recursively applies the \textit{biased-continuation attack} introduced by \citet{HaitnerOmri11}.\footnote{Called the ``random continuation attack'' in  \cite{HaitnerOmri11}.} The biased-continuation attacker $\rcAP{\cpi}{1}$ -- playing the role of $\Ac$ -- works as follows: in each of $\HonA$'s turns, $\rcAP{\cpi}{1}$ picks a random continuation of $\cpi$, whose output it induces is equal to one, and plays the current turn accordingly. The $i$'th  biased-continuation attacker $\rcAP{\cpi}{i}$, formally described below,   uses the same strategy but the random continuation taken is  of the protocol $(\rcAP{\cpi}{i-1},\HonB)$.

Moving to the formal discussion, for a protocol  $\cpi = \HonAHonB$, we defined  its biased continuator $\RandomCont_\cpi$ as follows.
\begin{definition}[biased continuator $\RandomCont_\cpi$]\label{def:sampler}
\item Input: $u \in \Vertices(\cpi) \setminus \Leaves(\cpi)$ and a bit $b\in \zo$
\item Operation:
\begin{enumerate}
  \item Choose  $\ell \la \LDist\cpi$ conditioned  that
   \begin{enumerate}
     \item $\ell \in  \desc(u)$, and
     \item $\Color_\cpi(\ell) = b$.\footnote{If no such $\ell$ exists, the algorithm returns an arbitrary leaf in $\desc(u)$.}
   \end{enumerate}

  \item Return $\ell_{\size{u}+1}$.
\end{enumerate}
\end{definition}

Let $\rcAP{\cpi}{0}\equiv \HonA$, and for integer $i> 0$ define:
\begin{algorithm}[recursive biased-continuation attacker $\rcAP{\cpi}{i}$]\label{alg:adversary}
\item  Input: transcript $u \in \zo^\ast$.
\item Operation:
\begin{enumerate}

\item If $u\in \Leaves(\cpi)$, output $\Color_\cpi(u)$ and halt.

\item  Set $\msg = \RandomCont_{(\rcAP{\cpi}{i-1},\HonB)}(u,1)$.

\item Send $\msg$ to $\HonB$.

\item If $u'=u\conc \msg\in \Leaves(\cpi)$, output $\Color_\cpi(u')$.\footnote{For the mere purpose of biasing $\HonB$'s output, there is no need for $\rcA{i}$ to output anything. Yet doing so helps us to simplify our recursion definitions (specifically, we use the fact that in $(\rcA{i},\HonB)$ the parties always have the same output).}

\end{enumerate}
\end{algorithm}
The attacker $\rcBP{\cpi}{i}$ attacking towards zero is  analogously defined (specifically, the call to the biased continuator $\RandomCont_{(\rcAP{\cpi}{i-1},\HonB)}(u,1)$ in \cref{alg:adversary} is changed to $\RandomCont_{(\HonA,\rcBP{\cpi}{i-1})}(u,0)$).\footnote{The subscript $\cpi$ is added to the notation (\ie $\rcAP{\cpi}{i}$), since the biased-continuation attack for $\Ac$ depends not only on the definition of the party $\Ac$, but also on the definition of $\Bc$, the other party in the protocol.}

It is relatively easy to show that the more recursions $\rcAP{\cpi}{i}$ and $\rcBP{\cpi}{i}$ do, the closer their success probability is to that of an all-powerful attacker, who can either bias the outcome to zero or to one. The important point of the following theorem is that, for any $\eps>0$, there exists a \emph{global} constant $\kappa=\kappa(\eps)$ (\ie independent of the underlying protocol), for which either $\rcAP{\cpi}{\kappa}$ or $\rcBP{\cpi}{\kappa}$ succeeds in its attack with probability at least $1-\eps$. This becomes crucial when trying to efficiently implement these adversaries (see \cref{sec:RealAttacker}), as each recursion call might induce a polynomial blowup in the running time of the adversary. Since $\kappa$ is constant (for a constant $\eps$), the recursive attacker is still efficient.

\begin{theorem}[main theorem, ideal version]\label{thm:MainIdeal}
For every $\eps \in (0,\frac12]$ there exists non-negative  integer $\kappa \in \widetilde{O}(1/\eps)$ such that for every protocol $\cpi= \HonAHonB$, either $\Val(\rcAP{\cpi}{\kappa},\HonB)>  1 - \eps$ or  $\Val(\HonA,\rcBP{\cpi}{\kappa})< \eps$.
\end{theorem}
The rest of this section is devoted to proving the above theorem.


In what follows, we typically omit the subscript $\cpi$ from the notation of the above attackers. Towards proving \cref{thm:MainIdeal} we show a strong (and somewhat surprising) connection between recursive biased-continuation attacks on a given protocol  and the optimal valid attack on this protocol. The latter is the best (unbounded) attack on this protocol, which sends only valid messages (ones that could have been sent by the honest party). Towards this goal we define sequences of measures over the leaves (\ie transcripts) of the protocol, connect these measures to the optimal attack, and then lower bound the success of the recursive biased-continuation attacks using these measures.

In the following we first observe some basic properties of the recursive biased-continuation attack. Next, we define the optimal valid attack, define a simple measure \wrt this attack, and analyze, as a warm-up, the success of recursive biased-continuation attacks on this measure. After arguing why considering the latter measure does not suffice, we define a sequence of measures, and then state, in \cref{sec:RCAltDomMeas}, a property of this sequence that yields \cref{thm:MainIdeal} as a corollary. The main body of this section deals with proving the aforementioned property.


\subsection{Basic Observations About $\RC{\HonA}{i}$}\label{sec:PropBiasCont}
We make two basic observations regarding the recursive biased-continuation attack. The first gives expression to the edge distribution this attack induces. The second is that this attack is stateless. We'll use these observations in the following sections; however, the reader might want to skip their straightforward proofs for now.

Recall that at each internal node in its control, $\rcA{1}$ picks a random continuation to one. We can also describe $\rcA{1}$'s behavior as follows: after seeing a transcript $u$, $\rcA{1}$ biases the probability of sending, \eg $0$ to $\HonB$: it does so proportionally to the ratio between the chance of having output one among all honest executions of the protocol that are consistent with the transcript $u \concat 0$, and the same chance but \wrt the transcript $u$.
The behavior of $\rcA{i}$ is analogous where $\rcA{i-1}$ replaces the role of $\Ac$ in the above discussion. Formally, we have the following claim.
\begin{claim}\label{claim:weights}
Let $\cpi = \HonAHonB$ be a protocol and let $\rcA{j}$ be according to \cref{alg:adversary}. Then
\begin{align*}
\EdgeDist_{(\rcA{i},\HonB)}(u,u b) = \EdgeDist_\cpi(u,u b)\cdot
\frac{\prod_{j=0}^{i-1}\Val((\rcA{j},\HonB)_{u b})}{\prod_{j=0}^{i-1}\Val((\rcA{j},\HonB)_{u})},\footnotemark
\end{align*}
\footnotetext{Recall that for a protocol $\cpi$ and a partial transcript $u$, we let $\EdgeDist_\cpi(u,u b)$ stand for the probability that the party controlling $u$ sends $b$ as the next message, conditioning that $u$ is the transcript of the execution thus far.}
for any $i\in \N$,  $\HonA$-controlled $u\in \Vertices(\cpi)$ and $b\in\set{0,1}$.
\end{claim}
This claim is a straightforward generalization of the proof of \cite[Lemma 12]{HaitnerOmri11}. However, for completeness and to give an example of our notations, a full proof is given below.
\begin{proof}
The proof is by induction on $i$. For $i=0$, recall that $\rcA{0}\equiv \HonA$, and hence $\EdgeDist_{(\rcA{0},\HonB)}(u,u b)=\EdgeDist_\cpi(u,u b)$, as required.

Assume the claim holds for $i-1$, and we want to compute $\EdgeDist_{(\rcA{i},\HonB)}(u,u b)$.
The definition of \cref{alg:adversary} yields that for any positive $i\in\N$, it holds that
\begin{align}\label{eq:CalcEdgeDist}
\EdgeDist_{(\rcA{i},\HonB)}(u,u b) &= \ppr{\ell\la \LDist{\ProArc{i-1}}}{\ell_{\size{u}+1}=b\condition \ell\in\desc(u)\land \Color_{\ProArcP{i-1}}(\ell)=1}\footnotemark\\
&= \frac{\ppr{\ell\la \LDist{\ProArc{i-1}}}{\ell_{\size{u}+1}=b\land \Color_{\ProArcP{i-1}}(\ell)=1 \condition \ell\in\desc(u)}}{\ppr{\ell\la \LDist{\ProArc{i-1}}}{\Color_{\ProArcP{i-1}}(\ell)=1 \condition \ell\in\desc(u)}} \nonumber\\
&= \EdgeDist_{{(\rcA{i-1},\HonB)}}(u,u b)\cdot \frac{\Val((\rcA{i-1},\HonB)_{u b})}{\Val((\rcA{i-1},\HonB)_{u})},\nonumber
\end{align}
\footnotetext{Recall that for a protocol $\cpi$, we let $\LDist{\cpi}$ stand for the leaf distribution of $\cpi$.}
where the last equality is by a simple chain rule, \ie since
\begin{align*}
\EdgeDist_{{(\rcA{i-1},\HonB)}}(u,u b) &= \ppr{\ell\la \LDist{\ProArc{i-1}}}{\ell_{\size{u}+1}=b \condition \ell\in\desc(u)}, \hbox{ and}\\
\Val((\rcA{i-1},\HonB)_{u b}) &= \ppr{\ell\la \LDist{\ProArc{i-1}}}{\Color_{\ProArcP{i-1}}(\ell)=1 \condition \ell\in\desc(u) \land \ell_{\size{u}+1}=b}.
\end{align*}

The proof is concluded by plugging the induction hypothesis into \cref{eq:CalcEdgeDist}.
\end{proof}

The following observation enables us to use induction when analyzing the power of  $\rcA{i}$.
\begin{proposition}\label{fact:switchOrder}
For every protocol $\cpi = (\HonA_\cpi,\HonB_\cpi)$,  $i\in \N$ and $b\in \zo$, it holds that $\Restrict{\rcAP{\cpi}{i},\Bc}{b}$ and  $\paren{\rcAP{\cpi_b}{i},\HonB_{\cpi_b}}$ are the same protocol, where $\cpi_b = (\HonA_{\cpi_b},\HonB_{\cpi_b})$.
\end{proposition}
\begin{proof}
Immediately follows from  $\rcAP{\cpi}{i}$ being stateless.
\end{proof}

\begin{remark}
Note that the party $\HonB_{\cpi_b}$, defined by the subprotocol $\cpi_b$ (specifically, by the edge distribution of the subtree $\Tree(\cpi_b)$), might not have an efficient implementation, even if $\Bc$ does have one. For the sake of the arguments we make in this section, however, it matters only that $\HonB_{\cpi_b}$ is well defined.
\end{remark}

\subsection{Optimal Valid Attacks}\label{sec:optimum}
When considering the optimal attackers for a given protocol, we restrict ourselves to valid attackers. Informally, we can say that, on each of its turns, a valid attacker sends a message from the set of  possible replies that the honest party might choose given the  transcript so far.

\begin{definition}[optimal valid attacker]\label{def:optimumAttacker}
Let $\cpi=\HonAHonB$ be a protocol. A deterministic algorithm $\HonA'$ playing the role of $\HonA$ in $\cpi$ is in $\cA^\ast$, if $\VerticesDist_\cpi(u) = 0 \implies \VerticesDist_{(\HonA',\HonB)}(u) = 0$ for any $u\in\Vertices(\cpi)$. The class  $\cB^\ast$ is analogously defined. Let $\BestA{\cpi}=\max_{\Ac'\in\cA^\ast}\set{\Val(\Ac',\HonB)}$ and $\BestB{\cpi}=\max_{\Bc'\in\cB^\ast}\set{1-\Val(\HonA,\Bc')}$.
\end{definition}

The following proposition is immediate.
\begin{proposition}\label{fact:perfectAdv}
Let $\cpi=\HonAHonB$ be a protocol and let $u\in \Vertices(\cpi)$. Then,
\begin{align*}
\BestA{\cpi_u} = \left\{
  \begin{array}{ll}
    \Color_\cpi(u) & \hbox{$u \in \Leaves(\cpi)$;}\\
    \max\set{\BestA{\cpi_{u b}} \colon   \EdgeDist_{\cpi}(u,ub)>0}, & \hbox{$u\notin \Leaves(\cpi)$ and $u$ is controlled by $\HonA$;}\\
    \EdgeDist_{\cpi}(u,u0)\cdot \BestA{\cpi_{u 0}} + \EdgeDist_{\cpi}(u,u1)\cdot \BestA{\cpi_{u 1}}, & \hbox{$u\notin \Leaves(\cpi)$ and $u$ is controlled by $\HonB$,}
  \end{array}
\right.
\end{align*}
and the analog conditions hold for $\BestB{\cpi_u}$.\footnotemark
\end{proposition}
\footnotetext{Recall that for a (possible partial) transcript $u$, $\cpi_u$ is the protocol $\cpi$, conditioned that $u_1, \ldots, u_{\size{u}}$ were the first $\size{u}$ messages.}

The following holds true for any  (bit value) protocol.
\begin{proposition}\label{lemma:perfectAdv}
Let $\cpi=\HonAHonB$ be a protocol with $\Val(\cpi) \in [0,1]$. Then either  $\BestA{\cpi}$ or $\BestB{\cpi}$ (but not both) is equal to $1$.
\end{proposition}

The somewhat surprising part is that \emph{only} one party has a valid winning strategy. Assume for simplicity that $\BestA{\cpi}=1$. Since $\HonA$ might accidentally mimic the optimal winning valid attacker, it follows that for any valid strategy $\HonB'$ for $\HonB$ there is a positive probability over the random choices of the honest $\HonA$ that the outcome is \emph{not} zero. Namely, it holds that $\BestB{\cpi}<1$. The formal proof follows a straightforward induction on the protocol's round complexity.
\begin{proof}[Proof of \cref{lemma:perfectAdv}]
The proof is by induction on the round complexity of $\cpi$. Assume that $\round(\cpi) = 0$ and let $\ell$ be the only node in $\Tree(\cpi)$. If $\Color_\cpi(\ell)=1$, the proof follows since $\BestA{\cpi}=1$ and $\BestB{\cpi}=0$. In the complementary case, \ie $\Color_\pi(\ell)=0$, the proof follows since $\BestA{\cpi}=0$ and $\BestB{\cpi}=1$.

Assume that the lemma holds for $\rnd$-round protocols and that $\round(\cpi) = \rnd+1$. If $\EdgeDist_{\cpi}(\EmptyString,b)=1$\footnote{Recall that $\EmptyString$ is the string representation of the root of $\Tree(\cpi)$.} for some $b\in\zo$, since $\cpi$ is a protocol, it holds that $\EdgeDist_{\cpi}(\EmptyString,1-b)=0$. Hence, by \cref{fact:perfectAdv} it holds that $\BestA{\cpi}=\BestA{\cpi_b}$ and $\BestB{\cpi}=\BestB{\cpi_b}$, regardless of the party controlling $\Root(\cpi)$. The proof follows from the induction hypothesis.

If $\EdgeDist_{\cpi}(\EmptyString,b)\notin\zo$ for both $b\in\zo$, the proof splits according to the following complementary cases:
\begin{description}
        \item[$\BestB{\cpi_0}<1$ and $\BestB{\cpi_1}<1$.] The induction hypothesis yields that $\BestA{\cpi_0}=1$ and $\BestA{\cpi_1}=1$. \cref{fact:perfectAdv} now yields that $\BestB{\cpi}<1$ and $\BestA{\cpi}=1$, regardless of the party controlling $\Root(\cpi)$.
        \item[$\BestB{\cpi_0}=1$ and $\BestB{\cpi_1}=1$.] The induction hypothesis yields that $\BestA{\cpi_0}<1$ and $\BestA{\cpi_1}<1$. \cref{fact:perfectAdv} now yields that $\BestB{\cpi}=1$ and $\BestA{\cpi}<1$, regardless of the party controlling $\Root(\cpi)$.
        \item[$\BestB{\cpi_0}=1$ and $\BestB{\cpi_1}<1$.] The induction hypothesis yields that $\BestA{\cpi_0}<1$ and $\BestA{\cpi_1}=1$. If $\HonA$ controls $\Root(\cpi)$, \cref{fact:perfectAdv} yields that $\BestA{\cpi}=1$ and $\BestB{\cpi}<1$. If $\HonB$ controls $\Root(\cpi)$, \cref{fact:perfectAdv} yields that $\BestA{\cpi}<1$ and $\BestB{\cpi}=1$. Hence, the proof follows.
        \item $\BestB{\cpi_0}<1$ and $\BestB{\cpi_1}=1$. The proof follows arguments similar to the previous case.
\end{description}
\end{proof}

In the next sections we show the connection between the optimal valid attack and recursive biased-continuation attacks, by connecting them both to a specific measure over the protocol's leaves, called here the ``dominated measure'' of a protocol.

\subsection{Dominated Measures}\label{sec:DomMeas}

Let $\cpi=\HonAHonB$ be a protocol with $\BestA{\cpi}=1$ (and thus, by \cref{fact:perfectAdv}, $\BestB{\cpi}<1$). In such a protocol, the optimal attacker for $\HonA$ always has a winning strategy, regardless of $\HonB$'s strategy  (honest or not). Our goal is to define a measure $\ADomMeas{\cpi}\colon\Leaves(\cpi)\to [0,1]$ that will capture the ``$1-\BestB{\cpi}$'' advantage that party $\HonA$ has over party $\HonB$. Specifically, we would like that $\eex{\LDist{\cpi}}{\ADomMeas{\cpi}}=1-\BestB{\cpi}$.

Recall that $\BestB{\cpi}$ is the expected outcome of the protocol $(\HonA,\HonB')$, where $\HonB'$ is the optimal attacker for $\HonB$. To achieve our goal, $\ADomMeas{\cpi}$ must ``behave'' similarly to the expected outcome of $(\HonA,\HonB')$. Naturally, such measure will be defined recursively. On $\HonA$-controlled nodes, its expected value (over a choice of a random leaf in the original protocol $\cpi$) should be the weighted average of the expected values of the lower-level measures --- similarly to the expected outcome of $(\HonA,\HonB')$ which is the weighted average of the expected outcomes of the sub-protocols. On $\HonB$-controlled nodes, the situation is trickier. $\HonB'$ chooses to send the message that minimizes the expected outcome of $(\HonA,\HonB')$. Assuming that the lower-level measures already behave like the expected outcome of $(\HonA,\HonB')$, $\HonB'$ actually choose the message for which the expected value of the lower-level measure is smaller. But, the expected value of $\ADomMeas{\cpi}$ remains the weighted average of the expected values of the lower-level measures. To fix this we lower the value of the lower-level measure whose expected outcome is larger, so that the expected value of both lower-level measures is equal.
The above discussion leads to the following measure over the protocol's leaves.

\begin{definition}[dominated measures]\label{def:Dominated}
The {\sf $\HonA$-dominated measure} of protocol $\cpi = \HonAHonB$, denoted  $\ADomMeas{\cpi}$, is a measure over $\Leaves(\cpi)$ defined by $\ADomMeas{\cpi}(\ell) = \Color_\cpi(\ell)$ if $\round(\cpi) =0$, and otherwise recursively defined by:
\begin{align*}
\ADomMeas{\cpi}(\ell) = \left\{
  \begin{array}{ll}
    0, & \EdgeDist_{\cpi}(\EmptyString,\ell_1)=0; \footnotemark\\
    \ADomMeas{\cpi_{\ell_1}}(\ell_{2,\ldots,\size{\ell}}), & \EdgeDist_{\cpi}(\EmptyString,\ell_1)=1;\\
  \ADomMeas{\cpi_{\ell_1}}(\ell_{2,\ldots,\size{\ell}}), &  \EdgeDist_{\cpi}(\EmptyString,\ell_1)\notin \zo \land  (\hbox{$\HonA$ controls $\Root(\cpi)$} \lor  \Smaller{\cpi}{\ell_1});\\
    \frac{\eex{\LDist{\cpi_{1-\ell_1}}}{\ADomMeas{\cpi_{1-\ell_1}}}}{\eex{\LDist{\cpi_{\ell_1}}}{\ADomMeas{\cpi_{\ell_1}}}} \cdot \ADomMeas{\cpi_{\ell_1}}(\ell_{2,\ldots,\size{\ell}}) , & \hbox{otherwise,}
  \end{array}
\right.
\end{align*}
\footnotetext{Recall that for transcript $\ell$, $\ell_1$ stands for the first messages sent in $\ell$.}
where $\Smaller{\cpi}{\ell_1}=1$  if $\eex{\LDist{\cpi_{\ell_1}}}{\ADomMeas{\cpi_{\ell_1}}} \leq \eex{\LDist{\cpi_{1-\ell_1}}}{\ADomMeas{\cpi_{1-\ell_1}}}$. Finally, we let $\ADomMeas{\perp}$ be the zero measure.

The {\sf $\HonB$-dominated measure} of protocol $\cpi$, denoted $\BDomMeas{\cpi}$, is analogously defined, except that $\BDomMeas{\cpi}(\ell) = 1-\Color_\cpi(\ell)$ if $\round(\cpi) =0$.
\end{definition}

\begin{example}\label{exmp:2hon}
	($\HonA$-dominated measure)

	Before continuing with the formal proof, we believe the reader might find the following concrete example useful. Let $\cpi=\HonAHonB$ be the protocol described in \cref{fig:HonPro} and assume for the sake of this example that $\alpha_0 < \alpha_1$. The $\HonA$-dominated measures of $\cpi$ and its subprotocols are given in \cref{fig:CalcDM}.

	We would like to highlight some points regarding  the calculations of the $\HonA$-dominated measures.
	The first  point we note is that $\ADomMeas{\cpi_{011}}(011)=1$ but $\ADomMeas{\cpi_{01}}(011)=0$. Namely, the $\HonA$-dominated measure of the subprotocol $\cpi_{011}$ assigns the leaf represented by the string $011$ with the value $1$, while the $\HonA$-dominated measure of the subprotocol $\cpi_{01}$ (for which $\cpi_{011}$ is a subprotocol) assigns the same leaf with the value $0$.
	This follows since $\eex{\LDist{\cpi_{010}}}{\ADomMeas{\cpi_{010}}}=0$ and $\eex{\LDist{\cpi_{011}}}{\ADomMeas{\cpi_{011}}}=1$, which yield that $\Smaller{\cpi_{01}}{1}=0$ (recall that $\Smaller{\cpi'}{b}=0$ iff the expected value of the $\HonA$-dominated measure of $\cpi'_b$ is larger than that of the $\HonA$-dominated measure of $\cpi'_{1-b}$). Hence, \cref{def:Dominated} \wrt $\cpi_{01}$ now yields that
	\begin{align*}
	\ADomMeas{\cpi_{01}}(011) &= \frac{\eex{\LDist{\cpi_{010}}}{\ADomMeas{\cpi_{010}}}}{\eex{\LDist{\cpi_{011}}}{\ADomMeas{\cpi_{011}}}} \cdot \ADomMeas{\cpi_{011}}(011) \\
	&= \frac{0}{1} \cdot 1 = 0.
	\end{align*}

	The second point we note is that $\ADomMeas{\cpi_{1}}(10)=1$ but $\ADomMeas{\cpi}(10)=\frac{\alpha_0}{\alpha_1}$ (recall that we assumed that $\alpha_0<\alpha_1$, so $\frac{\alpha_0}{\alpha_1}< 1$). This follows similar arguments to the previous point; it holds that $\eex{\LDist{\cpi_{0}}}{\ADomMeas{\cpi_{0}}}=\alpha_0$ and $\eex{\LDist{\cpi_{1}}}{\ADomMeas{\cpi_{1}}}=\alpha_1$, which yield that $\Smaller{\cpi}{1}=0$ (since $\alpha_0<\alpha_1$). \cref{def:Dominated} \wrt $\cpi$ now yields that
	\begin{align*}
	\ADomMeas{\cpi}(10) &= \frac{\eex{\LDist{\cpi_{0}}}{\ADomMeas{\cpi_{0}}}}{\eex{\LDist{\cpi_{1}}}{\ADomMeas{\cpi_{1}}}} \cdot \ADomMeas{\cpi_{1}}(10) \\
	&= \frac{\alpha_0}{\alpha_1} \cdot 1 = \frac{\alpha_0}{\alpha_1}.
	\end{align*}

	The third and final point we note is that $\eex{\LDist{\cpi}}{\ADomMeas{\cpi}} = 1 - \BestB{\cpi}$. By the assumption that $\alpha_0<\alpha_1$, it holds that $\BestB{\cpi}=1-\alpha_0$. Independently, let us calculate the expected value of the $\HonA$-dominated measure. Since $\Supp\paren{\ADomMeas{\cpi}} = \set{00,01}$, it holds that
	\begin{align*}
	\eex{\LDist{\cpi}}{\ADomMeas{\cpi}} &= \VerticesDist_\cpi(00)\cdot \ADomMeas{\cpi}(00) + \VerticesDist_\cpi(10)\cdot \ADomMeas{\cpi}(10) \\
	&= \beta \cdot \alpha_0 \cdot 1 + (1-\beta)\cdot \alpha_1 \cdot \frac{\alpha_0}{\alpha_1} \\
	&= \alpha_0.
	\end{align*}
	Hence, $\eex{\LDist{\cpi}}{\ADomMeas{\cpi}} = 1 - \BestB{\cpi}$.

	\begin{figure}
		\centering
		\begin{subfigure}{.48\textwidth}
			\centering
			\begin{tikzpicture}[->,>=stealth',level/.style={sibling distance = 5cm/#1, level distance = 1.5cm}]
			\begin{scope}
			\node [B-node] (B) {$\HonB$}
			child{ node [A-node] (A0) {$\HonA$}
				child{ node [leaf] (L00) {$1$} edge from parent node[above left] {$\alpha_0$}}
				child{ node [B-node] (B01) {$\HonB$}
					child{ node [leaf] (L010) {$0$} edge from parent node[above left] {$\beta_{01}$}}
					child{ node [leaf,thick] (L011) {$1$} edge from parent node[above right] {$1-\beta_{01}$}}
					edge from parent node[above right] {$1- \alpha_0$}
				}
				edge from parent node[above left] {$\beta$}
			}
			child{ node [A-node] (A1) {$\HonA$}
				child{ node [leaf] (L10) {$1$} edge from parent node[above left] {$\alpha_1$}}
				child{ node [leaf] (L11) {$0$} edge from parent node[above right] {$1-\alpha_1$}
				}
				edge from parent node[above right] {$1-\beta$}
			}
			;
			\end{scope}
			\end{tikzpicture}
			\caption{Protocol $\cpi=\HonAHonB$. The label of an internal node denotes the name of the party controlling it, and that of a leaf denotes its value. The label on an edge leaving a node $u$ to node $u'$ denotes the probability that a random execution of $\cpi$ visits $u'$ once in $u$. Finally, all nodes are represented as strings from the root of $\cpi$, even when considering subprotocols (\eg the string representations of the leaf with the thick borders is $011$).}
			\label{fig:HonPro}
		\end{subfigure}%
		\hfill
		\begin{subfigure}{.48\textwidth}
			\centering
			\renewcommand{\arraystretch}{1.2}
			\begin{tabular}{c|c|c|c|c|c|}
				\cline{2-6} & \multicolumn{5}{ c| }{Leaves} \\
				\cline{1-6}  \multicolumn{1}{ |c| }{measures} & 00 & 010 & 011 & 10 & 11 \\
				\cline{1-6} \multicolumn{1}{ |c| }{$\ADomMeas{\cpi_{00}}$} & 1 &  &  &  &  \\
				\cline{1-6} \multicolumn{1}{ |c| }{$\ADomMeas{\cpi_{010}}$} &  & 0 &  &  &  \\
				\cline{1-6} \multicolumn{1}{ |c| }{$\ADomMeas{\cpi_{011}}$} &  &  & 1 & &  \\
				\cline{1-6} \multicolumn{1}{ |c| }{$\ADomMeas{\cpi_{01}}$} &  & 0 & 0 & &  \\
				\cline{1-6} \multicolumn{1}{ |c| }{$\ADomMeas{\cpi_{0}}$} & 1 & 0 & 0 & &  \\
				\cline{1-6} \multicolumn{1}{ |c| }{$\ADomMeas{\cpi_{10}}$} &  &  &  & 1 &  \\
				\cline{1-6} \multicolumn{1}{ |c| }{$\ADomMeas{\cpi_{11}}$} & &  &  &  & 0 \\
				\cline{1-6} \multicolumn{1}{ |c| }{$\ADomMeas{\cpi_{1}}$} & & & & 1 & 0 \\
				\cline{1-6} \multicolumn{1}{ |c| }{$\ADomMeas{\cpi}$} & 1 & 0 & 0 & $\alpha_0/\alpha_1$ & 0 \\
				\cline{1-6}
			\end{tabular}
			\caption{Calculating the $\HonA$-dominated measure of $\cpi$. The $\HonA$-dominated measure of a subprotocol $\cpi_u$, is only defined over the leaves in the subtree $\Tree\paren{\cpi_u}$.}
			\label{fig:CalcDM}
		\end{subfigure}
		\caption{An example of a (coin-flipping) protocol is given on the left, and an example of how to calculate its $\HonA$-dominated measure is given on the right.}
		\label{fig:exmpl1}
	\end{figure}
\end{example}

Note that the $\HonA$-dominated measure is  \emph{$\HonB$-immune}---if  $\Bc$ controls a node $u$, the expected value of the measure is that of the lowest measure of the subprotocols $\cpi_{u0}$ and $\cpi_{u1}$. Where if $\Ac$ controls a node $u$, the expected value of the $\Ac$-dominated  measure is the weighted average of the measures of the same subprotocols (according to the edge distribution). In both cases, the $\Ac$-dominated measure indeed ``captures" the  behavior of the optimal attacker for $\Bc$. This  observation is formally stated as   the following lemma:
\begin{lemma}\label{lemma:ExpDomMeas}
Let $\cpi=\HonAHonB$ be a protocol and let $\ADomMeas{\cpi}$ be its $\HonA$-dominated measure. Then
$\BestB{\cpi}=1-\eex{\LDist\cpi}{\ADomMeas{\cpi}}$.
\end{lemma}
In particular, since $\BestA{\cpi}=1$ iff $\BestB{\cpi}<1$ (\cref{fact:perfectAdv}), it holds that $\BestA{\cpi}=1$ iff $\eex{\LDist\cpi}{\ADomMeas{\cpi}}>0$.

Towards proving \cref{lemma:ExpDomMeas}, we first note that the definition of $\ADomMeas{\cpi}$ ensures three important properties.
\begin{proposition}\label{prop:DomMeasPro}
Let $\cpi$ be a protocol with $\EdgeDist_\cpi(\EmptyString,b)\notin\zo$ for both $b\in\zo$. Then
\begin{enumerate}
  \item\label{item:Amaximal} ($\HonA$-maximal) $\HonA$ controls $\Root(\cpi) \implies$ $\Restrict{\ADomMeas{\cpi}}{b} \equiv \ADomMeas{\cpi_b}$ for both  $b\in\zo$.\footnote{Recall that for a measure $\Meas\colon \Leaves(\cpi)\to [0,1]$ and a bit $b$, $\Restrict{\Meas}{b}$ is the measure induced by $\Meas$ when restricted to $\Leaves(\cpi_b)\subseteq\Leaves(\cpi)$.}
  \item\label{item:Bminimal} ($\HonB$-minimal) $\HonB$ controls $\Root(\cpi) \implies \Restrict{\ADomMeas{\cpi}}{b} \equiv \left\{
                                                 \begin{array}{ll}
                                                   \ADomMeas{\cpi_b}, & \hbox{$\Smaller{\cpi}{b}=1$;} \\
                                                   \frac{\eex{\LDist{\cpi_{1-b}}}{\ADomMeas{\cpi_{1-b}}}}{\eex{\LDist{\cpi_{b}}}{\ADomMeas{\cpi_{b}}}} \cdot \ADomMeas{\cpi_{b}}, & \hbox{otherwise.}
                                                 \end{array}
                                               \right.$
  \item\label{item:Bimmune} ($\HonB$-immune) $\HonB$ controls $\Root(\cpi) \implies \eex{\LDist{\cpi_0}}{\Restrict{\ADomMeas{\cpi}}{0}} = \eex{\LDist{\cpi_1}}{\Restrict{\ADomMeas{\cpi}}{1}}$.
\end{enumerate}
\end{proposition}

Namely, if $\HonA$ controls $\Root(\cpi)$, the \emph{$\HonA$-maximal} property of $\ADomMeas{\cpi}$ (the $\HonA$-dominated measure of $\cpi$) ensures that the restrictions of this measure to the subprotocols of $\cpi$ are the $\HonA$-dominated measures of these subprotocols. In the complementary case, \ie $\HonB$ controls $\Root(\cpi)$, the \emph{$\HonB$-minimal} property of $\ADomMeas{\cpi}$ ensures that for at least one subprotocol of $\cpi$, the restriction of this measure to this subprotocol is equal to the $\HonA$-dominated measure of the subprotocol. Finally, the \emph{$\HonB$-immune} property of $\ADomMeas{\cpi}$ ensures that the expected values of the measures derived by restricting $\ADomMeas{\cpi}$ to the subprotocols of $\cpi$ are equal (and hence, they are also equal to the expected value of $\ADomMeas{\cpi}$).

\begin{proof}[Proof of \cref{prop:DomMeasPro}]
The proof of \cref{item:Amaximal,item:Bminimal} ($\HonA$-maximal and $\HonB$-minimal) immediately follows from \cref{def:Dominated}.

Towards proving \cref{item:Bimmune} ($\HonB$-immune), we will assume that $\HonB$ controls $\Root(\cpi)$. If $\Smaller{\cpi}{0}=\Smaller{\cpi}{1}=1$, the proof again follows immediately from  \cref{def:Dominated}. In the complementary case, \ie $\Smaller{\cpi}{b}=0$ and $\Smaller{\cpi}{1-b}=1$ for some $b\in\zo$, it holds that
\begin{align*}
\eex{\LDist{\cpi_b}}{\Restrict{\ADomMeas{\cpi}}{b}} &= \eex{\LDist{\cpi_b}}{\frac{\eex{\LDist{\cpi_{1-b}}}{\ADomMeas{\cpi_{1-b}}}}{\eex{\LDist{\cpi_{b}}}{\ADomMeas{\cpi_{b}}}} \cdot \ADomMeas{\cpi_{b}}}\\
&= \frac{\eex{\LDist{\cpi_{1-b}}}{\ADomMeas{\cpi_{1-b}}}}{\eex{\LDist{\cpi_{b}}}{\ADomMeas{\cpi_{b}}}} \cdot \eex{\LDist{\cpi_{b}}}{\ADomMeas{\cpi_{b}}}\\
&= \eex{\LDist{\cpi_{1-b}}}{\ADomMeas{\cpi_{1-b}}}\\
&= \eex{\LDist{\cpi_{1-b}}}{\Restrict{\ADomMeas{\cpi}}{1-b}},
\end{align*}
where the first and last equalities follow from the $\HonB$-minimal property of $\ADomMeas{\cpi}$ (\cref{item:Bminimal}).
\end{proof}

We are now ready to prove \cref{lemma:ExpDomMeas}.
\begin{proof}[Proof of \cref{lemma:ExpDomMeas}]
The proof is by  induction on the round complexity of $\cpi$.

Assume that $\round(\cpi) = 0$ and let $\ell$ be the only node in $\Tree(\cpi)$. If $\Color_\cpi(\ell)=1$, then by \cref{def:Dominated} it holds that
$\ADomMeas{\cpi}(\ell)=1$, implying that $\eex{\LDist\cpi}{\ADomMeas{\cpi}}=1$. The proof follows since in this case, by \cref{lemma:perfectAdv}, $\BestB{\cpi}=0$. In the complementary case, \ie $\LeafValue(\ell)=0$, by \cref{def:Dominated} it holds that $\ADomMeas{\cpi}(\ell)=0$, implying that $\eex{\LDist\cpi}{\ADomMeas{\cpi}}=0$. The proof follows since in this case, by \cref{lemma:perfectAdv}, $\BestB{\cpi}=1$.

Assume that the lemma holds for $\rnd$-round protocols and that $\round(\cpi) = \rnd+1$.  For $b\in \zo$ let $\alpha_b \eqdef \eex{\LDist{\cpi_b}}{\DomMeas{\cpi_b}{\HonA}}$. The induction hypothesis yields that $\BestB{\cpi_b}=1-\alpha_b$ for both $b\in\zo$. If $\EdgeDist_{\cpi}(\EmptyString,b)=1$ for some $b\in\zo$ (which also means that $\EdgeDist_{\cpi}(\EmptyString,1-b)=0$), the proof follows since \cref{fact:perfectAdv} yields that $\BestB{\cpi}=\BestB{\cpi_b}=1-\alpha_b$, where \cref{def:Dominated} yields that $\eex{\LDist\cpi}{\ADomMeas{\cpi}}=\eex{\LDist{\cpi_b}}{\ADomMeas{\cpi_b}}=\alpha_b$.

Assume $\EdgeDist_{\cpi}(\EmptyString,b)\notin\zo$ for both $b\in\zo$ and let $p\eqdef\EdgeDist_\cpi(\EmptyString,0)$. The proof splits according to who controls the root of $\cpi$.

    \begin{description}
    \item[$\HonA$ controls $\Root(\cpi)$.] \cref{def:Dominated} yields  that
    \begin{align*}
    \eex{\LDist\cpi}{\ADomMeas{\cpi}} &= p\cdot \eex{\LDist{\cpi_0}}{\Restrict{\ADomMeas{{\cpi}}}{0}} + (1-p)\cdot \eex{\LDist{\cpi_1}}{\Restrict{\ADomMeas{\cpi}}{1}}\\
    &= p\cdot \eex{\LDist{\cpi_0}}{\ADomMeas{{\cpi_0}}} + (1-p)\cdot  \eex{\LDist{\cpi_1}}{\ADomMeas{{\cpi_1}}}\\ 
    &=p\cdot \alpha_0 + (1-p)\cdot \alpha_1,
    \end{align*}
    where the second equality follows from the $\HonA$-maximal property of $\ADomMeas{\cpi_b}$ (\cref{prop:DomMeasPro}(\ref{item:Amaximal})). Using \cref{fact:perfectAdv} we conclude that
    \begin{align*}
    \BestB{\cpi} &= p\cdot \BestB{\cpi_0} + (1-p)\cdot \BestB{\cpi_1} \\
    &= p \cdot (1-\alpha_0) + (1-p)\cdot (1-\alpha_1) \\
    &= 1-(p\cdot \alpha_0 + (1-p)\cdot \alpha_1) \\
    & = 1 - \eex{\LDist\cpi}{\ADomMeas{\cpi}}.
    \end{align*}

    \item[$\HonB$ controls $\Root(\cpi)$.] We assume that $\alpha_0\leq\alpha_1$ (the complementary case is analogous). \cref{fact:perfectAdv} and the induction hypothesis yield that $\BESTB=1-\alpha_0$.  Hence, it is left to show that $\eex{\LDist\cpi}{\ADomMeas{\cpi}}=\alpha_0$. The assumption that $\alpha_0\leq\alpha_1$ yields that $\Smaller{\cpi}{0}=1$. Thus, by the $\HonB$-minimal property of $\ADomMeas{\cpi}$ (\cref{prop:DomMeasPro}(\ref{item:Bminimal})), it holds that $\Restrict{\ADomMeas{\cpi}}{0}\equiv\ADomMeas{\cpi_0}$. It follows that $\eex{\LDist{\cpi_0}}{\Restrict{\ADomMeas{\cpi}}{0}}=\alpha_0$, and the $\HonB$-immune property of $\ADomMeas{\cpi}$ (\cref{prop:DomMeasPro}(\ref{item:Bimmune})) yields that $\eex{\LDist{\cpi_1}}{\Restrict{\ADomMeas{\cpi}}{1}}=\alpha_0$. To conclude the proof, we compute
        \begin{align*}
        \eex{\LDist\cpi}{\ADomMeas{\cpi}}&= p\cdot \eex{\LDist{\cpi_0}}{\Restrict{\ADomMeas{\cpi}}{0}} + (1-p)\cdot\eex{\LDist{\cpi_1}}{\Restrict{\ADomMeas{\cpi}}{1}}\\
        &=p\cdot \alpha_0 + (1-p)\cdot\alpha_0 \\
        & = \alpha_0.
        \end{align*}
    \end{description}
\end{proof}

\cref{lemma:ExpDomMeas} connects the success of the optimal attack to the expected value of the dominated measure. In the next section we analyze the success of the recursive biased-continuation attack using this expected value.  Unfortunately, this analysis does not seem to suffice for our goal. In \cref{sec:AltDomMeas} we generalize the dominated measure described above to a sequence of (alternating) dominated measures, where in \cref{sec:RCAltDomMeas} we use this new notion to prove that the recursive biased continuation is indeed a good attack.

\subsection{Warmup ---  Proof Attempt Using a (Single) Dominated Measure}\label{sec:SingleDomMeas}

As mentioned above, the  approach described in this section falls too short to serve our goals. Yet we describe it here as a detailed overview for the more complicated proof, given in following sections (\wrt a sequence of dominated measures). Specifically, we sketch a proof of the following lemma, which
relates the performance of the recursive biased-continuation attacker playing the role of $\HonA$, to the performance of the optimal (valid) attacker playing the role of $\HonB$. The proof (see below) is via the $\HonA$-dominated measure of $\cpi$ defined above.\footnote{The formal proof of \cref{lemma:SingleMeasManyIter} follows its stronger variant, \cref{lemma:IdealMainLemmaSim}, introduced in \cref{sec:RCAltDomMeas}.}
\begin{lemma}\label{lemma:SingleMeasManyIter}
Let $\cpi=\HonAHonB$ be a protocol with $\Val(\cpi)>0$, let $k\in\N$ and let $\rcA{k}$ be according to \cref{alg:adversary}. Then
\begin{align*}
\Val(\ProArc{k})\geq \frac{1- \BestB{\cpi}}{\prod_{i=0}^{k-1}\Val(\ProArc{i})}.
\end{align*}
\end{lemma}
The proof of the above lemma is a direct implication of the next lemma.
\begin{lemma}\label{lemma:SingleMeasManyIterDet}
Let $\cpi=\HonAHonB$ be a protocol with $\Val(\cpi)>0$, let $k\in\N$ and let $\rcA{k}$ be according to \cref{alg:adversary}. Then
\begin{align*}
\eex{\LDist{\ProArc{k}}}{\ADomMeas{\cpi}}\geq \frac{\eex{\LDist{\cpi}}{\ADomMeas{\cpi}}}{\prod_{i=0}^{k-1}\Val(\ProArc{i})}.
\end{align*}
\end{lemma}
\begin{proof}[Proof of \cref{lemma:SingleMeasManyIter}]
Immediately follows \cref{lemma:SingleMeasManyIterDet,lemma:ExpDomMeas,fact:ValueExp} (we can use \cref{fact:ValueExp} since by \cref{def:Dominated}, $\ADomMeas{\cpi}(\ell)=0$ for every $\ell\in\Leaves_0(\cpi)$).
\end{proof}

We begin by sketching the proof of the following lemma, which is a special case of \cref{lemma:SingleMeasManyIterDet}. Later we explain how to generalize the proof below to derive \cref{lemma:SingleMeasManyIterDet}.
\begin{lemma}\label{lemma:claim1}
Let $\cpi=\HonAHonB$ be a protocol with $\Val(\cpi)>0$ and let $\rcA{1}$ be according to \cref{alg:adversary}. Then $\eex{\LDist{\ProArc{1}}}{\ADomMeas{\cpi}}\geq \frac{\eex{\LDist{\cpi}}{\ADomMeas{\cpi}}}{\Val(\cpi)}$.
\end{lemma}
\begin{proofsketch}
The proof is by induction on the round complexity of $\cpi$. The base case (\ie $\round(\cpi)=0$) is straightforward. Assume that the lemma holds for $\rnd$-round protocols and that $\round(\cpi) = \rnd+1$. For $b\in\zo$ let $\alpha_b \eqdef \eex{\LDist{\cpi_b}}{\ADomMeas{\cpi_b}}$ and let $p\eqdef\EdgeDist_{\cpi}(\EmptyString,0)$.

If $\Root(\cpi)$ is controlled by $\HonA$, the $\HonA$-maximal property of $\ADomMeas{\cpi}$ (\cref{prop:DomMeasPro}(\ref{item:Amaximal})) yields that $\eex{\LDist{\cpi}}{\ADomMeas{\cpi}}=p\cdot\alpha_0 + (1-p)\cdot \alpha_1$. It holds that
\begin{align}\label{eq:sketch1}
\eex{\LDist{\ProArc{1}}}{\ADomMeas{\cpi}} &= \EdgeDist_\ProArcP{1}(\EmptyString,0)\cdot \eex{\LDist{\ProArcP{1}_0}}{\Restrict{\ADomMeas{\cpi}}{0}} + \EdgeDist_\ProArcP{1}(\EmptyString,1)\cdot \eex{\LDist{\ProArcP{1}_1}}{\Restrict{\ADomMeas{\cpi}}{1}} \\
&= p\cdot \frac{\Val(\cpi_0)}{\Val(\cpi)} \cdot \eex{\LDist{\ProArcP{1}_0}}{\Restrict{\ADomMeas{\cpi}}{0}} + (1-p)\cdot \frac{\Val(\cpi_1)}{\Val(\cpi)} \cdot \eex{\LDist{\ProArcP{1}_1}}{\Restrict{\ADomMeas{\cpi}}{1}}, \nonumber
\end{align}
where the second equality follows from \cref{claim:weights}. Since $\rcA{1}$ is stateless (\cref{fact:switchOrder}), we can write \cref{eq:sketch1} as
\begin{align}\label{eq:sketch2}
\eex{\LDist{\ProArc{1}}}{\ADomMeas{\cpi}} &= p\cdot \frac{\Val(\cpi_0)}{\Val(\cpi)} \cdot \eex{\LDist{\ProArcb{1}{0}}}{\Restrict{\ADomMeas{\cpi}}{0}} + (1-p)\cdot \frac{\Val(\cpi_1)}{\Val(\cpi)} \cdot \eex{\LDist{\ProArcb{1}{1}}}{\Restrict{\ADomMeas{\cpi}}{1}}.
\end{align}
The $\HonA$-maximal property of $\ADomMeas{\cpi}$ and \cref{eq:sketch2} yield that
\begin{align}\label{eq:sketch3}
\eex{\LDist{\ProArc{1}}}{\ADomMeas{\cpi}} &= p\cdot \frac{\Val(\cpi_0)}{\Val(\cpi)} \cdot \eex{\LDist{\ProArcb{1}{0}}}{\ADomMeas{\cpi_0}} + (1-p)\cdot \frac{\Val(\cpi_1)}{\Val(\cpi)} \cdot \eex{\LDist{\ProArcb{1}{1}}}{\ADomMeas{\cpi_1}}.
\end{align}
Applying the induction hypothesis on the right-hand side of \cref{eq:sketch3} yields that
\begin{align*}
\eex{\LDist{\ProArc{1}}}{\ADomMeas{\cpi}} &\geq p\cdot \frac{\Val(\cpi_0)}{\Val(\cpi)} \cdot \frac{\alpha_0}{\Val(\cpi_0)} + (1-p)\cdot \frac{\Val(\cpi_1)}{\Val(\cpi)} \cdot \frac{\alpha_1}{\Val(\cpi_1)} \\
&=  \frac{p\cdot \alpha_0+(1-p)\cdot\alpha_1}{\Val(\cpi)}\\
&= \frac{\eex{\LDist{\cpi}}{\ADomMeas{\cpi}}}{\Val(\cpi)},
\end{align*}
which concludes the proof for the case that $\HonA$ controls $\Root(\cpi)$.

If $\Root(\cpi)$ is controlled by $\HonB$, and assuming that $\alpha_0\leq\alpha_1$ (the complementary case is analogous), it holds that $\Smaller{\cpi}{0}=1$. Thus, by the $\HonB$-minimal property of $\ADomMeas{\cpi}$ (\cref{prop:DomMeasPro}(\ref{item:Bminimal})), it holds that $\Restrict{\ADomMeas{\cpi}}{0}\equiv \ADomMeas{\cpi_0}$ and $\Restrict{\ADomMeas{\cpi}}{1}\equiv \frac{\alpha_0}{\alpha_1}\ADomMeas{\cpi_1}$. Hence, the $\HonB$-immune property of $\ADomMeas{\cpi}$ (\cref{prop:DomMeasPro}(\ref{item:Bimmune})) yields that $\eex{\LDist{\cpi}}{\ADomMeas{\cpi}}=\alpha_0$. In addition, since $\HonB$ controls $\Root(\cpi)$, the distribution of the edges $(\EmptyString,0)$ and $(\EmptyString,1)$ has not changed. It holds that
\begin{align}\label{eq:sketch4}
\eex{\LDist{\ProArc{1}}}{\ADomMeas{\cpi}} &= p\cdot \eex{\LDist{\ProArcP{1}_0}}{\Restrict{\ADomMeas{\cpi}}{0}} + (1-p) \cdot \eex{\LDist{\ProArcP{1}_1}}{\Restrict{\ADomMeas{\cpi}}{1}} \\
&\overset{(1)}{=} p \cdot \eex{\LDist{\ProArcb{1}{0}}}{\Restrict{\ADomMeas{\cpi}}{0}} + (1-p) \cdot \eex{\LDist{\ProArcb{1}{1}}}{\Restrict{\ADomMeas{\cpi}}{1}} \nonumber\\
&= p \cdot \eex{\LDist{\ProArcb{1}{0}}}{\ADomMeas{\cpi_0}} + (1-p) \cdot \eex{\LDist{\ProArcb{1}{1}}}{\frac{\alpha_0}{\alpha_1}\ADomMeas{\cpi_1}} \nonumber\\
&= p \cdot \eex{\LDist{\ProArcb{1}{0}}}{\ADomMeas{\cpi_0}} + (1-p)\cdot \frac{\alpha_0}{\alpha_1} \cdot \eex{\LDist{\ProArcb{1}{1}}}{\ADomMeas{\cpi_1}}, \nonumber
\end{align}
where (1) follows since $\rcA{1}$ is stateless (\cref{fact:switchOrder}). Applying the induction hypothesis on the right-hand side of \cref{eq:sketch4} yields that
\begin{align*}
\eex{\LDist{\ProArc{1}}}{\ADomMeas{\cpi}} &\geq p \cdot \frac{\alpha_0}{\Val(\cpi_0)} + (1-p)\cdot \frac{\alpha_0}{\alpha_1} \cdot \frac{\alpha_1}{\Val(\cpi_1)} \\
&= \alpha_0 \left( \frac{p}{\Val(\cpi_0)} + \frac{1-p}{\Val(\cpi_1)} \right)\\
&\overset{(1)}{\geq} \frac{\eex{\LDist{\cpi}}{\ADomMeas{\cpi}}}{\Val(\cpi)},
\end{align*}
which concludes the proof for the case that $\HonB$ controls $\Root(\cpi)$, and where (1) holds since
\begin{align}\label{eq:claim1sketch}
\frac{p}{\Val(\cpi_0)} + \frac{1-p}{\Val(\cpi_1)}\geq \frac{1}{\Val(\cpi)}.
\end{align}
\end{proofsketch}

The proof of \cref{lemma:SingleMeasManyIterDet} follows from similar arguments to those used above for proving \cref{lemma:claim1}.\footnote{The proof sketch given for \cref{lemma:claim1} is almost a formal proof, lacking only consideration of the base case and the extreme cases in which $\EdgeDist_\cpi(\EmptyString,b)=1$ for some $b\in\zo$.} Informally, we proved  \cref{lemma:claim1} by showing that $\rcA{1}$ ``assigns'' more weight to the dominated measure than $\HonA$ does. A natural step is to consider $\rcA{2}$ and to see if it assigns more weight to the dominated measure than $\rcA{1}$ does. It turns out that one can turn this intuitive argument into a formal proof, and prove \cref{lemma:SingleMeasManyIter} by  repeating this procedure \wrt many recursive biased-continuation attacks.\footnote{The main additional complication in  the proof of \cref{lemma:SingleMeasManyIter} is that the simple argument used to derive \cref{eq:claim1sketch} is replaced with the more general argument, described in \cref{lemma:calculus1}.}

\paragraph{The shortcoming of \cref{lemma:SingleMeasManyIter}.} Given a protocol $\cpi = \HonAHonB$, we are interested in the minimal  value of $\kappa$  for which $\rcA{\kappa}$ biases the value of the protocol towards one with  probability of at least $0.9$ (as a concrete example). Following \cref{lemma:SingleMeasManyIter}, it suffices to find a value $\kappa$ such that
\begin{align}
 \Val(\ProArc{\kappa}) \geq \frac{1-\BestB{\cpi}}{\prod_{i=0}^{\kappa-1}\Val(\ProArc{i})} \geq 0.9.
\end{align}
Using worst case analysis,  it suffices to find $\kappa$ such that  $(1-\BestB{\cpi})/(0.9)^{\kappa}\geq 0.9$, where the latter dictates that
\begin{align}\label{eq:warmup}
\kappa \geq \frac{\log\paren{\frac{1}{1-\BestB{\cpi}}}}{\log\paren{\frac{1}{0.9}}}.
\end{align}
Recall that our ultimate goal is to implement an \emph{efficient} attack on any coin-flipping protocol, under the mere assumption that one-way functions do not exist. Specifically, we would like to do so by giving an efficient version of the recursive biased-continuation attack. At the very least, due to the recursive nature of the attack, this requires the protocols $(\rcA{1},\HonB),\dots, (\rcA{\kappa-1},\HonB)$ be efficient in comparison to the basic protocol.  The latter  efficiency restriction  together with the recursive definition of $\rcA{\kappa}$ dictates that $\kappa$ (the number of recursion calls) be constant.

Unfortunately, \cref{eq:warmup} reveals that if $\BestB{\cpi}\in 1-o(1)$, we need to take $\kappa\in\omega(1)$, yielding an inefficient  attack.

\subsection{Back to the Proof --- Sequence of Alternating Dominated Measures}\label{sec:AltDomMeas}
Let $\cpi = \HonAHonB$ be a protocol and let $\Meas$ be a measure over the leaves of $\cpi$. Consider the variant of $\cpi$ whose parties act identically to the parties in $\cpi$, but  with the following tweak: when the execution  reaches a leaf $\ell$, the protocol restarts  with probability $\Meas(\ell)$. Namely, a random execution of the resulting (possibly inefficient) protocol is distributed like a random execution of $\cpi$, conditioned on not ``hitting'' the measure $\Meas$.\footnote{For concreteness, one might like to consider the case where $\Meas$ is a set.} The above is formally captured by the  definition below.

\subsubsection{Conditional Protocols}
\begin{definition}[conditional protocols]\label{def:CondProtocol}
Let $\cpi$ be an $\rnd$-round protocol and let $\Meas$ be a measure over $\Leaves(\cpi)$ with $\Exp_{\LDist{\cpi}}[\Meas]<1$. The $\rnd$-round $\Meas$-{\sf conditional protocol of} $\cpi$, denoted $\CondPro{\cpi}{\Meas}$, is defined by the color function $\Color_{\paren{\CondPro{\cpi}{\Meas}}} \equiv  \Color_\cpi$, and the edge distribution function $\EdgeDist_{\paren{\CondPro{\cpi}{\Meas}}}$ is defined by
\begin{align*}
\EdgeDist_{\paren{\CondPro{\cpi}{\Meas}}}(u,u b) =
\begin{cases} 0, & \Exp_{\LDist{\cpi_{u}}}[\Meas] = 1;\footnotemark  \\
\EdgeDist_\cpi(u,u b) \cdot \frac{1- \Exp_{\LDist{\cpi_{u b}}}[\Meas]}{1- \eex{\LDist{\cpi_{u}}}{\Meas}},& \mbox{otherwise}.
\end{cases},
\end{align*}
\footnotetext{Note that this case does not affect the resulting protocol, and is defined only to simplify future discussion.}
for every $u\in \Vertices(\cpi)\setminus\Leaves(\cpi)$ and $b\in\zo$. The controlling scheme of the protocol $\CondPro{\cpi}{\Meas}$ is the same as in $\cpi$.

If $\Exp_{\LDist{\cpi}}[\Meas]=1$ or $\cpi = \perp$,  we set $\CondPro{\cpi}{\Meas} = \perp$.
\end{definition}

\begin{example}
(Conditional Protocol)

Once again we consider the protocol $\cpi$ from \cref{fig:HonPro}. In \cref{fig:exmpl2} we present the conditional protocol $\cpi' = \CondPro{\cpi}{\ADomMeas{\cpi}}$, namely the protocol derived when protocol $\cpi$ is conditioned not to ``hit'' the $\HonA$-dominated measure of $\cpi$.
We would like to highlight some points regarding this conditional protocol.

The first  point we note is the changes in the edge distribution. Consider the root of $\cpi_0$ (\ie the node $0$). According to the calculations in \cref{fig:CalcDM}, it holds that $\eex{\LDist{\cpi_{00}}}{\ADomMeas{\cpi}} = \ADomMeas{\cpi}(00) = 1$ and that $\eex{\LDist{\cpi_{0}}}{\ADomMeas{\cpi}} = \alpha_0$. Hence, \cref{def:CondProtocol} yields that
\begin{align*}
\EdgeDist_{\paren{\CondPro{\cpi}{\ADomMeas{\cpi}}}}(0,00) &= \alpha_0 \cdot \frac{1-\eex{\LDist{\cpi_{00}}}{\ADomMeas{\cpi}}}{1-\eex{\LDist{\cpi_{0}}}{\ADomMeas{\cpi}}} \\
&=  \alpha_0 \cdot \frac{0}{1-\alpha_0} \\
&= 0.
\end{align*}
Note that the above change makes the leaf $00$ inaccessible in $\cpi'$. This occurs since $\ADomMeas{\cpi}(00)=1$. Similar calculations yield the changes in the distribution of the edges leaving the root of $\cpi_1$ (\ie the node $1$).

The second point we note is that the conditional protocol is in fact a protocol. Namely, for every node, the sum of the probabilities of the edges leaving it is one.
This is easily seen from \cref{fig:exmpl2}.

The third point we note is that the edge distribution of the root of $\cpi$ does not change at all. This follows from \cref{def:CondProtocol} and the fact that
\begin{align*}
\eex{\LDist{\cpi_0}}{\ADomMeas{\cpi}} = \eex{\LDist{\cpi_{1}}}{\ADomMeas{\cpi}} = \eex{\LDist{\cpi}}{\ADomMeas{\cpi}} = \alpha_0.
\end{align*}

The fourth point we note is that in the conditional protocol, an optimal valid attacker playing the role of $\HonB$ can bias the outcome towards zero with probability one. Namely, $\BestB{\CondPro{\cpi}{\ADomMeas{\cpi}}}=1$. Such an attacker will send $0$ as the first message, after which $\HonA$ must send $1$ as the next message, and then the attacker will send $0$. The outcome of this interaction is the value of the leaf $010$, which is $0$.

In the rest of the section we show that the above observations can actually be generalize to statements regarding any conditional protocol.
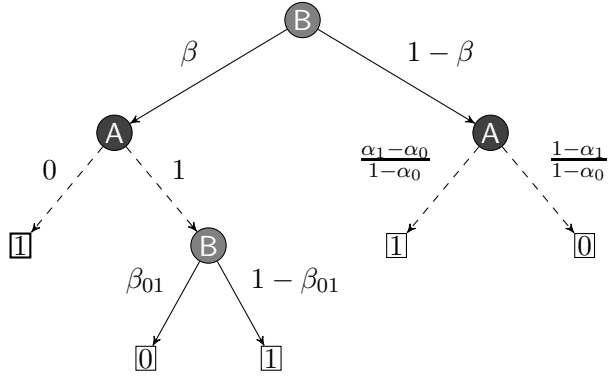
\begin{figure}
\centering
\begin{tikzpicture}[->,>=stealth',level/.style={sibling distance = 5cm/#1, level distance = 1.5cm,}]
\begin{scope}
\node [B-node] (B) {$\HonB$}
    child{ node [A-node] (A0) {$\HonA$}
        child{ node[leaf,thick] (L00) {$1$} edge from parent[dashed] node[above left] {$0$}}
        child{ node [B-node] (B01) {$\HonB$}
            child[solid]{ node [leaf] (L010) {$0$} edge from parent node[above left] {$\beta_{01}$}}
			child[solid]{ node [leaf] (L011) {$1$} edge from parent node[above right] {$1-\beta_{01}$}}
            edge from parent[dashed] node[above right] {$1$}
		}
        edge from parent node[above left] {$\beta$}
    }
    child{ node [A-node] (A1) {$\HonA$}
        child{ node [leaf] (L10) {$1$} edge from parent[dashed] node[above left] {$\frac{\alpha_1-\alpha_0}{1-\alpha_0}$}}
        child{ node [leaf] (L11) {$0$} edge from parent[dashed] node[above right] {$\frac{1-\alpha_1}{1-\alpha_0}$}
		}
        edge from parent node[above right] {$1-\beta$}
    }
;
\end{scope}

\end{tikzpicture}
\caption{The conditional protocol $\cpi' = \CondPro{\cpi}{\ADomMeas{\cpi}}$ of $\cpi$ from \cref{fig:HonPro}. Dashed edges are such that their distribution has changed. Note that due to this change, the leaf $00$ (the leftmost leaf, marked by a thick border) is \emph{inaccessible} in $\cpi'$. The $\HonB$-dominated measure of $\cpi'$ assigns a value of $1$ to the leaf $010$, and value of $0$ to all other leaves.}
\label{fig:exmpl2}
\end{figure}
\end{example}

The next proposition shows that the $\Meas$-conditional protocol is indeed a protocol. It also shows a relation between the leaf distribution of the $\Meas$-conditional protocol and the original protocol. Using this relation we conclude that the set of possible transcripts of the $\Meas$-conditional protocol is a subset the original protocol's possible transcripts and that if $\Meas$ assigns a value of $1$ to some transcript, then this transcript is inaccessible by the $\Meas$-conditional protocol.

\begin{proposition}\label{prop:CondPro}
Let $\cpi$ be a protocol and let $\Meas$ be a measure over $\Leaves(\cpi)$ with $\eex{\LDist{\cpi}}{\Meas}<1$. Then
\newcounter{eqn}
\renewcommand*{\theeqn}{\arabic{eqn}.}
\newcommand{\num}{\refstepcounter{eqn}\text{\theeqn}\quad}
\begin{fleqn}[1.3em]
\begin{alignat*}{3}
  \num& \forall u\in \Vertices(\cpi)\setminus\Leaves(\cpi)\colon \ \ &\VerticesDist_{\paren{\CondPro{\cpi}{\Meas}}}(u)>0 &\implies \EdgeDist_{\paren{\CondPro{\cpi}{\Meas}}}(u,u 0) + \EdgeDist_{\paren{\CondPro{\cpi}{\Meas}}}(u,u 1) =1; \\
  \num& \forall  \ell\in\Leaves(\cpi) \colon \ \ &\VerticesDist_{\paren{\CondPro{\cpi}{\Meas}}}(\ell)&\makebox[1cm][c]{$=$}\VerticesDist_{\cpi}(\ell)\cdot\frac{1-\Meas(\ell)}{1-\eex{\LDist{\cpi}}{\Meas}}; \\
  \num& \forall  \ell\in\Leaves(\cpi) \colon \ \ &\VerticesDist_{\paren{\CondPro{\cpi}{\Meas}}}(\ell)>0 &\implies \VerticesDist_{\cpi}(\ell)>0; \text{ and} \\
  \num& \forall \ell\in\Leaves(\cpi)\colon \ \  &\Meas(\ell)=1 &\implies \VerticesDist_{\paren{\CondPro{\cpi}{\Meas}}}(\ell)=0.
\end{alignat*}
\end{fleqn}
\end{proposition}
\begin{proof}
The first two items immediately follow from \cref{def:CondProtocol}. The last two items follow the second item.
\end{proof}
In addition to the above properties, \cref{def:CondProtocol} guarantees the following ``locality'' property of the $\Meas$-conditional protocol.
\begin{proposition}\label{prop:CondProLocal}
Let $\cpi$ be a protocol and let $\Meas$ be a measure over $\Leaves(\cpi)$. Then $\Restrict{\CondPro{\cpi}{\Meas}}{u}=\CondPro{\cpi_u}{\Restrict{\Meas}{u}}$ for every $u\in\Vertices(\cpi)\setminus\Leaves(\cpi)$.
\end{proposition}
\begin{proof}
Immediately follows from \cref{def:CondProtocol}.
\end{proof}

\cref{prop:CondProLocal} helps us to apply induction on conditional protocols. Specifically, we use it to prove the following lemma, which relates the (dominated measure)-conditional protocol to the optimal (valid) attack.
\begin{lemma}\label{claim:SwitchRoles}
Let $\cpi=\HonAHonB$ be a protocol with $\Val(\cpi)<1$. Then $\BestB{\CondPro{\cpi}{\ADomMeas{\cpi}}}=1$.
\end{lemma}
This lemma justifies yet again the name of the $\HonA$-dominated measure. Not only that this measure give a precise quantity to the advantage of the optimal attacker when taking the role of $\HonA$ over the one taking the role of $\HonB$ (\cref{lemma:ExpDomMeas}), but when we condition on not ``hitting'' this measure, the optimal attacker  taking the role of $\HonA$ no longer wins with probability one.

As an intuitive explanation, assume that $\BestA{\CondPro{\cpi}{\ADomMeas{\cpi}}}=1$. By \cref{lemma:perfectAdv}, it holds that $\BestB{\CondPro{\cpi}{\ADomMeas{\cpi}}}<1$, and so there exists an $\HonA$-dominated measure $M$ in the conditional protocol $\CondPro{\cpi}{\ADomMeas{\cpi}}$. Let the measure $M'$  be the ``union'' of $\ADomMeas{\cpi}$ and $M$. It holds that $M'$ (like $\ADomMeas{\cpi}$ itself) is $\HonA$-maximal, $\HonB$-minimal and $\HonB$-immune in $\cpi$, and that $\eex{\LDist{\cpi}}{M'}>\eex{\LDist{\cpi}}{\ADomMeas{\cpi}}$. Following similar arguments to those in the proof of \cref{lemma:ExpDomMeas}, it  also holds that $\eex{\LDist{\cpi}}{M'} = 1-\BestB{\cpi}$. But \cref{lemma:ExpDomMeas} already showed that $1-\BestB{\cpi} = \eex{\LDist{\cpi}}{\ADomMeas{\cpi}}$, a contradiction (in essence, \cref{lemma:ExpDomMeas} shows that $\ADomMeas{\cpi}$ is the ``only'' $\HonA$-maximal, $\HonB$-minimal and $\HonB$-immune measure in $\cpi$). 

\begin{proof}[Proof of \cref{claim:SwitchRoles}]
First, we note that \cref{fact:ValueExp} yields that $\eex{\LDist\cpi}{\ADomMeas{\cpi}}\leq \Val(\cpi)<1$, and hence $\CondPro{\cpi}{\ADomMeas{\cpi}}\neq \perp$ (\ie is a protocol).
The rest of the proof is by induction on the round complexity of $\cpi$.

Assume that $\round(\cpi) = 0$ and let $\ell$ be the only node in $\Tree(\cpi)$. Since it is assumed that $\Val(\cpi)<1$, it must be the case that $\Color_\cpi(\ell)=0$. The proof follows since $\ADomMeas{\cpi}(\ell)=0$, and thus $\CondPro{\cpi}{\ADomMeas{\cpi}}=\cpi$, and since $\BestB{\cpi}=1$.

Assume the lemma holds for $\rnd$-round protocols and that $\round(\cpi) = \rnd+1$. If $\EdgeDist_{\cpi}(\EmptyString,b)=1$ for some $b\in\zo$, \cref{def:Dominated} yields that $\Restrict{\ADomMeas{\cpi}}{b} = \ADomMeas{\cpi_b}$. Moreover, \cref{def:CondProtocol} yields that $\EdgeDist_{\paren{\CondPro{\cpi}{\ADomMeas{\cpi}}}}(\EmptyString,b)=1$. It holds that
\begin{align}\label{eq:SwitchRoles1}
\BestB{\CondPro{\cpi}{\ADomMeas{\cpi}}} &\overset{(1)}{=} \BestB{\Restrict{\CondPro{\cpi}{\ADomMeas{\cpi}}}{b}} \\
&\overset{(2)}{=} \BestB{\CondPro{\cpi_b}{\Restrict{\ADomMeas{\cpi}}{b}}} \nonumber\\
&= \BestB{\CondPro{\cpi_b}{\ADomMeas{\cpi_b}}} \nonumber\\
&\overset{(3)}{=} 1, \nonumber
\end{align}
where $(1)$ follows from \cref{fact:perfectAdv}, $(2)$ follows from \cref{prop:CondProLocal}, and $(3)$ follows from the induction hypothesis.

In the complementary case, \ie  $\EdgeDist_{\cpi}(\EmptyString,b)\notin\zo$ for both $b\in\zo$, the proof splits according to who controls the root of $\cpi$.

\paragraph{$\HonA$ controls $\Root(\cpi)$.} The assumption that $\Val(\cpi)<1$ dictates that $\Val(\cpi_0)<1$ or $\Val(\cpi_1)<1$. Consider the following complimentary cases.
\begin{description}
  \item[$\Val(\cpi_0),\Val(\cpi_1)<1$:] \cref{fact:perfectAdv} yields that
  \begin{align*}
\lefteqn{\BestB{\CondPro{\cpi}{\ADomMeas{\cpi}}}}\\
&\overset{(1)}{=} \EdgeDist_{\paren{\CondPro{\cpi}{\ADomMeas{\cpi}}}}(\EmptyString,0)\cdot \BestB{\Restrict{\CondPro{\cpi}{\ADomMeas{\cpi}}}{0}} + \EdgeDist_{\paren{\CondPro{\cpi}{\ADomMeas{\cpi}}}}(\EmptyString,1)\cdot \BestB{\Restrict{\CondPro{\cpi}{\ADomMeas{\cpi}}}{1}} \\
&\overset{(2)}{=} \EdgeDist_{\paren{\CondPro{\cpi}{\ADomMeas{\cpi}}}}(\EmptyString,0)\cdot \BestB{\CondPro{\cpi_0}{\Restrict{\ADomMeas{\cpi}}{0}}} + \EdgeDist_{\paren{\CondPro{\cpi}{\ADomMeas{\cpi}}}}(\EmptyString,1)\cdot \BestB{\CondPro{\cpi_1}{\Restrict{\ADomMeas{\cpi}}{1}}} \\
    &\overset{(3)}{=} \EdgeDist_{\paren{\CondPro{\cpi}{\ADomMeas{\cpi}}}}(\EmptyString,0)\cdot \BestB{\CondPro{\cpi_0}{\ADomMeas{\cpi_0}}} + \EdgeDist_{\paren{\CondPro{\cpi}{\ADomMeas{\cpi}}}}(\EmptyString,1)\cdot \BestB{\CondPro{\cpi_1}{\ADomMeas{\cpi_1}}} \\
    &\overset{(4)}{=} 1,
  \end{align*}
  where (1) follows from \cref{fact:perfectAdv}, (2) from \cref{prop:CondProLocal}, (3) follows from by the $\HonA$-maximal property of $\ADomMeas{\cpi}$ (\cref{prop:DomMeasPro}(\ref{item:Amaximal})), and (4) follows from the induction hypothesis.

  \item[$\Val(\cpi_0)<1$, $\Val(\cpi_1)=1$:] By \cref{def:CondProtocol}, it holds that
  \begin{align*}
  \EdgeDist_{\paren{\CondPro{\cpi}{\ADomMeas{\cpi}}}}(\EmptyString,1) &= \EdgeDist_\cpi(\EmptyString,1)\cdot \frac{1-\eex{\LDist{\cpi_1}}{\Restrict{\ADomMeas{\cpi}}{1}}}{1-\eex{\LDist{\cpi}}{\ADomMeas{\cpi}}} \\
&\overset{(1)}{=} \EdgeDist_\cpi(\EmptyString,1)\cdot \frac{1-\eex{\LDist{\cpi_1}}{\ADomMeas{\cpi_1}}}{1-\eex{\LDist{\cpi}}{\ADomMeas{\cpi}}} \\
&\overset{(2)}{=}0,
  \end{align*}
where the (1) follows from the $\HonA$-maximal property of $\ADomMeas{\cpi}$, and (2) follows since $\Val(\cpi_1)=1$, which yields that $\eex{\LDist{\cpi_1}}{\ADomMeas{\cpi_1}}=1$. Since $\CondPro{\cpi}{\ADomMeas{\cpi}}$ is a protocol (\cref{prop:CondPro}), it holds that $\EdgeDist_{\paren{\CondPro{\cpi}{\ADomMeas{\cpi}}}}(\EmptyString,0) = 1$. The proof now follows from \cref{eq:SwitchRoles1}.

  \item[$\Val(\cpi_0)=1$, $\Val(\cpi_1)<1$:] The proof in analogous to the previous case.
\end{description}

\paragraph{$\HonB$ controls $\Root(\cpi)$.} Assume for simplicity that $\Smaller{\cpi}{0}=1$, namely that $\eex{\LDist{\cpi_0}}{\ADomMeas{\cpi_0}} \leq \eex{\LDist{\cpi_1}}{\ADomMeas{\cpi_1}}$ (the other case is analogous). It must hold that $\Val(\cpi_0)<1$ (otherwise, it holds that $\eex{\LDist{\cpi_0}}{\ADomMeas{\cpi_0}}=\eex{\LDist{\cpi_1}}{\ADomMeas{\cpi_1}}=1$, which yields that $\Val(\cpi_1)=1$, and thus $\Val(\cpi)=1$). Hence, $\eex{\LDist{\cpi_0}}{\ADomMeas{\cpi_0}}<1$, and \cref{def:CondProtocol} yields that $\EdgeDist_{\paren{\CondPro{\cpi}{\ADomMeas{\cpi}}}}(\EmptyString,0) > 0$. By \cref{fact:perfectAdv}, it holds that
\begin{align*}
\BestB{\CondPro{\cpi}{\ADomMeas{\cpi}}} &\geq \BestB{\Restrict{\CondPro{\cpi}{\ADomMeas{\cpi}}}{0}} \\
&\overset{(1)}{=} \BestB{\CondPro{\cpi_0}{\Restrict{\ADomMeas{\cpi}}{0}}} \\
&\overset{(2)}{=} \BestB{\CondPro{\cpi_0}{\ADomMeas{\cpi_0}}} \\
&\overset{(3)}{=} 1,
\end{align*}
where (1) follows from \cref{prop:CondProLocal}, (2) follows from the $\HonB$-minimal property of $\ADomMeas{\cpi}$ (\cref{prop:DomMeasPro}(\ref{item:Bminimal})), and (3) follows from the induction hypothesis.
\end{proof}

Let $\cpi=\HonAHonB$ be a protocol in which an optimal adversary playing the role of $\HonA$ biases the outcome towards one with probability one. \cref{claim:SwitchRoles} shows that in the conditional protocol $\cpi_{(\HonB,0)} \eqdef \CondPro{\cpi}{\ADomMeas{\cpi}}$, an optimal adversary playing the role of $\HonB$ can bias the outcome towards zero with probability one. Repeating this procedure \wrt $\cpi_{(\HonB,0)}$ results in the protocol  $\cpi_{(\HonA,1)} \eqdef \CondPro{\cpi_{(\HonB,0)}}{\ADomMeas{\cpi_{(\HonB,0)}}}$, in which again an optimal adversary playing the role of $\HonA$ can bias the outcome towards one with probability one. This procedure is stated formally in \cref{def:mDMSequence}.

\subsubsection{Sequence of Dominated Measures}
Given a protocol $\HonAHonB$, order the pairs $\set{(\HonC,j)}_{(\HonC,j) \in \set{\HonA,\HonB} \times \N}$ according to the sequence $(\HonA,0),(\HonB,0),(\HonA,1),(\HonB,1)$ and so on.

\begin{notation}
Let $\HonAHonB$ be a protocol. For $j\in \Z$ let $\pred(\HonA,j) = (\HonB,j-1)$ and $\pred(\HonB,j) = (\HonA,j)$, and let $\succe$ be the inverse operation of $\pred$ (\ie  $\succe(\pred(\HonC,j))=(\HonC,j)$). For pairs $(\HonC,j), (\HonC',j') \in \set{\HonA,\HonB} \times \Z$, we write
\begin{itemize}
  \item $(\HonC,j)$ is {\sf less than or equal to} $(\HonC',j')$ , denoted $(\HonC,j) \preceq (\HonC',j')$, if $\exists \set{(\HonC_1,j_1),\ldots,(\HonC_n,j_n)}$ such that $(\HonC,j)=(\HonC_1,j_1)$, $(\HonC',j')=(\HonC_n,j_n)$ and $(\HonC_i,j_i)=\pred(\HonC_{i+1},j_{i+1})$ for any $i\in[n-1]$.

  \item $(\HonC,j)$ is {\sf less than} $(\HonC',j')$, denoted $(\HonC,j) \prec (\HonC',j')$, if $(\HonC,j) \preceq (\HonC',j')$ and $(\HonC,j) \neq (\HonC',j')$.
\end{itemize}
 Finally, for $(\HonC,j)\succeq(\HonA,0)$,  let $[(\HonC,j)]\eqdef \set{(\HonC',j') \colon (\HonA,0) \preceq (\HonC',j') \preceq (\HonC,j)}$.

\end{notation}

\begin{definition}(dominated measures sequence)\label{def:mDMSequence}
For a protocol $\cpi=\HonAHonB$ and $(\HonC,j)\in\set{\HonA,\HonB} \times \N$, the protocol $\cpi_{(\HonC,j)}$ is defined by
\begin{align*}
 \cpi_{(\HonC,j)} =\left\{
  \begin{array}{ll}
    \cpi, & (\HonC,j)=(\HonA,0); \\
    \CondProP{\cpi_{(\HonC',j') = \pred(\HonC,j)}}{\DomMeas{\cpi_{(\HonC',j')}}{\HonC'}}, & \hbox{otherwise.\footnotemark}
  \end{array}
\right.
\end{align*}
\footnotetext{Note that if $\eex{\LDist{\icpi{(\HonC,j)}{}}}{\DomMeas{\icpi{(\HonC,j)}{}}{\HonC}}=1$, \cref{def:CondProtocol} yields that $\icpi{\succe(\HonC,j)}{}=\perp$. In fact, since we defined $\CondPro{\perp}{\Meas}=\perp$ for any measure $\Meas$ (also in \cref{def:CondProtocol}), it follows that $\icpi{(\HonC',j')}{}=\perp$ for any $(\HonC',j')\succ (\HonC,j)$.}

Define the {\sf $(\HonC,j)$ dominated measures sequence of $\cpi$, denoted $(\HonC,j)$-$\DMS{\cpi}$}, by $\set{\DomMeas{\cpi_{(\HonC',j')}}{\HonC'}}_{(\HonC',j')\in [(\HonC,j)]}$. Finally, for $\dmsi\in\N$, let $\CombineMeas{\cpi}{\HonC}{\dmsi} \equiv \sum_{j=0}^\dmsi \DomMeas{\icpi{(\HonC,j)}{}}{\HonC}\prod_{t=0}^{j-1}\paren{1-\DomMeas{\icpi{(\HonC,t)}{}}{\HonC}}$.
\end{definition}

We show that $\CombineMeas{\cpi}{\HonA}{\dmsi}$ is a measure (\ie its range is $[0,1]$) and that its support is a subset of the $1$-leaves of $\cpi$. We also give an explicit expression for its expected value (analogous to the expected value of $\ADomMeas{\cpi}$ given in \cref{lemma:ExpDomMeas}).
\begin{lemma}\label{lemma:propCombineMeas}
Let $\cpi=\HonAHonB$ be a protocol, let $\dmsi\in\N$, and let $\CombineMeas{\cpi}{\HonA}{\dmsi}$ be as in \cref{def:mDMSequence}. It holds that
\begin{enumerate}
  \item\label{item:CombineMeas1}  $\CombineMeas{\cpi}{\HonA}{\dmsi}$ is a measure over $\Leaves_1(\cpi)$:

  \begin{enumerate}
    \item $\CombineMeas{\cpi}{\HonA}{\dmsi}(\ell)\in[0,1]$ for every $\ell\in\Leaves(\cpi)$, and
    \item $\Supp\paren{\CombineMeas{\cpi}{\HonA}{\dmsi}}\subseteq\Leaves_1(\cpi)$.
  \end{enumerate}

  \item\label{item:CombineMeas2} $\eex{\LDist{\cpi}}{\CombineMeas{\cpi}{\HonA}{\dmsi}}=\sum_{j=0}^z \alpha_j\cdot \prod_{t=0}^{j-1}(1-\beta_t)(1-\alpha_t)$, where $\alpha_j=1-\BestB{\cpi_{(\HonA,j)}}$, $\beta_j=1-\BestA{\cpi_{(\HonB,j)}}$ and $\BestA{\perp}=\BestB{\perp}=1$.
\end{enumerate}
\end{lemma}
\begin{proof}
We prove the above two items separately.
\begin{description}
\item[Proof of \cref{item:CombineMeas1}.] Let $\ell \in \Leaves_0(\cpi)$. Since $\ADomMeas{\icpi{(\HonA,j)}{}}(\ell)=0$ for every $j\in(\dmsi)$, it holds that $\CombineMeas{\cpi}{\HonA}{\dmsi}(\ell)=0$. Let $\ell\in\Leaves_1(\cpi)$. Since $\CombineMeas{\cpi}{\HonA}{\dmsi}(\ell)$ is a sum of non-negative numbers, it follows that its value is non-negative. It is left to argue that $\CombineMeas{\cpi}{\HonA}{\dmsi}(\ell)\leq 1$. Since $\ADomMeas{\icpi{(\HonA,\dmsi)}{}}$ is a measure, note that $\ADomMeas{\icpi{(\HonA,\dmsi)}{}}(\ell) \leq 1$. Thus
\begin{align*}
\CombineMeas{\cpi}{\HonA}{\dmsi}(\ell)  &= \sum_{j=0}^{\dmsi}\ADomMeas{\icpi{(\HonA,j)}{}}(\ell) \cdot\prod_{t=0}^{j-1}\paren{1-\ADomMeas{\icpi{(\HonA,t)}{}}(\ell)}\\
&\leq \prod_{t=0}^{\dmsi-1}\paren{1-\ADomMeas{\icpi{(\HonA,t)}{}}(\ell)} + \sum_{j=0}^{\dmsi-1}\ADomMeas{\icpi{(\HonA,j)}{}}(\ell) \cdot\prod_{t=0}^{j-1}\paren{1-\ADomMeas{\icpi{(\HonA,t)}{}}(\ell)} \\
&= \paren{\sum_{\I\subseteq (\dmsi-1)}(-1)^{\size{\I}}\cdot\prod_{t\in\I}\ADomMeas{\icpi{(\HonA,t)}{}}(\ell)} + \sum_{j=0}^{\dmsi-1}\ADomMeas{\icpi{(\HonA,j)}{}}(\ell) \cdot \paren{\sum_{\I\subseteq (j-1)}(-1)^{\size{\I}}\cdot\prod_{t\in\I}\ADomMeas{\icpi{(\HonA,t)}{}}(\ell)} \\
&= \paren{\sum_{\I\subseteq (\dmsi-1)}(-1)^{\size{\I}}\cdot\prod_{t\in\I}\ADomMeas{\icpi{(\HonA,t)}{}}(\ell)} + \paren{\sum_{\emptyset\neq\I\subseteq (\dmsi-1)}(-1)^{\size{\I}+1}\cdot\prod_{t\in\I}\ADomMeas{\icpi{(\HonA,t)}{}}(\ell)} \\
&= 1.
\end{align*}

\item[Proof of \cref{item:CombineMeas2}.] By linearity of expectation, it suffices to prove that
\begin{align}\label{eq:PropCombine}
\eex{\LDist{\cpi}}{\ADomMeas{\cpi_{(\HonA,j)}}\cdot\prod_{t=0}^{j-1}\paren{1-\ADomMeas{\cpi_{(\HonA,t)}}}} = \alpha_j\cdot \prod_{t=0}^{j-1}(1-\beta_t)(1-\alpha_t)
\end{align}
for any $j\in(z)$. Fix $j\in(z)$.  If $\cpi_{(\HonA,j)}=\perp$, then by \cref{def:Dominated} it holds that $\ADomMeas{\cpi_{(\HonA,j)}}$ is the zero measure, and both sides of \cref{eq:PropCombine} equal $0$.

In the following we assume that $\cpi_{(\HonA,j)}\neq\perp$. We first note that  $\eex{\LDist{\cpi_{(\HonC,t)}}}{\DomMeas{\cpi_{(\HonC,t)}}{\HonC}}<1$ for any $(\HonC,t)\in[\pred(\HonA,j)]$ (otherwise, it must be that $\cpi_{(\HonA,j)}=\perp$). Thus, \cref{lemma:ExpDomMeas} yields that $\alpha_t,\beta_t<1$ for every $t\in(j-1)$. Hence, recursively applying \cref{prop:CondPro}(2) yields that
\begin{align}\label{eq:propCombineMeas1}
\VerticesDist_{\paren{\cpi_{(\HonA,j)}}}(\ell) = \VerticesDist_{\cpi}(\ell)\cdot \prod_{t=0}^{j-1}\frac{1-\ADomMeas{\cpi_{(\HonA,t)}}(\ell)}{1-\alpha_t}\cdot \frac{1-\BDomMeas{\cpi_{(\HonB,t)}}(\ell)}{1-\beta_t}
\end{align}
for every $\ell\in\Leaves(\cpi)$. Moreover, for $\ell\in\Supp\paren{\cpi_{(\HonA,j)}}$, \ie $\VerticesDist_{\paren{\cpi_{(\HonA,j)}}}(\ell) > 0$, we can manipulate \cref{eq:propCombineMeas1} to get that
\begin{align}\label{eq:propCombineMeas2}
\VerticesDist_{\cpi}(\ell) = \VerticesDist_{\paren{\cpi_{(\HonA,j)}}}(\ell) \cdot \prod_{t=0}^{j-1} \frac{1-\alpha_t}{1-\ADomMeas{\cpi_{(\HonA,t)}}(\ell)}\cdot \frac{1-\beta_t}{1-\BDomMeas{\cpi_{(\HonB,t)}}(\ell)}
\end{align}
for every $\ell\in\Supp\paren{\cpi_{(\HonA,j)}}$.

It follows that
\begin{align*}
\lefteqn{\eex{\LDist{\cpi}}{\ADomMeas{\cpi_{(\HonA,j)}} \cdot \prod_{t=0}^{j-1}\paren{1-\ADomMeas{\cpi_{(\HonA,t)}}}}}\\
&= \sum_{\ell\in\Leaves(\cpi)} \VerticesDist_{\cpi}(\ell)\cdot \paren{\ADomMeas{\cpi_{(\HonA,j)}}(\ell)\cdot\prod_{t=0}^{j-1}\paren{1-\ADomMeas{\cpi_{(\HonA,t)}}(\ell)}} \\
&\overset{(1)}{=} \sum_{\ell\in\Supp\paren{\cpi_{(\HonA,j)}}\cap \Leaves_1(\cpi)} \VerticesDist_{\cpi}(\ell)\cdot \paren{\ADomMeas{\cpi_{(\HonA,j)}}(\ell)\cdot\prod_{t=0}^{j-1}\paren{1-\ADomMeas{\cpi_{(\HonA,t)}}(\ell)}} \\
&\overset{(2)}{=} \sum_{\ell\in\Supp\paren{\cpi_{(\HonA,j)}}\cap \Leaves_1(\cpi)} \VerticesDist_{\paren{\cpi_{(\HonA,j)}}}(\ell) \cdot \prod_{t=0}^{j-1} \frac{1-\alpha_t}{1-\ADomMeas{\cpi_{(\HonA,t)}}(\ell)}\cdot \frac{1-\beta_t}{1-\BDomMeas{\cpi_{(\HonB,t)}}(\ell)}\\
&\quad  \cdot\paren{\ADomMeas{\cpi_{(\HonA,j)}}(\ell)\cdot\prod_{t=0}^{j-1}\paren{1-\ADomMeas{\cpi_{(\HonA,t)}}(\ell)}} \\
&\overset{(3)}{=} \sum_{\ell\in\Supp\paren{\cpi_{(\HonA,j)}}\cap \Leaves_1(\cpi)} \VerticesDist_{\paren{\cpi_{(\HonA,j)}}}(\ell) \cdot \ADomMeas{\cpi_{(\HonA,j)}}(\ell) \cdot\prod_{t=0}^{j-1} \paren{1-\alpha_j} \paren{1-\beta_j}\\
&= \alpha_j\cdot\prod_{t=0}^{j-1}(1-\beta_t)(1-\alpha_t),
\end{align*}
concluding the proof. (1) follows since  \cref{def:Dominated} yields that $\ADomMeas{\cpi_{(\HonA,j)}}(\ell)=0$ for any $\ell\notin\Supp\paren{\cpi_{(\HonA,j)}}\cap \Leaves_1(\cpi)$, (2) follows from \cref{eq:propCombineMeas2} and (3) follows since $\BDomMeas{\cpi_{(\HonB,t)}}(\ell)=0$ for every $\ell\in\Leaves_1(\cpi)$ and $t\in(j-1)$. \qedhere
\end{description}
\end{proof}

Using dominated measure sequences, we manage to give an improved bound for the success probability of the recursive biased-continuation attacks (compared to the bound of \cref{lemma:claim1}, which uses a single dominated measure).  The improved analysis yields that a constant number of recursion calls of the biased-continuation attack is successful in biasing the protocol to an arbitrary constant close to either $0$ or $1$.

\subsection{Improved Analysis Using Alternating Dominated Measures}\label{sec:RCAltDomMeas}
We are finally ready to state two main lemmas, whose proofs -- given in the next two sections -- are the main technical contribution of \cref{sec:IdealAttacker}, and then show how to use them to prove \cref{thm:MainIdeal}.

The first lemma is analogous to \cref{lemma:SingleMeasManyIter}, but applied on the sequence of the dominated measures, and not just on a single dominated measure.

\begin{lemma}\label{lemma:IdealMainLemmaSim}
For a protocol $\cpi=\HonAHonB$ with $\Val(\cpi)>0$ and $\dmsi\in\N$, it holds that
\begin{align*}
\Val(\rcA{k},\HonB) \geq \eex{\LDist{\rcA{k},\HonB}}{\CombineMeas{\cpi}{\HonA}{\dmsi}} \geq \frac{\eex{\LDist{\cpi}}{\CombineMeas{\cpi}{\HonA}{\dmsi}}}{\prod_{i=0}^{k-1}\Val(\rcA{i},\HonB)} \cdot \paren{1-\sum_{j=0}^{\dmsi-1}\beta_j}^k
\end{align*}
for every $k\in\N$, where $\beta_j= 1-\BestA{\cpi_{(\HonB,j)}}$, letting $\BestA{\perp}=1$.
\end{lemma}
The above states that the recursive biased-continuation attacker biases the outcome of the protocol by a bound similar to that given in \cref{lemma:SingleMeasManyIter}, but applied \wrt $\CombineMeas{\cpi}{\HonA}{\dmsi}$, instead of $\ADomMeas{\cpi}$ in \cref{lemma:SingleMeasManyIter}. This is helpful since the expected value of $\CombineMeas{\cpi}{\HonA}{\dmsi}$ is strictly larger than that of $\ADomMeas{\cpi}$. However, since $\CombineMeas{\cpi}{\HonA}{\dmsi}$ is defined \wrt a sequence of conditional protocols, we must ``pay'' the term $\paren{1-\sum_{j=0}^{\dmsi-1}\beta_j}^k$ in order to get this bound in the original protocol.

The following lemma states that \cref{lemma:IdealMainLemmaSim} provides a sufficient bound. Specifically, it shows that if we take a long enough sequence of conditional protocols, the expected value of the measure $\CombineMeas{\cpi}{\HonA}{\dmsi}$ will be sufficiently large, while the payment term mentioned above will be kept sufficiently small.

\begin{lemma}\label{lemma:CombineMeasHitConst}
Let $\cpi=\HonAHonB$ be a protocol. Then for every $c\in (0,\frac12]$ there exists  $\dmsi = \dmsi(c,\cpi)\in\N$ (possibly exponential large) such that:
\begin{enumerate}
  \item\label{item:convergence1} $\eex{\LDist{\cpi}}{\CombineMeas{\cpi}{\HonA}{\dmsi}}\geq c\cdot (1-2c)$ and $\sum_{j=0}^{\dmsi-1}\beta_j < c$; or
  \item\label{item:convergence2} $\eex{\LDist{\cpi}}{\CombineMeas{\cpi}{\HonB}{\dmsi}}\geq c\cdot (1-2c)$ and $\sum_{j=0}^{\dmsi}\alpha_j < c$,
\end{enumerate}
where $\alpha_j= 1-\BestB{\cpi_{(\HonA,j)}}$ and $\beta_j= 1-\BestA{\cpi_{(\HonB,j)}}$.
\end{lemma}

To derive \cref{thm:MainIdeal}, we take a sequence of the dominated measures that is long enough so that its accumulated weight will be sufficiently large. Furthermore, the weight of the dominated measures that precede the final dominated measure in the sequence is small (otherwise, we would have taken a shorter sequence), so the parties are ``missing'' these measures with high probability. The formal proof of \cref{thm:MainIdeal} is given next, and the proofs of \cref{lemma:IdealMainLemmaSim,lemma:CombineMeasHitConst} are given in \cref{sec:ProvingIdealMainLemma,sec:DomMeasHitConst} respectively.

\subsubsection{Proving Theorem \ref{thm:MainIdeal}}
\begin{proof}[Proof of \cref{thm:MainIdeal}]
If $\Val(\cpi)=0$, \cref{thm:MainIdeal} trivially holds. Assume that $\Val(\cpi)>0$, let $\dmsi$ be the minimum integer guaranteed by \cref{lemma:CombineMeasHitConst} for $c=\eps/2$, and let $\kappa=\ceil{\frac{\log\paren{\frac{2}{\eps}}}{\log\paren{\frac{1-\eps/2}{1-\eps}}}}\in\widetilde{O}\paren{1/\eps}$.

If $\dmsi$ satisfies \cref{item:convergence1} of \cref{lemma:CombineMeasHitConst}, assume towards a contradiction that $\Val(\ProArc{\kappa})\leq1-\eps$. \cref{lemma:IdealMainLemmaSim} yields that
\begin{align*}
\Val(\rcA{\kappa},\HonB) &\geq \frac{\eex{\LDist{\cpi}}{\CombineMeas{\cpi}{\HonA}{\dmsi}}}{\prod_{i=0}^{\kappa-1}\Val(\rcA{i},\HonB)} \cdot \paren{1-\sum_{j=0}^{\dmsi-1}\beta_j}^\kappa \\
&> \frac{\eps(1-\eps)}{2} \cdot \paren{\frac{1-\eps/2}{1-\eps}}^\kappa \\
&\geq 1-\eps,
\end{align*}
and a contradiction is derived.

If $\dmsi$ satisfies \cref{item:convergence2} of \cref{lemma:CombineMeasHitConst}, an analogous argument to the above yields that $\Val(\HonA,\rcB{\kappa})\leq \eps$.
\end{proof}

\subsection{Proving Lemma \ref{lemma:IdealMainLemmaSim}}\label{sec:ProvingIdealMainLemma}
\subsubsection{Outline}
We would like to follow the proof's outline of \cref{lemma:claim1}, which is a
special case of \cref{lemma:IdealMainLemmaSim} for $k=1$ and $z=0$ (\ie only
a single dominated measure instead of a sequence).

The proof of \cref{lemma:claim1} was done through the following steps: (1) we
applied the induction hypothesis to the sub-protocols $\cpi_0$ and $\cpi_1$ \wrt
their $\HonA$-dominated measures, $\ADomMeas{\cpi_0}$ and $\ADomMeas{\cpi_1}$;
(2) we related, using \cref{prop:DomMeasPro}, $\ADomMeas{\cpi_0}$ and
$\ADomMeas{\cpi_1}$ to $\Restrict{\ADomMeas{\cpi}}{0}$ and
$\Restrict{\ADomMeas{\cpi}}{1}$, were the latter are the restrictions of the
$\HonA$-dominated measure of $\cpi$ to $\cpi_0$ and $\cpi_1$; (3) if $\HonA$
controls the root, then we used the properties of $\rcA{1}$ (specifically, the
way it changes the edges distribution) to complete the proof; (4) if $\HonB$
controls the root, then we used a convexity-type argument to complete the proof.

Lets try to extend the above outline for a sequence of two dominated measures. It will be useful to consider a specific protocol, presented in \cref{fig:OutlineProofPi} (this protocol is an instantiation of the protocol we have been using thus far for the examples). Recall that the $\HonA$-dominated measure of $\cpi=\cpi_{(\HonA,0)}$ assigns $\ADomMeas{\cpi}(00)=1$ (the left-most leaf), $\ADomMeas{\cpi}(10)=1/2$ (the second to the right-most leaf), and zero to the rest of the leaves. Using $\ADomMeas{\cpi}$, we can now compute $\cpi_{(\HonB,0)}$, presented in \cref{fig:OutlineProofCondPi}.  Now, consider the sequence of two dominated measures for $\cpi_1$, presented in \cref{fig:OutlineProofPi2}. The $\HonA$-dominated measure of $\cpi_1=\paren{\cpi_1}_{(\HonA,0)}$ assigns $\ADomMeas{\cpi_1}(10)=1$ and $\ADomMeas{\cpi_1}(11)=0$, and using it we can compute $\paren{\cpi_1}_{(\HonB,0)}$, presented in \cref{fig:OutlineProofCondPi2}.

The first step of the outline above is to apply the induction hypothesis to the
sub-protocol $\cpi_1$. When trying to extend this outline for proving
\cref{lemma:IdealMainLemmaSim} we face a problem, since
$\paren{\cpi_1}_{(\HonB,0)}$ is not the same protocol as
$\Restrict{\cpi_{(\HonB,0)}}{1}$. The latter is a consequence of the fact that
$\Restrict{\ADomMeas{\cpi}}{1}\neq\ADomMeas{\cpi_1}$. In fact, we implicitly
faced the same problem in the proof of \cref{lemma:claim1}, where we used
\cref{prop:DomMeasPro} to show that
$\Restrict{\ADomMeas{\cpi}}{1}=(1/2)\cdot \ADomMeas{\cpi_1}$, and thus still
enabling us to use the induction hypothesis.  At this point we observe that the
proof of \cref{lemma:claim1} can also be viewed differently. Instead of applying
the induction hypothesis on $\ADomMeas{\cpi_1}$ and use \cref{prop:DomMeasPro},
we can apply the induction hypothesis directly to the measure
$(1/2)\cdot \ADomMeas{\cpi_1}$. This requires strengthening of the statement of
the lemma to consider \emph{submeasures} of of dominated measures, namely,
measures of the form $\eta\cdot M$, for $0\leq \eta\leq 1$ and $M$ being some
dominated measure.

Using sequence of dominated submeasures is the path we take for proving \cref{lemma:IdealMainLemmaSim}. The outline of the proof is as follows:
\begin{enumerate}
\item Define $(\cpi,\vetas)$-dominated submeasures sequence, where $\vetas$ is a vector of real values in $[0,1]$ (\cref{def:refinements}).
\item Extend the statement of \cref{lemma:IdealMainLemmaSim} to handle dominated submeasures sequences (\cref{lemma:IdealMainLemmaDet}).
\item Given $\vetas$, carefully define $\vetas_0$ and $\vetas_1$ such that the restrictions of the $(\cpi,\vetas)$-dominated submeasures sequence are exactly the measures used in $(\cpi_0,\vetas_0)$-dominated submeasures sequence and in $(\cpi_1,\vetas_1)$-dominated submeasure sequence (\cref{def:etaP,claim:RestrictFinal}).
\item Apply the induction hypothesis to the $(\cpi_0,\vetas_0)$-dominated submeasures sequence and the $(\cpi_1,\vetas_1)$-dominated submeasures sequence.
\item  If $\HonA$ controls the root, then use the properties of $\rcA{1}$ to complete the proof.
\item If $\HonB$ controls the root, then use a convexity-type argument to complete the proof.
\end{enumerate}
The formal proof, given below, follows precisely this outline. Unlike in the proof of \cref{lemma:claim1}, the last two steps are not trivial, and require careful analysis.

\begin{figure}
\centering
\begin{subfigure}{.48\textwidth}
\centering
\begin{tikzpicture}[->,>=stealth',level/.style={sibling distance = 5cm/#1, level distance = 1.5cm}]
\begin{scope}
\node [B-node] (B) {$\HonB$}
    child{ node [A-node] (A0) {$\HonA$}
        child{ node [leaf] (L00) {$1$} edge from parent node[above left] {$1/4$}}
        child{ node [B-node] (B01) {$\HonB$}
            child{ node [leaf] (L010) {$0$} edge from parent node[above left] {$1/2$}}
			child{ node [leaf,thick] (L011) {$1$} edge from parent node[above right] {$1/2$}}
            edge from parent node[above right] {$3/4$}
		}
        edge from parent node[above left] {$1/2$}
    }
    child{ node [A-node] (A1) {$\HonA$}
        child{ node [leaf] (L10) {$1$} edge from parent node[above left] {$1/2$}}
        child{ node [leaf] (L11) {$0$} edge from parent node[above right] {$1/2$}
		}
        edge from parent node[above right] {$1/2$}
    }
;
\end{scope}
\end{tikzpicture}
\caption{Protocol $\cpi=\cpi_{(\HonA,0)}$.}
\label{fig:OutlineProofPi}
\end{subfigure}%
 \hfill
\begin{subfigure}{.48\textwidth}
\centering
\begin{tikzpicture}[->,>=stealth',level/.style={sibling distance = 5cm/#1, level distance = 1.5cm}]
\begin{scope}
\node [B-node] (B) {$\HonB$}
    child{ node [A-node] (A0) {$\HonA$}
        child{ node [leaf] (L00) {$1$} edge from parent node[above left] {$0$}}
        child{ node [B-node] (B01) {$\HonB$}
            child{ node [leaf] (L010) {$0$} edge from parent node[above left] {$1/2$}}
			child{ node [leaf,thick] (L011) {$1$} edge from parent node[above right] {$1/2$}}
            edge from parent node[above right] {$1$}
		}
        edge from parent node[above left] {$1/2$}
    }
    child{ node [A-node] (A1) {$\HonA$}
        child{ node [leaf] (L10) {$1$} edge from parent node[above left] {$1/3$}}
        child{ node [leaf] (L11) {$0$} edge from parent node[above right] {$2/3$}
		}
        edge from parent node[above right] {$1/2$}
    }
;
\end{scope}
\end{tikzpicture}
\caption{Protocol $\cpi_{(\HonB,0)}$.}
\label{fig:OutlineProofCondPi}
\end{subfigure}
\caption{An example of a coin-flipping protocol to the left and its conditional protocol tp the right, when conditioning not to ``hit'' the $\HonA$-dominated measure.}
\label{fig:outline1}
\end{figure}
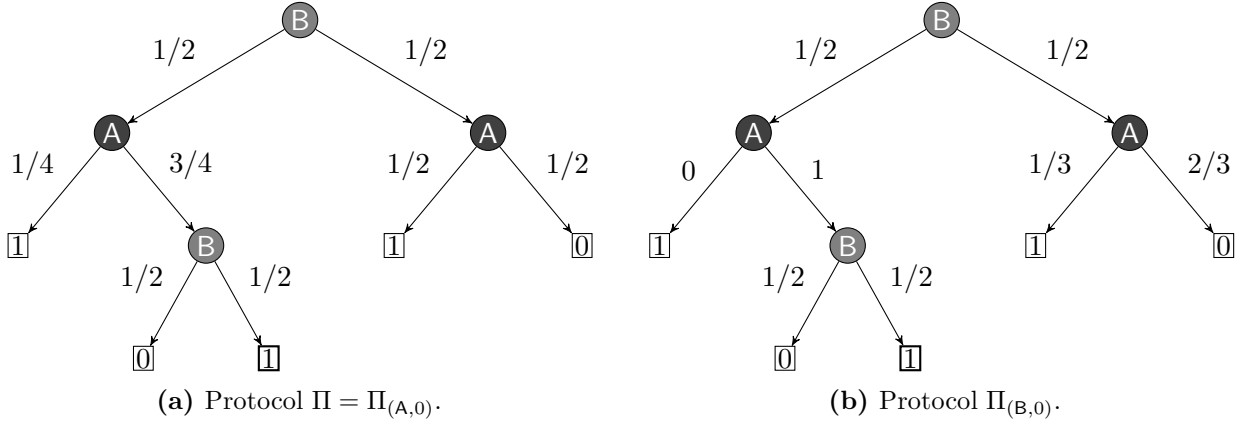

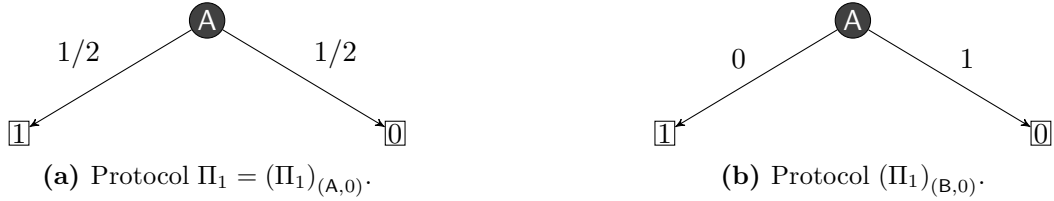
\begin{figure}
\centering
\begin{subfigure}{.48\textwidth}
\centering
\begin{tikzpicture}[->,>=stealth',level/.style={sibling distance = 5cm/#1, level distance = 1.5cm}]
\begin{scope}
\node [A-node] (A1) {$\HonA$}
    child{ node [leaf] (L10) {$1$} edge from parent node[above left] {$1/2$} }
    child{ node [leaf] (L11) {$0$} edge from parent node[above right] {$1/2$} }
;
\end{scope}
\end{tikzpicture}
\caption{Protocol $\cpi_1=\paren{\cpi_1}_{(\HonA,0)}$.}
\label{fig:OutlineProofPi2}
\end{subfigure}%
 \hfill
\begin{subfigure}{.48\textwidth}
\centering
\begin{tikzpicture}[->,>=stealth',level/.style={sibling distance = 5cm/#1, level distance = 1.5cm}]
\begin{scope}
\node [A-node] (A1) {$\HonA$}
    child{ node [leaf] (L10) {$1$} edge from parent node[above left] {$0$} }
    child{ node [leaf] (L11) {$0$} edge from parent node[above right] {$1$} }
;
\end{scope}
\end{tikzpicture}
\caption{Protocol $\paren{\cpi_1}_{(\HonB,0)}$.}
\label{fig:OutlineProofCondPi2}
\end{subfigure}
\caption{The sub-protocol $\cpi_1$ of the protocol from \cref{fig:outline1} and its conditional protocol.}
\label{fig:outline2}
\end{figure}

\subsubsection{Formal Proof of Lemma \ref{lemma:IdealMainLemmaSim}}
The proof of \cref{lemma:IdealMainLemmaSim} is an easy implication of \cref{lemma:propCombineMeas} and the following key lemma, defined \wrt sequences of \emph{submeasures} of the dominated measure.

\begin{definition}(dominated submeasure sequence)\label{def:refinements}
For a protocol $\cpi=\HonAHonB$, a pair $(\HonC^\ast,j^\ast)\in\set{\HonA,\HonB} \times \N$ and $\vetas = \set{\eta_{(\HonC,j)}\in[0,1]}_{(\HonC,j) \in [(\HonC^\ast,j^\ast)]}$, define the protocol $\icpih{(\HonC,j)}{\vetas}$ by
\begin{align*}
\icpih{(\HonC,j)}{\vetas} \eqdef \left\{
  \begin{array}{ll}
    \cpi, & (\HonC,j)=(\HonA,0); \\
    \CondProP{\icpih{(\HonC',j')=\pred(\HonC,j)}{\vetas}}{\idsm{(\HonC',j')}{\cpi}{\vetas}}, & \hbox{otherwise.}
  \end{array}
\right.,
\end{align*}
where $\idsm{(\HonC',j')}{\cpi}{\vetas} \equiv \eta_{(\HonC',j')}\cdot\DomMeas{\icpi{(\HonC',j')}{\vetas}}{\HonC'}$.
For  $(\HonC,j)\in[(\HonC^\ast,j^\ast)]$, define the {\sf $(\HonC,j,\vetas)$-dominated measure sequence of $\cpi$}, denoted $(\HonC,j,\vetas)$-$\DMS{\cpi}$, as $\set{\idsm{(\HonC',j')}{\cpi}{\vetas}}_{(\HonC',j')\in [(\HonC,j)]}$, and let $\imu{(\HonC,j)}{\cpi}{\vetas}=\eex{\LDist{\icpih{(\HonC,j)}{\vetas}}}{\idsm{(\HonC,j)}{\cpi}{\vetas}}$.\footnote{Note that for $\vetas =(1,1,1,\dots,1)$, \cref{def:refinements} coincides with \cref{def:mDMSequence}.}

Finally, let $\FinalMeas{\cpi}{\HonC}{\vetas} \equiv \sum_{j\colon (\HonC,j)\in[(\HonC^\ast,j^\ast)]}\idsm{(\HonC,j)}{\cpi}{\vetas} \cdot\prod_{t=0}^{j-1}\left(1-\idsm{(\HonC,t)}{\cpi}{\vetas}\right)$.
\end{definition}

\begin{lemma}\label{lemma:IdealMainLemmaDet}
Let $\cpi=\HonAHonB$ be a protocol with $\Val(\cpi)>0$,  let $\dmsi\in \N$ and let $\vetas=\set{\eta_{(\HonC,j)}\in[0,1]}_{(\HonC,j)\in [(\HonA,\dmsi)]}$. For $j \in (\dmsi)$, let $\alpha_j = \imu{(\HonA,j)}{\cpi}{\vetas}$, and for $j \in (\dmsi-1)$, let $\beta_{j} = \imu{(\HonB,j)}{\cpi}{\vetas}$. Then
\begin{align*}
\eex{\LDist{\ProArc{k}}}{\FinalMeas{\cpi}{\HonA}{\vetas}} \geq \frac{\sum_{j=0}^\dmsi \alpha_j \cdot \prod_{t=0}^{j-1}(1-\beta_t)^{k+1}(1-\alpha_t)}{\prod_{i=0}^{k-1}\Val{(\ProArc{i})}}
\end{align*}
for any positive $k\in\N$.
\end{lemma}

The proof of \cref{lemma:IdealMainLemmaDet} is given below, but we first use it to prove \cref{lemma:IdealMainLemmaSim}.
\begin{proof}[Proof of \cref{lemma:IdealMainLemmaSim}]
Let $\eta_{(\HonC,j)}=1$ for every $(\HonC,j)\in[(\HonA,\dmsi)]$ and let $\vetas=\set{\eta_{(\HonC,j)}}_{(\HonC,j)\in[(\HonA,\dmsi)]}$. It follows that $\FinalMeas{\cpi}{\HonA}{\vetas}\equiv\CombineMeas{\cpi}{\HonA}{\dmsi}$. Applying \cref{lemma:IdealMainLemmaDet} yields that
\begin{align}\label{eq:SimFromDet1}
\eex{\LDist{\ProArc{k}}}{\CombineMeas{\cpi}{\HonA}{\dmsi}} \geq \frac{\sum_{j=0}^\dmsi \alpha_j \cdot \prod_{t=0}^{j-1}(1-\beta_t)^{k+1}(1-\alpha_t)}{\prod_{i=0}^{k-1}\Val{(\ProArc{i})}},
\end{align}
where $\alpha_j = \imu{(\HonA,j)}{\cpi}{\vetas}$ and $\beta_{j} = \imu{(\HonB,j)}{\cpi}{\vetas}$. Multiplying the $j$'th summand of the right-hand side of \cref{eq:SimFromDet1} by $\prod_{t=j}^{\dmsi-1}(1-\beta_j)^k \leq 1$ yields that
\begin{align}\label{eq:SimFromDet2}
\eex{\LDist{\ProArc{k}}}{\CombineMeas{\cpi}{\HonA}{\dmsi}} &\geq \frac{\sum_{j=0}^\dmsi \alpha_j \cdot \prod_{t=0}^{j-1}(1-\beta_t)(1-\alpha_t)}{\prod_{i=0}^{k-1}\Val{(\ProArc{i})}} \cdot\prod_{t=0}^{\dmsi-1}(1-\beta_t)^k \\
&\geq \frac{\sum_{j=0}^\dmsi \alpha_j \cdot \prod_{t=0}^{j-1}(1-\beta_t)(1-\alpha_t)}{\prod_{i=0}^{k-1}\Val{(\ProArc{i})}} \cdot\paren{1-\sum_{t=0}^{\dmsi-1}\beta_t}^k, \nonumber
\end{align}
where the second inequality follows since $\beta_j\geq 0$ and $(1-x)(1-y)\geq 1-(x+y)$ for any $x,y\geq 0$. By \cref{lemma:ExpDomMeas} and the definition of $\vetas$ it follows that $\imu{(\HonA,j)}{\cpi}{\vetas} = 1- \BestB{\icpi{(\HonA,j)}{}}$ and $\imu{(\HonB,j)}{\cpi}{\vetas} = 1- \BestA{\icpi{(\HonB,j)}{}}$. Hence, plugging \cref{lemma:propCombineMeas} into \cref{eq:SimFromDet2} yields that
\begin{align}
\eex{\LDist{\rcA{k},\HonB}}{\CombineMeas{\cpi}{\HonA}{\dmsi}} \geq \frac{\eex{\LDist{\cpi}}{\CombineMeas{\cpi}{\HonA}{\dmsi}}}{\prod_{i=0}^{k-1}\Val(\rcA{i},\HonB)} \cdot \paren{1-\sum_{t=0}^{\dmsi-1}\beta_t}^k.
\end{align}
Finally, the proof is concluded, since by \cref{lemma:propCombineMeas,fact:ValueExp} it immediately follows that $\Val(\rcA{k},\HonB) \geq \eex{\LDist{\rcA{k},\HonB}}{\CombineMeas{\cpi}{\HonA}{\dmsi}}$.
\end{proof}

\subsubsection{Proving Lemma \ref{lemma:IdealMainLemmaDet}}\label{sec:ProvingMainIdeal}
\begin{proof}[Proof of \cref{lemma:IdealMainLemmaDet}]
In the following we fix a protocol $\cpi$, real vector $\vetas=\set{\eta_{(\HonC,j)}}_{(\HonC,j)\in[(\HonA,z)]}$ and a positive integer $k$. We also assume for simplicity that $\icpih{(\HonA,z)}{\vetas}$ is not the undefined protocol, \ie $\icpih{(\HonA,\dmsi)}{\vetas}\neq \perp$.\footnote{If this assumption does not hold, let $\dmsi'\in(\dmsi-1)$ be the largest index such that $\icpih{(\HonA,\dmsi')}{\vetas}\neq \perp$, and let $\vetas'=\set{\eta_{(\HonC,j)}}_{(\HonC,j)\in[(\HonA,\dmsi')]}$. It follows from \cref{def:Dominated} that $\idsm{(\HonA,j)}{\cpi}{\vetas}$ is the zero measure for any $\dmsi' < j \leq  \dmsi$, and thus $\FinalMeas{\HonA}{\cpi}{\vetas'}\equiv \FinalMeas{\HonA}{\cpi}{\vetas}$. Moreover, the fact that $\alpha_j=0$ for any $\dmsi' < j \leq  \dmsi$ suffices to validate the assumption.} The proof is by induction on the round complexity of $\cpi$.

\paragraph{Base case.}
Assume $\round(\cpi) = 0$ and let $\ell$ be the only node in $\Tree(\cpi)$. For $j \in (\dmsi)$, \cref{def:refinements} yields that $\Color_{\icpih{(\HonA,j)}{\vetas}}(\ell) = \Color_\cpi(\ell) = 1$, where the last equality holds since, by assumption, $\Val(\cpi)>0$. It follows  \cref{def:Dominated} that $\DomMeas{\icpih{(\HonA,j)}{\vetas}}{\HonA}(\ell)=1$ and \cref{def:refinements} that $\idsm{(\HonA,j)}{\cpi}{\vetas}(\ell)=\eta_{(\HonA,j)}$. Hence, it holds that $\alpha_j=\eta_{(\HonA,j)}$. Similarly, for $j \in (\dmsi-1)$ it holds that $\idsm{(\HonB,j)}{\cpi}{\vetas}(\ell)=0$ and thus $\beta_{j}=0$. Clearly, $\ProArcP{k}=\cpi$ and $\Val(\ProArc{i})=1$ for every $i\in[k-1]$. We conclude that
\begin{align*}
\eex{\LDist{\ProArc{k}}}{\FinalMeas{\HonA}{\cpi}{\vetas}} =& \eex{\LDist\cpi}{\FinalMeas{\HonA}{\cpi}{\vetas}} \\
=& \sum_{j=0}^{\dmsi}\idsm{(\HonA,j)}{\cpi}{\vetas}(\ell) \cdot\prod_{t=0}^{j-1}\left(1-\idsm{(\HonA,t)}{\cpi}{\vetas}(\ell) \right) \\
=& \sum_{j=0}^{\dmsi}\eta_{(\HonA,j)} \cdot\prod_{t=0}^{j-1}\left(1-\eta_{(\HonA,t)} \right) \\
=& \sum_{j=0}^{\dmsi}\alpha_j \cdot\prod_{t=0}^{j-1}\left(1-\alpha_t \right) \\
=& \frac{\sum_{j=0}^\dmsi \alpha_j\prod_{t=0}^{j-1}(1-\beta_t)^{k+1}(1-\alpha_t)}{\prod_{i=0}^{k-1}\Val{(\ProArc{i})}}.
\end{align*}

\paragraph{Induction step.}
Assume the lemma holds for $\rnd$-round protocols and  that $\round(\cpi) = \rnd+1$.
We prove it by the following steps: (1) we define two real vectors $\veta{0}$ and $\veta{1}$ such that the restriction of  $\FinalMeas{\HonA}{\cpi}{\vetas}$ to $\cpi_0$ and $\cpi_1$ is equal to $\FinalMeas{\HonA}{\cpi_0}{\veta{0}}$ and $\FinalMeas{\HonA}{\cpi_1}{\veta{1}}$ respectively; (2) we apply the induction hypothesis on the two latter measures; (3) if $\HonA$ controls $\Root(\cpi)$, we use the properties of $\rcA{k}$ -- as stated in \cref{claim:weights} -- to derive the lemma, whereas if $\HonB$ controls $\Root(\cpi)$, we derive the lemma from \cref{lemma:calculus1}.

All claims given in the context of this proof are proven in \cref{sec:MainIdealMissingProofs}. We defer handling the case that $\EdgeDist_{\cpi}(\EmptyString,b)\in\zo$ for some $b\in\zo$ (see the end of this proof) and assume for now that $\EdgeDist_{\cpi}(\EmptyString,0),\EdgeDist_{\cpi}(\EmptyString,1)\in(0,1)$. The real vectors $\veta{0}$ and $\veta{1}$ are defined as follows.
\begin{definition}\label{def:etaP}
Let $\veta{b}=\set{\eta^b_{(\HonC,j)}}_{(\HonC,j)\in[(\HonA,\dmsi)]}$, where for $(\HonC,j)\in[(\HonA,\dmsi)]$ and $b\in\zo$, let
\begin{align*}
        \eta^b_{(\HonC,j)} = \left\{
                                   \begin{array}{ll}
                                      0 & \EdgeDist_{\icpih{(\HonC,j)}{\vetas}}(\EmptyString,b)=0;\\
                                      \eta_{(\HonC,j)} & \EdgeDist_{\icpih{(\HonC,j)}{\vetas}}(\EmptyString,b)=1;\\
                                      \eta_{(\HonC,j)} & \EdgeDist_{\icpih{(\HonC,j)}{\vetas}}(\EmptyString,b)\notin\zo \land  (\HonC \hbox{ controls }\Root(\cpi)\lor \Smaller{\icpih{(\HonC,j)}{\vetas}}{b});\\
                                      \frac{\DMExRes{(\HonC,j)}{1-b}}{\DMExRes{(\HonC,j)}{b}} \cdot \eta_{(\HonC,j)} &  \hbox{otherwise;}
                                   \end{array}
                                 \right.,
\end{align*}
where $\DMExRes{(\HonC,j)}{b}=\eex{\LDist{\left(\icpih{(\HonC,j)}{\vetas}\right)_{b}}}{\DomMeas{\left(\icpih{(\HonC,j)}{\vetas}\right)_{b}}{\HonC}}$ and $\Smaller{\icpih{(\HonC,j)}{\vetas}}{b}=1$  if $\DMExRes{(\HonC,j)}{b} \leq \DMExRes{(\HonC,j)}{1-b}$.\footnote{Note that the  definition of $\eta^b$ follows the same lines of the definition of the dominated measure (given in \cref{def:Dominated}).}
\end{definition}

Given the real vector $\veta{b}$, consider the dominated submeasure sequence $\veta{b}$ induces on the subprotocol $\cpi_b$. At first glance, the relation of this submeasure sequence to the dominated submeasure sequence $\vetas$ induces on $\cpi$, is unclear; nonetheless, we manage to prove the following key observation.
\begin{claim}\label{claim:RestrictFinal}
It holds that $\FinalMeas{\HonA}{\cpi_b}{\veta{b}} \equiv \Restrict{\FinalMeas{\HonA}{\cpi}{\vetas}}{b}$ for both $b\in\zo$.
\end{claim}
Namely, taking $(\HonA,\dmsi,\veta{b})$-$\DMS{\cpi_b}$ -- the dominated submeasures defined \wrt $\cpi_b$ and $\veta{b}$ -- and constructing from it the measure $\FinalMeas{\HonA}{\cpi_b}{\veta{b}}$, results in the same measure as taking $(\HonA,\dmsi,\vetas)$-$\DMS{\cpi}$ -- the dominated submeasures defined \wrt $\cpi$ and $\vetas$ -- and constructing from it the measure $\FinalMeas{\HonA}{\cpi}{\vetas}$ while restricting the latter to $\cpi_b$.

Given the above fact, we can use our induction hypothesis on the subprotocols $\cpi_0$ and $\cpi_1$ \wrt the real vectors $\veta{0}$ and $\veta{1}$, respectively. For $b\in\zo$ and $j\in(\dmsi)$, let $\alpha^b_j \eqdef \DMEx{\cpi_b}{\veta{b}}{(\HonA,j)} \ (\eqdef \eex{\LDist{\icpib{(\HonA,j)}{\veta{b}}{b}}}{\idsm{(\HonA,j)}{\cpi_b}{\veta{b}}})$, and for $j\in(\dmsi-1)$ let $\beta^b_j \eqdef \DMEx{\cpi_b}{\veta{b}}{(\HonB,j)}$. Assuming that $\Val(\cpi_1)>0$, then
\begin{align}\label{eq:MainLemmaIndRight}
\eex{\LDist{\Restrict{\ProArc{k}}{1}}}{\Restrict{\FinalMeas{\HonA}{\cpi}{\vetas}}{1}} \overset{(1)}{=} \eex{\LDist{\ProArcb{k}{1}}}{\FinalMeas{\HonA}{\cpi_1}{\veta{1}}}
\overset{(2)}{\geq} \frac{\sum_{j=0}^{\dmsi} \alpha^1_j\prod_{t=0}^{j-1}(1-\beta^1_t)^{k+1}(1-\alpha^1_t)}{\prod_{i=0}^{k-1}\Val\paren{\ProArcP{i}_1}}.
\end{align}
where (1) follows from \cref{fact:switchOrder,claim:RestrictFinal}, and (2) follows from the induction hypothesis. Similarly, if $\Val(\cpi_0)>1$, then
\begin{align}\label{eq:MainLemmaIndLeft}
\eex{\LDist{\Restrict{\ProArc{k}}{0}}}{\Restrict{\FinalMeas{\HonA}{\cpi}{\vetas}}{0}} = \eex{\LDist{\ProArcb{k}{0}}}{\FinalMeas{\HonA}{\cpi_0}{\veta{0}}}
\geq \frac{\sum_{j=0}^{\dmsi} \alpha^0_j\prod_{t=0}^{j-1}(1-\beta^0_t)^{k+1}(1-\alpha^0_t)}{\prod_{i=0}^{k-1}\Val\paren{\ProArcP{i}_0}}.
\end{align}

In the following we use the fact that the dominated submeasure sequence of one of the subprotocols is at least as long as the submeasure sequence of the protocol itself. Specifically, we show the following.

\begin{definition}\label{def:maxInd}
For $b\in\zo$, let $\dmsi^b = \min\set{\set{j\in(\dmsi) \colon \alpha^b_j = 1 \lor \beta^b_j = 1}\cup \set{\dmsi}}$.
\end{definition}
Assuming \wlg (and throughout the proof of the lemma) that $\dmsi^1\leq \dmsi^0$, we have the following claim (proven in \cref{sec:MainIdealMissingProofs}).
\begin{claim}\label{claim:dmsiequal}
Assume that $\dmsi^1\leq \dmsi^0$, then $\dmsi^0=\dmsi$.
\end{claim}

We are now ready to prove the lemma by separately considering which party controls the root of $\cpi$.
\begin{description}

\item[$\HonA$ controls $\Root(\cpi)$ and $\Val(\cpi_0),\Val(\cpi_1)>0$.] Under these assumptions, we can apply the induction hypothesis on both subtrees (namely, we can use  \cref{eq:MainLemmaIndLeft,eq:MainLemmaIndRight}). Let $p = \EdgeDist_{\cpi}(\EmptyString,0)$. Compute
\begin{align}\label{eq:MainLemmaA1}
\lefteqn{\eex{\LDist{\ProArc{k}}}{\FinalMeas{\HonA}{\cpi}{\vetas}}}\\
&= \EdgeDist_{\ProArcP{k}}(\EmptyString,0)\cdot \eex{\LDist{\ProArcP{k}_0}}{\Restrict{\FinalMeas{\HonA}{\cpi}{\vetas}}{0}} + \EdgeDist_{\ProArcP{k}}(\EmptyString,1)\cdot \eex{\LDist{\ProArcP{k}_1}}{\Restrict{\FinalMeas{\HonA}{\cpi}{\vetas}}{1}} \nonumber\\
&\overset{(1)}{=} p \cdot \frac{\prod_{i=0}^{k-1}\Val\paren{\ProArcP{i}_{0}}}{\prod_{i=0}^{k-1}\Val\paren{\ProArc{i}}} \cdot \eex{\LDist{\ProArcP{k}_0}}{\Restrict{\FinalMeas{\HonA}{\cpi}{\vetas}}{0}}\nonumber\\ &\quad +(1-p) \cdot \frac{\prod_{i=0}^{k-1}\Val\paren{\ProArcP{i}_{1}}}{\prod_{i=0}^{k-1}\Val(\ProArc{i})} \cdot \eex{\LDist{\ProArcP{k}_1}}{\Restrict{\FinalMeas{\HonA}{\cpi}{\vetas}}{1}}\nonumber \\
&\overset{(2)}{\geq} p \cdot \frac{\prod_{i=0}^{k-1}\Val\paren{\ProArcP{i}_{0}}}{\prod_{i=0}^{k-1}\Val\paren{\ProArc{i}}} \cdot \frac{\sum_{j=0}^{\dmsi} \alpha^0_j\prod_{t=0}^{j-1}(1-\beta^0_t)^{k+1}(1-\alpha^0_t)}{\prod_{i=0}^{k-1}\Val\paren{\ProArcP{i}_0}}\nonumber \\
&\quad + (1-p) \cdot \frac{\prod_{i=0}^{k-1}\Val\paren{\ProArcP{i}_{1}}}{\prod_{i=0}^{k-1}\Val(\ProArc{i})} \cdot \frac{\sum_{j=0}^{\dmsi} \alpha^1_j\prod_{t=0}^{j-1}(1-\beta^1_t)^{k+1}(1-\alpha^1_t)}{\prod_{i=0}^{k-1}\Val\paren{\ProArcP{i}_1}}\nonumber \\
&= \frac{p \cdot \left(\sum_{j=0}^\dmsi \alpha^0_j\prod_{t=0}^{j-1}(1-\beta^0_t)^{k+1}(1-\alpha^0_t)\right)}{\prod_{i=0}^{k-1}\Val(\ProArc{i})}
+ \frac{(1-p)\cdot\left(\sum_{j=0}^{\dmsi} \alpha^1_j\prod_{t=0}^{j-1}(1-\beta^1_t)^{k+1}(1-\alpha^1_t)\right)}{\prod_{i=0}^{k-1}\Val(\ProArc{i})},\nonumber
\end{align}
where (1) follows from \Cref{claim:weights} and (2) follows from \cref{eq:MainLemmaIndRight,eq:MainLemmaIndLeft}.

Our next step is to establish a connection between the above $\set{\alpha^0_j,\alpha^1_j}_{j\in(z)}$ and $\set{\beta^0_j,\beta^1_j}_{j\in(z-1)}$ to $\set{\alpha_j}_{j\in(z)}$ and $\set{\beta_j}_{j\in(z-1)}$ (appearing in the lemma's statement). We prove the following claims.
\begin{claim}\label{clm:RelatingInductionConstants}
If $\HonA$ controls $\Root(\cpi)$, it holds that $\beta^0_j=\beta_j$ for every $j\in(z-1)$ and $\beta^1_j=\beta_j$ for every $j\in(z^1-1)$.
\end{claim}
It is a direct implication of \cref{prop:DomMeasPro} that $\beta^0_j = \beta^1_j=\beta_j$ for $j\in(\dmsi^1-1)$. Moreover, $\beta^0_j=\beta_j$ for every $\dmsi^1 \leq j \leq \dmsi-1$. The latter is harder to grasp without the technical proof of the claim, which is provided in \cref{sec:MainIdealMissingProofs}.

\begin{claim}\label{claim:AlphaIs1}
If $\HonA$ controls $\Root(\cpi)$ and $\dmsi^1<\dmsi$, it holds that $\alpha^1_{\dmsi^1}=1$.
\end{claim}
By \cref{clm:RelatingInductionConstants} it follows that as long as an undefined protocol was not reached in one of the subprotocols, then $\beta^0_j=\beta^1_j=\beta_j$. Assuming that $\dmsi^1<\dmsi$ and $\beta^1_{\dmsi^1}=1$, it would have followed that $\beta_{\dmsi^1}=1$, and an undefined protocol is reached in the original protocol before $\dmsi$, a contradiction to our assumption. (Again, see \cref{sec:MainIdealMissingProofs} for the formal proof.)

\cref{clm:RelatingInductionConstants,claim:AlphaIs1,eq:MainLemmaA1} yield that
\begin{align}\label{eq:MainLemmaA2}
\eex{\LDist{\ProArc{k}}}{\FinalMeas{\HonA}{\cpi}{\vetas}}\geq \frac{\sum_{j=0}^{\dmsi}\prod_{t=0}^{j-1}(1-\beta_t)^{k+1} \!\left(p\cdot\alpha^0_j \prod_{t=0}^{j-1}(1-\alpha^0_t) +  (1-p)\cdot\alpha^1_j \!\cdot \prod_{t=0}^{j-1}(1-\alpha^1_t) \right)}{\prod_{i=0}^{k-1}\Val(\ProArc{i})}.
\end{align}

The proof of this case is concluded by plugging the next claim into \cref{eq:MainLemmaA2}.
\begin{claim}\label{claim:alphaRelation}
If $\HonA$ controls $\Root(\cpi)$ it holds that
  \begin{align*}
  \alpha_j \cdot \prod_{t=0}^{j-1}(1-\alpha_t) = p\cdot\alpha^0_j \cdot \prod_{t=0}^{j-1}(1-\alpha^0_t) +  (1-p)\cdot\alpha^1_j \cdot\prod_{t=1}^{j-1}(1-\alpha^1_t)
  \end{align*}
  for any $j\in(\dmsi)$.
\end{claim}
\cref{claim:alphaRelation} is proven in \cref{sec:MainIdealMissingProofs}, but informally it holds since the probability of visiting the left-hand [\resp right-hand] subprotocol in the conditional protocol $\icpih{(\HonA,j)}{\vetas}$ (in which $\alpha_j$ is defined) is $p\cdot \prod_{t=0}^{j-1}(1-\alpha^0_t)/\prod_{t=0}^{j-1}(1-\alpha_t)$ [\resp $(1-p)\cdot \prod_{t=0}^{j-1}(1-\alpha^1_t)/\prod_{t=0}^{j-1}(1-\alpha_t)$]. Since $\alpha_j$ is defined to be the expected value of some measure in the above conditional protocol, its value is a linear combination of $\alpha^0_j$ and $\alpha^1_j$, with the coefficient being the above probabilities.

\item[$\HonA$ controls $\Root(\cpi)$ and $\Val(\cpi_0)>\Val(\cpi_1)=0$.]
Under these assumptions, we can still use the induction hypothesis for the left-hand subprotocol $\cpi_0$, where for right-hand  subprotocol $\cpi_1$, we argue the following.
\begin{claim}\label{claim:APi0is0}
If $\Val(\cpi_1)=0$, it holds that  $\Restrict{\FinalMeas{\HonA}{\cpi}{\vetas}}{1}\equiv 0$.\footnote{That is, $\Restrict{\FinalMeas{\HonA}{\cpi}{\vetas}}{1}$ is the zero measure.}
\end{claim}
\cref{claim:APi0is0} holds since according to \cref{claim:RestrictFinal} we can simply argue that $\FinalMeas{\HonA}{\cpi_1}{\veta{1}}$ is the zero measure, and this holds since the latter measure is a combination of $\HonA$-dominated measures, all of which are the zero measure in a zero-value protocol.

Using \cref{claim:APi0is0}, similar computations to the ones in \cref{eq:MainLemmaA1} yield that
\begin{align}\label{eq:MainLemmaA3}
\lefteqn{\eex{\LDist{\ProArc{k}}}{\FinalMeas{\HonA}{\cpi}{\vetas}}}\\
&= \EdgeDist_{\ProArcP{k}}(\EmptyString,0)\cdot \eex{\LDist{\ProArcP{k}_0}}{\Restrict{\FinalMeas{\HonA}{\cpi}{\vetas}}{0}} + \EdgeDist_{\ProArcP{k}}(\EmptyString,1)\cdot \eex{\LDist{\ProArcP{k}_1}}{\Restrict{\FinalMeas{\HonA}{\cpi}{\vetas}}{1}} \nonumber\\
&\geq p \cdot \frac{\prod_{i=0}^{k-1}\Val\paren{\ProArcP{i}_{0}}}{\prod_{i=0}^{k-1}\Val\paren{\ProArc{i}}} \cdot \frac{\sum_{j=0}^{\dmsi} \alpha^{0}_j\prod_{t=0}^{j-1}(1-\beta^{0}_t)^{k+1}(1-\alpha^{0}_t)}{\prod_{i=0}^{k-1}\Val\paren{\ProArcP{i}_{0}}} \nonumber\\
&= \frac{p \cdot \left(\sum_{j=0}^{\dmsi} \alpha^{0}_j\prod_{t=0}^{j-1}(1-\beta^{0}_t)^{k+1}(1-\alpha^{0}_t)\right)}{\prod_{i=0}^{k-1}\Val(\ProArc{i})}.\nonumber
\end{align}
Using a similar argument to that of \cref{eq:MainLemmaA2}, combining \cref{clm:RelatingInductionConstants,eq:MainLemmaA3} yields that
\begin{align}\label{eq:MainLemmaA4}
\eex{\LDist{\ProArc{k}}}{\FinalMeas{\HonA}{\cpi}{\vetas}} \geq \frac{\sum_{j=0}^{\dmsi}\prod_{t=0}^{j-1}(1-\beta_t)^{k+1} \left[p\cdot\alpha^{0}_j \prod_{t=0}^{j-1}(1-\alpha^{0}_t)\right]}{\prod_{i=0}^{k-1}\Val(\ProArc{i})}.
\end{align}
The proof of this case is concluded by plugging the next claim (proven in \cref{sec:MainIdealMissingProofs}) into \cref{claim:alphaRelation}, and plugging the result into \cref{eq:MainLemmaA4}.
\begin{claim}\label{claim:alphaIsZero}
If $\Val(\cpi_1)=0$, it holds that $\alpha^1_j=0$ for every $j\in(\dmsi)$.
\end{claim}

\item[$\HonA$ controls $\Root(\cpi)$ and $\Val(\cpi_1)>\Val(\cpi_0)=0$.]
The proof of the lemma under these assumptions is analogous to the previous case.

\end{description}

We have concluded the proof for cases in which $\HonA$ controls $\Root(\cpi)$, and now proceed to prove the cases in which $\HonB$ controls $\Root(\cpi)$. Roughly speaking, $\HonA$ and $\HonB$ switched roles, and claims true before regarding $\beta_j$ are now true for $\alpha_j$, and vice versa. Moreover, the analysis above relies on the probabilities that the recursive biased-continuation attacker visits the subprotocols $\cpi_0$ and $\cpi_1$ when it plays the role of $\HonA$ and controls $\Root(\cpi)$. When $\HonB$ controls $\Root(\cpi)$, however, these probabilities do not change (namely, they remain $p$ and $1-p$ respectively). To overcome this difficulty we use a convex type argument stated in \cref{lemma:calculus1}.

\begin{description}
 \item[$\HonB$ controls $\Root(\cpi)$ and $\Val(\cpi_0),\Val(\cpi_1)>0$.] In this case \cref{eq:MainLemmaIndLeft,eq:MainLemmaIndRight} hold.

    Compute
\begin{align}\label{eq:MainLemmaB1}
\lefteqn{\eex{\LDist{\ProArc{k}}}{\FinalMeas{\HonA}{\cpi}{\vetas}}}\\
&= p\cdot \eex{\LDist{\Restrict{\ProArc{k}}{0}}}{\Restrict{\FinalMeas{\HonA}{\cpi}{\vetas}}{0}} + (1-p)\cdot \eex{\LDist{\Restrict{\ProArc{k}}{1}}}{\Restrict{\FinalMeas{\HonA}{\cpi}{\vetas}}{1}} \nonumber\\
&\geq p \cdot \frac{\sum_{j=0}^{\dmsi} \alpha^0_j\prod_{t=0}^{j-1}(1-\beta^0_t)^{k+1}(1-\alpha^0_t)}{\prod_{i=0}^{k-1}\Val\paren{\Restrict{\ProArc{i}}{0}}}
+ (1-p) \cdot \frac{\sum_{j=0}^{\dmsi} \alpha^1_j\prod_{t=0}^{j-1}(1-\beta^1_t)^{k+1}(1-\alpha^1_t)}{\prod_{i=0}^{k-1}\Val\paren{\Restrict{\ProArc{i}}{1}}},\nonumber
\end{align}
where the inequality follows from \cref{eq:MainLemmaIndRight,eq:MainLemmaIndLeft}.
If $\HonB$ controls $\Root(\cpi)$, we can prove the next claims (proven in \cref{sec:MainIdealMissingProofs}), analogous to \cref{clm:RelatingInductionConstants,claim:AlphaIs1}.
\begin{claim}\label{clm:RelatingInductionConstantsB}
If $\HonB$ controls $\Root(\cpi)$, it holds that $\alpha^0_j=\alpha_j$ for every $j\in(\dmsi)$ and that $\alpha^1_j = \alpha_j$ for every $j\in(\dmsi^1)$.
\end{claim}
\begin{claim}\label{claim:BetaIs1}
If $\HonB$ controls $\Root(\cpi)$ and $\dmsi^1<\dmsi$, it holds that $\beta^1_{\dmsi^1}=1$.
\end{claim}
\cref{clm:RelatingInductionConstantsB} and \cref{eq:MainLemmaB1} yield that
\begin{align}\label{eq:MainLemmaB2}
\lefteqn{\eex{\LDist{\ProArc{k}}}{\FinalMeas{\HonA}{\cpi}{\vetas}}} \\
&\geq \sum_{j=0}^{\dmsi} \alpha_j\prod_{t=0}^{j-1}(1-\alpha_t)\left(p\cdot\frac{\prod_{t=0}^{j-1}(1-\beta^0_t)^{k+1}}{\prod_{i=0}^{k-1}\Val\paren{\Restrict{\ProArc{i}}{0}}} + (1-p)\cdot\frac{\prod_{t=0}^{j-1}(1-\beta^1_t)^{k+1}}{\prod_{i=0}^{k-1}\Val\paren{\Restrict{\ProArc{i}}{1}}} \right). \nonumber
\end{align}
Applying the convex type inequality given in  \cref{lemma:calculus1} for each summand in the right-hand side of \cref{eq:MainLemmaB2} \wrt $x=\prod_{t=0}^{j-1}(1-\beta^0_t)$, $y=\prod_{t=0}^{j-1}(1-\beta^1_t)$, $a_i=\Val(\ProArc{i-1}_0)$, $b_i=\Val(\ProArc{i-1}_1)$, $p_0=p$ and $p_1=1-p$, and plugging into \cref{eq:MainLemmaB2} yield that
\begin{align}
\eex{\LDist{\ProArc{k}}}{\FinalMeas{\HonA}{\cpi}{\vetas}}
\geq \frac{\sum_{j=0}^{\dmsi} \alpha_j\prod_{t=0}^{j-1}(1-\alpha_t)\paren{p\cdot\prod_{t=0}^{j-1}(1-\beta^0_t) + (1-p)\cdot\prod_{t=0}^{j-1}(1-\beta^1_t)}^{k+1}}{\prod_{i=0}^{k-1}\left(p\cdot\Val\paren{\Restrict{\ProArc{i}}{0}} + (1-p)\cdot \Val\paren{\Restrict{\ProArc{i}}{1}}\right)}.
\end{align}
We conclude the proof of this case by observing that for every $i\in(k-1)$ it holds that $\Val\paren{\ProArc{i}} = p\cdot\Val\paren{\Restrict{\ProArc{i}}{0}} + (1-p)\cdot \Val\paren{\Restrict{\ProArc{i}}{1}}$, and using the next claim (proven in \cref{sec:MainIdealMissingProofs}), analogous to \cref{claim:alphaRelation}.
\begin{claim}\label{claim:betaRelation}
If $\HonB$ controls $\Root(\cpi)$, it holds that
  \begin{align*}
  \prod_{t=0}^{j-1}(1-\beta_t) = p\cdot \prod_{t=0}^{j-1}(1-\beta^0_t) +  (1-p)\cdot\prod_{t=0}^{j-1}(1-\beta^1_t).
  \end{align*}
\end{claim}

\item[$\HonB$ controls $\Root(\cpi)$ and $\Val(\cpi_0)>\Val(\cpi_1)=0$.] In this case, \cref{clm:RelatingInductionConstants,clm:RelatingInductionConstantsB} yield that $\alpha_j=0$ for any $j\in(\dmsi^1)$. Hence, it suffices to prove that
\begin{align}\label{eq:MainLemmaB3}
 \eex{\LDist{\ProArc{k}}}{\FinalMeas{\HonA}{\cpi}{\vetas}} \geq \frac{\sum_{j=\dmsi^1+1}^{\dmsi} \alpha_j\prod_{t=0}^{j-1}(1-\beta_t)^{k+1}(1-\alpha_t)}{\prod_{i=0}^{k-1}\Val{(\ProArc{i})}}.
\end{align}
Thus, the proof immediately follows if $\dmsi^1=\dmsi$, and in the following we  assume that $\dmsi^1<\dmsi$.

As in \cref{eq:MainLemmaB1}, compute
\begin{align}\label{eq:MainLemmaB4}
\eex{\LDist{\ProArc{k}}}{\FinalMeas{\HonA}{\cpi}{\vetas}} &= p\cdot \eex{\LDist{\Restrict{\ProArc{k}}{0}}}{\Restrict{\FinalMeas{\HonA}{\cpi}{\vetas}}{0}} + (1-p)\cdot \eex{\LDist{\Restrict{\ProArc{k}}{1}}}{\Restrict{\FinalMeas{\HonA}{\cpi}{\vetas}}{1}} \\
&\geq p \cdot \frac{\sum_{j=0}^{\dmsi} \alpha^{0}_j\prod_{t=0}^{j-1}(1-\beta^{0}_t)^{k+1}(1-\alpha^{0}_t)}{\prod_{i=0}^{k-1}\Val\paren{\Restrict{\ProArc{i}}{0}}},\nonumber
\end{align}
where the inequality follows \cref{eq:MainLemmaIndLeft,claim:APi0is0}. \cref{clm:RelatingInductionConstantsB} now yields
\begin{align}\label{eq:MainLemmaB5}
\eex{\LDist{\ProArc{k}}}{\FinalMeas{\HonA}{\cpi}{\vetas}} \geq \sum_{j=0}^{\dmsi} \alpha_j\prod_{t=0}^{j-1}(1-\alpha_t)\cdot\frac{p\cdot \prod_{t=0}^{j-1}(1-\beta^{0}_t)^{k+1}}{\prod_{i=0}^{k-1}\Val\paren{\Restrict{\ProArc{i}}{0}}},
\end{align}
where \cref{clm:RelatingInductionConstantsB} yields
\begin{align}\label{eq:MainLemmaB6}
\eex{\LDist{\ProArc{k}}}{\FinalMeas{\HonA}{\cpi}{\vetas}} \geq \sum_{j=\dmsi^1+1}^{\dmsi} \alpha_j\prod_{t=0}^{j-1}(1-\alpha_t)\cdot\frac{p\cdot \prod_{t=0}^{j-1}(1-\beta^{0}_t)^{k+1}}{\prod_{i=0}^{k-1}\Val\paren{\Restrict{\ProArc{i}}{0}}}.
\end{align}
Multiplying both the numerator and the denominator for every summand of \cref{eq:MainLemmaB6} with $p^k$ yields
\begin{align}\label{eq:MainLemmaB7}
\eex{\LDist{\ProArc{k}}}{\FinalMeas{\HonA}{\cpi}{\vetas}} \geq \sum_{j=\dmsi^1+1}^{\dmsi} \alpha_j\prod_{t=0}^{j-1}(1-\alpha_t)\cdot\frac{\paren{p\cdot \prod_{t=0}^{j-1}(1-\beta^{0}_t)}^{k+1}}{\prod_{i=0}^{k-1}p\cdot\Val\paren{\Restrict{\ProArc{i}}{0}}}.
\end{align}
\cref{eq:MainLemmaB3}, and hence the proof of this case, is derived by observing that $\Val{(\ProArc{i})} = p\cdot \Val\paren{\Restrict{\ProArc{i}}{0}}$ for every $i\in(k-1)$,\footnote{Recall that if $\Val\paren{\HonA,\HonB}=0$, then $\Val\paren{\rcA{i},\HonB}=0$ for every $i\in\N$.} and plugging \cref{claim:BetaIs1,claim:betaRelation} into \cref{eq:MainLemmaB7}.

\item[$\HonB$ controls $\Root(\cpi)$ and $\Val(\cpi_1)>\Val(\cpi_0)=0$.]
Analogously to \cref{claim:alphaIsZero}, it holds that $\alpha^0_j=0$ for every $j\in(\dmsi)$. \cref{clm:RelatingInductionConstantsB} yields that $\alpha_j=0$ for every $j\in(\dmsi)$. The proof of this case trivially follows since
\begin{align*}
\frac{\sum_{j=0}^\dmsi \alpha_j\prod_{t=0}^{j-1}(1-\beta_t)^{k+1}(1-\alpha_t)}{\prod_{i=0}^{k-1}\Val{(\ProArc{i})}}=0.
\end{align*}
\end{description}

The above case analysis concludes the proof of the lemma when assuming that $\EdgeDist_\cpi(\EmptyString,b)\notin\zo$ for both $b\in\zo$. Assume that  $\EdgeDist_\cpi(\EmptyString,b)=1$ for some $b\in\zo$. Since, by assumption, $\Val(\cpi)>0$, it follows that $\Val(\cpi_b)>0$. Moreover, the definition of conditional protocols (\cref{def:CondProtocol}) yields that $\EdgeDist_{\icpih{(\HonC,j)}{\vetas}}(\EmptyString,b)=1$ and $\EdgeDist_{\icpih{(\HonC,j)}{\vetas}}(\EmptyString,1-b)=0$ for any $(\HonC,j)\in[(\HonA,\dmsi)]$ (regardless of which party controls $\Root(\cpi)$). By defining $\veta{b}=\vetas$, the definition of the dominated measure (\cref{def:Dominated}) yields that $\alpha_j=\alpha^b_j$ for every $j\in(\dmsi)$ and that $\beta_j=\beta^b_j$ for every $j\in(\dmsi-1)$. The proof of this case immediately follows from the induction hypothesis on $\cpi_b$.
\end{proof}

\subsubsection{Missing Proofs}\label{sec:MainIdealMissingProofs}
This section is dedicated to proving deferred statements used  in the proof of \cref{lemma:IdealMainLemmaDet}. We assume a fixed protocol $\cpi$, fixed real vector $\vetas=\paren{\eta_{(\HonA,0)}, \eta_{(\HonB,0)}, \ldots, \eta_{(\HonB,\dmsi-1)}, \eta_{(\HonA,\dmsi)}}$ and a fixed positive integer $k$. We also assume that $\icpih{(\HonA,\dmsi)}{\vetas}\neq \perp$, $\dmsi^1\leq\dmsi^0$ and $\EdgeDist_{\cpi}(\EmptyString,b)\in(0,1)$ for both $b\in\zo$. Recall that we defined two real vectors $\veta{0}$ and $\veta{1}$ (\cref{def:etaP}), and for $b\in\zo$ we defined $\alpha^b_j \eqdef \DMEx{\cpi_b}{\veta{b}}{(\HonA,j)} \ (\eqdef \eex{\LDist{\icpib{(\HonA,j)}{\veta{b}}{b}}}{\idsm{(\HonA,j)}{\cpi_b}{\veta{b}}})$ for $j\in(\dmsi)$, and  $\beta^b_j \eqdef \DMEx{\cpi_b}{\veta{b}}{(\HonB,j)}$, for $j\in(\dmsi-1)$.

We begin with the following proposition, which underlies many of the claims to follow.
\begin{proposition}\label{claim:ConRestrict}
For $b\in \zo$ and $(\HonC,j)\in[(\HonA,\dmsi)]$, it holds that
\begin{enumerate}
  \item \label{item:protocol} $\Restrict{\icpih{(\HonC,j)}{\vetas}}{b}=\icpib{(\HonC,j)}{\veta{b}}{b}$; and
  \item \label{item:measure} $\Restrict{\idsm{(\HonC,j)}{\cpi}{\vetas}}{b} \equiv \idsm{(\HonC,j)}{\cpi_b}{\veta{b}}$.
\end{enumerate}
\end{proposition}
Namely, the restriction of $\icpih{(\HonC,j)}{\vetas}$ (the $(\HonC,j)$'th conditional protocol \wrt $\cpi$ and $\vetas$) to its $b$'th subtree is equal to the $(\HonC,j)$'th conditional protocol defined \wrt $\cpi_b$ ($b$'th subtree of $\cpi$) and $\veta{b}$. Moreover, the result of multiplying the  $\HonC$-dominated measure of $\icpih{(\HonC,j)}{\vetas}$ by $\eta_{(\HonC,j)}$, and then restricting it to the subtree $\Restrict{\icpih{(\HonC,j)}{\vetas}}{b}$, is equivalent to multiplying the $\HonC$-dominated measure of $\icpib{(\HonC,j)}{\veta{b}}{b}$ by $\eta^b_{(\HonC,j)}$.\footnote{Note that \cref{item:protocol} is not immediate. Protocol $\Restrict{\icpih{(\HonC,j)}{\vetas}}{b}$ is a restriction of a protocol defined on the root of $\cpi$, whereas $\icpib{(\HonC,j)}{\veta{b}}{b}$ is a protocol defined on the root of $\cpi_b$.}
\begin{proof}[Proof of \cref{claim:ConRestrict}]
The proof is by induction on the ordered pairs $[(\HonA,\dmsi)]$.

\paragraph{Base case.} Recall that the first pair of $[(\HonA,\dmsi)]$ is $(\HonA,0)$. \cref{def:refinements} yields that $\icpih{(\HonA,0)}{\vetas}=\cpi$ and that $\icpib{(\HonA,0)}{\veta{b}}{b}=\cpi_b$, yielding that \cref{item:protocol} holds for $(\HonA,0)$. As for \cref{item:measure}, by \cref{def:Dominated} and the assumption that $\EdgeDist_{\cpi}(\EmptyString,b)\in(0,1)$ for both $b\in\zo$, it holds that
\begin{align*}
      \Restrict{\idsm{(\HonA,0)}{\cpi}{\vetas}}{b} \equiv \Restrict{\eta_{(\HonA,0)}\cdot\DomMeas{\cpi}{\HonA}}{b} \equiv \left\{
                                   \begin{array}{ll}
                                      \eta_{(\HonA,0)}\cdot \DomMeas{\cpi_b}{\HonA} &  \HonA \hbox{ controls }\Root(\cpi)\lor \Smaller{\cpi}{b};\\
                                      \eta_{(\HonA,0)}\cdot\frac{\DMExRes{(\HonA,0)}{1-b}}{\DMExRes{(\HonA,0)}{b}} \cdot  \DomMeas{\cpi_b}{\HonA} &  \hbox{otherwise.}
                                   \end{array}
                                 \right.
\end{align*}
The proof that \cref{item:measure} holds for $(\HonA,0)$ now follows from \cref{def:etaP}.

\paragraph{Induction step.} Fix $(\HonC,j)\in[(\HonA,\dmsi)]$ and assume the claim holds for $\pred(\HonC,j)$. Using the induction hypothesis, we first prove \cref{item:protocol} for $(\HonC,j)$. Next, using the fact that \cref{item:protocol} holds for $(\HonC,j)$, we prove \cref{item:measure}.
\begin{description}
  \item[Proving \cref{item:protocol}.] By \cref{def:refinements}, it holds that
      \begin{align*}
      \Restrict{\icpih{(\HonC,j)}{\vetas}}{b} &= \Restrict{\CondProP{\icpih{\pred(\HonC,j)}{\vetas}}{\idsm{\pred(\HonC,j)}{\cpi}{\vetas}}}{b} \\
      &= \CondPro{\Restrict{\icpih{\pred(\HonC,j)}{\vetas}}{b}}{\Restrict{\idsm{\pred(\HonC,j)}{\cpi}{\vetas}}{b}} \\
      &\overset{(1)}{=} \CondProP{\icpib{\pred(\HonC,j)}{\veta{b}}{b}}{\idsm{\pred(\HonC,j)}{\cpi_b}{\veta{b}}} \\
      &= \icpib{(\HonC,j)}{\veta{b}}{b},
      \end{align*}
      where (1) follows from the induction hypothesis.

  \item[Proving \cref{item:measure}.] Similarly to the base case, \cref{def:Dominated} yields that
    \begin{align*}
    \Restrict{\idsm{(\HonC,j)}{\cpi}{\vetas}}{b} \equiv \left\{
                           \begin{array}{ll}
                              0 & \EdgeDist_{\icpih{(\HonC,j)}{\vetas}}(\EmptyString,b)=0;\\
                              \eta_{(\HonC,j)}\cdot \DomMeas{\left(\icpih{(\HonC,j)}{\vetas}\right)_b}{\HonC}  & \EdgeDist_{\icpih{(\HonC,j)}{\vetas}}(\EmptyString,b)=1;\\
                              \eta_{(\HonC,j)}\cdot \DomMeas{\left(\icpih{(\HonC,j)}{\vetas}\right)_b}{\HonC}  & \EdgeDist_{\icpih{(\HonC,j)}{\vetas}}(\EmptyString,b)\notin\zo \land \\
                                                    & \paren{\HonC\hbox{ controls }\Root(\cpi) \lor  \Smaller{\icpih{(\HonC,j)}{\vetas}}{b}};\\
                              \eta_{(\HonC,j)} \cdot \frac{\DMExRes{(\HonC,j)}{1-b}}{\DMExRes{(\HonC,j)}{b}} \cdot  \DomMeas{\left(\icpih{(\HonC,j)}{\vetas}\right)_{b}}{\HonC} &  \hbox{otherwise,}
                               \end{array}
                             \right.
    \end{align*}
    and the  proof follows from \cref{item:protocol,def:etaP}.
\end{description}
\end{proof}

Recall that the real numbers $\alpha^b_j$ and $\beta^b_j$ were defined to be the expected values of the $(\HonA,j)$'th and $(\HonB,j)$'th dominated measures in the sequence $(\HonA,\dmsi,\veta{b})$-$\DMS{\cpi_b}$, respectively (see the proof of \cref{lemma:IdealMainLemmaDet}). Following \cref{claim:ConRestrict}, we could equivalently define $\alpha^b_j$ and $\beta^b_j$ \wrt the sequence $(\HonA,\dmsi,\vetas)$-$\DMS{\cpi}$.
\begin{proposition}\label{claim:ChangeContext}
For both $b\in\zo$, it holds that
\begin{enumerate}
  \item $\alpha^b_j=\eex{\LDist{\Restrict{\icpih{(\HonA,j)}{\vetas}}{b}}}{\Restrict{\idsm{(\HonA,j)}{\cpi}{\vetas}}{b}}$ for every $j\in(\dmsi)$; and
  \item $\beta^b_{j}=\eex{\LDist{\Restrict{\icpih{(\HonB,j)}{\vetas}}{b}}}{\Restrict{\idsm{(\HonB,j)}{\cpi}{\vetas}}{b}}$ for every $j\in(\dmsi-1)$.
\end{enumerate}
\end{proposition}
\begin{proof}
Immediately follows \cref{claim:ConRestrict}.
\end{proof}

\cref{claim:ChangeContext} allows us to use \cref{prop:DomMeasPro} in order to analyze the connections between $\alpha^0_j$ and $\alpha^1_j$ to $\alpha_j$, and similarly between $\beta^0_j$ and $\beta^1_j$ to $\beta_j$. Towards this goal, we analyze the edge distribution of the conditional protocols defined in the procedure that generates the  measure sequence $(\HonA,\dmsi,\vetas)$-$\DMS{\cpi}$.

\begin{proposition}\label{claim:EdgeDistAna}
The following holds for both $b\in\zo$.
\begin{enumerate}
  \item $\HonA$ controls $\Root(\cpi)$ $\implies$  \label{claim:EdgeDistAna:1}
  \begin{enumerate}
    \item  $\EdgeDist_{\icpih{(\HonA,j)}{\vetas}}(\EmptyString,b) = \EdgeDist_{\cpi}(\EmptyString,b)\cdot\frac{\prod_{t=0}^{j-1}\left(1-\alpha^b_t\right)}{\prod_{t=0}^{j-1}\left(1-\alpha_t\right)}$ for all $j\in (\dmsi)$.
    \item $\EdgeDist_{\icpih{(\HonB,j)}{\vetas}}(\EmptyString,b) = \EdgeDist_{\cpi}(\EmptyString,b)\cdot\frac{\prod_{t=0}^{j}\left(1-\alpha^b_t\right)}{\prod_{t=0}^{j}\left(1-\alpha_t\right)}$ for all $j\in (\dmsi-1)$.
  \end{enumerate}

  \item  $\HonB$ controls $\Root(\cpi)$ $\implies$ \label{claim:EdgeDistAna:2}
  \begin{enumerate}
    \item $\EdgeDist_{\icpih{(\HonA,j)}{\vetas}}(\EmptyString,b) = \EdgeDist_{\cpi}(\EmptyString,b)\cdot\frac{\prod_{t=0}^{j-1}\left(1-\beta^b_t\right)}{\prod_{t=0}^{j-1}\left(1-\beta_t\right)}$ for all $j\in (\dmsi)$.
    \item $\EdgeDist_{\icpih{(\HonB,j)}{\vetas}}(\EmptyString,b) = \EdgeDist_{\cpi}(\EmptyString,b)\cdot\frac{\prod_{t=0}^{j-1}\left(1-\beta^b_t\right)}{\prod_{t=0}^{j-1}\left(1-\beta_t\right)}$ for all $j\in (\dmsi-1)$.
  \end{enumerate}
\end{enumerate}
\end{proposition}
\begin{proof}
We prove \cref{claim:EdgeDistAna:1} using induction on the ordered pairs $[(\HonA,\dmsi)]$. The proof of \cref{claim:EdgeDistAna:2} is analogous.

\paragraph{Base case.} The proof follows since according to \cref{def:refinements}, it holds that $\icpih{(\HonA,0)}{\vetas}=\cpi$.

\paragraph{Induction step.} Fix $(\HonC,j)\in[(\HonA,\dmsi)]$ and assume the claim holds for $\pred(\HonC,j)$. The proof splits according to which party $\HonC$ is.
\begin{description}
  \item[Case $\HonC = \HonA$.] If $\EdgeDist_{\icpih{(\HonB,j-1)}{\vetas}}(\EmptyString,b)=0$, \cref{def:CondProtocol} yields that $\EdgeDist_{\icpih{(\HonA,j)}{\vetas}}(\EmptyString,b)=0$. The proof follows since, by the induction hypothesis, it holds that
      \begin{align*}
      \EdgeDist_{\icpih{(\HonA,j)}{\vetas}}(\EmptyString,b) = \EdgeDist_{\icpih{(\HonB,j-1)}{\vetas}}(\EmptyString,b) = \EdgeDist_{\cpi}(\EmptyString,b)\cdot\frac{\prod_{t=0}^{j-1}\left(1-\alpha^b_t\right)}{\prod_{t=0}^{j-1}\left(1-\alpha_t\right)}.
      \end{align*}
      In the complementary case, \ie $\EdgeDist_{\icpih{(\HonB,j-1)}{\vetas}}(\EmptyString,b)>0$, \cref{prop:DomMeasPro,def:Dominated} yield that $\beta_{j-1} = \beta^b_{j-1}$. It must be the case that $\beta_{j-1} = \beta^b_{j-1}<1$, since otherwise, according to \cref{def:refinements}, it holds that $\icpih{(\HonA,j)}{\vetas}=\perp$, a contradiction to the assumption that $\icpih{(\HonA,\dmsi)}{\vetas}\neq \perp$. The proof follows since in this case \cref{def:CondProtocol,claim:ChangeContext} yield that
      \begin{align*}
      \EdgeDist_{\icpih{(\HonA,j)}{\vetas}}(\EmptyString,b) &= \EdgeDist_{\icpih{(\HonB,j-1)}{\vetas}}(\EmptyString,b) \cdot \frac{1-\beta^b_{j-1}}{1-\beta_{j-1}} \\
      &= \EdgeDist_{\icpih{(\HonB,j-1)}{\vetas}}(\EmptyString,b) \\
      &= \EdgeDist_{\cpi}(\EmptyString,b)\cdot\frac{\prod_{t=0}^{j-1}\left(1-\alpha^b_t\right)}{\prod_{t=0}^{j-1}\left(1-\alpha_t\right)},
      \end{align*}
      where the last equality follows the induction hypothesis.
  \item[Case $\HonC = \HonB$.] It must be that case that $\alpha_{j}<1$, since otherwise, similarly to the previous case and according to \cref{def:refinements}, it holds that $\icpih{(\HonB,j)}{\vetas}=\perp$, a contradiction to the assumption that $\icpih{(\HonA,\dmsi)}{\vetas}\neq \perp$. The proof follows since in this case \cref{def:CondProtocol,claim:ChangeContext} yield that
      \begin{align*}
      \EdgeDist_{\icpih{(\HonB,j)}{\vetas}}(\EmptyString,b) &= \EdgeDist_{\icpih{(\HonA,j)}{\vetas}}(\EmptyString,b) \cdot \frac{1-\alpha^b_{j}}{1-\alpha_{j}} \\
      &= \EdgeDist_{\cpi}(\EmptyString,b)\cdot\frac{\prod_{t=0}^{j-1}\left(1-\alpha^b_t\right)}{\prod_{t=0}^{j-1}\left(1-\alpha_t\right)} \cdot \frac{1-\alpha^b_{j}}{1-\alpha_{j}} \\
      &= \EdgeDist_{\cpi}(\EmptyString,b)\cdot\frac{\prod_{t=0}^{j}\left(1-\alpha^b_t\right)}{\prod_{t=0}^{j}\left(1-\alpha_t\right)},
      \end{align*}
      where the second equality follows from the induction hypothesis.
\end{description}
\end{proof}

Using the above propositions, we now turn our focus to proving the claims in the proof of \cref{lemma:IdealMainLemmaDet}. To facilitate reading and tracking the proof, we cluster claims together according to their role in the proof of \cref{lemma:IdealMainLemmaDet}.

\subsubsubsection{Proving \cref{claim:RestrictFinal,claim:dmsiequal}}
\begin{proof}[Proof of \cref{claim:RestrictFinal}]
For $b\in\zo$ it holds that
\begin{align*}
\FinalMeas{\HonA}{\cpi_b}{\veta{b}} &\equiv \sum_{j=0}^{\dmsi}\idsm{(\HonA,j)}{\cpi_b}{\veta{b}} \cdot\prod_{t=0}^{j-1}\left(1-\idsm{(\HonA,t)}{\cpi_b}{\veta{b}}\right)\\
&\equiv \sum_{j=0}^{\dmsi}\Restrict{\idsm{(\HonA,j)}{\cpi}{\vetas}}{b} \cdot\prod_{t=0}^{j-1}\left(1- \Restrict{\idsm{(\HonA,t)}{\cpi}{\vetas}}{b}\right) \\
&\equiv \Restrict{\FinalMeas{\HonA}{\cpi}{\vetas}}{b},
\end{align*}
where the second equivalence follows from \cref{claim:ConRestrict}.
\end{proof}

\begin{proof}[Proof of \cref{claim:dmsiequal}]
Assume towards a contradiction that $\dmsi^0<\dmsi$. By the definition of $\dmsi^0$ (\cref{def:maxInd}) and the definition of conditional protocols (\cref{def:CondProtocol}), it follows that $\icpib{(\HonA,\dmsi^0+1)}{\veta{0}}{0}=\perp$. Since (by assumption) $\dmsi^1\leq\dmsi^0$ , it also holds that $\icpib{(\HonA,\dmsi^0+1)}{\veta{1}}{1}=\perp$. Hence,  \cref{claim:ConRestrict} yields that $\Restrict{\icpih{(\HonA,\dmsi^0+1)}{\vetas}}{0},\Restrict{\icpih{(\HonA,\dmsi^0+1)}{\vetas}}{1}=\perp$. Namely, the function describing $\icpih{(\HonA,\dmsi^0+1)}{\vetas}$ does not correspond to any two-party execution when restricting it to the subtrees $\Tree(\cpi_0)$ and $\Tree(\cpi_1)$. Hence, the aforementioned function does not correspond to a two-party execution (over $\Tree(\cpi)$), in contradiction to the assumption that $\icpih{(\HonA,\dmsi)}{\vetas}\neq\perp$.
\end{proof}

\subsubsubsection{Proving \cref{claim:AlphaIs1,clm:RelatingInductionConstants,claim:alphaRelation}}
The following proofs rely on the next observation. As long as $\alpha^b_j<1$ and $\beta^b_j<1$, \cref{claim:EdgeDistAna} ensures that there is a positive probability to visit both the left and the right subtree of the $(\HonC,j)$'th conditional protocol.

\begin{proof}[Proof of \cref{claim:AlphaIs1}]
Assume that $\HonA$ controls $\Root(\cpi)$ and that $\dmsi^1<\dmsi$. Assume towards a contradiction that $\alpha^1_{\dmsi^1}<1$. Since $\dmsi^1\leq\dmsi^0$ (by assumption), it follows that $\alpha^0_{\dmsi^1}<1$ as well. The definition of $\dmsi^1$ (\cref{def:maxInd}) yields that $\beta^1_{\dmsi^1}=1$. However, \cref{claim:EdgeDistAna} yields that $\EdgeDist_{\icpih{(\HonB,j)}{\vetas}}(\EmptyString,b)\in(0,1)$ for both $b\in\zo$, and thus \cref{prop:DomMeasPro,claim:ChangeContext} yield that $\beta_{\dmsi^1}=1$. Now, \cref{def:refinements} yields that $\icpih{(\HonA,\dmsi^1+1)}{\vetas}=\perp$, a contradiction to the assumption that $\icpih{(\HonA,\dmsi)}{\vetas}\neq\perp$.
\end{proof}

\begin{proof}[Proof of \cref{clm:RelatingInductionConstants}]
For $j\in(\dmsi^1-1)$, it holds that $\EdgeDist_{\icpih{(\HonB,j)}{\vetas}}(\EmptyString,b)\in(0,1)$ for both $b\in\zo$. Thus, $\beta^0_j=\beta^1_j=\beta_j$ is a direct implication of \cref{claim:ConRestrict,prop:DomMeasPro}.

For $\dmsi^1 \leq \dmsi-1$, \cref{claim:AlphaIs1,claim:EdgeDistAna} yield that $\EdgeDist_{\icpih{(\HonB,j)}{\vetas}}(\EmptyString,0)=1$. Since, by \cref{def:etaP}, it holds that $\eta_{(\HonB,j)}=\eta^0_{(\HonB,j)}$, \cref{def:Dominated,claim:ConRestrict} yield that $\beta^0_j=\beta_j$.
\end{proof}

\begin{proof}[Proof of \cref{claim:alphaRelation}]
The proof immediately follows from \cref{claim:ChangeContext,claim:EdgeDistAna}.
\end{proof}

\subsubsubsection{Proving \cref{claim:APi0is0,claim:alphaIsZero}}
\begin{proof}[Proof of \cref{claim:APi0is0}]
By \cref{def:Dominated} it holds that $\idsm{(\HonA,j)}{\cpi_1}{\veta{1}}\equiv 0$ for every $j\in(\dmsi)$. \cref{def:refinements} yields that $\FinalMeas{\HonA}{\cpi_1}{\veta{1}} \equiv 0$. The proof follows from \cref{claim:RestrictFinal}.
\end{proof}

\begin{proof}[Proof of \cref{claim:alphaIsZero}]
Follows similar arguments to the above proof of \cref{claim:APi0is0}, together with \cref{claim:ChangeContext}.
\end{proof}

\subsubsubsection{Proving \cref{clm:RelatingInductionConstantsB,claim:BetaIs1,claim:betaRelation}}
The proofs of the rest of the claims stated in the proof of \cref{lemma:IdealMainLemmaDet} are analogous to the claims proven above. Specifically, \cref{clm:RelatingInductionConstantsB} is analogous to \cref{clm:RelatingInductionConstants}, \cref{claim:BetaIs1} is analogous to \cref{claim:AlphaIs1}, and \cref{claim:betaRelation} is analogous to \cref{claim:alphaRelation}.

\subsection{Proving Lemma \ref{lemma:CombineMeasHitConst}}\label{sec:DomMeasHitConst}

\cref{lemma:CombineMeasHitConst} immediately follows by the next lemma.

\begin{lemma}\label{cor:DomMeasHit1}
For every protocol $\cpi$, there exists $(\HonC,j) \in \set{\HonA,\HonB} \times \N$ such  that
\begin{align*}
\eex{\LDist{\cpi_{(\HonC,j)}}}{\DomMeas{\cpi_{(\HonC,j)}}{\HonC}}=1.
\end{align*}
\end{lemma}

The proof of \cref{cor:DomMeasHit1} is given below, but first we use it to derive \cref{lemma:CombineMeasHitConst}.
\begin{proof}[Proof of \cref{lemma:CombineMeasHitConst}]
Let $\dmsi$ be the minimal integer such that $\sum_{j=0}^{\dmsi}\alpha_j \geq c$ or $\sum_{j=0}^{\dmsi}\beta_j \geq c$. Note that such $\dmsi$ is guaranteed to exist by \cref{cor:DomMeasHit1} and since by \cref{lemma:ExpDomMeas} it holds that $\alpha_j=\eex{\LDist{\cpi_{(\HonA,j)}}}{\DomMeas{\cpi_{(\HonA,j)}}{\HonA}}$ and $\beta_j=\eex{\LDist{\cpi_{(\HonB,j)}}}{\DomMeas{\cpi_{(\HonB,j)}}{\HonB}}$. The proof splits to the following cases.
\begin{description}
  \item[Case $\sum_{j=0}^{\dmsi}\alpha_j \geq c$.] By the choice of $\dmsi$ it holds that $\sum_{j=0}^{\dmsi-1}\alpha_j < c$ and $\sum_{j=0}^{\dmsi-1}\beta_j < c$. \cref{lemma:propCombineMeas} yields that
\begin{align*}
\eex{\LDist{\cpi}}{\CombineMeas{\cpi}{\HonA}{\dmsi}} &= \sum_{j=0}^{\dmsi} \alpha_j\prod_{t=0}^{j-1}(1-\beta_t)(1-\alpha_t) \\
&\overset{(1)}{\geq} \sum_{j=0}^{\dmsi} \alpha_j\prod_{t=0}^{\dmsi-1}(1-\beta_t)(1-\alpha_t) \\
&\overset{(2)}{\geq} \paren{\sum_{j=0}^{\dmsi} \alpha_j}\cdot \paren{1-\sum_{j=0}^{\dmsi-1} \beta_j}\cdot\paren{1-\sum_{j=0}^{\dmsi-1} \alpha_j} \\
&\overset{(3)}{\geq} c\cdot(1-2c),
\end{align*}
where (1) follows from multiplying the $j$'th summand by $\prod_{t=j}^{\dmsi-1}(1-\beta_t)(1-\alpha_t)\leq 1$ and (2) and (3) follow since $(1-x)(1-y)\geq 1-(x+y)$ for any $x,y\geq 0$. Hence, $\dmsi$ satisfies \cref{item:convergence1}.

  \item[Case $\sum_{j=0}^{\dmsi}\alpha_j < c$.] By the choice of $\dmsi$ it holds that $\sum_{j=0}^{\dmsi}\beta_j \geq c$ and $\sum_{j=0}^{\dmsi-1}\beta_j < c$. Similar arguments to the previous case show that $\dmsi$ satisfies \cref{item:convergence2}.
\end{description}
\end{proof}

Towards proving \cref{cor:DomMeasHit1} we prove that there is always a leaf for which the value of the dominated measure is $1$.
\begin{claim}\label{claim:DomLeaf}
Let $\cpi$ be a protocol with $\BestA{\cpi}=1$. Then there exists $\ell\in\Leaves_1(\cpi)$ such that $\DomMeas{\cpi}{\HonA}(\ell)=1$.
\end{claim}
\begin{proof}
The proof is by  induction on the round complexity of $\cpi$.

Assume that $\round(\cpi) = 0$ and let $\ell$ be the only node in $\Tree(\cpi)$. Since $\BestA{\cpi}>0$, it must be the case that $\Color_{\cpi}(\ell)=1$. The proof follows since \cref{def:Dominated} yields that $\DomMeas{\cpi}{\HonA}(\ell)=1$.

Assume that $\round(\cpi) = \rnd+1$ and that the lemma holds for $\rnd$-round protocols. If $\EdgeDist_{\cpi}(\EmptyString,b)=1$ for some $b\in\zo$, then by \cref{fact:perfectAdv} it holds that  $\BestA{\cpi_b}=\BestA{\cpi}=1$. This allows us to apply the induction hypothesis on $\cpi_b$, which yields that there exists $\ell\in\Leaves_1(\cpi_b)$ such that $\DomMeas{\cpi_b}{\HonA}(\ell)=1$. In this case, according to \cref{def:Dominated}, $\DomMeas{\cpi}{\HonA}(\ell)=\DomMeas{\cpi_b}{\HonA}(\ell)=1$, and the proof follows.

In the following we assume that $\EdgeDist_{\cpi}(\EmptyString,b)\in(0,1)$ for any $b\in\zo$.
We conclude the proof using  the following case analysis.
\begin{description}
    \item[$\HonA$ controls $\Root(\cpi)$.] According to \cref{fact:perfectAdv}, there exists $b\in\zo$ such that $\BestA{\cpi_b}=\BestA{\cpi}=1$. This allows us to apply the induction hypothesis on $\cpi_b$, which yields that there exists $\ell\in\Leaves_1(\cpi_b)$ such that $\DomMeas{\cpi_b}{\HonA}(\ell)=1$. The $\HonA$-maximal property of $\ADomMeas{\cpi}$ (\cref{prop:DomMeasPro}(\ref{item:Amaximal})) yields that $\DomMeas{\cpi}{\HonA}(\ell)=\DomMeas{\cpi_b}{\HonA}(\ell)=1$, and the proof for this case follows.

    \item[$\HonB$ controls $\Root(\cpi)$.] According to \cref{fact:perfectAdv}, $\BestA{\cpi_b}=\BestA{\cpi}=1$ for both  $b\in\zo$. This allows us to apply the induction hypothesis on $\cpi_0$ and $\cpi_1$, which yields that there exists $\ell_0\in\Leaves_1(\cpi_0)$ and $\ell_1\in\Leaves_1(\cpi_1)$ such that $\DomMeas{\cpi_0}{\HonA}(\ell_0)=1$ and $\DomMeas{\cpi_1}{\HonA}(\ell_1)=1$. The $\HonB$-minimal property of $\ADomMeas{\cpi}$ (\cref{prop:DomMeasPro}(\ref{item:Bminimal})) yields that there exists $b\in\zo$ such that $\DomMeas{\cpi}{\HonA}(\ell_b)=\DomMeas{\cpi_b}{\HonA}(\ell_b)=1$ (the bit $b$ for which $\Smaller{\cpi}{b}=1$), and the proof for this case follows.
\end{description}
This concludes the case analysis and the proof follows.
\end{proof}

We can now derive \cref{cor:DomMeasHit1}. \cref{claim:DomLeaf,prop:DomMeasPro} yield that the number of possible transcripts of $\icpi{(\HonC,j)}{}$ shrinks as $(\HonC,j)$ grows. Specifically, at least one possible transcript of $\icpi{(\HonA,j)}{}$ whose common outcome is $1$ (the transcript represented by the leaf is guaranteed to exist from \cref{claim:DomLeaf}) is \emph{not} a possible transcript of $\icpi{(\HonB,j)}{}$. Similarly, at least one possible transcript of $\icpi{(\HonB,j-1)}{}$ whose common outcome is $0$ is not a possible transcript of $\icpi{(\HonA,j)}{}$.
Since the number of possible transcripts of $\cpi$ is finite (though might be exponentially large), there exists $j\in\N$ such that either the common outcome of all possible transcripts $\icpi{(\HonA,j)}{}$ is $1$ or the common outcome of all possible transcripts of $\icpi{(\HonB,j)}{}$ is $0$. The expected value of the $\HonA$-dominated measure of $\icpi{(\HonA,j)}{}$ or the $\HonB$-dominated measure of $\icpi{(\HonB,j)}{}$ will be $1$. The formal proof is given next.

\begin{proof}[Proof of \cref{cor:DomMeasHit1}]
Assume towards a contradiction that  $\eex{\LDist{\cpi_{(\HonC,j)}}}{\DomMeas{\cpi_{(\HonC,j)}}{\HonC}}<1$ for every $(\HonC,j) \in \set{\HonA,\HonB} \times \N$. It follows that $\icpi{(\HonC,j)}{}\neq\perp$ for every such $(\HonC,j)$.  For a pair $(\HonC,j)\in \set{\HonA,\HonB} \times \N$, recursively define $\cL_{(\HonC,j)} \eqdef \cL_{\pred(\HonC,j)}\cup \cS_{(\HonC,j)}$, where $\cS_{(\HonC,j)} \eqdef \set{\ell \in \Leaves(\cpi) \colon \DomMeas{\cpi_{(\HonC,j)}}{\HonC}(\ell) = 1}$ and $\cL_{(\HonB,-1)}\eqdef\emptyset$.
The following claim (proven below) shows two properties of $\cS_{(\HonC,j)}$.
\begin{claim}\label{claim:convergence1}
It holds that $\cS_{(\HonC,j)}\neq \emptyset$ and $\cL_{\pred(\HonC,j)}\cap \cS_{(\HonC,j)}=\emptyset$ for every $(\HonC,j)\succeq(\HonB,0)$.
\end{claim}
\cref{claim:convergence1} yields that $\size{\cL_{(\HonC,j)}}>\size{\cL_{\pred(\HonC,j)}}$ for every $(\HonC,j)\succeq(\HonB,0)$, a contradiction to the fact that $\cL_{(\HonC,j)}\subseteq \Leaves(\cpi)$ for every $(\HonC,j)$.
\end{proof}

\begin{proof}[Proof of \cref{claim:convergence1}]
Let $(\HonC,j)\succeq(\HonB,0)$. By \cref{claim:SwitchRoles} it holds that $\BEST{\HonC}{\icpi{(\HonC,j)}{}} = 1$.\footnote{Note that this might not hold for $\icpi{(\HonA,0)}{}=\cpi$. Namely, it might be the case that $\BestB{\cpi} = 1$. In this case $\ADomMeas{\cpi}$ is the zero measure, $\icpi{(\HonB,0)}{}=\cpi$ and $\cS_{(\HonA,0)}=\emptyset$.} Hence, \cref{claim:DomLeaf} yields that $\cS_{(\HonC,j)}\neq\emptyset$.

Towards proving the second property, let $\ell'\in\cL_{\pred(\HonC,j)}$, and let $(\HonC',j')\in[\pred(\HonC,j)]$ such that $\ell'\in\cS_{(\HonC',j')}$. By the definition of $\cS_{(\HonC',j')}$, it holds that $\DomMeas{\cpi_{(\HonC',j')}}{\HonC'}(\ell') = 1$. By \cref{prop:CondPro} it holds that $\ell'\notin\Supp\paren{\LDist{\icpi{(\HonC'',j'')}{}}}$ for every $(\HonC'',j'')\succ(\HonC',j')$. Since $(\HonC,j)\succ\pred(\HonC,j)\succeq(\HonC',j')$, it holds that $\ell'\notin\Supp\paren{\LDist{\icpi{(\HonC,j)}{}}}$. By \cref{def:Dominated} it holds that $\DomMeas{\cpi_{(\HonC,j)}}{\HonC}(\ell) = 0$ for every $\ell\notin\Supp\paren{\LDist{\icpi{(\HonC,j)}{}}}$, and thus $\ell'\notin\cS_{(\HonC,j)}$. Hence, $\cL_{\pred(\HonC,j)}\cap \cS_{(\HonC,j)}=\emptyset$.
\end{proof}


\newcommand{\Neigh}{{\mathcal{B}\mathsf{order}}}
\newcommand{\neigh}{{\mathsf{border}}}

\newcommand{\prprrA}{{\HonA_{\delta,p;\bar r}^{(i,p)}}}
\newcommand{\Fail}{{\mathcal{F}\mathsf{ail}}}
\newcommand{\realA}[2]{\HonA^{(#1)}_{#2}}
\newcommand{\realB}[2]{\HonB^{(#1)}_{#2}}

\newcommand{\NumOfContPrime}{\left\lceil\frac{\log(1/\xi)}{\log(1/(1-\delta'))}\right\rceil}
\newcommand{\EffPruHead}{{\widetilde{\MathAlg{A}}}}
\newcommand{\NumOfSam}{\left\lceil\frac{\ln\paren{2^{m}/\xi}}{\xi^2/2}\right\rceil}
\newcommand{\NumOfCont}{\left\lceil\frac{\log(1/\xi)}{\log(1/(1-\delta))}\right\rceil}
\newcommand{\hRandomCont}{{\widehat{\RandomCont}}}
\newcommand{\BC}{\MathAlgX{C}}
\newcommand{\Est}{\MathAlg{Est}}

\section{Efficiently Biasing Coin-Flipping Protocols}\label{sec:RealAttacker}
In \cref{sec:IdealAttacker}, we showed that for any  coin-flipping protocol and   $\eps\in(0,\frac12]$, applying   the biased-continuation attack recursively for $\kappa=\kappa(\eps)$  times, biases the honest party's outcome by (at least) $1/2-\eps$. Implementing this attack, however, requires access to a sampling algorithm (\ie the biased continuator $\RandomCont$; see \cref{def:sampler}), which we do not know how to efficiently implement even when assuming OWFs do not exist. In this section, we show that the inexistence  of  OWFs does suffice to implement an \emph{approximation} of the biased-continuation attack that can be used to implement a strong enough variant of the aforementioned attack.

The outline of this section is as follows. In \cref{sec:EfficientRCAttacker} we define  the \emph{approximated (recursive) biased-continuation attacker}, an approximated variant of the (ideal) recursive biased-continuation attacker defined in \cref{sec:IdealAttacker}. We  show that this approximated attacker does well as lone as it does not visits \emph{low-value nodes} --- the expected protocol's outcome conditioned on  visiting  the nodes (transcripts) is close to zero. In \cref{sec:AttackPrunedProtocol}, we define a special class of protocols, called \emph{approximately pruned protocols}, that have (almost) no low-value nodes. We conclude that the approximated attacker does well when it attacks approximately pruned protocols, and argue about the implementation of this attacker. In \cref{sec:PruningInTheHead}, we define the \emph{pruning-in-the-head attacker} that behaves as if the protocol it is attacking is pruned, and by doing so manages to make use of the  recursive approximated biased-continuation attacker to attack \emph{any} protocol. In \cref{sec:ProtocolInv} we argue about the implementation of the pruning-in-the-head attacker. Finally in \cref{sec:EfficeinrAttack}, we show that the assumption that OWFs do not exist implies that the above attacker can be implemented efficiently, yielding that the outcome on \emph{any} coin-flipping protocol can be efficiently biased to be arbitrarily close to $0$ or $1$.

Throughout the section, as it was the case in \cref{sec:IdealAttacker}, we prove statements \wrt attackers that, when playing the role of the left-hand party of the protocol (\ie $\HonA$), are trying to bias the common output of the protocol towards one, and, when playing the role of the right-hand party of the protocol (\ie $\HonB$), are trying to bias the common output of the protocol towards zero. All statements have analogues ones \wrt the opposite attack goals.

\subsection{The Approximated  Biased-Continuation Attacker}\label{sec:EfficientRCAttacker}
We start with defining the recursive approximated biased-continuation attacker, an approximated variant of the recursive biased-continuation attacker defined in \cref{sec:IdealAttacker}, and state our bound on its success probability. The rest of the section will be devoted to proving this bound.

\paragraph{Defining the attacker.}\label{sec:real:app:def}
The approximated recursive biased-continuation attacker is using an approximated version  of the biased continuator $\RandomCont$ (see \cref{def:sampler}). The approximated biased continuator  is only guaranteed to works well when applied on nodes whose value (\ie the probability that the protocol outcome is $1$ given that the current transcript is the node's label) is not too close to the borders.  The motivation for using this weaker biased continuator is that, as we see later,  it can be efficiently implemented assuming the in-existence of  OWFs. In the following let $\RandomCont_\cpi$ be as in \cref{def:sampler}.

\begin{definition}[low-value and high-value nodes]\label{def:ILowHigh}
For a protocol $\cpi=(\HonA,\HonB)$ and $\delta\in[0,1]$, let
\begin{itemize}
\item $\low{\cpi}{\delta}=\set{u\in\Vertices(\cpi)\setminus\Leaves(\cpi): \Val(\cpi_{u})\leq\delta}$, and

\item $\high{\cpi}{\delta}=\left\{u\in\Vertices(\cpi)\setminus\Leaves(\cpi): \Val(\cpi_{u})\geq1-\delta\right\}$.

\end{itemize}
For $\HonC\in\set{\HonA,\HonB}$, let $\low{\cpi}{\delta,\HonC}=\low{\cpi}{\delta}\cap\ctrl{\cpi}{\HonC}$ and similarly let $\high{\cpi}{\delta,\HonC}=\high{\cpi}{\delta}\cap\ctrl{\cpi}{\HonC}$.\footnote{Recall that $\ctrl{\cpi}{\HonC}$ denotes the nodes in $\Tree(\cpi)$ controlled by party $\HonC$ (see \cref{def:ProtocolTree}).}
\end{definition}

\begin{definition}[approximated  biased continuator $\RandomCont_{\Pi}^{\xi,\delta}$]\label{def:AppBiassCSampler}
Algorithm \BC is a {\sf $(\xi,\delta)$-biased-continuator} for an $\rnd$-round protocol $\cpi$ if the following hold.
\begin{enumerate}\setlength{\leftmargini}{3pt}
 \item $\ppr{\ell \la \LDist{\cpi}}{\exists i\in (\rnd-1) \colon \SDP{\BC(\ell_{1,\ldots,i},1)}{\RandomCont_\cpi(\ell_{1,\ldots,i},1)} > \xi \land \ell_{1,\ldots,i}\notin\low{\cpi}{\delta}} \leq \xi$,\\ and
 \item $\ppr{\ell \la \LDist{\cpi}}{\exists i\in (\rnd-1)  \colon  \SDP{\BC(\ell_{1,\ldots,i},0)}{\RandomCont_\cpi(\ell_{1,\ldots,i},0)} > \xi \land \ell_{1,\ldots,i}\notin\high{\cpi}{\delta}} \leq \xi$.
\end{enumerate}
Let $\RandomCont_\Pi^{\xi,\delta}$ be an arbitrary  (but fixed) $(\xi,\delta)$-biased-continuator of $\Pi$.
\end{definition}

The recursive approximated biased-continuation attacker is identical to that defined in \cref{sec:IdealAttacker}, except that it uses the approximated biased-continuator sampler and not the ideal one.

Let $\rcAP{\cpi}{0,\xi,\delta}\equiv \HonA$, and for integer $i> 0$ define:
\begin{algorithm}[approximated recursive biased-continuation attacker $\realA{i,\xi,\delta}{\cpi}$]\label{def:itAppAtt}
	\item  Parameters: integer $i>0$, $\xi,\delta\in(0,1)$.
	\item  Input: transcript $u \in \zo^\ast$.
	\item Operation:
	\begin{enumerate}

		\item If $u\in \Leaves(\cpi)$, output $\Color_\cpi(u)$ and halt.

		\item  Set $\msg = \RandomCont_{\paren{\realA{i-1,\xi,\delta}{\cpi},\HonB}}^{\xi,\delta}(u,1)$.

		\item Send $\msg$ to $\HonB$.

		\item If $u'=u\conc \msg\in \Leaves(\cpi)$, output $\Color_\cpi(u')$.

	\end{enumerate}
\end{algorithm}
In the following we sometimes refer to the base (non-recursive) version of the above algorithm, \ie $\realA{1,\xi,\delta}{\cpi}$, as the approximated biased-continuation attacker. When clear from the context, we will remove the protocol name (\ie $\cpi$) from the subscript of the above attacker. (As a rule of thumb, in statements and definitions we explicitly write the protocols to which the algorithms refer, whereas in proofs and informal discussions we usually omit them.) 

\paragraph{The attacker's success probability.}
We would like to bound the difference between the biased-continuation attacker and its approximated variant defined above. Following \cref{def:AppBiassCSampler}, if the approximated biased continuator $\RandomCont^{\xi,\delta}$ is called on non-low-value nodes (transcripts), both attackers are given similar answers, so the difference between them will be small. Hence, as long as the probability of hitting low-value nodes under $\HonA$'s control is small (note that only nodes under $\HonA$'s control are queried), we expect that the recursive approximated biased-continuation attacker will do well. This is formally put in the next lemma.

\begin{lemma}\label{lemma:goodWOLowValueSimple}
	For any     $\delta\in (0,1/4]$ and   $k\in\N$, there exists a  polynomial $p_{k,\delta}$  such that the following holds. Let  $\cpi=(\HonA,\HonB)$ be an $\rnd$-round protocol, and assume that   $\ppr{\LDist{\cpi}}{\descP{\low{\cpi}{1.5\delta',\HonA}}}\leq \alpha$ for some $\delta \le\delta'\le\frac14$.\footnote{$\descP{\cS}$ is the set of nodes with ancestor in $\cS$ (see \cref{def:BinaryTrees}).} Then for any  $\xi,\mu\in(0,1)$, it holds that
	\begin{align*}
	\SDP{\LDist{\rcAP{\cpi}{k},\HonB}}{\LDist{\realA{k,\xi,\delta'}{\cpi},\HonB}} &\leq  \phiPruE{k,\delta}(\alpha,\xi,m,\delta',\mu)\eqdef (\alpha+\xi)\cdot p_{k,\delta}(m,1/\delta',1/\mu) + \mu.
	\end{align*}
\end{lemma}

The fact that the lemma assumes a bound \wrt $\low{\cpi}{1.5\delta',\HonA}$ (and not $\low{\cpi}{\delta',\HonA}$) is of technical nature, and is not significant to the understating of the statement.

We will use \cref{lemma:goodWOLowValueSimple} as follows: the constants $\delta$, $\delta'$ and $k$ will be set according to  the (constant) bias of the protocol. Then we choose $\mu\in o(1)$. Finally, we are free to choose $\alpha$ and $\xi$ to be $1/p$ for large enough polynomial $p$, such that $p \gg p_{k,\delta}(m,1/\delta',1/\mu)$.

In addition to \cref{lemma:goodWOLowValueSimple}, the following lemma will  be useful when considering pruned protocols in the next section.

\begin{lemma}\label{claim:SmallStaySmallApproxSimple}
		For any    $\delta\in (0,1/4]$ and   $k\in\N$, there exists a polynomial $q_{k,\delta}$  such that the following holds. Let $\cpi = \paren{\HonA,\HonB}$ and $\cpi'=\paren{\HonC,\HonD}$ be two $\rnd$-round protocols and let $\cF$ be  a frontier of $\cU$, for some $\cU\subseteq \Vertices(\cpi)$. Assume  $\SDP{\LDist{\cpi}}{\LDist{\cpi'}} \leq \eps$, $\ppr{\LDist{\cpi}}{\descP{\low{\cpi}{1.5\delta',\HonA}}}\leq \alpha$ for some $\delta \le\delta'\le\frac14$, and  $\ppr{\LDist{\cpi'}}{\descP{\cF}}\leq\beta$. Then for any  $\xi,\mu\in(0,1)$, it holds that
	\begin{align*}
	\ppr{\LDist{\pruAttack{k,\delta',\xi}{\cpi},\HonB}}{\descP{\cF}} \leq  \phiBalE{k,\delta}(\alpha,\beta,\eps,m,\delta',\mu) + \phiPruE{k,\delta}(\alpha,\xi,m,\delta',\mu),
	\end{align*}
	for
	\begin{align*}
	\phiBalE{k,\delta}(\alpha,\beta,\eps,m,\delta',\mu) \eqdef \paren{\alpha + \beta + \eps}\cdot q_{k,\delta}(m,1/\delta',1/\mu) + \mu.
	\end{align*}
\end{lemma}

Namely, \cref{claim:SmallStaySmallApproxSimple} asserts that if the transcripts of $\cpi$ and $\cpi'$ are close, the probability of hitting low-value nodes in $\cpi$ under the control of the left-hand party is small and the probability of hitting a frontier  $\cF$ in $\cpi'$ is small as well. Then the probability of hitting this frontier in $\cpi$   when the recursive approximated biased-continuation attacker is taking the role of the left-hand party in $\cpi$ is small as well.

\paragraph{Outline for the proof of \cref{lemma:goodWOLowValueSimple}.}
Proving \cref{lemma:goodWOLowValueSimple} actually turns out to be quite challenging. The lemma assumes that the probability, according to the \emph{honest distribution} of leaves (\ie $\LDist{\cpi}$), to generate a low-value node under $\HonA$'s control is small. The queries the attacker makes, however, might be chosen from a different distribution, making some nodes much more likely to be queried than before. We call such nodes ``unbalanced''. If low-value nodes under $\HonA$'s control were a large fraction of the unbalanced ones, then \cref{def:AppBiassCSampler} guarantee nothing about the answers of the approximated biased continuator $\RandomCont^{\xi,\delta}$. Indeed, the main technical contribution of this section is to show that low-value nodes under $\HonA$'s control are only small fraction of the unbalanced ones.

A natural approach for proving \cref{lemma:goodWOLowValueSimple} is to use induction on $k$. The base case  when $k=1$ holds since $\RandomCont_\Pi^{\xi,\delta'}$, used by $\realA{1,\xi,\delta'}{\cpi}$, is a $(\xi,\delta')$-biased-continuator of $\Pi$. Moving to the induction step, we assume the lemma is true for $k-1$. Namely, we assume that
\begin{align}\label{eq:outlineHyp}
\SDP{\LDist{\rcAP{\cpi}{k-1},\HonB}}{\LDist{\realA{k-1,\xi,\delta'}{\cpi},\HonB}} \text{ is small.}
\end{align}

The first step is to apply the ideal biased-continuation attacker on the left-hand side part of both protocols. We will show that even after applying the attacker, the protocols remain close. Namely, we will prove the following statement.
\begin{align}\label{eq:outlineRobust}
&\SDP{\LDist{\rcAP{\cpi}{k-1},\HonB}}{\LDist{\realA{k-1,\xi,\delta'}{\cpi},\HonB}} \text{ is small} \\
&\implies \SDP{\LDist{\paren{\rcAP{\cpi}{k-1}}^{(1)},\HonB}}{\LDist{\paren{\realA{k-1,\xi,\delta'}{\cpi}}^{(1)},\HonB}} \text{ is small as well}. \nonumber
\end{align}
Putting differently, to prove \cref{eq:outlineRobust} we show that the biased-continuation attacker is ``robust'' --- it does not make similar protocols dissimilar.

The second step it to show that applying the ideal biased-continuation attacker on the right-hand side protocol is similar to applying the approximated biased-continuation attacker on the same protocol. Namely, we will prove the following statement.
\begin{align}\label{eq:outlineIdeal2Real}
\SDP{\LDist{\paren{\realA{k-1,\xi,\delta'}{\cpi}}^{(1)},\HonB}}{\LDist{\paren{\realA{k-1,\xi,\delta'}{\cpi}}^{(1,\xi,\delta')},\HonB}} \text{ is small}
\end{align}
Putting differently, to prove the ``ideal to real'' reduction described in  \cref{eq:outlineRobust} we show that the approximated biased-continuation attacker is a good approximation to its ideal variant.

In fact, both the ``robustness'' property (\cref{eq:outlineRobust}) and the   ``ideal to real'' reduction (\cref{eq:outlineIdeal2Real}) require the additional assumption that the probability of hitting low-value nodes under the control of the left-hand side party is small. Following the induction hypothesis (\cref{eq:outlineHyp}) showing this assumption to be true reduces to showing that the recursive  ideal biased-continuation attacker  hit low-value nodes under its control with only small probability (specifically, we need this to hold for $k-1$ recursions). The lemma assumes that the probability of hitting such nodes in the original protocol is small, namely that the set of $\HonA$-controlled low-value nodes is of low density. We will show that the recursive  ideal biased-continuation attacker does not increase the density of any sets by much.

The outline of this section is as follows. In \cref{sec:VisitUnbal} we formally define unbalanced nodes \wrt the \emph{non-recursive} attacker, and show that low-value nodes under $\HonA$'s control are only small fraction of them. This connection between unbalanced nodes to low-value ones underlines all the other results in this section.
In \cref{sec:Robust} we state and prove the ``robustness'' property. In \cref{sec:real:app:bound} we analyze the ``ideal to real'' reduction. In \cref{sec:lowRemainLow} we show that when it is applied recursively, the ideal biased-continuation attacker does not increase the probability of hitting low-density sets. Finally, in \cref{sec:ProofsOfgoodWOLowValue} we give the proofs of  \cref{lemma:goodWOLowValueSimple,claim:SmallStaySmallApproxSimple}.

\subsubsection{Unbalanced Nodes}\label{sec:VisitUnbal}
For non low-value and non high-value transcripts, \cref{def:AppBiassCSampler}  guarantees that  when queried on transcripts chosen according to the honest distribution of leaves (\ie $\LDist{\cpi}$), there is only a small statistical distance between the answers of the biased continuator $\RandomCont$ and it approximated variant $\RandomContXD{\xi}{\delta}$. The queries the biased-continuation attacker makes, however, might be chosen from a different distribution, making some transcripts much more likely to be queried than before. We call such transcripts ``unbalanced''.
\begin{definition}[unbalanced nodes]\label{def:unbalNodes}
For a protocol $\cpi=\HonAHonB$ and $\gamma\geq1$, let
 $\unbal{\cpi}{i}{\gamma}{}= \set{u\in\Vertices(\cpi) \setminus \Leaves(\cpi) \colon \VerticesDist_{\paren{\rcAP{\cpi}{1},\HonB}}(u)\geq \gamma\cdot\VerticesDist_{(\Ac,\HonB)}(u)}$, where $\rcAP{\cpi}{1}$ is as in \cref{alg:adversary} and $\VerticesDist$ as in \cref{def:ProtocolTree}.\footnote{$\VerticesDist_{\Tau}(u)$ is the probability that node (transcript) $u$ is reached in an (honest) execution of protocol $\Tau$.}
\end{definition}
Namely, $\unbal{\cpi}{i}{\gamma}{}$ are those nodes that a random  execution of $(\rcA{1},\HonB)$ visits with probability at least  $\gamma$ times the probability that a random execution of $\Pi$ does.

Given a protocol $\cpi=\HonAHonB$, we would like to understand what makes a node unbalanced. Let $u$ be a $\gamma$-unbalanced node, \ie $\VerticesDist_{(\rcA{1},\HonB)}(u) \geq \gamma\cdot \VerticesDist_{(\HonA,\HonB)}(u)$. By the edge distribution of $\paren{\rcA{1},\HonB}$ (\cref{claim:weights}), it follows that
\begin{align}\label{eq:unbal}
\frac{\VerticesDist_{(\rcA{1},\HonB)}(u)}{\VerticesDist_{(\HonA,\HonB)}(u)} = \prod_{\substack{0\leq i \leq \size{u}-1 \colon \\ u_{1,\ldots,i}\in\ctrl{\cpi}{\HonA}}}\frac{\Val(\cpi_{u_{1,\ldots,i+1}})}{\Val(\cpi_{u_{1,\ldots,i}})} \geq \gamma.
\end{align}

Hence, if $\gamma$ is large, one of the terms of the product in \cref{eq:unbal} must be large. Since the value of any sub-protocol is at most one, the numerator of each term cannot be large. It then must be the case that the denominator of at least one of those terms is close to zero, \ie that $u$ has a low-value ancestor controlled by $\HonA$.\footnote{This discussion is not entirely accurate, but it gives a good intuition for why unbalanced nodes relate to low-value ones. Indeed, the actual statement (\cref{lem:ProbVisitUnBal}) shows this discussion to hold only with high probability, which suffices for our needs.}

The following key lemma formulates the above intuition, and shows that the biased-continuation attacker does not bias the original distribution of the protocol by too much, unless it has previously visited a low-value node controlled by $\Ac$.

\begin{lemma}\label{lem:ProbVisitUnBal}
Let $\cpi=\HonAHonB$ be a protocol and let $\rcAP{\cpi}{1}$ be as in \cref{alg:adversary}. Then for every $\delta\in(0,\frac12]$ there exists a constant $c=c(\delta)>0$, such that for every $\delta'\geq \delta$ and $\gamma> 1$:
\begin{align*}
\ppr{\LDist{\rcAP{\cpi}{1},\HonB}}{\descP{\unbal{\cpi}{1}{\gamma}{} \setminus \propDescP{\low{\cpi}{\delta',\HonA}}}}\leq \frac{2}{\gamma^{c}}.\footnotemark
\end{align*}
\end{lemma}
\footnotetext{Recall that for $\cS\subseteq\Vertices(\cpi)$, $\propDescP{\cS}$ stands for the set of nodes which have an ancestor in $\cS$, but are not in $\cS$ itself (see \cref{def:BinaryTrees}).}

Namely, the probability of reaching a $\gamma$-unbalanced node which does not have a $\delta'$-low ancestor, for $\delta'\geq \delta$, is some inverse polynomial in $\gamma$. The proof of \cref{lem:ProbVisitUnBal} is given below. Looking ahead, we will apply this lemma for some $\gamma\in\poly(n)$, where $n$ is the security parameter given to the parties. At a high level, $\RandomContXD{\xi}{\delta}$ gives a good (enough) approximation for the biased continuator $\RandomCont$ when called on nodes that are at most $\poly(n)$-unbalanced. This lemma is useful since it gives a $1/\poly(n)$ bound for the probability that $\RandomContXD{\xi}{\delta}$ is called on nodes that are more than $\poly(n)$-unbalanced. Another important point is that the inverse polynomial (\ie $c$) depends only on $\delta$ (and is independent of $\gamma$ and $\delta'$). This becomes crucial when analyzing the success probability of the approximated biased-continuation attacker.

\cref{lem:ProbVisitUnBal} allows us to bound the probability that the (ideal) biased-continuation attacker hits unbalanced nodes with the probability that the \emph{original} protocol hits $\HonA$-controlled low-value nodes. Indeed, consider the first time $(\rcA{1},\HonB)$ reaches a $\gamma$-unbalanced node $u$. If  an $\HonA$-controlled low-value ancestor node was reached before reaching $u$, then this ancestor cannot be $\gamma$-unbalanced, and thus the probability of hitting it (and in turn hitting $u$) is bounded by $\gamma$ times the probability of the original protocol hitting $\HonA$-controlled low-value nodes. In the complementary case, in which no $\HonA$-controlled low-value node was reached before reaching $u$,  the probability of hitting $u$ is bounded by \cref{lem:ProbVisitUnBal}. This analysis is where we use that \cref{lem:ProbVisitUnBal} is proven \wrt \emph{proper} descendants of low-value nodes. The above discussion is stated formally next.

\begin{lemma}\label{cor:UnbalBound}
Let $\cpi=\HonAHonB$ be a protocol, let $\delta\in(0,\frac12]$, and let $c=c(\delta)$ be according  \cref{lem:ProbVisitUnBal}. Then
\begin{align*}
\ppr{\LDist{\rcAP{\cpi}{1},\HonB}}{\descP{\UnBal{\cpi}{\gamma}}} \leq \gamma\cdot\ppr{\LDist{\HonA,\HonB}}{\descP{\low{\cpi}{\delta',\HonA}}} + \frac{2}{\gamma^c},
\end{align*}
for any $\delta'\geq\delta$ and $\gamma>1$.
\end{lemma}
\begin{proof}
By \cref{prop:UnBalLowValueG}, it holds that
\begin{align}\label{eq:UnBalLowValue}
\descP{\UnBal{\cpi}{\gamma}} \subseteq \descP{\low{\cpi}{\delta',\HonA} \setminus \UnBal{\cpi}{\gamma}} \cup \descP{\UnBal{\cpi}{\gamma} \setminus \propDescP{\low{\cpi}{\delta',\HonA}}}.
\end{align}

We can now compute
\begin{align*}
\ppr{\LDist{\rcA{1},\HonB}}{\descP{\UnBal{\cpi}{\gamma}}} &\leq \ppr{\LDist{\rcA{1},\HonB}}{\descP{\low{\cpi}{\delta',\HonA} \setminus \UnBal{\cpi}{\gamma}}}\\
&\quad +\ppr{\LDist{\rcA{1},\HonB}}{\descP{\UnBal{\cpi}{\gamma} \setminus \propDescP{\low{\cpi}{\delta',\HonA}}}} \\
&\leq \gamma\cdot\ppr{\LDist{\HonA,\HonB}}{\descP{\low{\cpi}{\delta',\HonA}}} + \frac{2}{\gamma^c},
\end{align*}
where the second inequality follows from the definition of $\UnBal{\cpi}{\gamma}$ and \cref{lem:ProbVisitUnBal}.
\end{proof}

The rest of this section is dedicated to proving \cref{lem:ProbVisitUnBal}.

\paragraph{Proving \cref{lem:ProbVisitUnBal}}
\begin{proof}[Proof of \cref{lem:ProbVisitUnBal}.]
The lemma is proven via the proving following facts:
\begin{enumerate}[(1)]
\item\label{item:unbal1} There exists $c>0$ such that
\begin{align}\label{eq:unbal2}
\ppr{\LDist{\rcA{1},\HonB}}{\descP{\unbal{\cpi}{1}{\gamma}{} \setminus \descP{\low{\cpi}{\delta,\HonA}}}}\leq \frac{2-\Val(\cpi)}{\gamma^{c}}
\end{align}
for every $\gamma>1$. Note that \cref{eq:unbal2} only considers descendants of $\low{\cpi}{\delta,\HonA}$, and not proper descendants.

\item\label{item:unbal2} For $\gamma>1$ it holds that
\begin{align}
\descP{\unbal{\cpi}{1}{\gamma}{} \setminus \propDescP{\low{\cpi}{\delta,\HonA}}} \subseteq \descP{\unbal{\cpi}{1}{\gamma}{} \setminus \descP{\low{\cpi}{\delta,\HonA}}}.\footnotemark
\end{align}
\footnotetext{It thus follows that  $\descP{\unbal{\cpi}{1}{\gamma}{} \setminus \propDescP{\low{\cpi}{\delta,\HonA}}} = \descP{\unbal{\cpi}{1}{\gamma}{} \setminus \descP{\low{\cpi}{\delta,\HonA}}}$.}
\item\label{item:unbal3} For  $\delta'>\delta$ it holds that
\begin{align}
\unbal{\cpi}{1}{\gamma}{} \setminus \propDescP{\low{\cpi}{\delta',\HonA}} \subseteq \unbal{\cpi}{1}{\gamma}{} \setminus \propDescP{\low{\cpi}{\delta,\HonA}}.
\end{align}
\end{enumerate}
It is clear that combining the above steps yields (a stronger version of) the lemma.

\medskip
\emph{Proof of (\ref{item:unbal1}):}
Fix  $\delta\in(0,\frac12]$ and let $c:=\alpha(\delta)$ be the value guaranteed in \cref{lemma:calculus2}. The proof is by induction on the round complexity of $\cpi$.

Assume $\round(\cpi)=0$ and let $\ell$ be the single leaf of $\cpi$. By \cref{def:unbalNodes}, $\ell \notin \unbal{\cpi}{1}{\gamma}{}$  and thus $\unbal{\cpi}{1}{\gamma}{} = \emptyset$. Hence, for every $\delta>0$,
\begin{align*}
\ppr{\LDist{\rcA{1},\HonB}}{\desc\left(\unbal{\cpi}{1}{\gamma})\setminus \desc(\low{\cpi}{\delta,\HonA})\right)}=\ppr{\LDist{\rcA{1},\HonB}}{\emptyset}=0\leq\frac{2-\Val(\cpi)}{\gamma^{c}}.
\end{align*}

Assume that \cref{eq:unbal2} holds for $\rnd$-round protocols and that $\round(\cpi) = \rnd+1$. If $\EdgeDist_{(\HonA,\HonB)}(\EmptyString,b)=1$ for some $b\in\zo$ (recall that $\EmptyString$ denotes the empty string), then
\begin{align*}
\ppr{\LDist{\rcA{1},\HonB}}{\desc\left(\unbal{\cpi}{1}{\gamma})\setminus \desc(\low{\cpi}{\delta,\HonA})\right)} &= \ppr{\LDist{\paren{\rcA{1},\HonB}_b}}{\desc\left(\unbal{\cpi_b}{1}{\gamma})\setminus \desc(\low{\cpi_b}{\delta,\HonA})\right)} \\
&= \ppr{\LDist{\ProArcb{1}{b}}}{\desc\left(\unbal{\cpi_b}{1}{\gamma})\setminus \desc(\low{\cpi_b}{\delta,\HonA})\right)},
\end{align*}
where the second equality follows \cref{fact:switchOrder}. The proof now follows from the induction hypothesis.

To complete the proof, we assume that  $\EdgeDist_{(\HonA,\HonB)}(\EmptyString,b)\notin\zo$ for both $b\in\zo$, and let $p=\EdgeDist_{(\HonA,\HonB)}(\EmptyString,0)$.
The proof splits according to who controls the root of $\cpi$.
\begin{description}
\item[$\HonB$ controls $\Root(\cpi)$.] We first note that
\begin{align}\label{eq:unbal3}
  \unbal{\cpi}{1}{\gamma}{}\setminus \descP{\low{\cpi}{\delta,\HonA}}
   = \paren{\unbal{\cpi_0}{1}{\gamma}{}\setminus \descP{\low{\cpi_0}{\delta,\HonA}}} \cup \paren{\unbal{\cpi_1}{1}{\gamma}{}\setminus \descP{\low{\cpi_1}{\delta,\HonA}}}.
\end{align}
\remove{
To see the above, let $u\in\Vertices(\cpi)$. First, note that since $\HonB$ controls $\Root(\cpi)$, it holds that $\EdgeDist_{(\rcA{1},\HonB)}(\EmptyString,b) = \EdgeDist_{(\HonA,\HonB)}(\EmptyString,b)$, and thus, if $u\neq \Root(\cpi)$ \Inote{I think it holds regardless}, it holds that $u\in\unbal{\cpi}{1}{\gamma}{}$ if and only if $u\in\unbal{\cpi_b}{1}{\gamma}{}$ \Inote{for both $b\in \zo$}. Assume that $u\in \unbal{\cpi}{1}{\gamma}{}\setminus \descP{\low{\cpi}{\delta,\HonA}}$. Since $\gamma>1$, it holds that $u\neq \Root(\cpi)$, and thus $u\in \unbal{\cpi_b}{1}{\gamma}{}$ \Inote{for some  $b\in \zo$}. Moreover, \Inote{why?} it follows that $u_1,\ldots,u_{1,\ldots,\size{u}}\notin\low{\cpi_b}{\delta,\HonA}$, and thus $u\in \unbal{\cpi_b}{1}{\gamma}{}\setminus \descP{\low{\cpi_b}{\delta,\HonA}}$. For the other direction, assume $u\in\unbal{\cpi_b}{1}{\gamma}{}\setminus\descP{\low{\cpi_b}{\delta,\HonA}}$ for some $b\in \zo$. As argued before, it holds that $u\in\unbal{\cpi}{1}{\gamma}{}$. Since $u_1,\ldots,u_{1,\ldots,\size{u}}\notin\low{\cpi_b}{\delta,\HonA}$ \Inote{why do you mentioned that?}, and since $\HonB$ controls $\Root(\cpi)$, it also holds that $\Root(\cpi)\notin\low{\cpi_b}{\delta,\HonA}$ \Inote{ $\Root(\cpi)$ is  a node of $\low{\cpi_b}{\delta,\HonA}$...}. Hence, $u\in\unbal{\cpi}{1}{\gamma}{}\setminus \descP{\low{\cpi}{\delta,\HonA}}$. This complete the proof of \cref{eq:unbal3}.
}

To see the above, first note $\descP{\low{\cpi}{\delta,\HonA}} \setminus \set{\Root(\cpi)}  = \descP{\low{\cpi_0}{\delta,\HonA}}  \cup \descP{\low{\cpi_1}{\delta,\HonA}}$, and since  $\HonB$ controls $\Root(\cpi)$, it holds that $\unbal{\cpi}{1}{\gamma}{}\setminus \set{\Root(\cpi)} = \unbal{\cpi_0}{1}{\gamma}{} \cup \unbal{\cpi_1}{1}{\gamma}{}$. Finally, since $\gamma>1$ it holds that $\Root(\cpi) \notin \unbal{\cpi}{1}{\gamma}{}$, and \cref{eq:unbal3} follows.

We can now  write
\begin{align*}
\lefteqn{\ppr{\LDist{\rcA{1},\HonB}}{\descP{\unbal{\cpi}{1}{\gamma}{}\setminus \descP{\low{\cpi}{\delta,\HonA}}}}}\\
&=\EdgeDist_{(\rcA{1},\HonB)}(\EmptyString,0)\cdot\ppr{\LDist{(\rcA{1},\HonB)_0}}{\desc\left(\unbal{\cpi_0}{1}{\gamma})\setminus \desc(\low{\cpi_0}{\delta,\HonA})\right)}\\
&\quad +\EdgeDist_{(\rcA{1},\HonB)}(\EmptyString,1)\cdot\ppr{\LDist{(\rcA{1},\HonB)_1}}{\desc\left(\unbal{\cpi_1}{1}{\gamma})\setminus \desc(\low{\cpi_1}{\delta,\HonA})\right)}\\
&=p\cdot \ppr{\LDist{\ProArcb{1}{0}}}{\desc\left(\unbal{\cpi_0}{1}{\gamma})\setminus \desc(\low{\cpi_0}{\delta,\HonA})\right)} \\
&\quad + (1-p)\cdot \ppr{\LDist{\ProArcb{1}{1}}}{\desc\left(\unbal{\cpi_1}{1}{\gamma})\setminus \desc(\low{\cpi_1}{\delta,\HonA})\right)}\\
&\leq p\cdot \frac{2-\Val(\cpi_0)}{\gamma^{c}}+(1-p)\cdot \frac{2-\Val(\cpi_1)}{\gamma^{c}}\\
&=\frac{2-\Val(\cpi)}{\gamma^{c}}.
\end{align*}
The first equality follows from \cref{eq:unbal3}, the second equality follows from \cref{fact:switchOrder}, and the inequality follows from the induction hypothesis.

\item[$\HonA$ controls $\Root(\cpi)$.]
If $\Val(\cpi)\leq\delta$, then $\Root(\cpi)\in\low{\cpi}{\delta,\HonA}$. Therefore, $\unbal{\cpi}{1}{\gamma}{}\setminus \descP{\low{\cpi}{\delta,\HonA}}=\emptyset$ and the proof follows from a similar argument as in the base case.

In the complementary case, \ie $\Val(\cpi)>\delta$, assume \wlg that $\Val(\cpi_0)\geq\Val(\cpi)\geq\Val(\cpi_1)$. We start with the case that $\Val(\cpi_1)>0$. For $b\in\zo$, let $\gamma_b \eqdef \frac{\Val(\cpi)}{\Val(\cpi_b)}\cdot \gamma$. By \cref{claim:weights}, for $u\in\Vertices(\cpi)$ with $u\neq\Root(\cpi)$ and $b=u_1$, it holds that
\begin{align*}
\frac{\VerticesDist_{(\rcA{1},\HonB)}(u)}{\VerticesDist_{(\HonA,\HonB)}(u)} = \frac{\EdgeDist_{(\HonA,\HonB)}(\EmptyString,b)}{\EdgeDist_{(\rcA{1},\HonB)}(\EmptyString,b)}\cdot \frac{\VerticesDist_{(\rcA{1},\HonB)_b}(u)}{\VerticesDist_{(\HonA,\HonB)_b}(u)} =
\frac{\Val(\cpi_b)}{\Val(\cpi)}\cdot \frac{\VerticesDist_{(\rcA{1},\HonB)_b}(u)}{\VerticesDist_{(\HonA,\HonB)_b}(u)}.
\end{align*}
Thus, $u\in\unbal{\cpi}{1}{\gamma}{}$ if and only if $u\in\unbal{\cpi_b}{1}{\gamma_b}{}$. Hence, using also the fact that $\Root(\cpi)\notin \low{\cpi}{\delta,\HonA}$ (since we assumed $\Val(\cpi)>\delta$), arguments similar to those used to prove \cref{eq:unbal3} yield that
\begin{align}\label{eq:unbal4}
\\
\unbal{\cpi}{1}{\gamma}{}\setminus \descP{\low{\cpi}{\delta,\HonA}}
   = \paren{\unbal{\cpi_0}{1}{\gamma_0}{}\setminus \descP{\low{\cpi_0}{\delta,\HonA}}} \cup \paren{\unbal{\cpi_1}{1}{\gamma_1}{}\setminus \descP{\low{\cpi_1}{\delta,\HonA}}}.\nonumber
\end{align}

Moreover, for $b\in \zo$ it holds that
\begin{align}\label{eq:unbal5}
\\
\ppr{\LDist{(\rcA{1},\HonB)_b} }{\desc\left(\unbal{\cpi_b}{1}{\gamma_b})\setminus \desc(\low{\cpi_b}{\delta,\HonA})\right)}&=\ppr{\LDist{\ProArcb{1}{b}}}{\desc\left(\unbal{\cpi_b}{1}{\gamma})\setminus \desc(\low{\cpi_b}{\delta,\HonA})\right)}\nonumber\\
&\leq\frac{2-\Val(\cpi_b)}{\gamma_b^c}\nonumber\\
&=\left(\frac{\Val(\cpi_b)}{\Val(\cpi)}\right)^{c}\cdot\frac{2-\Val(\cpi_b)}{\gamma^c}.\nonumber
\end{align}
The first equality follows from \cref{fact:switchOrder}.
The inequality follows from the next case analysis: if $\gamma_b>1$, then  it follows from the induction hypothesis applied \wrt $\cpi_b$, $\delta$ and $\gamma_b$; if $\gamma_b\leq 1$, then  it follows since $\frac{2-\Val(\cpi_b)}{\gamma_b^c} \geq 1$ and since the left-hand side of the inequality is a probability mass. Hence,
\begin{align}
\lefteqn{\ppr{\LDist{\rcA{1},\HonB}}{\descP{\unbal{\cpi}{1}{\gamma}{}\setminus \descP{\low{\cpi}{\delta,\HonA}}}}}\\
&=\EdgeDist_{(\rcA{1},\HonB)}(\EmptyString,0)\cdot\ppr{\LDist{(\rcA{1},\HonB)_0}}{\descP{\unbal{\cpi_0}{1}{\gamma_0}{}\setminus \descP{\low{\cpi_0}{\delta,\HonA}}}}\nonumber\\
&\quad +\EdgeDist_{(\rcA{1},\HonB)}(\EmptyString,1)\cdot\ppr{\LDist{(\rcA{1},\HonB)_1}}{\descP{\unbal{\cpi_1}{1}{\gamma_1}{}\setminus \descP{\low{\cpi_1}{\delta,\HonA}}}}\nonumber\\
&\leq p\cdot\left(\frac{\Val(\cpi_0)}{\Val(\cpi)}\right)^{1+c}\cdot\frac{2-\Val(\cpi_0)}{\gamma^c}+ (1-p)\cdot\left(\frac{\Val(\cpi_1)}{\Val(\cpi)}\right)^{1+c}\cdot\frac{2-\Val(\cpi_1)}{\gamma^c},\nonumber
\end{align}
where the equality follows from \cref{eq:unbal4}, and the inequality follows from \cref{eq:unbal5} together with \cref{claim:weights}.
Letting $y = \frac{\Val(\cpi_0)}{\Val(\cpi)} - 1$ , $x= \Val(\cpi)$ and $\lambda=\frac{p}{1-p}$, and noting that $\lambda y=\left(\frac{\Val(\cpi_0)}{\Val(\cpi)}-1\right)\cdot\frac{p}{1-p}=\frac{p\cdot\Val(\cpi_0)-p\cdot\Val(\cpi)}{\Val(\cpi)-p\cdot\Val(\cpi)}\leq\frac{p\cdot\Val(\cpi_0)}{\Val(\cpi)}\leq1$,  \cref{lemma:calculus2} yields (after multiplying by $\frac{1-p}{\gamma^c}$) that
\begin{align}
p\cdot\left(\frac{\Val(\cpi_0)}{\Val(\cpi)}\right)^{1+c}\cdot\frac{2-\Val(\cpi_0)}{\gamma^c}+ (1-p)\cdot\left(\frac{\Val(\cpi_1)}{\Val(\cpi)}\right)^{1+c}\cdot\frac{2-\Val(\cpi_1)}{\gamma^c}\leq\frac{2-\Val(\cpi)}{\gamma^c},
\end{align}
completing the proof for the case  $\Val(\cpi_1)>0$.

It is left to argue the case that $\Val(\cpi_1)=0$. In this case, according to \cref{claim:weights}, it holds that $\EdgeDist_{(\rcA{1},\HonB)}(\EmptyString,0)=1$ and $\EdgeDist_{(\rcA{1},\HonB)}(\EmptyString,1)=0$. Hence, there are no unbalanced nodes in $\cpi_1$, \ie $\unbal{\cpi}{1}{\gamma}{}\setminus \descP{\low{\cpi}{\delta,\HonA}} \cap \Vertices(\cpi_1)=\emptyset$. As before, let $\gamma_0 \eqdef \frac{\Val(\cpi)}{\Val(\cpi_0)}\cdot \gamma = p\cdot\gamma$ (the latter equality holds since $\Val(\cpi) = p\cdot \Val(\cpi_0)$.) Arguments similar to those used to prove \cref{eq:unbal4} yield that
\begin{align}
\unbal{\cpi}{1}{\gamma}{}\setminus \descP{\low{\cpi}{\delta,\HonA}}    = \unbal{\cpi_0}{1}{\gamma_0}{}\setminus \descP{\low{\cpi_0}{\delta,\HonA}}.
\end{align}

It follows that
\begin{align*}
\lefteqn{\ppr{\LDist{\rcA{1},\HonB}}{\descP{\unbal{\cpi}{1}{\gamma}{}\setminus \descP{\low{\cpi}{\delta,\HonA}}}}}\\
&=\EdgeDist_{(\rcA{1},\HonB)}(\EmptyString,0)\cdot\ppr{\LDist{(\rcA{1},\HonB)_0}}{\descP{\unbal{\cpi_0}{1}{\gamma_0}{}\setminus \descP{\low{\cpi_0}{\delta,\HonA}}}}\nonumber\\
&\leq \paren{\frac{1}{p}}^{1+c}\cdot\frac{2-\Val(\cpi_0)}{\gamma^c}.
\end{align*}
Applying \cref{lemma:calculus2} with the same parameters as above completes the proof.
\end{description}

\medskip
\emph{Proof of (\ref{item:unbal2}):}
Fix $\gamma>1$ and recall that for a set $\cS\subset\Vertices(\cpi)$, $\frnt{\cS}$ stands for the frontier of $\cS$, \ie the set of nodes belong to $\cS$, whose ancestors do not belong to $\cS$ (see \cref{def:BinaryTrees}). We prove that
\begin{align}\label{eq:unbalTwo1}
\frnt{\unbal{\cpi}{1}{\gamma}{} \setminus \propDescP{\low{\cpi}{\delta,\HonA}}} \subseteq \unbal{\cpi}{1}{\gamma}{} \setminus \descP{\low{\cpi}{\delta,\HonA}},
\end{align}
and the proof of (\ref{item:unbal2}) follows.

Let $u\in\frnt{\unbal{\cpi}{1}{\gamma}{} \setminus \propDescP{\low{\cpi}{\delta,\HonA}}}$. We prove \cref{eq:unbalTwo1} by showing that $u\notin\low{\cpi}{\delta,\HonA}$. Since $\gamma > 1$ and $u\in \unbal{\cpi}{1}{\gamma}{}$, it is clear that $u\neq\Root(\cpi)$. Let $w$ be the parent of $u$. By the choice of $u$, it follows that $w\notin\unbal{\cpi}{1}{\gamma}{}$, and thus $\VerticesDist_{(\rcA{1},\HonB)}(w) < \gamma\cdot\VerticesDist_{(\Ac,\HonB)}(w)$. We write
\begin{align}
\gamma\cdot \VerticesDist_{(\HonA,\HonB)}(w) \cdot \EdgeDist_{(\rcA{1},\HonB)}(w,u) &> \VerticesDist_{(\rcA{1},\HonB)}(w) \cdot \EdgeDist_{(\rcA{1},\HonB)}(w,u)\\
&=  \VerticesDist_{(\rcA{1},\HonB)}(u) \nonumber\\
&\geq \gamma\cdot  \VerticesDist_{(\HonA,\HonB)}(u)\nonumber \\
&= \gamma\cdot  \VerticesDist_{(\HonA,\HonB)}(w) \cdot \EdgeDist_{(\HonA,\HonB)}(w,u).\nonumber
\end{align}
We conclude that $\EdgeDist_{(\HonA,\HonB)}(w,u) < \EdgeDist_{(\rcA{1},\HonB)}(w,u)$, and thus it must be the case that $w$ is controlled by $\HonA$. By \cref{claim:weights}, it holds that $\EdgeDist_{(\rcA{1},\HonB)}(w,u) = \EdgeDist_{(\HonA,\HonB)}(w,u) \cdot\frac{\Val(\cpi_u)}{\Val(\cpi_w)}$, and thus $\Val(\cpi_u) > \Val(\cpi_w)$. Finally, observe that $w\notin \low{\cpi}{\delta,\HonA}$, since otherwise $u\in \propDescP{\low{\cpi}{\delta,\HonA}}$. It follows that $\Val(\cpi_w) > \delta$, and hence $\Val(\cpi_u) > \delta$, as required.

\medskip
\emph{Proof of (\ref{item:unbal3}):}
Note that for every $\delta'\geq\delta$ it holds that $\low{\cpi}{\delta,\HonA} \subseteq \low{\cpi}{\delta',\HonA}$. Hence, $\unbal{\cpi}{1}{\gamma}) \setminus \propDesc(\low{\cpi}{\delta',\HonA})\subseteq\unbal{\cpi}{1}{\gamma}) \setminus \propDesc(\low{\cpi}{\delta,\HonA})$, and the proof follows.
\end{proof}

\subsubsection{The Biased-Continuation Attacker is Robust}\label{sec:Robust}
Consider what happens when the biased-continuations attacker attacks a protocol $\cpi=(\HonA,\HonB)$. This attacker chooses a random $1$-leaf according to $\LDist{\cpi}$, the leaf distribution of $\cpi$. If these was another protocol $\cpi'$ that was close (in the leaf-distribution sense) to $\cpi$, then the attacker can instead sample from $\LDist{\cpi'}$, while making similar decisions throughout its operation. So, the biased-continuation attacker is robust to the distribution from which it samples. This is formally put in the next lemma.

\begin{lemma}[robustness lemma]\label{lem:RCofStatCloseProt}
Let $\cpi=(\HonA,\HonB)$ and $\cpi'=(\HonC,\HonD)$ be two $\rnd$-round protocols, let $\delta\in(0,\frac12]$, and let $c=c(\delta)$ be according to  \cref{lem:ProbVisitUnBal}. Assuming $\SDP{\LDist{\cpi}}{\LDist{\cpi'}}\leq\alpha$,  $\Color_{\cpi}\equiv\Color_{\cpi'}$, and  $\cpi$ and $\cpi'$ have the same control scheme, it holds that
\begin{align*}
\SD\bigg(\LDist{\rcAP{\cpi}{1},\HonB},\LDist{\rcCP{\cpi'}{1},\HonD}\bigg)&\leq\frac{ 3\cdot m\cdot\gamma}{\delta'}\cdot
\paren{ \alpha +\ppr{\LDist{\HonA,\HonB}}{\desc\left(\low{\cpi}{\delta',\HonA}\cup\low{\cpi'}{\delta',\HonC}\right)}} +\frac{2}{\gamma^c},
\end{align*}
for every $\delta'\geq\delta$ and $\gamma\geq1$, where $\rcA{1}$ and $\rcC{1}$ are as in \cref{alg:adversary}.
\end{lemma}
Namely, the biased-continuation attacker does not make similar protocols too dissimilar. The rest of this section is dedicated to proving \cref{lem:RCofStatCloseProt}.

\newcommand{\CtrlEQ}{{\mathcal{C}trlEQ}}
\begin{proof}
We use \cref{lemma:SDmQueriesAlg}. Define the random function $f$ given an element from $\Vertices(\cpi)\cup \set{\perp}$ as follows: given $u\in\Vertices(\cpi)$, if $\HonA$ controls $u$ return $\ell\la\LDist{\cpi_u}$ such that $\Color_\cpi(\ell)=1$ (if no such node exists, return an arbitrary node in $\descP{u}$); otherwise, \ie if $\HonB$ controls $u$, return $\ell\la\LDist{\cpi_u}$. Finally, given $\perp$, $f$ return $\perp$. The random function $g$ given an element from $\Vertices(\cpi)\cup \set{\perp}$ is analogously defined \wrt protocol $\cpi'$.\footnote{The sets $\Vertices(\cpi)$ and $\Vertices(\cpi')$, as well as the sets $\Leaves(\cpi)$ and $\Leaves(\cpi')$, are identical, as the both describe nodes in the complete binary tree of height $m$. See \cref{section:Preliminaries} for further details.}
For function $\phi$ with range in $\Leaves(\cpi)$, let $\sfH^\phi$ be the following algorithm:

\begin{algorithm}[$\sfH$]\label{alg:robust}
\item State: node $u$, set to $\EmptyString$ at the start of the execution.
\item Operation:
\begin{enumerate}
\item Repeat for $m$ times:
	\begin{enumerate}
	\item Set $\ell = \phi(u)$.
	\item Set $u = u\concat \ell_{i}$, where $i$ is the current iteration.
	\end{enumerate}

\item Output $u$.
\end{enumerate}
\end{algorithm}
It is easy to verify that $\sfH^f \equiv \LDist{\rcAP{\cpi}{1},\HonB}$ and $\sfH^g \equiv \LDist{\rcCP{\cpi'}{1},\HonD}$. Hence, it suffices to upper-bound $\SDP{\sfH^f}{\sfH^g}$.

For $i\in[m]$, let $P_i$ to be $i$'th node in a random execution of $\cpi$ (such a node consists of $i-1$ bits). We use the next claim, proven below.
\begin{claim}\label{claim:robustSDfg}
$\eex{u\la P_i}{\SDP{f(u)}{g(u)}} \leq \frac{2\alpha}{\delta'}+\ppr{\LDist{\cpi}}{\descP{\low{\cpi}{\delta',\HonA}\cup\low{\cpi'}{\delta',\HonC}}}$.
\end{claim}

Let $Q_i$ denote the $i$'th query to $f$ in a random execution of $\sfH^{f}$ (note that by construction, such a query always exists) and let $Q=(Q_1,\ldots,Q_m)$.
By construction, for $u\in\Vertices(\cpi)$ with $\size{u}=i-1$, $Q_i(u)$ is the probability that $u$ is visited in a random execution of $\paren{\rcAP{\cpi}{1},\HonB}$. We get
\begin{align*}
\ppr{(q_1,\ldots,q_m) \la Q}{\exists i\in [m] \colon  q_i\neq \perp \land Q_i(q_i) > \gamma\cdot P_i(q_i)} &=\ppr{\LDist{\rcA{1},\HonB}}{\descP{\UnBal{\cpi}{\gamma} }}\\
&\leq\gamma\cdot\ppr{\LDist{\cpi}}{\descP{\low{\cpi}{\delta',\HonA}}}+\frac{2}{\gamma^{c}} ,
\end{align*}
where the inequality  follows from \cref{cor:UnbalBound}.

The proof of \cref{lem:RCofStatCloseProt} now follows by \cref{lemma:SDmQueriesAlg}, letting $k=m$, $a=\frac{2\alpha}{\delta'}+\ppr{\LDist{\cpi}}{\desc\left(\low{\cpi}{\delta',\HonA}\cup\low{\cpi'}{\delta',\HonC}\right)}$, $\lambda=\gamma$ and $b=\gamma\cdot\ppr{\LDist{\cpi}}{\descP{\low{\cpi}{\delta',\HonA}}}+\frac{2}{\gamma^{c}} $.
\end{proof}

\begin{proof}[Proof of \cref{claim:robustSDfg}]
Let $\Vertices_{i}(\cpi)=\set{v\in\Vertices(\cpi) \colon \size{v}=i-1}$, $\Vertices_{i}^\HonA(\cpi) = \Vertices_{i}(\cpi) \cap \ctrl{\cpi}{\HonA}$ and $\Vertices_{i}^\HonB(\cpi) = \Vertices_{i}(\cpi) \cap \ctrl{\cpi}{\HonB}$. Compute
\begin{align}\label{eq:SDfg}
\eex{u\la P_i}{\SDP{f(u)}{g(u)}} &= \sum_{u\in \Vertices_{i}(\cpi)} P_i(u)\cdot \SDP{f(u)}{g(u)} \\
&= \sum_{u\in \Vertices_{i}^\HonA(\cpi)} P_i(u)\cdot \SDP{f(u)}{g(u)} + \sum_{u\in \Vertices_{i}^\HonB(\cpi)} P_i(u)\cdot \SDP{f(u)}{g(u)}.\nonumber
\end{align}

In the rest of the proof we show that
\begin{align}\label{eq:SDfgA}
\sum_{u\in \Vertices_{i}^\HonA(\cpi)} P_i(u)\cdot \SDP{f(u)}{g(u)} &\leq \frac{1}{\delta'}\cdot \sum_{u\in\Vertices_{i}^\HonA(\cpi)} P_i(u)\cdot \SDP{\LDist{\cpi_u}}{\LDist{\cpi'_u}} \\
&\quad + \ppr{\LDist{\cpi}}{\descP{\low{\cpi}{\delta',\HonA}\cup\low{\cpi'}{\delta',\HonC}}},\nonumber
\end{align}
that
\begin{align}\label{eq:SDfgB}
\sum_{u\in \Vertices_{i}^\HonB(\cpi)} P_i(u)\cdot \SDP{f(u)}{g(u)} \leq \sum_{u\in\Vertices_{i}^\HonB(\cpi)} P_i(u)\cdot \SDP{\LDist{\cpi_u}}{\LDist{\cpi'_u}},
\end{align}
and that
\begin{align}\label{eq:SDu}
\sum_{u\in\Vertices_{i}(\cpi)} P_i(u)\cdot \SDP{\LDist{\cpi_u}}{\LDist{\cpi'_u}}\leq 2 \cdot  \SDP{\LDist{\cpi}}{\LDist{\cpi'}}.
\end{align}
Plugging \cref{eq:SDfgA,eq:SDfgB,eq:SDu} into \cref{eq:SDfg} completes the proof \cref{claim:robustSDfg}.

\medskip
\emph{Proof of \cref{eq:SDfgA}:}
Let $u\in\Vertices_i^\HonA(\cpi)$. By the definition of $f$, and since $u$ is under $\HonA$'s control, it follows that $\ppr{}{f(u)=\ell} = \LDist{\cpi_u}(\ell)/\Val(\cpi_u)$ if $\Color_\cpi(\ell)=1$, and $\ppr{}{f(u)=\ell} = 0$ otherwise. Since $\cpi$ and $\cpi'$ have the same control scheme, the same holds for $g(u)$ \wrt $\cpi'$. Let $\cS'_u\subseteq\Leaves_1(\cpi)$ be the set with  $\SDP{f(u)}{g(u)} = \sum_{\ell\in\cS'_u} \paren{\ppr{}{f(u)=\ell} -  \ppr{}{g(u)=\ell}} = \sum_{\ell\in\Leaves_1(\cpi)\setminus\cS'_u} \paren{\ppr{}{g(u)=\ell} -  \ppr{}{f(u)=\ell}}$.\footnote{Note that it must be the case that $\cS'_u\subseteq\Leaves_1(\cpi)$, since $\ppr{}{f(u)=\ell}=\ppr{}{g(u)=\ell}=0$, for every $\ell$ with $\Color_{\cpi}(\ell)=0$, which follows from the assumption that $\Color_{\cpi}\equiv \Color_{\cpi'}$.} Define $\cS_u\subseteq\Leaves_1(\cpi)$ as follows: if $\Val(\cpi_u) \geq \Val(\cpi'_u)$ let $\cS_u=\cS'_u$; otherwise let $\cS_u= \Leaves_1(\cpi)\setminus\cS'_u$.  It follows that
\begin{align}\label{eq:SDfgd}
\sum_{u \in \Vertices_i^\HonA(\cpi)} P_i(u) \cdot \SDP{f(u)}{g(u)}
& \leq  \sum_{\substack{u \in \Vertices_i^\HonA(\cpi)\colon \\ \Val(\cpi_u)\geq \Val(\cpi'_u) \geq \delta'}} P_i(u) \cdot \sum_{\ell\in\cS_u} \paren{ \frac{\LDist{\cpi_u}(\ell)}{\Val(\cpi_u)} -  \frac{\LDist{\cpi'_u}(\ell)}{\Val(\cpi'_u)}} \\
&\quad + \sum_{\substack{u \in \Vertices_i^\HonA(\cpi)\colon \\ \Val(\cpi'_u) > \Val(\cpi_u) \geq \delta'}} P_i(u) \cdot \sum_{\ell\in\cS_u} \paren{ \frac{\LDist{\cpi'_u}(\ell)}{\Val(\cpi'_u)} -  \frac{\LDist{\cpi_u}(\ell)}{\Val(\cpi_u)}} \nonumber\\
&\quad + \sum_{\substack{u \in \Vertices_i^\HonA(\cpi) \colon \\ \Val(\cpi_u) < \delta' \lor \Val(\cpi'_u) < \delta'}} P_i(u). \nonumber
\end{align}
Assume  $\Val(\cpi_u)\geq \Val(\cpi'_u)$. The definition of $\cS_u$ implies that $\LDist{\cpi_u}(\ell)/\Val(\cpi_u) \geq \LDist{\cpi'_u}(\ell)/\Val(\cpi'_u)$ for every $\ell\in\cS_u$. But since $\Val(\cpi_u)/\Val(\cpi'_u)\geq 1$, the latter yields that  $\LDist{\cpi_u}(\ell) \geq \LDist{\cpi'_u}(\ell)$ for every $\ell\in\cS_u$. Using this observation, we bound the first summand in the right-hand side of \cref{eq:SDfgd}.
\begin{align}\label{eq:SDfgd1}
\lefteqn{\sum_{\substack{u \in \Vertices_i^\HonA(\cpi)\colon \\ \Val(\cpi_u)\geq \Val(\cpi'_u) \geq \delta'}} P_i(u) \cdot \sum_{\ell\in\cS_u} \paren{ \frac{\LDist{\cpi_u}(\ell)}{\Val(\cpi_u)} -  \frac{\LDist{\cpi'_u}(\ell)}{\Val(\cpi'_u)}}} \\
& \leq \sum_{\substack{u \in \Vertices_i^\HonA(\cpi)\colon \\ \Val(\cpi_u)\geq \Val(\cpi'_u) \geq \delta'}} \frac{P_i(u)}{\Val(\cpi')} \cdot \sum_{\ell\in\cS_u} \paren{\LDist{\cpi_u}(\ell) - \LDist{\cpi'_u}(\ell)} \nonumber\\
& \leq \frac{1}{\delta'}\sum_{\substack{u \in \Vertices_i^\HonA(\cpi)\colon \\ \Val(\cpi_u)\geq \Val(\cpi'_u) \geq \delta'}} P_i(u) \cdot \sum_{\ell\in\cS_u} \paren{\LDist{\cpi_u}(\ell) - \LDist{\cpi'_u}(\ell)} \nonumber\\
&\leq \frac{1}{\delta'} \sum_{\substack{u \in \Vertices_i^\HonA(\cpi)\colon \\ \Val(\cpi_u)\geq \Val(\cpi'_u) \geq \delta'}} P_i(u)\cdot \SDP{\LDist{\cpi_u}}{\LDist{\cpi'_u}},\nonumber
\end{align}
where the second inequality follows since $\sum_{\ell\in\cS_u} \paren{\LDist{\cpi_u}(\ell) - \LDist{\cpi'_u}(\ell)} \geq 0$, as argued above.
Similar calculations, and using the symmetry of statistical distance, we bound the second summand in the right-hand side of \cref{eq:SDfgd}:
\begin{align}\label{eq:SDfgd2}
\sum_{\substack{u \in \Vertices_i^\HonA(\cpi)\colon \\ \Val(\cpi'_u)\geq \Val(\cpi_u) \geq \delta'}} P_i(u) \cdot \sum_{\ell\in\cS_u} \paren{ \frac{\LDist{\cpi'_u}(\ell)}{\Val(\cpi'_u)} -  \frac{\LDist{\cpi_u}(\ell)}{\Val(\cpi_u)}}
\leq \frac{1}{\delta'} \sum_{\substack{u \in \Vertices_i^\HonA(\cpi)\colon \\ \Val(\cpi'_u)\geq \Val(\cpi_u) \geq \delta'}} P_i(u)\cdot \SDP{\LDist{\cpi_u}}{\LDist{\cpi'_u}}.
\end{align}
Finally, to bound the third summand in the right-hand side of \cref{eq:SDfgd}, we note that it sums over (not all) $u\in\low{\cpi}{\delta',\HonA}\cup\low{\cpi'}{\delta',\HonC}$. Since $P_i$ simply samples a random partial transcript from $\cpi$, it follows that
\begin{align}\label{eq:SDfgd3}
\sum_{\substack{u \in \Vertices_i^\HonA(\cpi) \colon \\ \Val(\cpi_u) < \delta' \lor \Val(\cpi'_u) < \delta'}} P_i(u) \leq \ppr{\LDist{\cpi}}{\descP{\low{\cpi}{\delta',\HonA}\cup\low{\cpi'}{\delta',\HonC}}}.
\end{align}
Plugging \cref{eq:SDfgd1,eq:SDfgd2,eq:SDfgd3} into \cref{eq:SDfgd} yields \cref{eq:SDfgA}.

\medskip
\emph{Proof of \cref{eq:SDfgB}:}
Since it is the right-hand party who controls $u$ in $\cpi$ and in $\cpi'$, it follows that $\SDP{f(u)}{g(u)} = \SDP{\LDist{\cpi_u}}{\LDist{\cpi'_u}}$, and \cref{eq:SDfgB} follows.

\medskip
\emph{Proof of \cref{eq:SDu}:}
Using the definition of $P_i$, we can write
\begin{align*}
\sum_{u\in\Vertices_{i}(\cpi)} P_i(u)\cdot \SDP{\LDist{\cpi_u}}{\LDist{\cpi'_u}} &=
\sum_{u\in\Vertices_{i}(\cpi)} \VerticesDist_{\cpi}(u)\cdot \frac12\sum_{\ell\in\Leaves(\cpi_u)} \size{\VerticesDist_{\cpi_u}(\ell) - \VerticesDist_{\cpi'_u}(\ell)}\\
&= \frac12\sum_{\ell\in\Leaves(\cpi)}\size{\VerticesDist_{\cpi}(\ell_{1,\ldots,i-1})\cdot\VerticesDist_{\cpi_{\ell_{1,\ldots,i-1}}}(\ell) -\VerticesDist_{\cpi}(\ell_{1,\ldots,i-1})\cdot\VerticesDist_{\cpi'_{\ell_{1,\ldots,i-1}}}(\ell)}\\
&= \SDP{\LDist{\cpi}}{\LDist{\cpi''}},
\end{align*}
for $\LDist{\cpi''}(\ell) \eqdef \VerticesDist_{\cpi}(\ell_{1,\ldots,i-1})\cdot\VerticesDist_{\cpi'_{\ell_{1,\ldots,i-1}}}(\ell)$.

We prove that $\SDP{\LDist{\cpi'}}{\LDist{\cpi''}}\leq\SDP{\LDist{\cpi'}}{\LDist{\cpi}}$, and \cref{eq:SDu} follows from the triangle inequality. Let $h$ be the random function that, given $\ell\in\Leaves(\cpi)$, returns $\ell'\la\LDist{\cpi'_{\ell_{1,\ldots,i-1}}}$. Therefore, $h(\LDist{\cpi'})\equiv \LDist{\cpi'}$ and $h(\LDist{\cpi})\equiv \LDist{\cpi''}$, and this completes the proof.

This completes the proof of \cref{eq:SDfgA,eq:SDfgB,eq:SDu}, and thus  the proof of \cref{claim:robustSDfg}.
\end{proof}

\subsubsection{The Success Probability of $\realA{1,\xi,\delta}{\cpi}$ --- The ``Ideal to Real'' Reduction}\label{sec:real:app:bound}
Consider an execution of $(\realA{1,\xi,\delta}{},\HonB)$. Such an execution asks the approximated biased continuator $\RandomContXD{\xi}{\delta}$ for continuations of transcripts under $\HonA$'s control, leading to $1$-leaves. Hence, as long as this execution generates neither low-value transcripts under $\HonA$'s control nor unbalanced transcripts, we expect the approximated biased-continuation attacker to do almost as well as its ideal variant.  This is formally put in the next lemma.

\begin{lemma}\label{lem:fromIdeal2Real}
Let $\cpi=(\HonA,\HonB)$ be an $\rnd$-round protocol and let $\delta\in(0,\frac12]$. Then
\begin{align*}
\SDP{\LDist{\rcAP{\cpi}{1},\HonB}}{\LDist{\realA{1,\xi,\delta}{\cpi},\HonB}}\leq \rnd\cdot\gamma\cdot\paren{2\xi+\ppr{\LDist{\HonA,\HonB}}{\desc(\low{\cpi}{\delta,\HonA})}} + \ppr{\LDist{\rcAP{\cpi}{1},\HonB}}{\descP{\unbal{\cpi}{1}{\gamma}{}}}
\end{align*}
for every $\gamma\geq 1$  and $\xi > 0$.
\end{lemma}
\begin{proof}
	We use \cref{lemma:SDmQueriesAlg}. For function  $\phi$, let $\sfH^\phi$ be an algorithm that outputs the transcript of a random execution of  $\paren{\rcAP{\cpi}{1},\HonB}$ in which $\rcAP{\cpi}{1}$'s calls to $\RandomCont_\cpi$ are sent to $\phi$ instead.\footnote{Note that $\sfH$ is not the same as \cref{alg:robust} defined in the proof of the robustness lemma (\cref{lem:RCofStatCloseProt}). There we considered two different underlying protocols, and needed to also argue about the different actions the honest (\ie non-attacked) parties took. Here, we have only one underlying protocol and thus care only about the calls the attacked party makes.}  Let $f$ and $g$ be the (random) functions $\RandomCont_\cpi$ and $\RandomCont_\cpi^{\xi,\delta}$ respectively, letting also $f(\perp)=g(\perp)=\perp$. By construction, it holds that
	\begin{align}
		\SDP{\LDist{\rcAP{\cpi}{1},\HonB}}{\LDist{\realA{1,\xi,\delta}{\cpi},\HonB}}=\SDP{\sfH^f}{\sfH^g}.
	\end{align}
	For $i\in[m]$, let $P'_i$ be the distribution of the $i$'th node under $\HonA$'s control in a random execution of $\cpi$, taking the value $\perp$ if no such node exists, and let $P_i = (P'_i,1)$, with $(\perp,1)=\perp$. By definition,
	\begin{align}
		\eex{q\la P_i}{\SD(f(q),g(q))} &=
		\eex{q\la P_i}{\SD(\RandomCont_\cpi(q),\RandomCont_\cpi^{\xi,\delta}(q))  \cdot 1_{\neg \perp} (q)} \\ &\leq 2 \xi+\ppr{\LDist{\cpi}}{\desc(\low{\cpi}{\delta,\HonA})},\nonumber
	\end{align}
	letting the indicator  $1_{\neg \perp}(q)$  take the value one if  $q\neq \perp$, and zero otherwise.

	Let $Q_i$ denote the $i$'th query to $f$ in a random execution of $\sfH^f$, taking the value  $\perp$ if no such query exists, and let $Q = (Q_1,
	\ldots,Q_m)$.  By definition,
	\begin{align}
		\ppr{(q_1,\ldots,q_m) \la Q}{\exists i\in [m] \colon  q_i\neq \perp \land \; Q_i(q_i) > \gamma\cdot P_i(q_i)} = \ppr{\LDist{\rcAP{\cpi}{1},\HonB}}{\descP{\unbal{\cpi}{1}{\gamma}{}}}.
	\end{align}
	Hence,  the proof follows by \cref{lemma:SDmQueriesAlg}, letting  $k:=m$, $a:=2 \xi+\ppr{\LDist{\cpi}}{\desc(\low{\cpi}{\delta,\HonA})}$, $\lambda:=\gamma$ and $b:=\ppr{\LDist{\rcAP{\cpi}{1},\HonB}}{\descP{\unbal{\cpi}{1}{\gamma}{}}}$.
\end{proof}

Our use of  \cref{lem:fromIdeal2Real} is via the following lemma that states that  the approximated biased-continuation attacker successfully biases protocols in which the probability of hitting $\HonA$-controlled low-value nodes is small.

\begin{lemma}\label{cor:approxBC}
Let $\cpi=(\HonA,\HonB)$ be an $\rnd$-round protocol, let $\delta\in(0,\frac12]$, and let $c=c(\delta)$ be according to \cref{lem:ProbVisitUnBal}, then
\begin{align*}
\SDP{\LDist{\rcAP{\cpi}{1},\HonB}}{\LDist{\realA{1,\xi,\delta'}{\cpi},\HonB}}&\leq 2\cdot\rnd\cdot\gamma\cdot\paren{\xi+\ppr{\LDist{\HonA,\HonB}}{\desc(\low{\cpi}{\delta',\HonA})}} + \frac{2}{\gamma^c}
\end{align*}
for any  $\delta'\geq \delta$, $\xi>0$ and $\gamma > 1$.
\end{lemma}
\begin{proof}
Follows by plugging \cref{cor:UnbalBound} into \cref{lem:fromIdeal2Real}.
\end{proof}

\subsubsection{Bounding the Probability of Hitting Low-Density Sets}\label{sec:lowRemainLow}
Our final step before proving \cref{lemma:goodWOLowValueSimple,claim:SmallStaySmallApproxSimple}, is showing that the recursive ideal biased-continuation attacker does not increase the probability of hitting any set by much. This is a generalization of \cref{cor:UnbalBound} to arbitrary sets of nodes (\ie not only unbalanced) and to the recursive version of the ideal biased-continuation attacker.

\cref{cor:UnbalBound} considered unbalanced nodes to be those that the probability of hitting them in the protocol which the (non-recursive) biased-continuation attacker take the role of $\HonA$ is $\gamma$-times higher than in the original protocol. When extending \cref{cor:UnbalBound} to the recursive version of the attacker, we  take different degree of ``unbalancedness'' for every level of the recursion. Specifically, we will (implicitly) define unbalanced nodes for the $i$'th level of the recursion, to be those nodes that the probability of hitting them in the protocol in which the $i$'th-level recursive attacker takes the role of $\HonA$, is $\gamma_i$-times higher than in the protocol which the $(i-1)$'th-level recursive attacker take the role of $\HonA$. The freedom to choose different degrees of ``unbalanceness'' for different levels of the recursion will be crucial when arguing that (a similar attack to) the biased-continuation attack can be can implemented efficiently assuming the in-existence of OWFs.


\begin{lemma}\label{prop:UnbalBounds}
Let $\cpi=\HonAHonB$ be a protocol, let $\delta\in(0,\frac12]$, and let $c=c(\delta)$ be according  \cref{lem:ProbVisitUnBal}. Then for any $\delta'\geq\delta$, every $k\in\N$, any $(\gamma_1,\ldots,\gamma_k)\in (1,\infty)^k$ and every $\cS\subseteq\Vertices(\cpi)$ it holds that
\begin{align*}
\ppr{\LDist{\rcAP{\cpi}{k},\HonB}}{\descP{\cS\cup\low{\cpi}{\delta',\HonA}}} \leq \ppr{\LDist{\HonA,\HonB}}{\descP{\cS\cup\low{\cpi}{\delta',\HonA}}} \cdot \prod_{i=1}^k\gamma_i + 2\cdot\sum_{i=1}^k \frac{\cdot\prod_{j=i+1}^{k}\gamma_j}{\gamma_i^c}.
\end{align*}
\end{lemma}

To prove \cref{prop:UnbalBounds} we will use the next claim.

\begin{claim}\label{cor:SmallStaySmall}
Let $\cpi=\HonAHonB$ be a protocol, let $\cS\subseteq\Vertices(\cpi)$, let $\delta\in(0,\frac12]$ and let $c=c(\delta)$ from \cref{lem:ProbVisitUnBal}. Then, for every $\delta'\geq\delta$ and  $\gamma>1$, it holds that
\begin{align*}
\ppr{\LDist{\rcAP{\cpi}{1},\HonB}}{\descP{\cS\cup\low{\cpi}{\delta',\HonA}}} \leq \gamma\cdot \ppr{\LDist{\HonA,\HonB}}{\descP{\paren{\cS\cup\low{\cpi}{\delta',\HonA}}\setminus \descP{\UnBal{\cpi}{\gamma}}}} +  \frac{2}{\gamma^c}.
\end{align*}
\end{claim}
\begin{proof}
Fix $\delta'\geq\delta$ and  $\gamma>1$.
Applying  \cref{prop:UnBal1} \wrt $\cA=\cS\cup\low{\cpi}{\delta',\HonA}$, $\cB=\low{\cpi}{\delta',\HonA}$ and $\cC=\UnBal{\cpi}{\gamma}$ yields that
\begin{align}\label{eq:propUnBal1}
\descP{\cS\cup\low{\cpi}{\delta',\HonA}}\subseteq \descP{\paren{\cS\cup\low{\cpi}{\delta',\HonA}}\setminus \descP{\UnBal{\cpi}{\gamma}}} \cup \descP{\UnBal{\cpi}{\gamma}\setminus\propDescP{\low{\cpi}{\delta',\HonA}}}.
\end{align}
It follows that
\begin{align*}
\ppr{\LDist{\rcA{1},\HonB}}{\descP{\cS}} &\leq \ppr{\LDist{\rcA{1},\HonB}}{\descP{\paren{\cS\cup\low{\cpi}{\delta',\HonA}}\setminus \descP{\UnBal{\cpi}{\gamma}}}}\\
&\quad +\ppr{\LDist{\rcA{1},\HonB}}{\descP{\UnBal{\cpi}{\gamma}\setminus\propDescP{\low{\cpi}{\delta',\HonA}}}} \\
&\leq \gamma\cdot \ppr{\LDist{\HonA,\HonB}}{\descP{\paren{\cS\cup\low{\cpi}{\delta',\HonA}}\setminus \descP{\UnBal{\cpi}{\gamma}}}} +  \frac{2}{\gamma^c},
\end{align*}
where the first inequality follows from \cref{eq:propUnBal1} and the second inequality follows from the definition of $\UnBal{\cpi}{\gamma}$ (\cref{def:unbalNodes}) and \cref{lem:ProbVisitUnBal}.
\end{proof}

We are now ready to prove \cref{prop:UnbalBounds}.

\begin{proof}[Proof of \cref{prop:UnbalBounds}.]
Fix $\delta'\ge\delta $ and $(\gamma_1,\ldots,\gamma_k) \in (1,\infty)^k$.
The proof is by induction on $k$. For $k=0$, the proof follows immediately from definition.

Assume the lemma holds for $k-1$; we prove it for $k$. For $i\in(k-1)$, let $\cpi^{(i)} = \paren{\rcAP{\cpi}{i},\HonB}$.
It is easy to verify that when the ideal biased continuation attacker takes the role of $\HonA$ in the protocol and tries to bias the outcome towards $1$, the value of every node cannot decrease. Namely, it holds that $\low{\cpi^{(i)}}{\delta',\HonA} \subseteq \low{\cpi^{(i-1)}}{\delta',\HonA}$ for every $i\in[k-1]$, and thus $\low{\cpi^{(k-1)}}{\delta',\HonA} \subseteq \low{\cpi}{\delta',\HonA}$. Applying \cref{cor:SmallStaySmall} \wrt the protocol $\cpi^{(k-1)}$, set $\cS\cup\low{\cpi}{\delta',\HonA}$ and $\gamma=\gamma_k$, yields that
\begin{align}\label{eq:propUnbalSmall1}
\lefteqn{\ppr{\LDist{\rcAP{\cpi^{(k-1)}}{1},\HonB}}{\descP{\cS\cup\low{\cpi}{\delta',\HonA}}}} \\
&\leq \gamma_k\cdot \ppr{\LDist{\cpi^{(k-1)}}}{\descP{\paren{\cS\cup\low{\cpi}{\delta',\HonA}\cup\low{\cpi^{(k-1)}}{\delta',\HonA}}\setminus \descP{\UnBal{\cpi^{(k-1)}}{\gamma_k}}}} +  \frac{2}{\gamma_k^c} \nonumber\\
&\leq  \gamma_k\cdot \ppr{\LDist{\cpi^{(k-1)}}}{\descP{\cS\cup\low{\cpi}{\delta',\HonA}}} +  \frac{2}{\gamma_k^c}. \nonumber
\end{align}
\cref{eq:propUnbalSmall1} together with the induction hypothesis now yield that
\begin{align*}
\ppr{\LDist{\rcAP{\cpi^{(k-1)}}{1},\HonB}}{\descP{\cS\cup\low{\cpi}{\delta',\HonA}}} &\leq \gamma_k \paren{\ppr{\LDist{\HonA,\HonB}}{\descP{\cS\cup\low{\cpi}{\delta',\HonA}}} \cdot \prod_{i=1}^{k-1}\gamma_i + 2\cdot\sum_{i=1}^{k-1} \frac{\cdot\prod_{j=i+1}^{k-1}\gamma_j}{\gamma_i^c}} \\
&\quad +  \frac{2}{\gamma_k^c} \\
&= \ppr{\LDist{\HonA,\HonB}}{\descP{\cS\cup\low{\cpi}{\delta',\HonA}}} \cdot \prod_{i=1}^k\gamma_i + 2\cdot\sum_{i=1}^k \frac{\cdot\prod_{j=i+1}^{k}\gamma_j}{\gamma_i^c}.
\end{align*}
Noting that $\paren{\rcAP{\cpi^{(k-1)}}{1},\HonB} = \paren{\rcAP{\cpi}{k},\HonB}$ concludes the proof.
\end{proof}

\subsubsection{Proving \cref{lemma:goodWOLowValueSimple,claim:SmallStaySmallApproxSimple}}\label{sec:ProofsOfgoodWOLowValue}
We are finally ready to prove \cref{lemma:goodWOLowValueSimple,claim:SmallStaySmallApproxSimple}. These proofs rely on the next lemma, a slight generalization to \cref{lemma:goodWOLowValueSimple}.

\begin{lemma}\label{lemma:goodWOLowValue}
	For any $\delta\in(0,1/4]$, exists a constant $c=c(\delta)$ such that the following holds. Let $\cpi=(\HonA,\HonB)$ be a $\rnd$-round protocol, and assume  $\ppr{\LDist{\cpi}}{\descP{\low{\cpi}{1.5\delta',\HonA}}}\leq \alpha$ for some  $\delta\le\delta'\le\frac14$. Then,  for every  $\xi\in(0,1)$, $k\in\N$ and $\vect{\gamma} = (\gamma_1,\ldots,\gamma_k)\in (1,\infty)^k$, it holds that
	\begin{align}
	\SDP{\LDist{\rcAP{\cpi}{k},\HonB}}{\LDist{\realA{k,\xi,\delta'}{\cpi},\HonB}} &\leq k\cdot\frac{30^k\cdot m^k \cdot \prod_{i=1}^k\gamma_i}{\delta'^{2k}}\cdot\paren{\alpha + \xi} \label{eq:goodWOLowValue1}\\
	&\quad + \sum_{i=1}^{k}2^{k-i+2}\cdot\frac{30^{k-i}\cdot m^{k-i}\cdot \prod_{j=i+1}^{k}\gamma_j}{\delta'^{2(k-i)}\cdot \gamma_{i}^c}. \label{eq:goodWOLowValue2}
	\end{align}
\end{lemma}

Before proving this lemma, we use is to derive  \cref{lemma:goodWOLowValueSimple,claim:SmallStaySmallApproxSimple}.

\paragraph{Proving \cref{lemma:goodWOLowValueSimple}}

\begin{proof}[Proof of \cref{lemma:goodWOLowValueSimple}]
	Fix $\delta\in(0,1/4]$, $k\in\N$ and $\delta'\in [\delta,1/4]$ for which $\ppr{\LDist{\cpi}}{\descP{\low{\cpi}{1.5\delta',\HonA}}}\leq \alpha$. Furthermore, fix $\xi\in(0,1)$ and $\mu\in(0,1)$ and let $c=c(\delta)$ be the constant guaranteed by  \cref{lemma:goodWOLowValue}. We begin by defining a vector $\vect{\gamma} = (\gamma_1,\ldots,\gamma_k)\in (1,\infty)^k$ \wrt   the sum in \cref{eq:goodWOLowValue2} is less than $\mu$. For $i\in [k]$, let
	\begin{align}
	t_i \eqdef 2^{k-i+2}\cdot\frac{30^{k-i}\cdot m^{k-i}\cdot \prod_{j=i+1}^{k}\gamma_j}{\delta'^{2(k-i)}\cdot \gamma_{i}^c}.
	\end{align}
	The sum in \cref{eq:goodWOLowValue2} can be now written as $\sum_{i=1}^k t_i$. We now define $\vect{\gamma}$ so that $t_i\leq \mu/2^i$ for every $i$, implying that $\sum_{i=1}^k t_i \leq \mu$. Let $\gamma_k \eqdef \ceil{(4\cdot 2^k/\mu)^{1/c}}$. Note that
	\begin{align}
	t_k = \frac{4}{\gamma_k^c} \leq \frac{\mu}{2^k}.
	\end{align}

The value of $\gamma_{k-1},\ldots,\gamma_1$ is set inductively. For $i\in[k-1]$,  let

\begin{align*}
\gamma_i \eqdef \ceil{\paren{2^{k-i+2}\cdot\frac{30^{k-i}\cdot m^{k-i}\cdot \prod_{j=i+1}^{k}\gamma_j}{\delta'^{2(k-i)}} \cdot \frac{2^i}{\mu}}^{1/c}}.
\end{align*}

By construction, it holds that  $\prod_{j=i+1}^{k}\gamma_j\in\poly(m,1/\delta',1/\mu)$, $\vect{\gamma} = (\gamma_1,\ldots,\gamma_k)\in (1,\infty)^k$ and that $\sum_{i=1}^k t_i \leq \mu$. The proof is  thus concluded by applying \cref{lemma:goodWOLowValue}.
\end{proof}

\paragraph{Proving \cref{claim:SmallStaySmallApproxSimple}}

\begin{proof}[Proof of \cref{claim:SmallStaySmallApproxSimple}]
	Fix $\delta\in(0,1/4]$, $k\in\N$ and $\delta'\in [\delta,1/4]$ for which $\ppr{\LDist{\cpi}}{\descP{\low{\cpi}{1.5\delta',\HonA}}}\leq \alpha$.
	Let $c=c(\delta)$ be the constant guaranteed by  \cref{lemma:goodWOLowValue}.
	Set $\vect{\gamma} = (\gamma_1,\ldots,\gamma_k)\in (1,\infty)^k$ in the same way it was set in the proof of \cref{lemma:goodWOLowValueSimple} above. By assumption, it holds that $\ppr{\LDist{\cpi}}{\descP{\cF}}\leq \beta + \eps$.
	Applying \cref{prop:UnbalBounds} yields that
	\begin{align*}
	\ppr{\LDist{\rcAP{\cpi}{k},\HonB}}{\descP{\cF}} &\leq \paren{\alpha + \beta+ \eps}\cdot \prod_{i=1}^k\gamma_i  + 2\cdot\sum_{i=1}^k \frac{\cdot\prod_{j=i+1}^{k}\gamma_j}{\gamma_i^c},
	\end{align*}
	and \cref{lemma:goodWOLowValue} now yields that
	\begin{align}
	\ppr{\LDist{\pruAttack{k,\delta',\xi}{\cpi},\HonB}}{\descP{\cF}} &\leq \paren{\alpha + \beta+ \eps}\cdot \prod_{i=1}^k\gamma_i  + 2\cdot\sum_{i=1}^k \frac{\cdot\prod_{j=i+1}^{k}\gamma_j}{\gamma_i^c} \label{eq:SmallStaySmallApproxSimple1}\\
	&\quad + k\cdot\frac{30^k\cdot m^k \cdot \prod_{i=1}^k\gamma_i}{\delta'^{2k}}\cdot\paren{\alpha + \xi}\label{eq:SmallStaySmallApproxSimple2}\\
	&\quad + \sum_{i=1}^{k}2^{k-i+2}\cdot\frac{30^{k-i}\cdot m^{k-i}\cdot \prod_{j=i+1}^{k}\gamma_j}{\delta'^{2(k-i)}\cdot \gamma_{i}^c}.\label{eq:SmallStaySmallApproxSimple3}
	\end{align}
	By the proof of \cref{lemma:goodWOLowValueSimple} above, the terms in \cref{eq:SmallStaySmallApproxSimple2,eq:SmallStaySmallApproxSimple3} are at most $\phiPruE{k,\delta}(\alpha,\xi,m,\delta',\mu)$. Moreover, the proof of \cref{lemma:goodWOLowValueSimple} also yields that the term in \cref{eq:SmallStaySmallApproxSimple3} is a most $\mu$ and that $\prod_{i=1}^{k}\gamma_i\in\poly(m,1/\delta',1/\mu)$. The proof is concluded by noting that the second term in the right-hand side of \cref{eq:SmallStaySmallApproxSimple1} is bounded from above by that in \cref{eq:SmallStaySmallApproxSimple3} and thus is also at most $\mu$.
\end{proof}


\paragraph{Proving \cref{lemma:goodWOLowValue}.}
\cref{lemma:goodWOLowValue} is proven by induction on $k$. The next lemma, which combines the results from the previous sections, will be useful to argue the induction step.

\begin{lemma}\label{lemma:IndStep}
	For every $\delta\in(0,1/4]$, exists a constant $c=c(\delta)$ such that the following holds. Let $\cpi=(\HonA,\HonB)$ and $\cpi'=(\HonC,\HonD)$ be two $\rnd$-round protocols with the same control scheme, and assume
 \begin{enumerate}
 \item\label{item:SingleAttack1} $\Color_{\cpi}\equiv\Color_{\cpi'}$,
 \item\label{item:SingleAttack2} $\SDP{\LDist{\cpi}}{\LDist{\cpi'}}\leq\beta$, and
 \item\label{item:SingleAttack3} $\ppr{\LDist{\cpi'}}{\descP{\low{\cpi'}{1.5\delta',\HonC}}}\leq \alpha$ for some $\delta \leq \delta' \leq \frac14$.
\end{enumerate}
Then, for every $\xi\in(0,1)$ and $\gamma > 1$, it holds that
\begin{align*}
\SDP{\LDist{\rcAP{\cpi}{1,\xi,\delta'},\HonB}}{\LDist{\pruAttGen{\HonC}{1}{\cpi'},\HonD}} \leq \frac{30\cdot m \cdot \gamma}{\delta'^2}\cdot\paren{\alpha+\xi+\beta}+\frac{4}{\gamma^c}.
\end{align*}
\end{lemma}

\begin{proof}
	Fix $\delta\in(0,1/4]$ and let $c=c(\delta)$ be according to  \cref{lem:ProbVisitUnBal}. Fix $\delta'\in [\delta,1/4]$ for which $\ppr{\LDist{\cpi}}{\descP{\low{\cpi}{1.5\delta',\HonA}}}\leq \alpha$. Furthermore, fix $\xi\in(0,1)$ and $\gamma>1$.

The proof proceeds in two steps. First, apply \cref{lem:RCofStatCloseProt} (robustness lemma) to show that after the (ideal) biased-continuation attacker takes the role of $\HonA$ and $\HonC$ in $\cpi$ and $\cpi'$ respectively, the leaf distributions of these protocols remain close. Second, apply \cref{cor:approxBC} (ideal-to-approximated biased-continuation attacker) to show that replacing the attacker of the left-hand party in $\cpi$ with its approximated variant, the leaf distributions of these protocols remain close.

In order to apply \cref{lem:RCofStatCloseProt}, we first need to bound $\ppr{\LDist{\cpi}}{\descP{\low{\cpi}{\delta',\HonA}\cup\low{\cpi'}{\delta',\HonC}}}$.
Let $\cF = \frnt{\low{\cpi}{\delta',\HonA}\cup\low{\cpi'}{\delta',\HonC}}$, let $\cF_1 = \set{u\in \cF \colon \Val((\cpi')_u) \geq 1.5\delta'} $, and let $\cF_2 = \set{u\in \cF \colon \Val((\cpi')_u) < 1.5\delta'}$. Since $\cF\subseteq \cF_1\bigcupdot\cF_2$, it suffices to bound $\ppr{\LDist{\cpi}}{\descP{\cF_1}}$ and $\ppr{\LDist{\cpi}}{\descP{\cF_2}}$, which we do separately.
\begin{description}
\item[Bounding $\cF_1$:] Nodes in $\cF_1$ must have small value in $\cpi$ but large value in $\cpi'$. Since $\LDist{\cpi}$ and $\LDist{\cpi'}$ are close, the probability of reaching such nodes is small.

Formally, since every node in $\cF_1$ must belong to $\low{\cpi}{\delta',\HonA}$, it follows that $\ppr{\LDist{\cpi}}{\Leaves_1(\cpi) \condition \descP{\cF_1}} \leq \delta'$. Assumption (\ref{item:SingleAttack1}) of the lemma and the definition of $\cF_1$ yield, however, that $\ppr{\LDist{\cpi'}}{\Leaves(\cpi) \condition \descP{\cF_1}} \geq 1.5\delta'$. It follows from \cref{prop:CloseProDiffValues} that
\begin{align*}
\ppr{\LDist{\cpi}}{\descP{\cF_1}} \leq \beta \cdot \frac{1+1.5\delta'}{0.5\delta'}\leq \frac{4\beta}{\delta'}.
\end{align*}
The last inequality holds since, by assumption,  $\delta'\leq 1/4$.

\item[Bounding $\cF_2$:] The definition of $\cF_2$, the assumption that $\cpi$ and $\cpi'$ have the same control scheme, and assumption (\ref{item:SingleAttack3}), yield that $\ppr{\LDist{\cpi'}}{\descP{\cF_2}}\leq \alpha$. Hence, the assumption that $\SDP{\LDist{\cpi}}{\LDist{\cpi'}}\leq \beta$ (assumption (\ref{item:SingleAttack2}) of the lemma) yields that $\ppr{\LDist{\cpi}}{\descP{\cF_2}}\leq \alpha + \beta$.
\end{description}

Combining the two bounds, it follows that $\ppr{\LDist{\cpi}}{\descP{\low{\cpi}{\delta',\HonA}\cup\low{\cpi'}{\delta',\HonC}}} \leq 5\beta / \delta' + \alpha$. We can apply \cref{lem:RCofStatCloseProt} and derive
\begin{align}\label{eq:applyRobust}
\SDP{\LDist{\pruAttGen{\HonA}{1}{\cpi},\HonB}}{\LDist{\pruAttGen{\HonC}{1}{\cpi'},\HonD}} &\leq \frac{3\cdot m\cdot\gamma}{\delta'}\cdot\paren{\beta+\frac{5\beta}{\delta'} + \alpha}+\frac{2}{\gamma^c}.
\end{align}
The next step is to apply \cref{cor:approxBC}. To do so we need to bound $\ppr{\LDist{\cpi}}{\descP{\low{\cpi}{\delta',\HonA}}}$, but since it is clear that $\ppr{\LDist{\cpi}}{\descP{\low{\cpi}{\delta',\HonA}}} \leq \ppr{\LDist{\cpi}}{\descP{\low{\cpi}{\delta',\HonA}\cup\low{\cpi'}{\delta',\HonC}}}$, it follows that $\ppr{\LDist{\cpi}}{\descP{\low{\cpi}{\delta',\HonA}}} \leq  5\beta / \delta' + \alpha$. Applying \cref{cor:approxBC}, we derive
 \begin{align}\label{eq:applyIdealToReal}
\SDP{\LDist{\pruAttGen{\HonA}{1}{\cpi},\HonB}}{\LDist{\rcAP{\cpi}{1,\xi,\delta'},\HonB}} &\leq  2\cdot m\cdot\gamma \cdot\paren{\xi+\frac{5\beta}{\delta'} + \alpha}+\frac{2}{\gamma^c}.
 \end{align}
Finally, applying the triangle inequality of statistical distance to \cref{eq:applyRobust,eq:applyIdealToReal} completes the proof of \cref{lemma:IndStep}.
\end{proof}

The proof of \cref{lemma:goodWOLowValue} now follows straightforward calculations.
\begin{proof}[Proof of \cref{lemma:goodWOLowValue}]
	Fix $\delta\in(0,1/4]$ and let $c=c(\delta)$ be according to \cref{lemma:IndStep}. Fix $\delta'\in [\delta,1/4]$ for which $\ppr{\LDist{\cpi}}{\descP{\low{\cpi}{1.5\delta',\HonA}}}\leq \alpha$. Furthermore, fix $\xi\in(0,1)$.

	The proof is by induction on $k$. For $k=0$, the proof follows immediately from definition.

Fix $k\in\N$ and let $(\gamma_1,\ldots,\gamma_k) \in (1,\infty)^k$. Assume the lemma holds for $k-1$; we prove it for $k$ by applying \cref{lemma:IndStep}. For $i\in(k)$, let $\cpi_1^{(i)} = \paren{\rcAP{\cpi}{i},\HonB}$ and let $\cpi_2^{(i)} =\paren{\realA{i,\xi,\delta'}{\cpi},\HonB}$. Using this notation, we can write $\cpi_1^{(k)} = \paren{\rcAP{\cpi_1^{(k-1)}}{1},\HonB}$ and $\cpi_2^{(k)} = \paren{\rcAP{\cpi_2^{(k-1)}}{1,\xi,\delta'},\HonB}$. Hence,
\begin{align}
\SDP{\LDist{\rcAP{\cpi}{k},\HonB}}{\LDist{\realA{k,\xi,\delta'}{\cpi},\HonB}} = \SDP{\LDist{\rcAP{\cpi_1^{(k-1)}}{1},\HonB}}{\LDist{\rcAP{\cpi_2^{(k-1)}}{1,\xi,\delta'},\HonB}}.
\end{align}
We would like to apply \cref{lemma:IndStep} \wrt $\cpi_1^{(k-1)}$ and $\cpi_2^{(k-1)}$. Indeed, these protocols share the same control scheme and common output function $\Color$, and the induction hypothesis gives us a bound for $\SDP{\LDist{\cpi_1^{(k-1)}}}{\LDist{\cpi_2^{(k-1)}}}$. It remains to bound $\ppr{\LDist{\cpi_1^{(k-1)}}}{\descP{\low{\cpi_1^{(k-1)}}{1.5\delta',\HonA}}}$.

As we argued before,\footnote{We used the same argument in the proof of \cref{prop:UnbalBounds}.} it is easy to verify that when the ideal biased continuation attacker takes the role of $\HonA$ in the protocol and tries to bias the outcome towards $1$, the value of every node cannot decrease. Namely, it holds that $\low{\cpi_1^{(i)}}{1.5\delta',\HonA} \subseteq \low{\cpi_1^{(i-1)}}{1.5\delta',\HonA}$ for every $i\in[k-1]$, and thus $\low{\cpi_1^{(k-1)}}{1.5\delta',\HonA} \subseteq \low{\cpi_1}{1.5\delta',\HonA}=\low{\cpi}{1.5\delta',\HonA}$.
It holds that
\begin{align}
\ppr{\LDist{\cpi_1^{(k-1)}}}{\descP{\low{\cpi_1^{(k-1)}}{1.5\delta',\HonA}}} &\leq \ppr{\LDist{\cpi_1^{(k-1)}}}{\descP{\low{\cpi}{1.5\delta',\HonA}}} \\
& \leq \ppr{\LDist{\cpi}}{\descP{\low{\cpi}{1.5\delta',\HonA}}} \cdot \prod_{i=1}^{k-1}\gamma_i + 2\cdot\sum_{i=1}^{k-1} \frac{\cdot\prod_{j=i+1}^{k-1}\gamma_j}{\gamma_i^c}\nonumber \\
&\leq \alpha\cdot \prod_{i=1}^{k-1}\gamma_i + 2\cdot\sum_{i=1}^{k-1} \frac{\cdot\prod_{j=i+1}^{k-1}\gamma_j}{\gamma_i^c}.\nonumber
\end{align}
The second inequality follows from applying \cref{prop:UnbalBounds} \wrt $1.5\delta'$ and the set $\low{\cpi}{1.5\delta',\HonA}$. By the induction hypothesis and \cref{lemma:IndStep} applied to $\cpi_1^{(k-1)}$ and $\cpi_2^{(k-1)}$ \wrt $\gamma_k$, it holds that
\begin{align*}
\lefteqn{\SDP{\LDist{\rcAP{\cpi}{k},\HonB}}{\LDist{\realA{k,\xi,\delta'}{\cpi},\HonB}}} \\
&\leq \frac{30\cdot m \cdot \gamma_k}{\delta'^2}\cdot \left( (k-1)\cdot\frac{30^{k-1}\cdot m^{k-1} \cdot \prod_{i=1}^{k-1}\gamma_i}{\delta'^{2(k-1)}}\cdot\paren{\xi + \alpha} \right. \\
& \left. \hspace{2.8cm} + \sum_{i=1}^{k-1}2^{k-i+1}\cdot\frac{30^{k-1-i}\cdot m^{k-1-i}\cdot \prod_{j=i+1}^{k-1}\gamma_j}{\delta'^{2(k-1-i)}\cdot \gamma_{i}^c} \right. \\
& \left. \hspace{2.8cm} + \alpha\cdot \prod_{i=1}^{k-1}\gamma_i + 2\cdot\sum_{i=1}^{k-1} \frac{\cdot\prod_{j=i+1}^{k-1}\gamma_j}{\gamma_i^c} + \xi \right) + \frac{4}{\gamma_k^c} \\
&= \frac{30\cdot m \cdot \gamma_k}{\delta'^2}\cdot \paren{ (k-1)\cdot\frac{30^{k-1}\cdot m^{k-1} \cdot \prod_{i=1}^{k-1}\gamma_i}{\delta'^{2(k-1)}}\cdot\paren{\xi + \alpha}
 +  \alpha\cdot \prod_{i=1}^{k-1}\gamma_i + \xi } \\
&\quad + \frac{30\cdot m \cdot \gamma_k}{\delta'^2}\cdot \left(\sum_{i=1}^{k-1}2^{k-i+1}\cdot\frac{30^{k-1-i}\cdot m^{k-1-i}\cdot \prod_{j=i+1}^{k-1}\gamma_j}{\delta'^{2(k-1-i)}\cdot \gamma_{i}^c} + 2\cdot\sum_{i=1}^{k-1} \frac{\cdot\prod_{j=i+1}^{k-1}\gamma_j}{\gamma_i^c} \right) + \frac{4}{\gamma_k^c}.
\end{align*}
The induction proof now follows by grouping together the summands in the parentheses. This concludes the proof of \cref{lemma:AttackPruned}.
\end{proof}

\subsection{Attacking Pruned Protocols}\label{sec:AttackPrunedProtocol}
In \cref{sec:EfficientRCAttacker} we showed that if in a protocol $\cpi=(\HonA,\HonB)$ the probability to visit $\HonA$-controlled low-value nodes is small, then the recursive approximated biased-continuation attacker (taking the role of $\HonA$) biases the outcome of the protocol towards one almost as well as its ideal variant does (a similar fact holds for the attacker taking the role of $\HonB$, trying to bias the outcome of the protocol towards zero, and the probability to visit $\HonB$-controlled high-value nodes is small). For some protocols, however, this probability might be arbitrarily large, so the analysis in \cref{sec:EfficientRCAttacker} does not suffice to argue that the recursive approximated biased-continuation attacker successfully biases \emph{any} protocol. In this section we define the pruned variant of a protocol so that the probability of hitting $\HonA$-controlled low-value nodes, as well as hitting $\HonB$-controlled high-value nodes is indeed small. Hence, \cref{lemma:goodWOLowValueSimple} yields that the recursive approximated biased-continuation attacker successfully biases the pruned variant of any protocol. In \cref{sec:PruningInTheHead}, we  exploit the above for attacking \emph{any} protocol by letting the attacker ``pretend'' it is  attacking a pruned variant, rather than the original protocol.

We  start with defining an ideal pruned variant of a protocol, in which there exist no $\HonA$-controlled low-value nodes and $\HonB$-controlled high-value nodes. This variant, however, might not be efficiently computed, even if OWFs do not exist. To cope with this efficiency issue, we consider an approximated variant of the pruned protocol, in which such nodes might exist, but the probability of hitting them is small. Finally, we apply the results from \cref{sec:EfficientRCAttacker} to argue that the recursive approximated biased-continuation attacker biases the outcome of the approximately pruned variant of any protocol.

\paragraph{Pruned protocols.}
In the pruned variant of protocol $\cpi=(\HonA,\HonB)$, the edge distribution remains intact, while the controlling scheme is changed, giving the control to $\HonB$ on  low-value nodes, and to $\HonA$ on high-value nodes.

\begin{definition}[the pruned variant of a protocol]\label{def:prunnedProt}
Let $\cpi=\HonAHonB$ be an $\rnd$-round protocol and let $\delta\in(0,\frac12)$. In the {\sf $\delta$-pruned variant of $\cpi$}, denoted by $\pru{\cpi}{\delta}=\paren{\pruparty{\HonA}{\delta}{\cpi},\pruparty{\HonB}{\delta}{\cpi}}$, the parties follow the protocol $\cpi$, where $\pruparty{\HonA}{\delta}{\cpi}$ and $\pruparty{\HonB}{\delta}{\cpi}$ take the roles of $\HonA$ and $\HonB$ respectively, with the following exception occurring the \emph{first time} the protocol's transcript $u$ is in $\low{\cpi}{\delta} \cup \high{\cpi}{\delta}$:

If $u\in \high{\cpi}{\delta}$, set $\HonC = \pruparty{\HonA}{\delta}{\cpi}$; otherwise set $\HonC = \pruparty{\HonB}{\delta}{\cpi}$. The party $\HonC$ takes control of the node $u$, samples a leaf $\ell \la \LDist{\cpi_u}$, and then, bit by bit, sends $\ell_{\size{u}+1,\ldots,\rnd}$ to the other party.
\end{definition}

Namely, the first time the value of the protocol is close to either $1$ or $0$, the party interested in this value (\ie $\pruparty{\HonA}{\delta}{\cpi}$ for $1$, and $\pruparty{\HonB}{\delta}{\cpi}$ for $0$) takes control and decides the outcome (without changing the value of the protocol). Hence, the protocol is effectively pruned at these nodes (each such node is effectively a parent of two leaves).

For every protocol $\cpi$, its pruned variant $\pru{\cpi}{\delta}$ is a well-defined protocol, so the analysis of \cref{sec:IdealAttacker} can be applied.\footnote{Note that in the pruned protocol, the parties' turns might not alternate (\ie the same party might send several consecutive bits), even if they do alternate in the original protocol. Rather, the protocol's control scheme (determining which party is active at a given point) is a function of the protocol's transcript and the original protocol's control scheme. Such schemes are consistent with the ones considered in the previous sections.}  As  mentioned above, the pruned variant of a protocol might \emph{not} be efficiently computed, even if OWFs do not exist, so  we move to consider  an  approximated variant of the pruned protocol.

\paragraph{Approximately pruned protocols.}
To define the approximated pruned protocols, we begin by defining two algorithms, both of which can be efficiently implemented assuming OWFs do not exist for an appropriate choice of parameters. The first algorithm samples an honest (\ie unbiased) continuation of the protocol.

\newcommand{\HC}{\MathAlg{HC}}
\begin{definition}[approximated honest continuation]\label{def:AppxHonC}
Let $\cpi$ be an $m$-round  protocol, and let  $\HonCont_\cpi$ be the algorithm that on  node $u \in \Vertices(\cpi)$ returns $\ell \la\LDist{\cpi_u}$. Algorithm $\HC$ is a {\sf $\xi$-Honest-Continuator} for $\cpi$, if $\ppr{\ell \la \LDist{\cpi}}{\exists i\in (\rnd-1) \colon \SDP{\HC(\ell_{1,\ldots,i})}{\HonCont_\cpi(\ell_{1,\ldots,i})} > \xi} \leq \xi$. Let $\HonCont_\cpi^{\xi}$ be an arbitrary (but fixed) $\xi$-honest-continuator for $\cpi$.
\end{definition}
The second algorithm estimates the value of a given transcript (\ie a node) of the protocol.

\begin{definition}[estimator]\label{def:AppxEst}
Let $\cpi$ be an $m$-round  protocol. A deterministic algorithm $\Est$ is a {\sf $\xi$-Estimator} for $\cpi$, if  $\ppr{\ell \la \LDist{\cpi}}{\exists i\in (\rnd-1) \colon\size{\Est(\ell_{1,\ldots,i}) - \Val(\cpi_{\ell_{1,\ldots,i}})} > \xi}\leq\xi$. Let $\Est_\cpi^{\xi}$ be an arbitrary (but fixed) $\xi$-estimator for $\cpi$.
\end{definition}

Using the above estimator, we define the approximated version of the low and high value nodes.

\begin{definition}[approximated low-value and high-value nodes]\label{def:realSmallLarge}
For protocol $\cpi$, $\delta\in(0,\frac12)$ and a deterministic real-value algorithm $\Est$, let

\begin{itemize}
\item $\low{\cpi}{\delta,\Est}=\set{u\in\Vertices(\cpi)\setminus\Leaves(\cpi) \colon \Est(u)\leq\delta}$;

\item $\high{\cpi}{\delta,\Est}=\set{u\in\Vertices(\cpi)\setminus\Leaves(\cpi) \colon \Est(u)\geq 1-\delta}$.
\end{itemize}
For $\xi\in[0,1]$, let $\low{\cpi}{\delta,\xi}=\low{\cpi}{\delta,\Est_\cpi^\xi}$.
\end{definition}

We can now define the approximately pruned protocol, which is the oracle variant of the ideal pruned protocol.
\begin{definition}[the approximately pruned variant of a protocol]\label{def:ApxPrunnedProt}
Let $\cpi=\HonAHonB$ be an $\rnd$-round protocol, let $\delta\in(0,\frac12)$, let $\HC$ be an algorithm, and let $\Est$ be a deterministic real value algorithm. The {\sf $(\delta,\Est,\HC)$-approximately pruned variant of $\cpi$}, denoted $\pru{\cpi}{\delta,\Est,\HC}=\paren{\pruparty{\HonA}{\delta,\Est,\HC}{\cpi},\pruparty{\HonB}{\delta,\Est,\HC}{\cpi}}$, is defined as follows.
\begin{itemize}
\item[Control Scheme:] the parties follow the control scheme of the protocol $\cpi$, where $\pruparty{\HonA}{\delta,\Est,\HC}{\cpi}$ and $\pruparty{\HonB}{\delta,\Est,\HC}{\cpi}$ take the roles of $\HonA$ and $\HonB$ respectively, with the following exception occurring  the \emph{first time} the protocol's transcript $u$ is in $\low{\cpi}{\delta,\Est} \cup \high{\cpi}{\delta,\Est}$:
if $u\in \high{\cpi}{\delta,\Est}$ set $\HonC = \pruparty{\HonA}{\delta,\Est,\HC}{\cpi}$; otherwise set $\HonC = \pruparty{\HonB}{\delta,\Est,\HC}{\cpi}$. The party $\HonC$ takes control of all nodes in $\descP{u}$ (\ie nodes for which $u$ is an ancestor).

\item[Execution:] for a protocol's transcript $u$ and a party $\HonC$ who controls $u$, $\HonC$ sets $\ell=\HC(u)$ and sends $\ell_{\size{u}+1}$ to the other party.\footnote{This happens to every transcript, even those that are not children of $\low{\cpi}{\delta,\Est} \cup \high{\cpi}{\delta,\Est}$.}
\end{itemize}

For $\delta\in(0,\frac12)$ and $\xi,\xi'\in[0,1]$, let $\pru{\cpi}{\delta,\xi,\xi'}=\pru{\cpi}{\delta,\Est_\cpi^\xi,\HonCont_\cpi^{\xi'}}$ and $\pru{\cpi}{\delta,\xi}=\pru{\cpi}{\delta,\xi,\xi}$, and the same notation is used for the parties of the pruned protocol.
\end{definition}

Namely, in $\pru{\cpi}{\delta,\xi}$, the parties follow the control scheme of $\cpi$ until reaching a node in $\low{\cpi}{\delta,\xi} \cup \high{\cpi}{\delta,\xi}$ for the first time. Upon reaching such a node, the control moves to (and stays with) $\HonA$ if $u\in \high{\cpi}{\delta,\xi}$, or $\HonB$ if $u\in \low{\cpi}{\delta,\xi}$. The fact that the messages sent by the parties are determined by the answers of $\HonCont_\cpi^{\xi}$, instead of by their random coins, makes them \emph{stateless} throughout the execution of the protocol. This fact will be crucial when implementing our final attacker.

\paragraph{Attacking approximately pruned protocols.}
We would like to argue about the success probability of the recursive approximated biased-continuation attacker when attacking approximately pruned protocols. To do so, we must first show that the probability of reaching $\HonA$-controlled low-value nodes in such protocols is low. By definition, it is impossible to reach such nodes in the \emph{ideal} pruned protocol. Thus, if the approximately pruned variant is indeed an approximation of the pruned variant of the protocol, we expect that probability of reaching $\HonA$-controlled low-value nodes in this protocol will be low. Unfortunately, this does not necessarily hold. This is because the value of each node in both protocols might not be the same, and because the control scheme of these protocols might be different. It turns out that the bound for the above probability depends on the probability of the original protocol visiting nodes whose value is close to the pruning threshold, \ie $\delta$ and $1-\delta$.

\begin{definition}\label{def:neighbrhood}
For protocol $\cpi$, $\xi\in(0,1)$ and $\delta\in(0,\frac12)$, let
\begin{align*}
\Neigh_{\cpi}^{\delta,\xi}=\set{u\in\Vertices(\cpi)\setminus \Leaves(\cpi)\colon \Val(\cpi_u)\in(\delta-\xi,\delta+\xi]\lor \Val(\cpi_u)\in[1-\delta-\xi,1-\delta+\xi)},
\end{align*}
and let $\neigh_{\cpi}(\delta,\xi) = \ppr{\LDist{\cpi}}{\descP{\Neigh_{\cpi}^{\delta,\xi}}}$.
\end{definition}
Namely, $\Neigh_{\cpi}^{\delta,\xi}$ are those nodes that are $\xi$-close  to the ``border'' between $\low{\cpi}{\delta} \cup \high{\cpi}{\delta}$  and the rest of the nodes. The intervals in the above definition are taken to be open in one side and close on the other for technical reason, and this fact is  insignificant for the understanding of the definition.

We can now state the main result of this section --- the recursive approximated biased-continuation attacker biases this approximated pruned protocol with similar success to that of the  recursive (ideal) biased-continuation attacker. Specifically, we have the following lemma, which is an application of \cref{lemma:goodWOLowValueSimple} to the approximately pruned protocol.

\begin{lemma}\label{lemma:AttackPruned}
Let $0<\delta\le  \delta' \le\frac14$, let $\xi\in(0,1)$ and let $\cpit=\paren{\HonAt,\HonBt} = \pru{\cpi}{2\delta',\xi}$ be the $(2\delta',\xi)$-approximately pruned variant of a $\rnd$-round protocol $\cpi$. Then,
\begin{align*}
\SDP{\LDist{\rcAP{\cpit}{k},\HonBt}}{\LDist{\realA{k,\xi,\delta'}{\cpit},\HonBt}}
&\leq \phiPruE{k,\delta}\paren{\neigh_{\cpi}(2\delta',\xi)+12\cdot m\cdot \xi/\delta',\xi,m,\delta',\mu},
\end{align*}
for every $k\in\N$ and $\mu\in(0,1)$.\footnote{See \cref{lemma:goodWOLowValueSimple} for the definition of $\phiPruE{k,\delta}$.}
\end{lemma}

The next lemma will also be useful ahead. It shows that if a set of nodes is reached with low probability in the original protocol, then the probability to reach the same set does not increase by mush when the recursive approximated biased-continuation attacker attacks that approximately pruned variance of the protocol. This is an immediate application of  \cref{claim:SmallStaySmallApproxSimple} to the approximately pruned protocol.

\begin{lemma}\label{claim:SmallSetAprrx}
Let $0< \delta\le\delta'\le\frac14$, let $\xi\in(0,1)$ and let $\cpit=\paren{\HonAt,\HonBt} = \pru{\cpi}{2\delta',\xi}$ be the $(2\delta',\xi)$-approximately pruned variant of an $\rnd$-round protocol $\cpi$. Let $\cF$ be a frontier with $\ppr{\LDist{\cpi}}{\descP{\cF}}\leq\beta$. Then
\begin{align*}
\ppr{\LDist{\pruAttack{k,\delta',\xi}{\cpit},\HonBt}}{\descP{\cF}} &\leq \phiBalE{k,\delta}\paren{\neigh_{\cpi}(2\delta',\xi)+12\cdot m\cdot \xi/\delta',\beta,2\cdot m \cdot \xi, m,\delta',\mu} \\
&\quad +\phiPruE{k,\delta}\paren{\neigh_{\cpi}(2\delta',\xi)+12\cdot m\cdot \xi/\delta',\xi,m,\delta',\mu},
\end{align*}
for every $k\in\N$ and $\mu\in(0,1)$.\footnote{See \cref{claim:SmallStaySmallApproxSimple} for the definition of $\phiBalE{k,\delta}$.}
\end{lemma}

Finally, in order for the above bounds to be useful, we need to show that $\neigh_{\cpi}(\delta,\xi)$ --- the probability in the original protocol of reaching nodes whose value is $\xi$-close to $\delta$ --- is small. Unfortunately, given a protocol and a pruning threshold, this probability might be large. We argue, however, that if we allow a small deviation from the pruning threshold, this probability is small.

\begin{lemma}\label{prop:BadIsSmall}
Let $\cpi$ be an $\rnd$-round protocol, let $\delta\in(0,\frac12]$, and let $\xi\in(0,1)$. If $\xi\leq\frac{\delta^2}{16m^2}$, then there exists $j\in\J\eqdef\set{0,1,\ldots,\left\lceil m/\sqrt{\xi}\right\rceil}$ such that $\neigh_{\cpi}(\delta',\xi)\leq m\cdot\sqrt{\xi}$ for $\delta'=\delta/2+j\cdot 2\xi\in[\frac{\delta}{2},\delta]$.
\end{lemma}

The rest of this section is dedicated to proving the above Lemmas. In \cref{sec:ProvingPruned1} we show useful properties of approximately pruned protocols and use them to prove \cref{lemma:AttackPruned,claim:SmallSetAprrx}. In \cref{sec:ProvingPruned2} we prove \cref{prop:BadIsSmall}.

\subsubsection{Proving \cref{lemma:AttackPruned,claim:SmallSetAprrx}}\label{sec:ProvingPruned1}

\paragraph{Properties of approximately pruned protocols.}
In order to prove \cref{lemma:AttackPruned,claim:SmallSetAprrx} we need to bound the probability of hitting $\HonA$-controlled low-value nodes with that of reaching nodes whose value is close to the pruning threshold in the original protocol (\ie $\neigh_{\cpi}(\delta,\xi)$). The first step is to show that the approximately pruned protocol is close (in leaf-distribution sense) to the original protocol.

\begin{lemma}\label{lemma:ApprxPruIdealPru}
Let $\cpi=(\HonA,\HonB)$ be an $\rnd$-round protocol. Then
\begin{align*}
\SDP{\LDist{\cpi}}{\LDist{\pru{\cpi}{\delta,\xi}}} \leq 2\cdot m \cdot \xi
\end{align*}
for every $\delta\in(0,1/2]$ and $\xi\in(0,1)$.
\end{lemma}

The proof of \cref{lemma:ApprxPruIdealPru} is a simple implication of the approximation guarantee of the honest-continuator. Note that the leaf distributions of $\cpi$ and $\pru{\cpi}{\delta}$ are identical, so the above lemma also shows that the leaf distributions of the ideal and approximated pruned protocols are close (\ie that the latter is indeed an approximation to the former). Also note that the above bound does not depend on $\delta$. 

\newcommand{\FailC}{\mathcal{F}\mathsf{ailCont}}
\begin{proof}
The proof is an application of \cref{lemma:SDmQueriesAlg}.
By definition, every message in $\pru{\cpi}{\delta,\xi,0}$ is set by calling a perfect honest-continuator for $\cpi$. Thus, $\LDist{\cpi} \equiv \LDist{\pru{\cpi}{\delta,\xi,0}}$, and it suffices to bound $\SDP{\LDist{\pru{\cpi}{\delta,\xi,0}}}{\LDist{\pru{\cpi}{\delta,\xi}}=\LDist{\pru{\cpi}{\delta,\xi,\xi}}}$, which we do by applying \cref{lemma:SDmQueriesAlg}.

For a function  $\phi$, let $\sfH^\phi$ be an algorithm that outputs the transcript of a random execution of $\pru{\cpi}{\delta,\Est_\cpi^\xi,\phi}$. Let $f$ and $g$ be the (random) functions $\HonCont_\cpi$ and $\HonCont_\cpi^{\xi}$ respectively, and let $f(\perp)=g(\perp)=\perp$. By construction, it holds that
\begin{align*}
\SDP{\LDist{\pru{\cpi}{\delta,\xi,0}}}{\LDist{\pru{\cpi}{\delta,\xi,\xi}}}=\SDP{\sfH^f}{\sfH^g}.
\end{align*}
For $i\in[m]$, let $P_i$ be $i$'th node in a random execution of $\cpi$ (such a node consists of $i-1$ bits), and let $\FailC_{\cpi}^{\xi,i}=\set{u\in\Vertices(\cpi) \colon \size{u}=i-1 \land \SDP{\HonCont_\cpi(u)}{\HonCont_\cpi^{\xi}(u)} > \xi}$. By definition,
\begin{align*}
\ppr{u\la P_i}{u\in\FailC_{\cpi}^{\xi,i}} &= \ppr{\ell\la \LDist{\cpi}}{\SDP{\HonCont(\ell_{1,\ldots,i-1})}{\HonCont_{\cpi}^{\xi}(\ell_{1,\ldots,i-1})} > \xi} \\
&\leq \ppr{\ell\la \LDist{\cpi}}{\exists i\in[m] \colon \SDP{\HonCont(\ell_{1,\ldots,i-1})}{\HonCont_{\cpi}^{\xi}(\ell_{1,\ldots,i-1})} > \xi}\nonumber\\
&\leq \xi,\nonumber
\end{align*}
and thus
\begin{align*}
\lefteqn{\eex{u \la P_i}{\SDP{f(u)}{g(u)}}} \\
&= \eex{u \la P_i}{\SDP{\HonCont_\cpi(u)}{\HonCont_{\cpi}^{\xi}(u)}} \nonumber\\
&= \ppr{u\la P_i}{u\in\FailC_{\cpi}^{\xi,i}}\cdot \eex{u \la P_i}{\SDP{\HonCont_\cpi(u)}{\HonCont_{\cpi}^{\xi}(u)} \condition u\in\FailC_{\cpi}^{\xi,i}}\nonumber \\
&\quad + \ppr{u\la P_i}{u\notin\FailC_{\cpi}^{\xi,i}}\cdot \eex{u \la P_i}{\SDP{\HonCont_\cpi(u)}{\HonCont_{\cpi}^{\xi}(u)} \condition u\notin\FailC_{\cpi}^{\xi,i}}\nonumber\\
&\leq \xi + \xi = 2\xi,\nonumber
\end{align*}
where the first equality follows since $P_i(\perp)=0$.

Let $Q_i$ denote the $i$'th query to $f$ in a random execution of $\sfH^{f}$ (note that by construction, such a query always exists) and let $Q=(Q_1,\ldots,Q_m)$. By definition, $Q_i\equiv P_i$, and thus
\begin{align*}
\ppr{(q_1,\ldots,q_m)\la Q}{\exists i \in [m] \colon q_i \neq \perp \land \; Q_i(q_i)>P_i(q_i)} = 0.
\end{align*}
The proof now follows by \cref{lemma:SDmQueriesAlg}, letting $k= m$, $a= 2\xi$, $\lambda= 1$ and $b= 0$.
\end{proof}

We can now bound the probability of hitting $\HonA$-controlled low-value nodes with that of reaching nodes whose value is close to the pruning threshold in the original protocol.

\begin{lemma}\label{lemma:LowApproxPruned}
Let $\delta\in(0,1/2)$, let $\eps\in (0,\delta)$, let $\xi\in(0,1)$ and let $\cpit=\paren{\HonAt,\HonBt} = \pru{\cpi}{\delta,\xi}$ be the $(\delta,\xi)$-approximately pruned variant of an $\rnd$-round protocol $\cpi$. Then
\begin{align*}
\ppr{\LDist{\cpit}}{\descP{\low{\cpit}{\delta-\eps,\HonAt}}} \leq \neigh_{\cpi}(\delta,\xi) + \frac{6\cdot m\cdot\xi}{\eps}.
\end{align*}
\end{lemma}

\newcommand{\FailE}{\mathcal{F}\mathsf{ailEst}}
\begin{proof}[Proof of \cref{lemma:LowApproxPruned}]
The proof is an application of \cref{lemma:ApprxPruIdealPru,prop:CloseProDiffValues}.
Let $\FailE_{\cpi}^{\xi} = \set{u\in\Vertices(\cpi) \colon \size{\Val(\cpi_u)-\Est_\cpi^\xi(u)}>\xi}$ and let $\cF=\frnt{\low{\cpit}{\delta-\eps,\HonAt}}\setminus \paren{\Neigh_{\cpi}^{\delta,\xi}\cup\FailE_\cpi^{\xi}}$.
It follows that
\begin{align}\label{eq:LowApproxPruned2}
\ppr{\LDist{\cpit}}{\descP{\low{\cpit}{\delta-\eps,\HonAt}}} \leq \ppr{\LDist{\cpit}}{\descP{\Neigh_{\cpi}^{\delta,\xi}\cup\FailE_\cpi^{\xi}}} + \ppr{\LDist{\cpit}}{\descP{\cF}}.
\end{align}
By \cref{lemma:ApprxPruIdealPru}, it holds that
\begin{align}\label{eq:LowApproxPruned3}
\ppr{\LDist{\cpit}}{\descP{\Neigh_{\cpi}^{\delta,\xi}\cup\FailE_\cpi^{\xi}}} \leq \neigh_{\cpi}(\delta,\xi) + 3\cdot m\cdot \xi.
\end{align}
Let $u\in \cF$. Since $u$ is under $\HonAt$'s control,  it holds that $\Est_\cpi^\xi(u)>\delta$. Since $u\notin \FailE_\cpi^{\xi}$,  it holds that $\Val(\cpi_u)>\delta-\xi$, and since $u\notin \Neigh_{\cpi}^{\delta,\xi}$, we have $\Val(\cpi_u)\geq\delta+\xi$. By definition, $\Val(\cpit_u)\leq \delta-\eps$. Thus,  $\ppr{\LDist{\cpit}}{\Leaves_1(\cpi)\condition \descP{\cF}}\leq \delta - \eps$ and $\ppr{\LDist{\cpi}}{\Leaves_1(\cpi)\condition \descP{\cF}}\geq \delta + \xi$. Finally, by \cref{lemma:ApprxPruIdealPru} it holds that $\SDP{\cpit}{\cpi}\leq 2\cdot m \cdot \xi$, and thus by \cref{prop:CloseProDiffValues} we have
\begin{align}\label{eq:LowApproxPruned4}
\ppr{\LDist{\cpit}}{\descP{\cF}} \leq 2\cdot m\cdot \xi\cdot \frac{1+ \delta - \eps}{\xi+\eps} \leq \frac{3\cdot m\cdot \xi}{\eps}.
\end{align}
Plugging \cref{eq:LowApproxPruned3,eq:LowApproxPruned4} into \cref{eq:LowApproxPruned2} completes the proof of the lemma.
\end{proof}

\paragraph{Proving \cref{lemma:AttackPruned}.}
\begin{proof}[Proof of \cref{lemma:AttackPruned}.]
Applying \cref{lemma:LowApproxPruned} to $\cpit$ and $\eps=0.5\delta'$ yields that
\begin{align}\label{eq:lowInPruned}
\ppr{\LDist{\cpit}}{\descP{\low{\cpit}{1.5\delta',\HonAt}}} \leq \neigh_{\cpi}(2\delta',\xi) + \frac{12\cdot m\cdot\xi}{\delta'}.
\end{align}
The proof now immediately follows from \cref{lemma:goodWOLowValueSimple}.
\end{proof}

\paragraph{Proving \cref{claim:SmallSetAprrx}.}
\begin{proof}[Proof of \cref{claim:SmallSetAprrx}.]
Immediately follows from plugging \cref{lemma:ApprxPruIdealPru,eq:lowInPruned} into \cref{claim:SmallStaySmallApproxSimple}.
\end{proof}

\subsubsection{Proving \cref{prop:BadIsSmall}}\label{sec:ProvingPruned2}
\begin{proof}[Proof of \cref{prop:BadIsSmall}.]
For $j\in\J$, let $\delta'(j) = \delta/2+j\cdot 2\xi$. From the definition of $\J$, it is clear that $\delta'(j)\in[\frac{\delta}{2},\delta]$ for every $j\in\J$. Hence, it is left to argue that $\exists j\in\J$ such that $\neigh_{\cpi}(\delta'(j),\xi)\leq m\cdot\sqrt{\xi}$.

For $i\in[m]$, let $\Neigh_{\cpi}^{\delta,\xi,i}=\set{u\in\Vertices(\cpi) \colon u\in\Neigh_{\cpi}^{\delta,\xi}\land \size{u}=i-1}$. It holds that
\begin{align}\label{eq:boundSum}
\ppr{\LDist{\cpi}}{\descP{\Neigh_{\cpi}^{\delta,\xi}}} &\leq \ppr{\LDist{\cpi}}{\descP{\cup_{i\in[m]}\Neigh_{\cpi}^{\delta,\xi,i}}}\\
&\leq \sum_{i=1}^m \ppr{\LDist{\cpi}}{\descP{\Neigh_{\cpi}^{\delta,\xi,i}}}. \nonumber
\end{align}
For every $i\in[m]$, let $\cN(i)=\set{j\in\J\colon \ppr{\LDist{\cpi}}{\descP{\Neigh_{\cpi}^{\delta'(j),\xi,i}}} > \sqrt{\xi}}$ and let $\cN=\cup_{i\in[m]} \cN(i)$. We use the following claim (proven below).
\begin{claim}\label{claim:BadIsSmall}
It holds that $\size{\cN(i)}<1/\sqrt{\xi}$ for every $i\in[m]$.
\end{claim}
\cref{claim:BadIsSmall} yields that $\size{\cN} \leq \sum_{i=1}^m \size{\cN(i)} < \frac{m}{\sqrt{\xi}} < \size{\J}$. Thus, $\exists j\in\J$ such that $j\notin\cN$. Set $\delta'=\delta'(j)$. It holds that $\ppr{\LDist{\cpi}}{\descP{\Neigh_{\cpi}^{\delta',\xi,i}}} \leq \sqrt{\xi}$ for every $i\in[m]$. Plugging it into \cref{eq:boundSum} yields that $\neigh_{\cpi}(\delta',\xi)=\ppr{\LDist{\cpi}}{\descP{\Neigh_{\cpi}^{\delta',\xi}}} \leq m\cdot \sqrt{\xi}$, completing the proof of \cref{prop:BadIsSmall}.
\end{proof}

\begin{proof}[Proof of \cref{claim:BadIsSmall}]
Assume towards a contradiction that there exists $i\in[m]$ such that $\size{\cN(i)}\geq 1/\sqrt{\xi}$. Let $P_i$ be the distribution over $\zo^i$, described by outputting $\ell_{i}$, for $\ell\la\LDist{\cpi}$. We get that $\ppr{\LDist{\cpi}}{\descP{\Neigh_{\cpi}^{\delta'(j),\xi,i}}} = P_i\paren{\Neigh_{\cpi}^{\delta'(j),\xi,i}}$. Since,  $\Neigh_{\cpi}^{\delta'(j),\xi,i}\cap\Neigh_{\cpi}^{\delta'(j'),\xi,i}=\emptyset$ for every $j\neq j'\in\J$, it holds that
\begin{align*}
1 &\geq \sum_{j\in\J} P_i\paren{\Neigh_{\cpi}^{\delta'(j),\xi,i}} \\
&\geq \sum_{j\in\cN(i)} P_i\paren{\Neigh_{\cpi}^{\delta'(j),\xi,i}} \\
&> \size{\cN(i)}\cdot \sqrt{\xi} \geq 1,
\end{align*}
and a contradiction is derived, where the last inequality follows the assumption that $\size{\cN(i)}\geq 1/\sqrt{\xi}$.
\end{proof}


\subsection{The Pruning-in-the-Head Attacker}\label{sec:PruningInTheHead}
In \cref{sec:AttackPrunedProtocol}, the recursive approximated biased-continuation attacker was shown to successfully biases the approximately pruned variant of any protocol. We now use this result to design an attacker that biases \emph{any} protocol. The new attacker applies the approximated biased-continuation attacker as if the attacked protocol is (approximately) pruned, until it reaches a low or high value node, and then it switches its behavior to act honestly (\ie as the protocol prescribes). Named after its strategy, we name it the \emph{pruning-in-the-head attacker}.

To make the discussion simpler, ee start with describing the  ideal (inefficient) variant of the pruning-in-the-head  attacker. Consider the ideal pruned variant of a protocol pruned at some threshold  $\delta$,   denoted by $\pru{\cpi}{\delta} = \paren{\pru{\HonA}{\delta},\pru{\HonB}{\delta}}$ (see \cref{def:prunnedProt}). Being a coin-flipping protocol, the results of \cref{sec:IdealAttacker} apply to $\pru{\cpi}{\delta}$. Specifically, \cref{thm:MainIdeal} yields that $\RC{\paren{\pru{\HonA}{\delta}}}{k}$ successfully biases $\pru{\cpi}{\delta}$ (as usual, for concreteness, we focus on the attacker for $\HonA$). For parameters $\delta$ and $k$, the \emph{ideal pruning-in-the-head attacker}, denoted $\RC{\HonA}{k,\delta}$, acts as follows: until reaching a pruned node according to $\delta$ (\ie a node whose value is lower than $\delta$ or higher than $1-\delta$), it acts like $\RC{\paren{\pru{\HonA}{\delta}}}{k}$; when reaching a pruned node, and in the rest of the execution, it acts like the honest party $\HonA$. Namely, $\RC{\HonA}{k,\delta}$ acts as if it is actually attacking the pruned variant of the protocol, instead of the original protocol.

We argue that $\RC{\HonA}{k,\delta}$ biases the original protocol almost as well as $\RC{\paren{\pru{\HonA}{\delta}}}{k}$ biases the ideal pruned protocol. Consider the protocols $\paren{\RC{\paren{\pru{\HonA}{\delta}}}{k},\pru{\HonB}{\delta}}$ and $\paren{\RC{\HonA}{k,\delta},\HonB}$. On unpruned nodes, both protocols act the same. On low-value nodes, the protocols might have different control schemes, but their outputs share the same distribution. On high-value nodes, the value of $\paren{\RC{\paren{\pru{\HonA}{\delta}}}{k},\pru{\HonB}{\delta}}$ might be as high as $1$, since $\RC{\paren{\pru{\HonA}{\delta}}}{k}$ attacks such nodes. On the other hand, in $\paren{\RC{\HonA}{k,\delta},\HonB}$, when a high-value node is reached, $\RC{\HonA}{k,\delta}$ acts honestly. However,  since this is a high-value node, its value is at least $1-\delta$. All in all, the values of the two protocols differ by at most $\delta$. Hence, $\RC{\HonA}{k,\delta}$ successfully attack (the non-pruned) protocol $\cpi$. This might not seem like a great achievement. An inefficient,  and much simpler, attack on protocol $\cpi$  was already presented in \cref{sec:IdealAttacker}. The point is that unlike the attack of \cref{sec:IdealAttacker}, the above attacker can be made efficient.

In the rest of this section we extend the above discussion for approximated attackers attacking approximately pruned protocols. Specifically, we give an approximated variant of $\RC{\HonA}{k,\delta}$ --- the pruning-in-the-head attacker --- and prove that it is a successful attacker by showing that it biases \emph{any} protocol $\cpi$ almost as well as the recursive approximated biased-continuation attacker biases the $\delta$-approximately pruned variant of $\cpi$ (the latter, by \cref{sec:AttackPrunedProtocol}, is a successful attack). In \cref{sec:ProtocolInv}, we show how to implement this attacker using only an honest continuator for the original protocol, which is the main step towards implementing it  efficiently assuming the inexistent of one-way functions (done in \cref{sec:EfficeinrAttack}).

\paragraph{The pruning-in-the-head attacker.}
Let $\cpi=(\HonA,\HonB)$ be a protocol. Recall that $\HonCont_\cpi^\xi$ stands for the arbitrarily fixed $\xi$-honest continuator for $\cpi$ (see \cref{def:AppxHonC}) and that $\Est_\cpi^\xi$ stands for the arbitrarily fixed $\xi$-estimator for $\cpi$ (see \cref{def:AppxEst}). Furthermore, recall that $\low{\cpi}{\delta,\Est_\cpi^\xi}$ [\resp $\high{\cpi}{\delta,\Est_\cpi^\xi}$] stands for the set of nodes for which $\Est_\cpi^\xi$ is at most $\delta$ [\resp at least $1-\delta$] (see \cref{def:realSmallLarge}) and that $\pru{\cpi}{\delta,\xi}$ stands for the $(\delta,\xi)$-approximately pruned variant of $\cpi$ (see \cref{def:ApxPrunnedProt}). Finally, recall that for a set of nodes $\cS\subseteq \Vertices(\cpi)$, $\descP{\cS}$ stands for those nodes that at least one of their predecessors belong to $\cS$ (see \cref{def:BinaryTrees}).

Let $\finalAttack{i,\xi,\delta}{\cpi}\equiv \HonA$ and for integer $i>0$ define:
\begin{algorithm}[the pruning-in-the-head attacker  $\finalAttack{i,\xi,\delta}{\cpi}$]\label{alg:rFinalPrunedAdv}
\item Parameters: integer $i>0$, $\xi,\delta\in(0,1)$.
\item Input: transcript $u \in \zo^\ast$.
\item Notation: let $\cpit = \pru{\cpi}{2\delta,\xi}$.
\item Operation:
\begin{enumerate}

\item If $u \in \Leaves(\cpi)$, output $\Color_\cpi(u)$ and halt.

\item Set $\msg$ as follows.
\begin{itemize}
 \item If $u \in \descP{\low{\cpi}{2\delta,\Est_\cpi^\xi}\cup\high{\cpi}{2\delta,\Est_\cpi^\xi}}$, set $\msg= \HonCont_\cpi^\xi(u)$.
 \item Otherwise,  set $\msg=  \pruAttack{i,\xi,\delta}{\cpit}(u)$ (see \cref{def:itAppAtt}).
\end{itemize}

\item Send $\msg$ to $\HonB$.

\item If $u'= u\conc \msg \in \Leaves(\cpi)$, output $\Color_\cpi(u')$.

\end{enumerate}
\end{algorithm}

The next lemma lower-bounds  the success probability of the pruning-in-the-head attacker. It states that if a given protocol $\cpi$ does not have many nodes whose value is close to $2\delta$, then the pruning-in-the-head attacker biases $\cpi$ almost as well as the approximated attacker biases the approximated pruned protocol.

Recall that $\neigh_{\cpi}(\delta,\xi)$ stands for the probability that $\cpi$ generates transcripts whose values are $\xi$-close to $\delta$ or to $1-\delta$ (see \cref{def:neighbrhood}).

\begin{lemma}[main lemma for the pruning-in-the-head attacker.]\label{lemma:final}
Let $0< \delta \leq \delta'\le \frac14$, let $\xi\in(0,1)$ and let $\cpit=\paren{\HonAt,\HonBt} = \pru{\cpi}{2\delta',\xi}$ be the $(2\delta',\xi)$-approximately pruned variant of an $\rnd$-round protocol $\cpi=(\HonA,\HonB)$ (see \cref{def:ApxPrunnedProt}). Then
\begin{align*}
\Val\paren{\finalAttack{k,\xi,\delta'}{\cpi},\HonB} &\geq \Val\paren{\rcAP{\cpit}{k},\HonBt} - 2\delta' - (m+2)\cdot \sqrt{\xi} \\
	&\quad -2\cdot\phiBalE{k,\delta}\paren{\neigh_{\cpi}(2\delta',\xi)+12\cdot m\cdot \xi/\delta',2\sqrt{\xi},2\cdot m \cdot \xi, m,\delta',\mu} \\
	&\quad -3\cdot\phiPruE{k,\delta}\paren{\neigh_{\cpi}(2\delta',\xi)+12\cdot m\cdot \xi/\delta',\xi,m,\delta',\mu},
\end{align*}
for every $k\in \N$ and $\mu\in(0,1)$, and for  $\phiPruE{k,\delta},\phiBalE{k,\delta}\in \poly$ be according to \cref{lemma:goodWOLowValueSimple,claim:SmallStaySmallApproxSimple} respectively.
\end{lemma}

The rest of this section is dedicated to proving \cref{lemma:final}.

\subsubsection{Proving \cref{lemma:final}}\label{subsec:PruningInTheHead Analysis}
The  proof follow the  proof we sketched  above for the ideal  pruning-in-the-head attacker. When moving to the approximated case, however, we need to consider \emph{failing transcripts} --- transcripts on which the approximating oracles fail to give a good approximation. As long as the approximated pruning-in-the-head attacker did not generate a failing transcript, it will succeed in biasing the protocol almost as well as its ideal variant. Thus, the  heart of the proof is  showing that the approximated pruning-in-the-head attacker  generates a failing transcript with only low probability. By definition, the probability of the original protocol to generate such failing transcripts is low, so we can use \cref{claim:SmallSetAprrx} to argue that the recursive approximated biased-continuation attacker, when attacking the \emph{approximated pruned protocol}, also generates  failing transcripts with only low probability. We use this fact to argue that the approximated pruning-in-the-head attacker such transcripts with only low probability as well.

The proof  handles  separately the failing transcripts  into that transcripts  \emph{precede} pruned transcripts, \ie the execution of the protocol has not pruned before generated these transcripts, and the rest of the failing transcripts (\ie failing transcripts  \emph{preceded by} pruned transcripts). Specifically, we make the following observations:
\begin{enumerate}
\item Failing transcripts that precede pruned transcripts (high- or low-value transcripts).

The probability of the approximated pruning-in-the-head attacker to reach these transcripts is the same as the recursive approximated biased-continuation attacker, which we already know is low.

\item Failing transcript preceded by pruned transcripts. We consider  the following two sub-cases.
\begin{enumerate}
\item The probability of the original protocol to generate pruned transcripts is low.

In this case, it suffices to show that the approximated pruning-in-the-head attacker generate pruned transcripts with low probability as well. By \cref{claim:SmallSetAprrx}, the probability of the recursive approximated biased-continuation attacker to generate pruned transcripts is low, and until reaching such transcripts, the approximated pruning-in-the-head attacker acts as the recursive approximated biased-continuation attacker.

\item The probability of the original protocol to generate pruned transcripts is high.

In this case, since, by definition, the overall probability of generating failing transcripts is low, the probability of the original protocol to generate failing transcripts \emph{given that the protocol reached a pruned transcript} is low. Once it reaches  a pruned transcript the pruning-in-the-head attacker behaves just like the original protocol. Thus, the probability the pruning-in-the-head attacker generates failing transcripts, even conditioning that is generates pruned transcript, is low.
\end{enumerate}
\end{enumerate}
All in all, we get that the probability that the approximated pruning-in-the-head attacker  generates failing transcripts is low, and thus the intuition from the ideal case applies.

Moving to the formal proof,  fix  $k>0$ (the proof for $k=0$ is immediate) and $\mu\in(0,1)$. To ease notation ahead, let $\gamma_\cpi(\delta',\xi) = \neigh_{\cpi}(2\delta',\xi)+12\cdot m\cdot \xi/\delta'$. We define four hybrid protocols to establish the above arguments step by step. The proof of the lemma will follow by showing that these hybrid protocols' expected outcomes as are close to one another.

\newcommand{\SafeLarge}{\mathcal{S}\mathsf{afe}\mathcal{L}\mathsf{arge}}
\newcommand{\FailCont}{\mathcal{F}\mathsf{ail}\mathcal{C}\mathsf{ont}}
\newcommand{\FailEst}{\mathcal{F}\mathsf{ail}\mathcal{E}\mathsf{st}}

Let $\FailCont \eqdef \set{u\in\Vertices(\cpi)\colon \SDP{\HonCont_\cpi^\xi(u)}{\HonCont(u)} > \xi}$, \ie   transcripts on which the approximated honest continuator $\HonCont_\cpi^\xi$ acts significantly different from the ideal honest continuator $\HonCont$. Let $\FailEst \eqdef \set{u\in\Vertices(\cpi) \colon \Val(\cpi_u)<1-2\delta'-\xi \land \Est_\cpi^\xi(u)>1-2\delta'}$, \ie   low-value transcripts which the approximated estimator $\Est_\cpi^\xi$ mistakenly estimates their value to be high, and let $\Fail \eqdef \FailCont \cup \FailEst$. Finally, let $\SafeLarge \eqdef \high{\cpi}{2\delta',\xi} \setminus \descP{\Fail}$, \ie    high-value transcripts that are not descendants of failing transcripts.

We are now ready to define the hybrid protocols, all of which share the common output function of the original protocol $\cpi$ (\ie the function determines the common output of full transcripts of $\cpi$, see \cref{def:ProtocolValue}).
\begin{itemize}
	\item \textbf{Protocol $\cpi_1$:} This protocol is just protocol $\paren{\pruAttack{k,\delta',\xi}{\cpit},\HonBt}$, \ie the approximated recursive biased-continuation attacker attacks the approximated pruned protocol.

	\item \textbf{Protocol $\cpi_2$:} Both parties act as in $\cpi_1$ until (if at all) the first time the protocol's transcript is in $\SafeLarge$. In the rest of the protocol, the parties act like in $\cpi$ (which also means following $\cpi$'s control scheme).

	\item \textbf{Protocol $\cpi_3$:} Both parties act as in $\cpi_2$ until (if at all) the first time the protocol's transcript is in $\Fail$. In the rest of the protocol, the parties act like in $\paren{\finalAttack{k,\xi,\delta'}{\cpi},\HonB}$ (which also means following $\paren{\finalAttack{k,\xi,\delta'}{\cpi},\HonB}$'s control scheme, which is identical to $\cpi$'s).

	\item \textbf{Protocol $\cpi_4$:} This protocol is just protocol $\paren{\finalAttack{k,\xi,\delta'}{\cpi},\HonB}$, \ie the approximated pruning-in-the-head attacker attacks the original protocol $\cpi$. (This is the protocol whose value we are trying to analyze.)
\end{itemize}
The proof of the lemma immediately  follows the next sequence of claims.

\begin{claim}\label{claim:Pi1Pi2}
	It holds that $\Val(\cpi_2) \geq \Val(\cpi_1) - 2\delta' -\xi$.
\end{claim}
\begin{proof}
	Note that  protocols $\cpi_1$ and $\cpi_2$ are identical until the first time the protocol's transcript is in $\SafeLarge$. Hence, we can couple random executions of protocols  $\cpi_1$ and $\cpi_2$, so that they are the same until   the first time the protocol's transcript is in $\SafeLarge$. Hence, for proving that claim it suffices to show that $\Val((\cpi_1)_u) - \Val((\cpi_2)_u) \leq 2\delta' +\xi$, for every $u\in\frnt{\SafeLarge}$.

	 Fix  $u\in\frnt{\SafeLarge}$. Since $u\in \high{\cpi}{2\delta',\xi}$, it holds that $\Est_\cpi^\xi(u) \geq 1- 2\delta'$. Since $u\notin \FailEst$, it holds that $\Val(\cpi_u)\geq 1-2\delta' - \xi$. Once visiting $u$, the parties in $\cpi_2$ act like in $\cpi$. Thus, it holds that $\Val((\cpi_2)_u) = \Val(\cpi_u)$. Since it is always the case that $\Val((\cpi_1)_u)\leq 1$, we have $\Val((\cpi_1)_u) - \Val((\cpi_2)_u) \leq 2\delta' + \xi$.
\end{proof}

\begin{claim}\label{claim:Pi2Pi3}
	It holds that
	\begin{align*}
	\Val(\cpi_3) &\geq \Val(\cpi_2) - \sqrt{\xi} - 2\cdot\phiBalE{k,\delta}\paren{\gamma_\cpi(\delta',\xi),2\sqrt{\xi},2\cdot m \cdot \xi,m,\delta',\mu}  -2\cdot\phiPruE{k,\delta}\paren{\gamma_\cpi(\delta',\xi),\xi,m,\delta',\mu}. \nonumber
	\end{align*}
\end{claim}
\begin{proof}
	We prove the claim by proving the following, stronger statement.
		\begin{align*}
	\SDP{\LDist{\cpi_2}}{\LDist{\cpi_3}} &\leq \sqrt{\xi} + 2\cdot\phiBalE{k,\delta}\paren{\gamma_\cpi(\delta',\xi),2\sqrt{\xi},2\cdot m \cdot \xi,m,\delta',\mu}  +2\cdot\phiPruE{k,\delta}\paren{\gamma_\cpi(\delta',\xi),\xi,m,\delta',\mu}. \nonumber
	\end{align*}
	Note that protocols $\cpi_2$ and $\cpi_3$ are identical until the first time the protocol's transcript is in $\Fail$. Hence, we can couple random executions of protocols $\cpi_2$ and $\cpi_3$, so that the executions are the same until   the first time the protocol's transcript is in $\Fail$. Thus, it suffices to show that
	\begin{align}\label{eq:P2P3_1}
	\ppr{\LDist{\cpi_2}}{\descP{\Fail}} &\leq \sqrt{\xi} + 2\cdot\phiBalE{k,\delta}\paren{\gamma_\cpi(\delta',\xi),2\sqrt{\xi},2\cdot m \cdot \xi, m,\delta',\mu} \\
	&\quad +2\cdot\phiPruE{k,\delta}\paren{\gamma_\cpi(\delta',\xi),\xi,m,\delta',\mu}. \nonumber
	\end{align}
	Let $\cF_1 = \Fail\cap \descP{\SafeLarge}$ and $\cF_2 = \Fail \setminus \cF_1$. Since
	\begin{align}\label{eq:P2P3_2}
	\ppr{\LDist{\cpi_2}}{\descP{\Fail}} \leq \ppr{\LDist{\cpi_2}}{\descP{\cF_1}} + \ppr{\LDist{\cpi_2}}{\descP{\cF_2}},
	\end{align}
	it suffices to bound the two summands in the right-hand side of \cref{eq:P2P3_2}. We begin by bounding the second summand. Since $\cF_2\subseteq \Fail$, and by the definitions of $\HonCont_\cpi^\xi$ and $\Est_\cpi^\xi$, it holds that
	\begin{align}\label{eq:P2P3_3}
	\ppr{\LDist{\cpi}}{\descP{\F_2}} \leq \ppr{\LDist{\cpi}}{\descP{\Fail}} \leq 2\xi.
	\end{align}
	As we did in the proof of the previous claim, we couple random executions of protocols $\cpi_1$ and $\cpi_2$, so that the executions are the same until the first time the protocol's transcript is in $\SafeLarge$. Since transcripts in $\cF_2$ are not descendants of $\SafeLarge$, it holds that
	\begin{align}\label{eq:P2P3_4}
	\ppr{\LDist{\cpi_2}}{\descP{\cF_2}} &= \ppr{\LDist{\cpi_1}}{\descP{\cF_2}} \\
	&= \ppr{\LDist{\pruAttack{k,\xi,\delta'}{\cpit},\HonBt}}{\descP{\cF_2}} \nonumber\\
	&\leq \phiBalE{k,\delta}\paren{\gamma_\cpi(\delta',\xi),2\xi,2\cdot m \cdot \xi, m,\delta',\mu} +\phiPruE{k,\delta}\paren{\gamma_\cpi(\delta',\xi),\xi,m,\delta',\mu}, \nonumber
	\end{align}
	where the second equality follows from the definition of $\cpi_1$, and the the inequality follows from \cref{claim:SmallSetAprrx}.

	We now bound the first summand in the right-hand side of \cref{eq:P2P3_2}. Let
	\begin{align*}
	\cS_1 = \set{u\in\frnt{\SafeLarge} \colon \ppr{\LDist{\cpi_u}}{\descP{\cF_1}} \geq \sqrt{\xi}},
	\end{align*}
	and
	\begin{align*}
	\cS_2 = \set{u\in\frnt{\SafeLarge} \colon 0<\ppr{\LDist{\cpi_u}}{\descP{\cF_1}} < \sqrt{\xi}}.
	\end{align*}
	Namely, $\cS_1$ are those nodes (transcripts) in the frontier of $\SafeLarge$ from which there is high probability (larger than $\sqrt{\xi}$) that $\cpi$ reaches $\cF_1$. On the other hand, $\cS_2$ are those nodes in the frontier of $\SafeLarge$ from which there is low probability (positive, but less than $\sqrt{\xi}$) that $\cpi$ reaches $\cF_1$. Using the above coupling between $\cpi_1$ and $\cpi_2$, it follows that
	\begin{align}\label{eq:P2P3_splitF1}
	\ppr{\LDist{\cpi_2}}{\descP{\cF_1}} \leq \ppr{\LDist{\cpi_1}}{\descP{\cS_1}} + \ppr{\LDist{\cpi_2}}{\descP{\cF_1 \cap\descP{\cS_2}}}.
	\end{align}
	Again, we bound each term in the right-hand side of the above equation separately. For the first term of \cref{eq:P2P3_splitF1}, it holds that
	\begin{align*}
	2\xi \geq \ppr{\LDist{\cpi}}{\descP{\Fail}} \geq \ppr{\LDist{\cpi}}{\descP{\cF_1}} \geq \ppr{\LDist{\cpi}}{\descP{\cS_1}} \cdot \sqrt{\xi},
	\end{align*}
	and thus $\ppr{\LDist{\cpi}}{\descP{\cS_1}} \leq 2\sqrt{\xi}$. Applying  \cref{claim:SmallSetAprrx} again yields that
	\begin{align}\label{eq:P2P3_S1}
	 \ppr{\LDist{\cpi_1}}{\descP{\cS_1}} &\leq \ppr{\LDist{\pruAttack{k,\xi,\delta'}{\cpit},\HonBt}}{\descP{\cS_2}} \\
	 &\leq \phiBalE{k,\delta}\paren{\gamma_\cpi(\delta',\xi),2\sqrt{\xi},2\cdot m \cdot \xi, m,\delta',\mu} +\phiPruE{k,\delta}\paren{\gamma_\cpi(\delta',\xi),\xi,m,\delta',\mu}. \nonumber
	\end{align}
	As for the second term of \cref{eq:P2P3_splitF1}, we write
	\begin{align}\label{eq:P2P3_S2}
	\ppr{\LDist{\cpi_2}}{\descP{\cF_1 \cap\descP{\cS_2}}} &= \sum_{u\in\cS_2} \ppr{\LDist{\cpi_2}}{\descP{u}}\cdot \ppr{\LDist{(\cpi_2)_u}}{\descP{\cF_2}} \\
	&= \sum_{u\in\cS_2} \ppr{\LDist{\cpi_2}}{\descP{u}}\cdot \ppr{\LDist{\cpi_u}}{\descP{\cF_2}} \nonumber\\
	&\leq \sqrt{\xi} \cdot \sum_{u\in\cS_2} \ppr{\LDist{\cpi_2}}{\descP{u}} \nonumber\\
	&\leq \sqrt{\xi}, \nonumber
	\end{align}
	where the second inequality follows form the definition of $\cpi_2$.

	Plugging \cref{eq:P2P3_S1,eq:P2P3_S2} into \cref{eq:P2P3_splitF1} yields that
	\begin{align}\label{eq:P2P3_boundfF1}
	\ppr{\LDist{\cpi_2}}{\descP{\cF_1}} \leq \sqrt{\xi} + \phiBalE{k,\delta}\paren{\gamma_\cpi(\delta',\xi),2\sqrt{\xi},2\cdot m \cdot \xi, m,\delta',\mu} +\phiPruE{k,\delta}\paren{\gamma_\cpi(\delta',\xi),\xi,m,\delta',\mu}.
	\end{align}
	\cref{eq:P2P3_1} follows by plugging \cref{eq:P2P3_boundfF1,eq:P2P3_4} into \cref{eq:P2P3_2}, and noting that replacing $\xi$ by $\sqrt{\xi}$ in the second variable of the function $\phiBalE{k,\delta}$ only increases it.
\end{proof}

\begin{claim}\label{claim:Pi3Pi4}
It holds that $\Val(\cpi_4)\geq \Val(\cpi_3) - m\cdot \xi$.
\end{claim}
\begin{proof}
	We prove the claim by proving the following, stronger, statement:
	\begin{align}
	\SDP{\LDist{\cpi_3}}{\LDist{\cpi_4}}\leq m\cdot \xi.
	\end{align}
	Let $\Larges = \high{\cpi}{2\delta,\Est_\cpi^\xi}$ and $\Smalls = \low{\cpi}{2\delta,\Est_\cpi^\xi}$.
	We start by defining two randomized functions $f,g\colon \Vertices(\cpi)\setminus\Leaves(\cpi)\to\Vertices(\cpi)$, to simulate $\cpi_3$ and $\cpi_4$ respectively. Let
	\begin{align*}
	f(u) = \begin{cases}
	\text{sample $\ell \la \paren{\pruAttack{k,\delta',\xi}{\cpit},\HonBt}_u$; return $\ell_{1,\ldots,\size{u}+1}$} & u\notin \descP{\Fail\cup \Larges\cup\Smalls} \\
	\HonCont(u)_{1,\ldots,\size{u}+1} & u\in\descP{\Larges}\setminus \descP{\Fail} \\
	\HonCont^\xi(u)_{1,\ldots,\size{u}+1} & u\in \descP{\Smalls}\setminus\descP{\Fail \cup \Larges} \\
	\text{sample $\ell \la \paren{\finalAttack{k,\xi,\delta'}{\cpi},\HonB}_u$; return $\ell_{1,\ldots,\size{u}+1}$} & u\in\descP{\Fail},
	\end{cases}
	\end{align*}
	and let
	\begin{align*}
	g(u) = \begin{cases}
	f(u) & u\notin \descP{\Fail\cup \Larges\cup\Smalls} \\
	\HonCont^\xi(u)_{1,\ldots,\size{u}+1} & u\in\descP{\Larges}\setminus \descP{\Fail}\cap \ctrl{\cpi}{\HonA} \\
	f(u) & u\in\descP{\Larges}\setminus \descP{\Fail}\cap \ctrl{\cpi}{\HonB} \\
	f(u) & u\in \descP{\Smalls}\setminus\descP{\Fail \cup \Larges}\cap \ctrl{\cpi}{\HonA} \\
	\HonCont(u)_{1,\ldots,\size{u}+1} & u\in \descP{\Smalls}\setminus\descP{\Fail \cup \Larges}\cap \ctrl{\cpi}{\HonB} \\
	f(u) & u\in\descP{\Fail}.
	\end{cases}.
	\end{align*}
	Namely, for $u\in\descP{\Larges}\setminus \descP{\Fail}\cap \ctrl{\cpi}{\HonA}$, $f(u) = \HonCont(u)_{1,\ldots,\size{u}+1}$ while $g(u)=\HonCont^\xi(u)_{1,\ldots,\size{u}+1}$, where for $u\in \descP{\Smalls}\setminus\descP{\Fail \cup \Larges}\cap \ctrl{\cpi}{\HonB}$, $f(u)=\HonCont^\xi(u)_{1,\ldots,\size{u}+1}$ while $g(u)=\HonCont(u)_{1,\ldots,\size{u}+1}$. For any other $u$, $f(u)=g(u)$.

	Let $\sfH^h$ be the process that repeatedly calls the function $h$ with the answer of the previous call, staring with $h(\Root(\cpi))$, until reaching a leaf. It is easy to verify that $\sfH^f\equiv \LDist{\cpi_3}$ and $\sfH^g\equiv \LDist{\cpi_4}$. Thus, is suffices to bound $\SDP{\sfH^f}{\sfH^g}$. By the definitions of $f$ and $g$, it holds that $\SDP{f(u)}{g(u)} = 0$ if $u\in \Fail$ and that $\SDP{f(u)}{g(u)}\leq \SDP{\HonCont^\xi(u)}{\HonCont(u)} \leq \xi$ if $u\notin \Fail$, where the last inequality follows form the definition of $\Fail$.
	The claim follows since $\sfH^h$ makes at most $m$ calls to $h$.
\end{proof}

Using the above claims, we can formally prove \cref{lemma:final}.
\begin{proof}[Proof of \cref{lemma:final}]
	Fix $k\in\N$ and $\mu\in(0,1)$.
	\cref{claim:Pi1Pi2,claim:Pi2Pi3,claim:Pi3Pi4} yields that
	\begin{align*}
	\Val\paren{\finalAttack{k,\xi,\delta'}{\cpi},\HonB} &\geq \Val\paren{\pruAttack{k,\delta',\xi}{\cpit},\HonBt} - 2\delta' - (m+2)\cdot \sqrt{\xi} \\
	&\quad -2\cdot\phiBalE{k,\delta}\paren{\gamma_\cpi(\delta',\xi),2\sqrt{\xi},2\cdot m \cdot \xi, m,\delta',\mu} \\
	&\quad -2\cdot\phiPruE{k,\delta}\paren{\gamma_\cpi(\delta',\xi),\xi,m,\delta',\mu}.
	\end{align*}
	The proof now follows from \cref{lemma:AttackPruned}.
\end{proof}


\subsection{Implementing the Pruning-in-the-Head Attacker Using an Honest Continuator}\label{sec:ProtocolInv}
The pruning-in-the-head attacker (\cref{alg:rFinalPrunedAdv}) uses the  honest continuator and the  estimator algorithms (see \cref{def:AppxHonC,def:AppxEst} respectively),  both defined \wrt the attacked (original) protocol. It also uses the recursive approximated biased-continuation attacker (see \cref{def:itAppAtt}), designed to attack the approximately pruned variant of the attacked protocol. In this section we show how to use a given honest continuator for implementing the other two algorithms the pruning-in-the-head attacker uses. It follows that implementing the pruning-in-the-head attacker reduces to implementing an honest continuator. In the next (and final) section we show how to implement such continuator assuming the in-existence of one-way functions.

We begin by showing that using an honest continuator and an estimator, one can implement a biased continuator for the approximated pruned protocol. In fact, due to the recursive nature of the attack, we need to implement a biased continuator for every level of the recursion, and not only for the approximated pruned protocol.

\begin{definition}\label{def:seqBC}
	Let $\cpi = (\HonA,\HonB)$ be a protocol, let $\delta,\xi\in(0,1)$, let $k\in\N$ and let $\set{\MathAlg{D}^{(i)}}_{i\in(k)}$ be a set of algorithms. Let $\cpi^{(0)}=\cpi$. For $i\in[k]$, let $\cpi^{(i)} = \paren{\realA{i,\xi,\delta}{\MathAlg{D}^{(i-1)}},\HonB}$, where $\realA{i,\xi,\delta}{\MathAlg{D}^{(i-1)}}$ acts as $\realA{i,\xi,\delta}{\cpi}$ (\cref{def:itAppAtt}) does, but with $\MathAlg{D}^{(i-1)}$ taking the role of   $\RandomCont_{\paren{\realA{i-1,\xi,\delta}{\cpi},\HonB}}^{\xi,\delta}$.\footnote{Recall that $\RandomCont_{\cpi'}^{\xi,\delta}$, is an arbitrary fixed $(\xi,\delta)$-biased-continuator of $\cpi'$.} The sequence $\set{\MathAlg{D}^{(i)}}_{i\in(k)}$ is  a {\sf $(\xi,\delta)$-biased-continuators-sequence for $\cpi$}, if algorithm $\MathAlg{D}^{(i)}$ is a $(\xi,\delta)$-biased-continuator of $\cpi^{(i)}$,  for every $i\in(k)$.
\end{definition}

\begin{lemma}\label{lemma:BCformHCandEst}
	Let $\cpi$ be an $\rnd$-round protocol. Let  $\delta\in(0,1/2)$ ,  let $\xi\in(0,1)$, let $\Est$ be a $[0,1]$-output deterministic algorithm, let $\HC$ be $\xi$-honest-continuator for $\cpi$, and let $\cpit = (\HonAt,\HonBt) = \pru{\cpi}{2\delta,\xi,\Est,\HC}$ be the $(\delta,\Est,\HC)$-approximately pruned variant of $\cpi$ (see \cref{def:ApxPrunnedProt}). Then for every $k\in\N$, there exists a sequence of algorithms $\set{\MathAlg{D}^{(i)}}_{i\in(k)}$ with the following properties:
	\begin{enumerate}
		\item $\set{\MathAlg{D}^{(i)}}_{i\in(k)}$ is a $(\xi,\delta)$-biased-continuators-sequence for $\cpit$.

		\item $\MathAlg{D}^{(k)}$'s running time is
		$O\paren{ m^{3(k+1)}\cdot\NumOfCont^{k+1} \cdot \paren{T_{\Est} + T_{\HC}}}$,
		for $T_{\Est}$ and $T_{\HC}$ being the running times of $\Est$ and $\HC$ are  respectively.
	\end{enumerate}
\end{lemma}
\cref{lemma:BCformHCandEst} is proven in \cref{sec:ImplementBiasedCont}.
Next, we show how to implement a \emph{randomized} estimator using an honest continuator.

\begin{lemma}\label{lemma:EstFromHonCont}
	Let $\cpi$ be an $\rnd$-round protocol,  let $\xi\in(0,1)$ and let $\HC$ be a $\xi/2$-honest continuator for $\cpi$. Then there exists a randomized algorithm $\Est_{\cpi}^{(\xi,\HC)}$ such that the following holds.
	\begin{enumerate}
		\item\label{item:EstSuccProb}
		$\ppr{r\la \zo^{\ell}}{\Est_{\cpi,r}^{(\xi,\HC)} \textrm{ is a $\xi$-estimator for $\cpi$}} \geq 1-\xi$,
		for $\ell$ being an upper bound on the number coins used by $\Est_{\cpi}^{(\xi,\HC)}$ including those used by $\HC$, and $\Est_{\cpi,r}^{(\xi,\HC)}$ being the deterministic algorithm defined by hard-wiring $r$ into the randomness of $\Est_{\cpi}^{(\xi,\HC)}$.

		\item  $\Est^{(\xi,\HC)}$'s running time  is $O\paren{m\cdot\NumOfSam \cdot T_\HC}$, for $T_\HC$ being the running time of $\HC$.
	\end{enumerate}
\end{lemma}

\cref{lemma:EstFromHonCont} is proven in \cref{sec:ImplementEstFromHonCont}.
Using the above implementations for a biased continuator and an estimator, we can define an implantation for the pruning-in-the-head attacker using only an honest continuator. Recall that the pruning-in-the-head attacker requires a \emph{deterministic} estimator. To get such an estimator, we randomly fix the coins of $\Est^{(\xi,\HC)}$.

\begin{definition}[algorithm $\EffPruHead_{\cpi}^{(k,\xi,\delta,\HC)}$]\label{def:effPruHead}
	Let $\delta\in(0,1/2)$, let $\xi\in(0,1)$ and  let $k>0$.  Let $\cpi$ be an $\rnd$-round protocol, let $\HC$ be an algorithm, and let $\ell$ be the number of coins used by algorithm $\Est^{(\xi,\HC)}_{\cpi}$ from \cref{lemma:EstFromHonCont}, including those used by algorithm $\HC$. For $r\in \zo^\ell$, let $\Est_r= \Est^{(\xi,\HC)}_{\cpi;r}$ be deterministic algorithm resulting from fixing $\Est^{(\xi,\HC)}_{\cpi;r}$ coins to $r$.

	Let $\cpit = (\HonAt,\HonBt) = \pru{\cpi}{2\delta,\xi,\Est_r,\HC}$ and $\RandomCont = \RandomCont_\cpit^{(\xi,\delta,\HC,k-1)}$, where $\set{\RandomCont_\cpit^{(\xi,\delta,\HC,i)}}_{i\in(k-1)}$ is the $(\xi,\delta)$-biased-continuators-sequence for $\cpit$, guaranteed to exists by \cref{lemma:BCformHCandEst}.  Algorithm $\EffPruHead_{\cpi;r}^{(k,\xi,\delta,\HC)}$ acts as  algorithm $\finalAttack{k,\xi,\delta}{\cpi}$ (see \cref{alg:rFinalPrunedAdv}), but with algorithms $\HC$, $\Est_r$ and $\RandomCont$,  taking the role of algorithms $\HonCont_\cpi^\xi$, $\Est_\cpi^\xi$ and $\RandomCont_{\paren{\realA{k-1,\xi,\delta}{\cpit},\HonBt}}^{\xi,\delta}$, respectively. Finally, algorithm $\EffPruHead_{\cpi}^{(k,\xi,\delta,\HC)}$ act as $\EffPruHead_{\cpi;r}^{(k,\xi,\delta,\HC)}$, for  $r\la\zo^\ell$.
\end{definition}

The analysis of algorithm $\finalAttack{k,\xi,\delta}{\cpi}$ given in previous sections for was  done \wrt $\HonCont_\cpi^\xi$, $\Est_\cpi^\xi$ and $\RandomCont_{\paren{\realA{k-1,\xi,\delta}{\cpit},\HonBt}}^{\xi,\delta}$, the arbitrary but fixed honest continuator, estimator and biased continuator (see \cref{def:AppxHonC,def:AppxEst,def:AppBiassCSampler}). \cref{lemma:EstFromHonCont} show that $\Est_r$ is $\xi$-estimator with high probability and \cref{lemma:BCformHCandEst} show that $\RandomCont$ is a $(\xi,\delta)$-biased-continuator. Since the above fixing was arbitrary, the results form previous sections can be applied to Algorithm $\EffPruHead_{\cpi}^{(k,\xi,\delta,\HC)}$ as well. We do so in the next lemma, which also analyzes $\EffPruHead_{\cpi}^{(k,\xi,\delta,\HC)}$'s running time.
\remove{
\begin{lemma}\label{lemma:TimeOfPruHead}
	Let $\cpi=(\HonA,\HonB)$ be an $\rnd$-round protocol, let $\delta\in(0,1/2)$, let $\xi\in(0,1)$,  let $k\in \N$, let $\HC$ be an algorithm, and let  $\ell$ begin the number of random coins used by $\EffPruHead_{\cpi}^{(k,\xi,\delta,\HC)}$.  Then for $r\in\zo^\ell$, the running time of $\EffPruHead_{\cpi;r}^{(k,\xi,\delta,\HC)}$ is at most
	\begin{align*}
	2\cdot k\cdot m^k\cdot\NumOfCont^{k}\cdot \NumOfSam \cdot T_\HC,
	\end{align*}
	for $T_\HC$ being the running time of runs in time $\HC$ \Inote{again, are we missing $O$?}
\end{lemma}
} 
\begin{lemma}\label{prop:main}
	Let $\cpi=(\HonA,\HonB)$ be an $\rnd$-round protocol, let $0< \delta \leq \delta'\le \frac14$, $\xi\in(0,1)$, $k>0$, and let $\HC$ be a  $\xi/2$-honest continuator for $\cpi$.  The following  holds \wrt Algorithm $\EffPruHead_\cpi^{(k,\xi,\delta',\HC)}$:
	\begin{enumerate}
		\item\label{item:FinalSuccProb} 
		\begin{align}\label{eq:FinalSuccProb}
		\Val\paren{\EffPruHead_\cpi^{(k,\xi,\delta',\HC)},\HonB} &\geq \Val\paren{\rcAP{\cpit}{k},\HonBt} - 2\delta' - (m+2)\cdot \sqrt{\xi} - \xi \\
		&\quad -2\cdot\phiBalE{k,\delta}\paren{\neigh_{\cpi}(2\delta',\xi)+12\cdot m\cdot \xi/\delta',2\sqrt{\xi},2\cdot m \cdot \xi, m,\delta',\mu} \nonumber\\
		&\quad -3\cdot\phiPruE{k,\delta}\paren{\neigh_{\cpi}(2\delta',\xi)+12\cdot m\cdot \xi/\delta',\xi,m,\delta',\mu}, \nonumber
		\end{align}
		for every $\mu\in(0,1)$, and for  $\phiPruE{k,\delta},\phiBalE{k,\delta}$ be according to \cref{lemma:goodWOLowValueSimple,claim:SmallStaySmallApproxSimple} respectively.

		\item\label{item:FinalTime}  $\EffPruHead_\cpi^{(k,\xi,\delta',\HC)}$'s running time is at most $O\paren{ m^{3k+5}\cdot\NumOfCont^{k}\cdot \NumOfSam \cdot T_\HC}$, for $T_\HC$ being the running time of $\HC$ .
	\end{enumerate}
\end{lemma}
Note the extra $\xi$ term in the right-hand side of \cref{eq:FinalSuccProb} compared to the term in \cref{lemma:final}. This term comes from the probability the estimator used by $\EffPruHead_\cpi^{(k,\xi,\delta',\HC)}$ is not a good one.

\begin{proof}
	We prove each item separately.

	\medskip
	\emph{Proof of (\ref{item:FinalSuccProb}):}  It holds that
	\begin{align}\label{eq:useEst}
	\Val\paren{\EffPruHead_\cpi^{(k,\xi,\delta',\HC)},\HonB} &\geq \ppr{}{\out\paren{\EffPruHead_{\cpi;r}^{(k,\xi,\delta',\HC)},\HonB} = 1 \bigm| \text{$\Est_r$ is a $\xi$-estimator}}\cdot \Pr[\text{$\Est_r$ is a $\xi$-estimator}] \\
	&\geq \ppr{}{\out\paren{\EffPruHead_{\cpi;r}^{(k,\xi,\delta',\HC)},\HonB} = 1 \bigm| \text{$\Est_r$ is a $\xi$-estimator}} - \xi, \nonumber
	\end{align}
	where the second inequality follows from \cref{lemma:EstFromHonCont}. The above probabilities are over the choice of $r$, the additional, if any,  coins of $\EffPruHead_{\cpi;r}^{(k,\xi,\delta',\HC)}$, and the  coins of $\HonB$.

	We would like to conclude the proof by apply \cref{lemma:final} to \cref{eq:useEst}.  \cref{lemma:final} is stated  for $\HonCont_\cpi^\xi$ and $\Est_\cpi^\xi$ ---   \emph{arbitrary} $\xi$-honest-continuator and $\xi$-estimator for the attacked (original) protocol ---  and for $\RandomCont_{\paren{\realA{k-1,\xi,\delta}{\cpit},\HonBt}}^{\xi,\delta}$ --- an \emph{arbitrary} $(\xi,\delta)$-biased-continuator for $\paren{\realA{k-1,\xi,\delta}{\cpit},\HonBt}$. By assumption and \cref{lemma:EstFromHonCont,lemma:BCformHCandEst}, $\HC$, $\Est_r$ and $\RandomCont$ are such hones-continuator, estimator and biased-continuator, respectively. Hence, the proof of this part followed by \cref{lemma:final}.

	\medskip
	\emph{Proof of (\ref{item:FinalTime}):}
	The proof is an easy implication of \cref{lemma:BCformHCandEst,lemma:EstFromHonCont}. By definition, $\EffPruHead_{\cpi}^{(k,\xi,\delta',\HC)}$ makes a single call to $\Est$, and then either calls $\RandomCont$ or $\HC$.\footnote{As written in \cref{alg:rFinalPrunedAdv}, $\EffPruHead_{\cpi;r}^{(k,\xi,\delta',\HC)}$ might make $\rnd$ calls to $\Est$ (checking whether $u\in\descP{\cF}$ in step $2$ of the algorithm). This, however, does not significantly effect the running time and can be easily avoided by having the attacker keep a state. Furthermore, the time it takes to sample coins for $\Est$ is bounded by $\Est$'s running time.} We focus on the former case, as the running time of $\RandomCont$ is longer than that of $\HC$. By \cref{lemma:EstFromHonCont}, the running time of  $\Est$ is $O\paren{m\cdot \NumOfSam\cdot T_\HC}$, and by \cref{lemma:BCformHCandEst} and since $\delta\leq \delta'$, the running time of $\RandomCont$ is at most $O\paren{ m^{3(k+1)}\cdot\NumOfCont^{k+1} \cdot \paren{T_{\Est} + T_{\HC}}}$. For every call to $\RandomCont$ and $\Est$, algorithm $\EffPruHead_{\cpi}^{(k,\xi,\delta',\HC)}$ makes at most $O(m)$ steps. Hence, $\EffPruHead_{\cpi}^{(k,\xi,\delta',\HC)}$'s running time is bounded by
	\begin{align*}
	&O\paren{m\cdot m^{3(k+1)}\cdot\NumOfCont^{k+1} \cdot \paren{\paren{m\cdot \NumOfSam\cdot T_\HC} + T_{\HC}}} \\
	&= O\paren{ m^{3k+5}\cdot\NumOfCont^{k}\cdot \NumOfSam \cdot T_\HC}.
	\end{align*}
\end{proof}

The rest of this section is dedicated to proving \cref{lemma:BCformHCandEst,lemma:EstFromHonCont}.

\subsubsection{Implementing the Biased-Continuation Attacker using Honest Continuator and Estimator --- Proving \cref{lemma:BCformHCandEst}}\label{sec:ImplementBiasedCont}

Our goal is to implement a sequence of biased continuators, denoted by $\set{\MathAlg{D}^{(i)}}_{i\in(k)}$, for the approximated pruned protocol $\cpit$, using only honest continuator $\HC$ and an estimator $\Est$ for the original (\ie un-pruned) protocol. We do so by a recursive construction.

Given $\set{\MathAlg{D}^{(i)}}_{i\in(k-1)}$, a sequence of efficient algorithms such that $\MathAlg{D}^{(i)}$ is a $(\xi,\delta)$-biased-continuator for $\cpit^{(i)} =\paren{\realA{i,\xi,\delta}{\MathAlg{D}^{(i-1)}},\HonBt}$, we construct $\MathAlg{D}^{(k)}$, an efficient  $(\xi,\delta)$-biased-continuator for $\cpit^{(k)}$, as follows.
The first step is to reduce the task of implementing a biased continuator for $\cpit^{(k)}$ to that of implementing a honest continuator for $\cpit^{(k)}$. This is done using the method of rejection sampling.  The second step is to reduce the task of implementing a honest continuator for $\cpit^{(k)}$ to that of efficiently computing  $\cpit^{(k)}$. A key observation to achieve this task is that $\cpit^{(k)}$ is \emph{stateless}, namely the parties do not keep state between the different rounds. And constructing honest continuator for stateless and efficiently computable protocols is a trivial task. Finally, we note that $\cpit^{(k)}$ is efficient, assuming that $\MathAlg{D}^{(k-1)}$, $\HC$ and $\Est$, are. The section follows this outline to formally prove \cref{lemma:BCformHCandEst}.

\paragraph{From honest continuation to biased continuation.}
Turning an honest continuator into a biased continuator is essentially an easy task; given a transcript $u$ and a bit $b$ toward which the continuator should bias,  sample sufficiently many honest continuations for $u$, and return the first continuation whose common output is $b$. Indeed, if the transcript's value (\ie expected outcome) is close enough to $b$, then \whp the above process indeed returns a biased continuation.

\begin{algorithm}[$\RandomCont_\cpi^{(\xi,\delta,\HC)}$]\label{alg:BiassedFromHonCont}
	\item Parameters: $\xi,\delta \in (0,1)$.
	\item Oracle:  $\HC$.
	\item Input:  $u \in \Vertices(\cpi)$ and $b\in\zo$.
	\item Operation:
	\begin{enumerate}
		\item For $i=1$ to $\NumOfCont$:
		\begin{enumerate}
			\item Set $\ell \eqdef \HC(u)$.
			\item If $\Color_\cpi(\ell)=b$, return $\ell_{\size{u}+1}$.
		\end{enumerate}
		\item Return $\perp$.
	\end{enumerate}
\end{algorithm}

\begin{claim}\label{claim:BiasFromHonest}
	Let $\cpi$ be an $\rnd$-round protocol, let $\xi,\xi',\delta\in(0,1)$, and let  $\HC$ be a $\xi'$-honest continuator for $\cpi$.  Then $\RandomCont_\cpi^{(\xi,\delta,\HC)}$ is a $((t+1)\cdot\xi' + \xi,\delta)$-biased continuator for $\cpi$, for  $t=\NumOfCont$.
\end{claim}
\begin{proof}

	Let  $\HonCont_\cpi$ be the algorithm that on input $u$ returns random element in $\LDist{\cpi_u}$, and recall the definition of $\RandomCont_\cpi$  from \cref{def:sampler}. As usual, we focus on proving the statement for algorithms trying to bias towards one, \ie $b=1$; the proof for the case that $b=0$ is analogous. We show that for every node $u\in\Vertices(\cpi)$ with $\SDP{\HC(u)}{\HonCont_\cpi(u)}\leq \xi'$ and $\Val(\cpi_u) \geq \delta$, it holds that
	\begin{align}\label{eq:BCformHC1}
	\SDP{\RandomCont_\cpi^{(\xi,\delta,\HC)}(u,1)}{\RandomCont_\cpi(u,1)} \leq t\cdot \xi' + \xi.
	\end{align}
	This suffices to complete the proof since $\HC$ is a $\xi'$-honest continuator for $\cpi$, and thus the probability that $\cpi$ generates a transcript $u$ such that $\SDP{\HC(u)}{\HonCont_\cpi(u)} > \xi'$ is at most $\xi'$. The following is an ``unbounded version'' of  algorithm $\RandomCont_\cpi^{(\xi,\delta,\HC)}(\cdot,1)$ defined above.
	\begin{algorithm}[$\hRandomCont$]
		\item Input: $u\in\Vertices(\cpi)$.
		\item Operation:
		\begin{enumerate}
			\item Do (forever):
			\begin{enumerate}
				\item Set $\ell \eqdef \HonCont_\cpi(u)$.
				\item If $\Color_\cpi(\ell)=1$, return $\ell_{\size{u}+1}$.
			\end{enumerate}
		\end{enumerate}
	\end{algorithm}

	It is not difficult to verify that the probability that $\hRandomCont(u)$ does not halt is zero for every $u$ with $\Val(\cpi_u)>0$. Fix $u$ with $\SDP{\HC(u)}{\HonCont_\cpi(u)}\leq \xi'$ and $\Val(\cpi_u) \geq \delta$. It holds that
	\begin{align}
	\RandomCont_\cpi(u,1) \equiv \hRandomCont(u).
	\end{align}

	The only  difference between $\hRandomCont(u)$ and algorithm $\RandomCont_\cpi^{(\xi,\delta,\HonCont_\cpi)}(u,1)$ (\ie $\HonCont_\cpi$ is taking the role of $\HC$ in \cref{alg:BiassedFromHonCont}) is the probability the latter output $\perp$. Hence,
	\begin{align}
	\SDP{\RandomCont_\cpi^{(\xi,\delta,\HonCont)}(u,1)}{\hRandomCont(u)}\leq\ppr{}{\RandomCont_\cpi^{(\xi,\delta,\HonCont)}(u,1)=\perp}.
	\end{align}
	Compute
	\begin{align*}
	\ppr{}{\RandomCont_\cpi^{(\xi,\delta,\HonCont)}(u,1)=\perp} &= \paren{\ppr{\ell\la\HonCont(u)}{\Color_\cpi(\ell)=0}}^t\\
	&\leq (1-\delta)^t\\
	&\leq \xi,
	\end{align*}
	where the first inequality follows since $\Val(\cpi_u)\geq \delta$ and the last inequality follows from the choice of $t$. Moreover, since $\RandomCont_\cpi^{(\xi,\delta,\HC)}$ makes $t$ calls to its oracle, the assumption that $\SDP{\HonCont(u)}{\HC(u)}\leq \xi'$ and a standard hybrid argument, yield that
	\begin{align*}
	\SDP{\RandomCont_\cpi^{(\xi,\delta,\HonCont)}(u,1)}{\RandomCont_\cpi^{(\xi,\delta,\HC)}(u,1)}\leq t\cdot\xi.
	\end{align*}
	A triangle inequality now completes the proof of \cref{eq:BCformHC1}, and thus of the claim.
\end{proof}

\paragraph{Honest continuator for stateless protocols.}
For stateless protocols (\ie the parties maintain no state), implementing (perfect) honest continuation is trivial.
\begin{algorithm}[$\rHonContStateless_{\cpi}$]\label{alg:realHonCont}
	\item Input:  transcript $u \in \zo^\ast$.
	\item Operation:
	\begin{enumerate}
		\item Set $t = u$.
		\item Repeat until $t\in\Leaves(\cpi)$:
		\begin{enumerate}
			\item Let $\HonC$ be the party that controls $t$.
			\item Sample uniformly at random coins $r_\HonC$ for this round.
			\item Set $t= t\concat \HonC(t;r_\HonC)$.
		\end{enumerate}
		\item Return $t$.
	\end{enumerate}
\end{algorithm}
\begin{claim}\label{claim:HonContForStatless}
	For a stateless protocol $\cpi$, algorithm  $\rHonContStateless_{\cpi}$ of \cref{alg:realHonCont} is a $0$-honest continuator.
\end{claim}
\begin{proof}
	Immediate.
\end{proof}

\paragraph{Proving \cref{lemma:BCformHCandEst}.}
We now use the above understanding (\cref{claim:BiasFromHonest,claim:HonContForStatless}) to prove \cref{lemma:BCformHCandEst}.
\begin{proof}[Proof of \cref{lemma:BCformHCandEst}]
	The proof is by induction on $k$. We show that the running time of $\MathAlg{D}^{(k)}$ is at most $c^{k+1}\cdot m^{2(k+1)}\cdot\NumOfCont^{k+1} \cdot \paren{T_{\Est} + T_{\HC}}$, for some constant $c>0$ to be determined by the analysis. The running time as stated in the lemma follows since $c^{k+1}\in O(m^{k+1})$.

	For the base case $k=0$, the  $(\xi,\delta)$-biased-continuator for $\cpit$ is defined by
	\begin{align*}
	\MathAlg{D}^{(0)} = \RandomCont_\cpi^{(\xi,\delta,\rHonContStateless_{\cpit})}.
	\end{align*}
	Namely,  $\MathAlg{D}^{(0)}$ is \cref{alg:BiassedFromHonCont} with  \cref{alg:realHonCont} being the  honest continuator. \cref{claim:BiasFromHonest,claim:HonContForStatless} and the fact that, by definition (recall \cref{def:ApxPrunnedProt}), $\cpit$ is stateless, yield that $\MathAlg{D}^{(0)}$ is indeed a  $(\xi,\delta)$-biased-continuator for $\cpit$. As for its running time, $\MathAlg{D}^{(0)}$ makes at most $\NumOfCont$ calls to $\rHonContStateless_{\cpit}$. Every time $\rHonContStateless_{\cpit}$ is called, it makes at most $m$ calls to $\Est$ and to $\HC$. Let $c>0$ be a constant such that the operations $\MathAlg{D}^{(0)}$ makes other than calling $\Est$ or $\HC$ take at most $c\cdot m$ steps per such call.\footnote{Since the input length to $\MathAlg{D}^{(0)}$ is at most $m$ it is easy to verify that such $c$ exists} Hence, the running time of $\MathAlg{D}^{(0)}$ is at most $c\cdot m^2\cdot \NumOfCont\cdot(T_\Est + T_\HC)$.

	Assume the lemma holds for $k-1$, namely that there exist a sequence of algorithms $\set{\MathAlg{D}^{(i)}}_{i\in(k-1)}$ such that $\MathAlg{D}^{(i)}$ is a $(\xi,\delta)$-biased-continuator for $\cpit^{(i)} =\paren{\realA{i,\xi,\delta}{\MathAlg{D}^{(i-1)}},\HonBt}$\footnote{Recall that $\realA{i,\xi,\delta}{\MathAlg{D}^{(i-1)}}$ was defined in \cref{def:seqBC}.} and $\MathAlg{D}^{(k-1)}$'s running time is at most $c^k\cdot m^{2k}\cdot\NumOfCont^{k} \cdot \paren{T_{\Est} + T_{\HC}}$.  Define
	\begin{align*}
	\MathAlg{D}^{(k)} = \RandomCont_\cpi^{(\xi,\delta,\rHonContStateless_{\cpit^{(k)}})}.
	\end{align*}
	Note that $\cpit^{(k)}$ is stateless: $\realA{k,\xi,\delta}{\MathAlg{D}^{(k-1)}}$ simply makes calls to $\MathAlg{D}^{(k-1)}$ and thus stateless, and $\HonBt$ is stateless by definition. As in the base case, \cref{claim:BiasFromHonest,claim:HonContForStatless} yield that $\MathAlg{D}^{(k)}$ is a $(\xi,\delta)$-biased-continuator for $\cpit^{(k)}$. As for the running time of $\MathAlg{D}^{(k)}$, the analysis is identical to the base case, but $\rHonContStateless_{\cpit^{(k)}}$ makes at most $m$ calls to $\MathAlg{D}^{(k-1)}$, $\HC$ or $\Est$. Since the assumed  bound on the running time of $\MathAlg{D}^{(k-1)}$ is much longer than $T_\HC$ and $T_\Est$, the running time of $\MathAlg{D}^{(k)}$ is at most
	\begin{align*}
	c\cdot m^2\cdot\NumOfCont\cdot \paren{c^k\cdot m^{2k}\cdot\NumOfCont^{k} \cdot \paren{T_{\Est} + T_{\HC}}} \\
	= c^{k+1}\cdot m^{2(k+1)}\cdot\NumOfCont^{k+1} \cdot \paren{T_{\Est} + T_{\HC}}
	\end{align*}
\end{proof}

\remove{
	\begin{lemma}\label{lemma:PropOfFinal}
		Let $\cpi$ be an $\rnd$-round protocol. Let  $\delta\in(0,1/2)$ ,  let $\xi\in(0,1)$, let $\Est$ be a $[0,1]$-output deterministic algorithm, let $\HC$ be $\xi$-honest-continuator for $\cpi$, and let $\cpit = \pru{\cpi}{2\delta,\xi,\Est,\HC}$ be the $(\delta,\Est,\HC)$-approximately pruned variant of $\cpi$ (see \cref{def:ApxPrunnedProt}).

		Then for every $k\in\N$,  there  exists an implantation of algorithm $\realA{k,\xi,\delta}{\cpit}$ (see \cref{def:itAppAtt}), denoted by $\realA{k,\xi,\delta,\Est,\HC}{\cpit}$, with running time
		\begin{align*}
		m^k\cdot\NumOfCont^k \cdot \paren{T_{\Est} + T_{\HC}},
		\end{align*}
		for $T_{\Est}$ and $T_{\HC}$ being the running times of $\Est$ and $\HC$ are  respectively.
	\end{lemma}
	\Inote{see my edits above}
} 

\subsubsection{Implementing Estimator using Honest Continuator --- Proving \cref{lemma:EstFromHonCont}}\label{sec:ImplementEstFromHonCont}
Turning an honest continuator into a \emph{randomized} estimator is straightforward: given a transcript $u$, sample many honest continuations from $u$ and return the mean of the parties' common outcome bit of these continuations.

\begin{algorithm}[$\Est_\cpi^{(\xi,\HC)}$]\label{alg:EstFromHonCont}
	\item Parameters: $\xi \in (0,1)$.
	\item Oracle: algorithm $\HC$.
	\item Input: transcript $u \in \Vertices(\cpi)$.
	\item Operation:
	\begin{enumerate}\vspace{-0.1in}
		\item Set $sum =0$ and $s=\NumOfSam$.
		\item For $i=1$ to $s$: $sum =sum+ \Color_\cpi(\HC(u))$.

		(each call to $\HC$ is with fresh random coins).

		\item Return $sum/s$.
	\end{enumerate}
\end{algorithm}

The number of calls $\Est_\cpi^{(\xi,\HC)}$  makes to $\HC$  is set so that for most choices of its coins, $\Est_\cpi^{(\xi,\HC)}$  returns a good estimation for the value of \emph{every} node. Thus, fixing, at random, the coins of $\Est_\cpi^{(\xi,\HC)}$, results with high probability in a good \emph{deterministic} estimator.

\begin{proof}[Proof of \cref{lemma:EstFromHonCont}]
	The running time of $\Est_{\cpi}^{(\xi,\HC)}$ follows immediately from its definition.\footnote{$\Est_{\cpi}^{(\xi,\HC)}$'s input length is at most $m$, so it makes at most $O(m)$ steps per call to $\HC$.} In the rest of the proof we show that \cref{item:EstSuccProb} holds, namely that with probability at least $1-\xi$ over fixing its coins at random $\Est_{\cpi}^{(\xi,\HC)}$ is a $\xi$-estimator for $\cpi$.

Let $\Est_r = \Est_{\cpi,r}^{(\xi,\HC)}$,  let $\mu_u = \eex{\ell\la\HC(u)}{\Color(\ell)}$, and let $Q_r$ denote the event that $\forall u\in\Vertices(\cpi)\setminus\Leaves(\cpi)$, it holds that a$ \size{\Est_r(u) - \mu_u} \leq \xi/2$. The proof is an immediate conclusion from the following two simple observations.
	\begin{enumerate}[(1)]
		\item\label{item:est1} Condition on $Q_r$ occurring, $\Est_r$ is a $\xi$-estimator for $\cpi$.

		\item\label{item:est2} $\ppr{r\la \zo^{\ell}}{\neg Q_r} \leq \xi$.
	\end{enumerate}

	\medskip
	\emph{Proof of (\ref{item:est1}):} Compute
	\begin{align}\label{eq:est1}
	\lefteqn{\ppr{\ell \la \LDist{\cpi}}{\exists i\in(m-1) \colon \size{\Est_r(\ell_{1,\ldots,i}) - \Val(\cpi_{\ell_{1,\ldots,i}})} > \xi}}\\
	&\leq \ppr{\ell \la \LDist{\cpi}}{\exists i\in(m-1) \colon \size{\Est_r(\ell_{1,\ldots,i}) - \mu_{\ell_{1,\ldots,i}}} > \xi/2 \lor \size{\mu_{\ell_{1,\ldots,i}} - \Val(\cpi_{\ell_{1,\ldots,i}})} > \xi/2} \nonumber\\
	&\leq \ppr{\ell \la \LDist{\cpi}}{\exists i\in(m-1) \colon \size{\Est_r(\ell_{1,\ldots,i}) - \mu_{\ell_{1,\ldots,i}}} > \xi/2} \nonumber\\
	&\quad+ \ppr{\ell \la \LDist{\cpi}}{\exists i\in(m-1) \colon \size{\mu_{\ell_{1,\ldots,i}} - \Val(\cpi_{\ell_{1,\ldots,i}})} > \xi/2}.\nonumber
	\end{align}
	Since, by assumption, $Q_r$ occurs, the first summand of the right-hand side of \cref{eq:est1} is zero.
	Furthermore,  since $\HC$ is a $\xi/2$-honest continuator for $\cpi$, we bound the second  summand of the right-hand side of \cref{eq:est1}:
	\begin{align*}\label{eq:est3}
	\lefteqn{\ppr{\ell \la \LDist{\cpi}}{\exists i\in(m-1) \colon \size{\mu_{\ell_{1,\ldots,i}} - \Val(\cpi_{\ell_{1,\ldots,i}})} > \xi/2}} \\
	& \leq \ppr{\ell \la \LDist{\cpi}}{\exists i\in(m-1) \colon \SDP{\HC(\ell_{1,\ldots,i})}{\HonCont_\cpi(\ell_{1,\ldots,i})} > \xi/2} \nonumber\\
	&\leq \xi/2 \leq \xi.\nonumber
	\end{align*}
	Plugging the above into \cref{eq:est1} completes the proof.

	\bigskip
	\emph{Proof of (\ref{item:est2}):} We use the following fact derived from Hoeffding's bound.
	\begin{fact}[sampling]\label{fact:Sampling}
		Let $t\geq\frac{\ln\paren{\frac{2}{\gamma}}}{2\cdot\eps^2}$, let $X_1,\ldots,X_t\in[0,1]$ be iid Boolean random variables, and let $\mu=\Exp[X_i]$. Then
		$\pr{\abs{\frac{1}{t}\sum_{i=1}^{t}X_i-\mu}\geq\eps}\leq \gamma$.

	\end{fact}
	Taking $\eps \eqdef \xi/2$ and $\gamma \eqdef \xi/2^m$ with \cref{fact:Sampling} yields that
	\begin{align}
	\ppr{r\la \zo^{\ell}}{\size{\Est_r(u) - \mu_u} > \xi/2} \leq \frac{\xi}{2^m}
	\end{align}
	for every $u\in\Vertices(\cpi)\setminus\Leaves(\cpi)$, and a  union bound yields that
	\begin{align*}
	\ppr{r\la \zo^{\ell}}{\neg Q_r} & = \ppr{r\la \zo^{\ell}}{\exists u\in\Vertices(\cpi)\setminus\Leaves(\cpi) \colon  \size{\Est_r(u) - \mu_u} > \xi/2}\\
	&\leq \sum_{u\in\Vertices(\cpi)\setminus\Leaves(\cpi)}\ppr{r\la \zo^{\ell}}{\size{\Est_r(u) - \mu_u} > \xi/2} \\
	&\leq \sum_{u\in\Vertices(\cpi)\setminus\Leaves(\cpi)}\frac{\xi}{2^m} = \xi.
	\end{align*}
\end{proof}


\subsection{Main Theorem --- Inexistence of OWF's Implies an Efficient Attacker}\label{sec:EfficeinrAttack}
We are finally ready to state and prove our main result -- the existence of any constant bias (even weak) coin-flipping protocol implies the existence of one-way functions.

In the following we consider both protocols and algorithms that get a security parameter, written in unary, as input (sometimes, in addition to other input), and  protocols and algorithms that do not get a security parameter, as we did in previous sections. We refer to the former type as parametrized and to the latter type as non-parametrized. It will be clear from the context whether we consider a parametrized or non-parametrized entity. In particular, a poly-time entity whose running time is measured as a function of its security parameter is  by definition parametrized. Given a parametrized protocol $\cpi$ and $n\in \N$, let $\cpi_n$ be its non-parametrized variant with the security parameter $1^n$ hardwired into the parties' code. We apply similar notation also for parametrized algorithms.

\begin{theorem}[main theorem, restatement of \cref{thm:mainInf}]\label{thm:main}
Assume one-way functions do not exist. For every \ppt coin-flipping protocol
$\cpi = \HonAHonB$ and $\eps>0$, there exist {\pptm}s $\cA$ and $\cB$ such that the following hold for infinitely many $n$'s.
\begin{enumerate}
 \item $\ppr{}{\out(\cA(1),\Bc)(1^n)=1} \geq 1 - \eps$ or $\ppr{}{\out(\HonA,\cB(0))(1^n)=0} \leq \eps$, and
 \item $\ppr{}{\out(\cA(0),\Bc)(1^n)=0} \leq \eps$ or $\ppr{}{\out(\Ac,\cB(1))(1^n)=1} \geq 1 - \eps$.
\end{enumerate}
\end{theorem}

The proof of \cref{thm:main} follows from \cref{thm:MainIdeal,prop:main} together with the following lemma that  shows how to implement an efficient honest continuator assuming OWFs do not exist.

\begin{lemma}\label{lemma:ProtocolInv}
Assume one-way functions do not exist. Then for any \ppt coin-flipping protocol $\cpi = \HonAHonB$ and $p\in\poly$, there exists a \pptm algorithm $\HC$ such that $\HC_n$ is a $1/p(n)$-honest continuator for $\cpi_n$ for infinitely many $n$'s.
\end{lemma}

The proof of \cref{lemma:ProtocolInv} is given below, but first  we use it to prove  \cref{thm:main}.


\paragraph{Proving \cref{thm:main}.}
\begin{proof}[Proof of \cref{thm:main}]
We focus on proving the first part of the theorem, where the second, symmetric, part follows the same arguments.

Let $\delta=\eps/8$, let $m(n) = \round(\cpi_n)$ and let $\xi(n) = 1/p(n)  < \frac{(2\delta)^2}{16m(n)^2}$ for some large enough $p\in\poly$ to be determined by the analysis. Let $\HC$ be the  algorithm guaranteed by \cref{lemma:ProtocolInv}, such that $\HC_n$ is an $\xi(n)/2$-honest continuator for $\cpi_n$ for every $n$ in an infinite set $\I\subseteq\N$. For  $n\in\I$, let $\delta'_n\in[\delta/2,\delta]$ be such that $\neigh_{\cpi_n}(2\delta'_n,\xi(n))\leq m(n)\cdot\sqrt{2\xi(n)}$, guaranteed to exist from \cref{prop:BadIsSmall}.\footnote{By the choice of $\xi$ and by  \cref{prop:BadIsSmall} there exists $\delta''\in[\delta,2\delta]$ such that $\neigh_{\cpi_n}(\delta'',\xi(n))\leq m(n)\cdot\sqrt{2\xi(n)}$. Now we can set $\delta'=\delta''/2$.}
Let $\cpit_n = \paren{\HonAt_n,\HonBt_n} = \pru{\cpi_n}{2\delta'_n,\xi}$ be the $(2\delta'_n,\xi)$-approximately pruned variant of $\cpi_n$. Let $\kappa=\kappa(\eps/2)$ be such that
$\Val\paren{\rcAP{\cpit_n}{k},\HonBt_n} > 1 - \eps/2$ or $\Val\paren{\HonAt_n,\rcBP{\cpit_n}{k}}< \eps/2$, guaranteed to exist for every $n\in\I$ from \cref{thm:MainIdeal}. Assume \wlg that there exists an infinite set $\I'\subseteq\I$ such that
\begin{align}
 \Val\paren{\rcAP{\cpit_n}{k},\HonBt_n}> 1 - \eps/2
\end{align}
for every $n\in\I'$ and let $\mu(n)=1/n$.

Let $r,s\in\poly$ such that the following two equations hold.
\begin{align*}
&\phiBalE{k,\delta/2}\paren{\neigh_{\cpi_n}(2\delta'_n,\xi(n))+12\cdot m(n)\cdot \xi(n)/\delta'_n,2\sqrt{\xi(n)},2\cdot m(n) \cdot \xi(n), m(n),\delta'_n,\mu(n)} \\
&= \paren{\neigh_{\cpi_n}(2\delta'_n,\xi(n))+12\cdot m(n)\cdot \xi(n)/\delta'_n + 2\sqrt{\xi(n)}  2\cdot m(n) \cdot \xi(n)} \cdot q_{\kappa,\delta/2}(m(n),1/\delta'_n,1/\mu(n)) + 1/\mu(n)\\
&\leq \sqrt{\xi(n)}\cdot r(n).
\end{align*}

And
\begin{align*}
&\phiPruE{\kappa,\delta/2}\paren{\neigh_{\cpi_n}(2\delta'_n,\xi(n))+12\cdot m(n)\cdot \xi(n)/\delta'_n,\xi(n),m(n),\delta'_n,\mu(n)} \\
& = \paren{\neigh_{\cpi_n}(2\delta'_n,\xi(n))+12\cdot m(n)\cdot \xi(n)/\delta'_n + \xi(n)}\cdot p_{\kappa,\delta/2}(m(n),1/\delta'_n,1/\mu(n)) +1/\mu(n) \\
&\leq \sqrt{\xi(n)}\cdot s(n),
\end{align*}
Note that by the setting of parameters thus far, such $r$ and $s$ exists. Finally, let  $\xi\in\poly$ be such that
\begin{align*}
(m(n)+2)\cdot \sqrt{\xi(n)} + \xi(n) + 2\cdot \sqrt{\xi(n)}\cdot r(n) + 3\cdot \sqrt{\xi(n)}\cdot s(n) \in o(1).
\end{align*}
By \cref{prop:main}(\ref{item:FinalSuccProb}),
\begin{align}
\Val\paren{\EffPruHead_{\cpi_n}^{(\kappa,\xi(n),\delta'_n,\HC_n)},\HonB_{\cpi_n}} &\geq \Val\paren{\rcAP{\cpit_n}{k},\HonBt_n} - 2\delta' -o(1) \geq 1-\frac{\eps}{2} - \frac{\eps}{4} - o(1).
\end{align}

We can now define out final adversary $\cA(1)$. Let $\cV = \set{\paren{\delta+j\cdot 2\xi}/2 \colon j\in\set{0,1,\ldots,\left\lceil m/\sqrt{\xi}\right\rceil}}$ be the set from \cref{prop:BadIsSmall} and recall that $\delta'_n\in\cV$.
Prior to interacting with $\HonB$, algorithm $\cA(1)$  estimates the value of
$\cpit_{\delta'}\eqdef\paren{\EffPruHead_{\cpi_n}^{(\kappa,\xi(n),\delta',\HC_n)},\HonB_{\cpi_n}}$, for every $\delta'\in\cV$, by running the latter protocol for polynomially-many times. Let $\delta_n^\ast$ be the value such that $\cpit_{\delta_n^*}$ is the maximum of all estimations. When interacting with $\HonB$, algorithm $\cA(1)$ behave as $\EffPruHead_{\cpi_n}^{(\kappa,\xi,\delta^\ast_n,\HC_n)}$.

Since $\delta'_n\in\cV$, it follows that $\ppr{}{\Val\paren{\cpit_{\delta_n^*}} \geq \Val\paren{\cpit_{\delta'_n}} - \eps/8} \geq 1-o(1)$, where the probability is over the coins on $\cA(1)$. Thus,
\begin{align}
&\ppr{}{\out(\cA(1),\Bc)(1^n)=1} \\
&\geq \ppr{}{\out(\cA(1),\Bc)(1^n)=1 \biggm| \Val\paren{\cpit_{\delta_n^*}} \geq \Val\paren{\cpit_{\delta'_n}} - \eps/8} \cdot \ppr{}{\Val\paren{\cpit_{\delta_n^*}} \geq \Val\paren{\cpit_{\delta'_n}} - \eps/8} \nonumber\\
&\geq (1-5\eps/8-o(1))\cdot (1-o(1)) \nonumber\\
&\geq 1-5\eps/8 - o(1) \geq 1-\eps,  \nonumber
\end{align}
for large enough $n\in\I'$.

The last step is to argue that $\cA(1)$ is efficient. By our choice of parameters, the fact that $\kappa$ is constant (\ie independent of $n$) and $\HC$ is \pptm, \cref{prop:main}(\ref{item:FinalTime}) yields that $\EffPruHead_{\cpi_n}^{(\kappa,\xi(n),\delta'_n,\HC_n)}$ is a \pptm. Since $\size{\cV}\in\poly(n)$, it follows that the running time of $\cA(1)$ is also is $\poly(n)$.
\end{proof}

It is left to prove \cref{lemma:ProtocolInv}.
\paragraph{Proving \cref{lemma:ProtocolInv}.}
\begin{proof}[Proof of \cref{lemma:ProtocolInv}]
Let $\rnd(n) = \round(\cpi_n)$, and let $\rho_\HonA(n)$ and $\rho_\HonB(n)$ be, respectively, the (maximal) number of random bits used by $\HonA$ and $\HonB$ on common input $1^n$. Consider the \emph{transcript function} $f_\cpi$ over $1^\ast\times \zo^{\rho_\HonA(n)}\times \zo^{\rho_\HonB(n)} \times (\rnd(n)-1)$, defined by
\begin{align}\label{eq:ProtocolFunction}
f_\cpi(1^n,r_\HonA,r_\HonB,i) &=1^n, \trans((\HonA(\cdot;r_\HonA),\HonB(\cdot;r_\HonB))(1^n))_{1,\ldots,i}.
\end{align}
Since $\cpi$ is a polynomial time protocol, it follows  \wlg that $m(n),\rho_\HonA(n),\rho_\HonB(n)\in\poly(n)$ and that $f_\cpi$ is computable in polynomial time.

Under the assumption that OWFs do not exist, the transcript function is not distributional one-way, \ie it has an inverter that returns a random preimage. We would like to argue that an algorithm that outputs the transcript induced by the randomness this inverter returns is an honest continuator. This is almost true, as this inverter guarantees to work for a random node of the protocol tree, and we require that an honest continuator work for all nodes in a random \emph{path} of the protocol tree. Still, since any path in the protocol tree is of polynomial length,  the lemma follows by a union bound. We now move to the formal proof.

Fix $p\in\poly$ and let $\Inv$ be the $1/(m\cdot p)$-inverter guaranteed to
exist by \cref{lemma:NoDistOWF}. Namely, $\Inv_n = \Inv(1^n,\cdot)$ is a
$1/(m(n)\cdot p(n))$-inverter for $f_\cpi(1^n,\cdot,\cdot,\cdot)$ for every $n$
within an infinite size index set
$\I\subseteq\N$.\footnote{\cref{lemma:NoDistOWF} is stated for functions whose
  domain is $\zn$ for every $n\in\N$, \ie functions defined for every input
  length. Although the transcript function is not defined for every input length
  (and has $1^n$ as an input), using the fact that it is defined on $\zo^{q(n)}$
  for some $q(n)\in\poly(n)$ and standard padding techniques,
  \cref{lemma:NoDistOWF} does in fact guarantee such an inverter.} By the
definition of $f_\cpi$, choosing a random preimage from $f_\cpi^{-1}(1^n,u)$ is
equivalent to choosing an element according to the distribution
$\paren{\ConsisDist_{\cpi_n}(u),\size{u}}$.\footnote{Recall that
  $\ConsisDist_\cpi(u)$ returns random coins for the parties, consistent with a
  random execution of $\cpi$ leading to $u$.} For a transcript $u$ and coins
$r_\HonA$ and $r_\HonB$ for $\HonA$ and $\HonB$ respectively, let
$f_u(r_\HonA,r_\HonB,\cdot)\eqdef u\concat
\paren{\trans(\HonA(\cdot;r_\HonA),\HonB(\cdot;r_\HonB))(1^n)}_{\size{u}+1,\ldots,\rnd(n)}$,
and let $\HC_n$ be the algorithm that, given input $u$, returns
$f_u(\Inv_n(u))$.\footnote{The function $f$ actually ignores its third argument.
  It is defined to take three arguments only to match the number of arguments in
  the output of $\Inv_n$.}
We show that $\HC_n$ is a $1/p(n)$-honest continuator for $\cpi_n$, for every
$n\in\I$.

Fix $n\in\I$. Let $m=m(n)$,  $p=p(n)$ and from now on we omit $n$ from notations. Note that $f_u(\ConsisDist_{\cpi}(u),\size{u})\equiv \LDist{\cpi_u} \equiv \HonCont_\cpi(u)$, 
and thus
\begin{align}
\SDP{\Inv(u)}{\paren{\ConsisDist_{\cpi}(u),\size{u}}} \geq \SDP{\HC(u)}{\HonCont_\cpi(u)},
\end{align}
for every transcript $u$. Let $I$ and $L$ be random variables distributed as $I \la (m-1)$ and $L\la \LDist{\cpi}$ respectively. Compute
\begin{align*}
\lefteqn{\ppr{}{\SDP{\Inv(L_{1,\ldots,I})}{\paren{\ConsisDist_{\cpi}(L_{1,\ldots,I}),I}} > \frac{1}{m\cdot p}}} \\
&= \sum_{j=0}^{m-1} \ppr{}{\SDP{\Inv(L_{1,\ldots,I})}{\paren{\ConsisDist_{\cpi}(L_{1,\ldots,I}),I}} > \frac{1}{m\cdot p} \mid I=j}\cdot \ppr{}{I=j} \\
&= \frac{1}{m} \sum_{j=0}^{m-1} \ppr{}{\SDP{\Inv(L_{1,\ldots,j})}{\paren{\ConsisDist_{\cpi}(L_{1,\ldots,j}),j}} > \frac{1}{m\cdot p}} \\
&\geq \frac{1}{m} \sum_{j=0}^{m-1} \ppr{}{\SDP{\HC(L_{1,\ldots,j})}{\HonCont_{\cpi}(L_{1,\ldots,j})} > \frac{1}{m\cdot p}} \\
&\geq \frac{1}{m} \sum_{j=0}^{m-1} \ppr{}{\SDP{\HC(L_{1,\ldots,j})}{\HonCont_{\cpi}(L_{1,\ldots,j})} > \frac{1}{p}} \\
& \geq \frac{1}{m} \ppr{}{\exists j\in(m-1) \colon \SDP{\HC(L_{1,\ldots,j})}{\HonCont_{\cpi}(L_{1,\ldots,j})} > \frac{1}{p}}.
\end{align*}
The proof now follows by the properties of $\Inv$.

\end{proof}

\bibliographystyle{abbrvnat}
\bibliography{crypto}

\begin{thebibliography}{26}
\providecommand{\natexlab}[1]{#1}
\providecommand{\url}[1]{\texttt{#1}}
\expandafter\ifx\csname urlstyle\endcsname\relax
  \providecommand{\doi}[1]{doi: #1}\else
  \providecommand{\doi}{doi: \begingroup \urlstyle{rm}\Url}\fi

\bibitem[Averbuch et~al.(1985)Averbuch, Blum, Chor, Goldwasser, and
  Micali]{ABCGM85}
B.~Averbuch, M.~Blum, B.~Chor, S.~Goldwasser, and S.~Micali.
\newblock How to implement {Bracha}'s ${O}(\log n)$ {Byzantine} agreement
  algorithm, 1985.
\newblock Unpublished manuscript.

\bibitem[Beimel et~al.(2010)Beimel, Omri, and Orlov]{BeimelOO10}
A.~Beimel, E.~Omri, and I.~Orlov.
\newblock Protocols for multiparty coin toss with dishonest majority.
\newblock In \emph{Advances in Cryptology -- CRYPTO 2010}, pages 538--557,
  2010.

\bibitem[Berman et~al.(2014)Berman, Haitner, and Tentes]{BermanHT14}
I.~Berman, I.~Haitner, and A.~Tentes.
\newblock Coin flipping of \emph{any} constant bias implies one-way functions.
\newblock In \emph{Symposium on Theory of Computing, {STOC} 2014, New York, NY,
  USA, May 31 - June 03, 2014}, pages 398--407, 2014.
\newblock \doi{10.1145/2591796.2591845}.
\newblock URL \url{http://doi.acm.org/10.1145/2591796.2591845}.

\bibitem[Berman et~al.(2018)Berman, Haitner, and Tentes]{BermanHT18-jacm}
I.~Berman, I.~Haitner, and A.~Tentes.
\newblock Coin flipping of any constant bias implies one-way functions.
\newblock \emph{J. ACM}, 65\penalty0 (3):\penalty0 14:1--14:95, Mar. 2018.
\newblock ISSN 0004-5411.
\newblock \doi{10.1145/2979676}.
\newblock URL \url{http://doi.acm.org/10.1145/2979676}.

\bibitem[Blum(1981)]{Blum81}
M.~Blum.
\newblock Coin flipping by telephone.
\newblock In \emph{Advances in Cryptology -- CRYPTO '81}, pages 11--15, 1981.

\bibitem[Chailloux and Kerenidis(2009)]{ChaillouxKerenidis09}
A.~Chailloux and I.~Kerenidis.
\newblock Optimal quantum strong coin flipping.
\newblock In \emph{Proceedings of the 50th Annual Symposium on Foundations of
  Computer Science (FOCS)}, pages 527--533, 2009.

\bibitem[Cleve(1986)]{Cleve86}
R.~Cleve.
\newblock Limits on the security of coin flips when half the processors are
  faulty.
\newblock In \emph{Proceedings of the 18th Annual ACM Symposium on Theory of
  Computing (STOC)}, pages 364--369, 1986.

\bibitem[Cleve and Impagliazzo(1993)]{CleveI93}
R.~Cleve and R.~Impagliazzo.
\newblock Martingales, collective coin flipping and discrete control processes
  (extended abstract).
\newblock
  \url{http://citeseerx.ist.psu.edu/viewdoc/summary?doi=10.1.1.51.1797}, 1993.

\bibitem[Dachman-Soled et~al.(2011)Dachman-Soled, Lindell, Mahmoody, and
  Malkin]{Dachman11}
D.~Dachman-Soled, Y.~Lindell, M.~Mahmoody, and T.~Malkin.
\newblock On the black-box complexity of optimally-fair coin tossing.
\newblock In \emph{Theory of Cryptography, 8th Theory of Cryptography
  Conference (TCC)}, volume 6597, pages 450--467, 2011.

\bibitem[Goldreich and Levin(1989)]{GoldreichL89}
O.~Goldreich and L.~A. Levin.
\newblock A hard-core predicate for all one-way functions.
\newblock In \emph{Proceedings of the 21st Annual ACM Symposium on Theory of
  Computing (STOC)}, pages 25--32, 1989.

\bibitem[Goldreich et~al.(1984)Goldreich, Goldwasser, and
  Micali]{GoldreichGoMi85}
O.~Goldreich, S.~Goldwasser, and S.~Micali.
\newblock On the cryptographic applications of random functions.
\newblock In \emph{Advances in Cryptology -- CRYPTO '84}, pages 276--288, 1984.

\bibitem[Goldreich et~al.(1986)Goldreich, Goldwasser, and
  Micali]{GoldreichGoMi86}
O.~Goldreich, S.~Goldwasser, and S.~Micali.
\newblock How to construct random functions.
\newblock \emph{J. {ACM}}, 33\penalty0 (4):\penalty0 792--807, 1986.
\newblock \doi{10.1145/6490.6503}.
\newblock URL \url{http://doi.acm.org/10.1145/6490.6503}.

\bibitem[Haitner and Omri(2011)]{HaitnerOmri11}
I.~Haitner and E.~Omri.
\newblock {Coin Flipping with Constant Bias Implies One-Way Functions}.
\newblock In \emph{Proceedings of the 52nd Annual Symposium on Foundations of
  Computer Science (FOCS)}, pages 110--119, 2011.

\bibitem[Haitner et~al.(2009)Haitner, Nguyen, Ong, Reingold, and
  Vadhan]{HaitnerNgOnReVa09}
I.~Haitner, M.~Nguyen, S.~J. Ong, O.~Reingold, and S.~Vadhan.
\newblock Statistically hiding commitments and statistical zero-knowledge
  arguments from any one-way function.
\newblock \emph{SIAM Journal on Computing}, 39\penalty0 (3):\penalty0
  1153--1218, 2009.

\bibitem[H{\aa}stad et~al.(1999)H{\aa}stad, Impagliazzo, Levin, and
  Luby]{HastadImLeLu99}
J.~H{\aa}stad, R.~Impagliazzo, L.~A. Levin, and M.~Luby.
\newblock A pseudorandom generator from any one-way function.
\newblock \emph{{SIAM} J. Comput.}, 28\penalty0 (4):\penalty0 1364--1396, 1999.
\newblock \doi{10.1137/S0097539793244708}.
\newblock URL \url{https://doi.org/10.1137/S0097539793244708}.
\newblock Preliminary versions in {\it STOC'89} and {\it STOC'90}.

\bibitem[Impagliazzo()]{ImpagliazzoPHD}
R.~Impagliazzo.
\newblock Pseudo-random generators for cryptography and for randomized
  algorithms.
\newblock \url{http://cseweb.ucsd.edu/~russell/format.ps}.
\newblock Ph.D. Thesis.

\bibitem[Impagliazzo and Luby(1989)]{ImpagliazzoLu89}
R.~Impagliazzo and M.~Luby.
\newblock One-way functions are essential for complexity based cryptography.
\newblock In \emph{Proceedings of the 30th Annual Symposium on Foundations of
  Computer Science (FOCS)}, pages 230--235, 1989.

\bibitem[Kitaev(2003)]{Kit03}
A.~Y. Kitaev.
\newblock Quantum coin-flipping.
\newblock Presentation at the 6th Workshop on Quantum Information Processing
  (QIP 2003), 2003.

\bibitem[Maji et~al.(2010)Maji, Prabhakaran, and Sahai]{Maji10}
H.~K. Maji, M.~Prabhakaran, and A.~Sahai.
\newblock {On the Computational Complexity of Coin Flipping}.
\newblock In \emph{Proceedings of the 51st Annual Symposium on Foundations of
  Computer Science (FOCS)}, pages 613--622, 2010.

\bibitem[Mochon(2007)]{Moc07}
C.~Mochon.
\newblock Quantum weak coin flipping with arbitrarily small bias.
\newblock arXiv:0711.4114, 2007.

\bibitem[Moran et~al.(2009)Moran, Naor, and Segev]{MoranNS09}
T.~Moran, M.~Naor, and G.~Segev.
\newblock An optimally fair coin toss.
\newblock In \emph{Theory of Cryptography, 6th Theory of Cryptography
  Conference (TCC)}, pages 1--18, 2009.

\bibitem[Naor(1991)]{Naor91}
M.~Naor.
\newblock Bit commitment using pseudorandomness.
\newblock \emph{J. Cryptology}, 4\penalty0 (2):\penalty0 151--158, 1991.
\newblock \doi{10.1007/BF00196774}.
\newblock URL \url{https://doi.org/10.1007/BF00196774}.
\newblock Preliminary version in {\it CRYPTO'89}.

\bibitem[Naor and Yung(1989)]{NaorYu89}
M.~Naor and M.~Yung.
\newblock Universal one-way hash functions and their cryptographic
  applications.
\newblock In \emph{Proceedings of the 21st Annual ACM Symposium on Theory of
  Computing (STOC)}, pages 33--43, 1989.

\bibitem[Roberts and Varberg(1973)]{RobertsV73}
A.~W. Roberts and D.~E. Varberg.
\newblock \emph{Convex Functions}.
\newblock Academic Press Inc, 1973.

\bibitem[Rompel(1990)]{Rompel90}
J.~Rompel.
\newblock One-way functions are necessary and sufficient for secure signatures.
\newblock In \emph{Proceedings of the 22nd Annual ACM Symposium on Theory of
  Computing (STOC)}, pages 387--394, 1990.

\bibitem[Zachos(1986)]{Zachos86}
S.~Zachos.
\newblock {Probabilistic Quantifiers, Adversaries, and Complexity Classes: An
  Overview}.
\newblock In \emph{Proceedings of the First Annual IEEE Conference on
  Computational Complexity}, pages 383--400, 1986.

\end{thebibliography}

\appendix
\section{Missing Proofs}\label{sec:MissProof}

\newcommand{\tg}{\widetilde{g}}
\subsection{Proving Lemma \ref{lemma:calculus1}}

\begin{lemma}[Restatement of \cref{lemma:calculus1}]
\CalcLemmaOne{1}
\end{lemma}

\begin{proof}
The lemma easily follows if one of the following holds: (1) $p_0=1,p_1=0$; (2) $p_0=0,p_1=1$; and (3) $x=y=0$. Assuming $1>p_0,p_1>0$ and $x+y>0$, dividing \cref{eq:CalcLemma} by its right-hand side (which is always positive) gives
\begin{align}\label{eq:CalcLemmaBeforeVarChange}
p_0\cdot \frac{\left(\frac{x}{(p_0x+p_1y)}\right)^{k+1}}{\prod_{i=1}^k \frac{a_i}{p_0a_i+p_1b_i}} + p_1\cdot \frac{\left(\frac{y}{(p_0x+p_1y)}\right)^{k+1}}{\prod_{i=1}^k \frac{b_i}{p_0a_i+p_1b_i}} \geq 1.
\end{align}
Define the following variable changes:
\begin{align*}
z= \frac{p_0 x}{p_0x+p_1y}  \qquad  c_i = \frac{p_0 a_i}{p_0a_i+p_1b_i} \quad \mbox{for $1\leq i\leq k$}.
\end{align*}
It follows that
\begin{align*}
1-z= \frac{p_1 y}{p_0x+p_1y}  \qquad  1-c_i = \frac{p_1 b_i}{p_0a_i+p_1b_i} \quad \mbox{for $1\leq i\leq k$}.
\end{align*}
Note that $0\leq z \leq 1$ and that $0< c_i < 1$ for every $1\leq i \leq k$. Plugging the above into \cref{eq:CalcLemmaBeforeVarChange}, it remains to show that
\begin{align}\label{eq:calc1eq2}
\frac{z^{k+1}}{\prod_{i=1}^k c_i} + \frac{(1-z)^{k+1}}{\prod_{i=1}^k (1-c_i)} \geq 1
\end{align}
for all $0\leq z \leq 1$ and $0< c_i < 1$.
\cref{eq:calc1eq2} immediately follows for $z=0,1$, and in the rest of the proof we show that it also holds for $z\in(0,1)$. 
Define $f(z,c_1,\ldots,c_k)\eqdef \frac{z^{k+1}}{\prod_{i=1}^k c_i} + \frac{(1-z)^{k+1}}{\prod_{i=1}^k (1-c_i)} - 1$. \cref{eq:calc1eq2} follows by showing that $f(z,c_1,\ldots,c_k)\geq 0$ for all $z\in(0,1)$ and $0<c_i<1$. Taking the partial derivative \wrt $c_i$ for $1\leq i \leq k$, it holds that
\begin{align*}
\frac{\partial}{\partial c_i}f = -\frac{z^{k+1}}{c_i^2\prod_{\substack{1\leq j\leq k \\ j\neq i}} c_j} + \frac{(1-z)^{k+1}}{(1-c_i)^2\prod_{\substack{1\leq j\leq k \\ j\neq i}} (1-c_j)}.
\end{align*}
Fix $0 < z < 1$, and let $f_z(c_1, \ldots, c_k)=f(z,c_1, \ldots, c_k)$. If $c_1=\ldots=c_k=z$, then for every $1\leq i \leq k$ it holds that $\frac{\partial}{\partial c_i}f_z(c_1,\ldots,c_k) = \frac{\partial}{\partial c_i}f(z,c_1,\ldots,c_k) = 0$. Hence, $f_z$ has a local extremum at $(c_1,\ldots,c_k)=(z,\ldots,z)$. Taking the second partial derivative \wrt $c_i$ for $1\leq i \leq k$, it holds that
\begin{align*}
\frac{\partial^2}{\partial c_i}f = \frac{2z^{k+1}}{c_i^3\prod_{\substack{1\leq j\leq k \\ j\neq i}} c_j} + \frac{2(1-z)^{k+1}}{(1-c_i)^3\prod_{\substack{1\leq j\leq k \\ j\neq i}} (1-c_j)} >0,
\end{align*}
and thus, $(c_1,\ldots,c_k)=(z,\ldots,z)$ is a local minimum of $f_z$.

The next step is to show that $(c_1,\ldots,c_k)=(z,\ldots,z)$ is a global minimum of $f_z$. This is done by showing that $f_z$ is convex when $0<c_i<1$. Indeed, consider the function $-\ln(x)$. This is a convex function in for $0<x<1$. Thus the function $\sum_{i=1}^k-\ln(c_i)$, which is a sum of convex functions, is also convex. Moreover, consider the function $e^x$. This is a convex function for any $x$. Hence, the function $e^{\sum_{i=1}^k-\ln(c_i)}=\frac{1}{\prod_{i=1}^k c_i}$, which is a composition of two convex functions, is also convex for $0<c_i<1$. Since $z$ is fixed, the function $\frac{z^{k+1}}{\prod_{i=1}^k c_i}$ is also convex. Similar argument shows that $ \frac{(1-z)^{k+1}}{\prod_{i=1}^k (1-c_i)}$ is also convex for $0< c_i <1$. This yields that $f_z$, which is a sum of two convex functions, is convex. It is known that a local minimum of a convex function is also a global minimum for that function \cite[Therorem A, Chapter \MakeUppercase{\romannumeral 5}]{RobertsV73}, and thus $(z,\ldots,z)$ is a global minimum of $f_z$.

Let $z',c_1',\ldots,c_k'\in(0,1)$. Since $(z',\ldots,z')$ is a global minimum of $f_{z'}$, it holds that $f(z',z',\ldots,z')=f_{z'}(z', \ldots, z')\leq f_{z'}(c_1', \ldots, c_k') = f(z',c_1', \ldots, c_k')$. But $f(z',z',\ldots,z')=0$, and thus $f(z',c_1', \ldots, c_k') \geq 0$. This shows that \cref{eq:calc1eq2} holds, and the proof is concluded.
\end{proof}

\subsection{Proving Lemma \ref{lemma:calculus2}}
\begin{lemma}[Restatement of \cref{lemma:calculus2}]
\CalcLemmaTwo{1}
\end{lemma}

\begin{proof}
Fix $\delta\in (0,\frac12]$. Rearranging the terms of \cref{eq:calculus2}, one can equivalently prove that for some $\alpha\in (0,1]$, it holds that
\begin{align}\label{eq:eq1}
x\cdot(1+\lambda-\lambda\cdot(1+y)^{2+\alpha}-(1-\lambda y)^{2+\alpha})\leq 2\cdot (1+\lambda-\lambda\cdot(1+y)^{1+\alpha}-(1-\lambda y)^{1+\alpha})
\end{align}
for all $x,\lambda$ and $y$ in the proper range. Note that the above trivially holds, regardless of the  choice of $\alpha\in (0,1]$, if $\lambda  y=0$ (both sides of the inequality are $0$). In the following we show that for the cases $\lambda y =1$ and $\lambda y \in (0,1)$, \cref{eq:eq1} holds for any small enough choice of $\alpha$. Hence, the proof follows by taking the small enough $\alpha$ for which the above cases hold simultaneously.

\begin{description}
  \item[$\lambda y =1$:] Let $z = \frac1\lambda +1  = y+1 >1$. Plugging in \cref{eq:eq1}, we need to find $\alpha_h\in(0,1]$ for which it holds that
  \begin{align}
  x \cdot \paren{1+\frac{1}{z-1} - \frac{z^{2+\alpha}}{z-1}} \leq 2\cdot\paren{1+\frac{1}{z-1} - \frac{z^{1+\alpha}}{z-1}}
  \end{align}
  for for all $z>1$ and $\alpha \in (0,\alpha_h)$.
  Equivalently, by multiplying both sides by $\frac{z-1}{z}$ -- which, since $z>1$, is always positive -- it suffices to find $\alpha_h\in(0,1]$ for which it holds that
\begin{align}\label{eq:eq3}
x\cdot(1-z^{1+\alpha})\leq 2\cdot (1-z^{\alpha})
\end{align}
for all $z>1$ and $\alpha \in (0,\alpha_h)$.

Since $1-z^{1+\alpha}<0$ for all $\alpha\geq0$ and $z>1$, and letting $h_\alpha(z) \eqdef \frac{z^{\alpha}-1}{z^{1+\alpha}-1}$, proving  \cref{eq:eq3} is equivalent to finding $\alpha_h \in (0,1]$ such that
\begin{align}\label{eq:eq4}
\delta\geq \sup_{z>1} \set{2\cdot h_\alpha(z)} = 2\cdot \sup_{z>1} \set{ h_\alpha(z)}
\end{align}
for all $z>1$ and $\alpha \in (0,\alpha_h)$.

Consider the function
\begin{align}\label{eq:hdef}
h(w)\eqdef \sup_{z>1}\set{h_w(z)}.
\end{align}
\cref{claim:hconverges} states that $\lim_{w\to 0^+}h(w) = 0$ (\ie $h(w)$ approaches $0$ when $w$ approaches $0$ from the positive side), and hence $2\cdot\lim_{w\to 0^+}h(w) = 0$. The proof of \cref{eq:eq4}, and thus the proof of this part, follows since there is now small enough $\alpha_h < 1$ for which $x\geq 2\cdot h(\alpha)$ for every $\alpha\in(0,\alpha_h]$ and $x\geq\delta$.

  \item[$\lambda y \in (0,1)$:] Consider the function
\begin{align}
g(\alpha,\lambda,y)\eqdef 1+\lambda-\lambda\cdot(1+y)^{2+\alpha}-(1-\lambda y)^{2+\alpha}.
\end{align}
\cref{claim:gIsNegative} states that for $\alpha\geq 0$, the function $g$ is negative over the given range of $\lambda$ and $y$. This allows us to complete the proof by finding $\alpha \in (0,1]$ for which
\begin{align}\label{eq:eq2}
\delta\geq 2\cdot\sup_{\lambda,y >0 ,  \lambda y<1}\set{f_\alpha(\lambda,y) \eqdef  \frac{1+\lambda-\lambda\cdot(1+y)^{1+\alpha}-(1-\lambda y)^{1+\alpha}}{1+\lambda-\lambda\cdot(1+y)^{2+\alpha}-(1-\lambda y)^{2+\alpha}}}.
\end{align}
Consider the function
\begin{align}
f(w)\eqdef\sup_{\lambda,y >0 ,  \lambda y<1}\set{f_w(\lambda,y)}.
\end{align}
\cref{claim:fconverges} states that $\lim_{w\to 0^+}h(w) = 0$, and hence $(1+\delta)\cdot\lim_{w\to 0^+}h(w) = 0$. The proof of \cref{eq:eq2}, and thus the proof of this part follows since there is now small enough $\alpha_f < 1$ for which $x\geq 2\cdot h(\alpha)$ for every $\alpha\in(0,\alpha_f]$ and $x\geq\delta$.
\end{description}
By setting $\alpha_{\min}=\min\set{\alpha_h,\alpha_f}$, it follows that $x\geq h(\alpha),f(\alpha)$ for any $\alpha \in (0,\alpha_{\min})$ and $x\geq\delta$, concluding the the proof of the claim.
\end{proof}

\begin{claim}\label{claim:hconverges}
$\lim_{w\to 0^+}h(w) = 0$.
\end{claim}
\begin{proof}
Simple calculations show that for fixed $w$, the function $h_w(z)$ is decreasing in the interval $(1,\infty)$. Indeed, fix some $w>0$, and consider the derivative of $h_w$
\begin{align}
h'_w(z) &= \frac{wz^{w-1}(z^{1+w}-1) - (1+w)z^w(z^w-1)}{(z^{1+w}-1)^2}\\
&= \frac{-z^{w-1}(z^{1+w} - (1+w)z + w)}{(z^{1+w}-1)^2}.\nonumber
\end{align}
Let $p(z)\eqdef z^{1+w} - (1+w)z + w$. Taking the derivative of $p$ and equaling it to $0$, we have that
\begin{align}
p'(z) &= (1+w)z^w - (1+w) = 0 \\
&\Longleftrightarrow z=1.\nonumber
\end{align}
Since $p''(1)=(1+w)w > 0$ for all $w>0$, it holds that  $z=1$ is the minimum of $p$ in $[1,\infty)$. Since $p(1)=0$, it holds that $p(a)>0$ for every $a\in (1,\infty)$. Thus, $h'_w(z) < 0$, and $h_w(z)$ is decreasing in the interval $(1,\infty)$. The latter fact yields that
\begin{align*}
\lim_{w\to 0^+}h(w) &= \lim_{w\to 0^+} \sup_{z>1}h_w(z) \\
&= \lim_{w\to 0^+} \lim_{z\to 1^+}  \frac{z^{w}-1}{z^{1+w}-1} \\
&= \lim_{w\to 0^+} \lim_{z\to 1^+}  \frac{wz^{w-1}}{(1+w)z^{w}} \\
&= \lim_{w\to 0^+} \frac{w}{1+w} \\
&= 0,
\end{align*}
where the third equality holds by L'H\^{o}pital's rule.
\end{proof}

\begin{claim}\label{claim:gIsNegative}
For all $\alpha\geq 0$ and $\lambda, y >0$ with $\lambda y < 1$, it holds that $g(\alpha,\lambda,y)<0$.
\end{claim}
\begin{proof}
Fix $\lambda, y >0$ with $\lambda y \leq 1$ and let $f(x) \eqdef g(x,\lambda,y)$. We first prove that $f$ is strictly decreasing in the range $[0,\infty)$, and then show that $f(0)< 0$, yielding that $g(\alpha,\lambda,y)<0$ for the given range of parameters. Taking the derivative of $f$, we have that
\begin{align}
f'(x)=-\lambda\cdot(1+y)^{2+x}\cdot \ln(1+y)+(1-\lambda y)^{2+x}\cdot\ln(1-\lambda y),
\end{align}
and since $\ln(1-\lambda y)<0$, it holds that $f'$ is a negative function. Hence, $f$ is strictly decreasing, and  takes its (unique) maximum over $[0,\infty)$ at $0$. We conclude the proof by noting that $f(0)=-\lambda \cdot y^2\cdot(1+\lambda) < 0$.
\end{proof}

\begin{claim}\label{claim:fconverges}
$\lim_{w\to 0^+}f(w) = 0$.
\end{claim}
\begin{proof}
Assume towards a contradiction that the claim does not hold. It follows that there exist $\eps >0$ and an infinite sequence $\set{w_i}_{i\in\N}$ such that $\lim_{i\to \infty}w_i = 0$ and $f(w_i)\geq \eps$ for every $i\in\N$. Hence, there exists an infinite sequence  of pairs $\set{(\lambda_i,y_i)}_{i\in\N}$,  such that for every $i\in \N$ it holds that $f(w_i)=f_{w_i}(\lambda_i,y_i) \geq \eps$, $\lambda_i,y_i> 0$ and $\lambda_i y_i \leq 1$.

If $\set{\lambda_i}_{i\in \N}$ is not bounded from above, we focus on a subsequence of  $\set{(\lambda_i,y_i)}$ in which $\lambda_i$ converges to $\infty$, and let $\lambda^\ast = \infty$. Similarly, if $\set{y_i}_{i\in \N}$ is not bounded from above, we focus on a subsequence of  $\set{(\lambda_i,y_i)}$ in which $y_i$ converges to $\infty$, and let $y^\ast = \infty$. Otherwise, by the Bolzano-Weierstrass Theorem, there exists a subsequence of  $\set{(\lambda_i,y_i)}$  in which both  $\lambda_i$ and  $y_i$ converge to some real values. We let $\lambda^\ast$ and $y^\ast$ be these values.

The rest of the proof splits according to the values of $\lambda^\ast$ and $y^\ast$. In each case we focus on the subsequence of $\set{(w_i,\lambda_i,y_i)}$ that converges to $(0,\lambda^\ast,y^\ast)$, and show that  $\lim_{i\to\infty}f_{w_i}(\lambda_i,y_i) =0$, in contradiction to the above assumption.

\begin{itemize}[\leftmargin=0em \itemindent=0em]
\item[$y^\ast = \infty$:]\noindent  First note that the assumption $y^\ast = \infty$ and the fact that $\lambda_i y_i \leq 1$ for every $i$ yield that $\lambda^\ast = 0$.

 For $c \in [0,1)$, the Taylor expansion with Lagrange remainder over the interval $[0,c]$ yields that
  \begin{align}\label{eq:taylor1}
  (1-c)^t = 1-tc + \frac{t(t-1)(1-s)^{t-2}}{2}c^2
  \end{align}
  for some $s\in (0,c)$. Consider the function
  \begin{align}\label{eq:g}
  g(t,\lambda,y)\eqdef 1+\lambda-\lambda\cdot(1+y)^t-(1-\lambda y)^t.
  \end{align}
  \cref{eq:taylor1} yields that
  \begin{align}\label{eq:g1}
  g(t,\lambda_i,y_i) &= 1+\lambda_i-\lambda_i\cdot(1+y_i)^t-\paren{1 - t\lambda_i y_i + \frac{t(t-1)(1-s_i)^{t-2}}{2}\lambda_i^2y_i^2}\\
  &= \lambda_i \paren{1 - (1+y_i)^t + ty - \frac{t(t-1)(1-s_i)^{t-2}}{2}\lambda_i y_i^2}\nonumber
  \end{align}
  for every index $i$ and some $s_i\in (0,\lambda_i y_i)$. We conclude that
  \begin{align*}
  \lim_{i\to\infty}f_{w_i}(\lambda_i,y_i) &= \lim_{i\to\infty} \frac{g(1+w_i,\lambda_i,y_i)}{g(2+w_i,\lambda_i,y_i)}\\
  &= \lim_{i\to\infty} \frac{1 - (1+y_i)^{1+w_i} + (1+w_i) y_i - \frac{(1+w_i)w_i(1-s_i)^{w_i-1}}{2}\lambda_i y_i^2}{1 - (1+y_i)^{2+w_i} + (2+w_i) y_i -  \frac{(2+w_i)(1+w_i)(1-s_i)^{w_i}}{2}\lambda_i y_i^2}\\
  &= \lim_{i\to\infty} \frac{\frac{1}{(1+y_i)^{2+w_i}} - \frac{(1+y_i)^{1+w_i}}{(1+y_i)^{2+w_i}} + \frac{(1+w_i) y_i}{(1+y_i)^{2+w_i}} - \frac{(1+w_i)w_i(1-s_i)^{w_i-1}\lambda_i y_i^2}{2(1+y_i)^{2+w_i}}}{\frac{1}{(1+y_i)^{2+w_i}} - 1 + \frac{(2+w_i) y_i}{(1+y_i)^{2+w_i}} -\frac{(2+w_i)(1+w_i)(1-s_i)^{w_i}\lambda_i y_i^2}{2(1+y_i)^{2+w_i}}}\\
  &= 0.
  \end{align*}

\noindent
  \item[$\lambda^\ast = \infty$:]\noindent Note that the assumption  $\lambda^\ast = \infty$ yields that $y^\ast =0$. For $c \in [0,1)$, the Taylor expansion with Lagrange remainder over the interval $[0,c]$ yields that

  \begin{align}\label{eq:taylor2}
  (1-c)^t = 1-tc + \frac{t(t-1)}{2}c^2 - \frac{t(t-1)(t-2)(1-s)^{t-3}}{6}c^3,
  \end{align}
  for some $s\in (0,c)$, and


  \begin{align}\label{eq:taylor3}
  (1+c)^t = 1+tc + \frac{t(t-1)}{2}c^2 + \frac{t(t-1)(t-2)(1+s')^{t-3}}{6}c^3,
  \end{align}
  for some $s'\in(0,c)$.

  Applying \cref{eq:taylor2,eq:taylor3} for the function $g$ of \cref{eq:g} yields that
  \begin{align}\label{eq:tg}
  \lefteqn{g(t,\lambda_i,y_i)}\\
   &= \tg(t,\lambda_i,y_i,s_i,s_i')\nonumber \\
   &\eqdef 1+\lambda_i -\lambda_i \paren{1+ty+\frac{t(t-1)}{2}y_i^2 + \frac{t(t-1)(t-2)(1+s_i')^{t-3}}{6}y_i^3}\nonumber\\
  &\quad - \paren{1-t\lambda_i y_i + \frac{t(t-1)}{2}\lambda_i^2y_i^2 + \frac{t(t-1)(t-2)(1-s_i)^{t-3}}{6}\lambda_i^3y_i^3} \nonumber\\
  &= - \frac{\lambda_i^2y_i^2}{6} \paren{\frac{3t(t-1)}{\lambda_i} + \frac{t(t-1)(t-2)(1+s_i')^{t-3}y_i}{\lambda_i} + 3t(t-1) + t(t-1)(t-2)(1-s_i)^{t-3}\lambda_i y_i} \nonumber
  \end{align}
for large enough index $i$ and some $s_i\in (0,\lambda_i y_i)$ and $s_i' \in (0,y_i)$. We conclude that
  \begin{align*}
  &\lim_{i\to\infty}f_{w_i}(\lambda_i,y_i)\\
  &= \lim_{i\to\infty} \frac{g(1+w_i,\lambda_i,y_i)}{g(2+w_i,\lambda_i,y_i)}\\
  &=\lim_{i\to\infty}\frac{\tg(1+w_i,\lambda_i,y_i,s_i,s_i')}{\tg(2+w_i,\lambda_i,y_i,s_i,s_i')}\\
  \\
  &= \lim_{i\to\infty} \frac{\frac{3(1+w_i)w_i}{\lambda_i} + \frac{(1+w_i)w_i(w_i-1)(1+s_i')^{w_i-1}y_i}{\lambda_i} + 3(1+w_i)w_i + (1+w_i)w_i(w_i-1)(1-s_i)^{w_i-2}\lambda_i y_i}{\frac{3(2+w_i)(1+w_i)}{\lambda_i} + \frac{(2+w_i)(1+w_i)w_i(1+s')^{w_i-1}y_i}{\lambda_i} + 3(2+w_i)(1+w_i) + (2+w_i)(1+w_i)w_i(1-s)^{w_i-1}\lambda_i y_i} \\
  &= \frac06 = 0,
  \end{align*}
  where the next-to-last equality holds since $\lambda_i y_i \leq 1$ for every $i$, and hence the last term of the numerator and denominator goes to $0$ when $i\to\infty$.

  \item[$\lambda^\ast,y^\ast  > 0$:]\noindent It holds that
  \begin{align*}
  \lim_{i\to\infty}f_{w_i}(\lambda_i,y_i) &= \lim_{i\to\infty} \frac{1+\lambda_i-\lambda_i\cdot(1+y_i)^{1+w_i}-(1-\lambda_i y_i)^{1+w_i}}{1+\lambda_i-\lambda_i\cdot(1+y_i)^{2+w_i}-(1-\lambda_i y_i)^{2+w_i}}\\
  &= \frac{1+\lambda^\ast - \lambda^\ast(1+y^\ast) - (1-\lambda^\ast y^\ast)}{1+\lambda^\ast - \lambda^\ast(1+y^\ast)^2 - (1-\lambda^\ast y^\ast)^2}\\
  &= 0.
  \end{align*}

  \item[$\lambda^\ast = 0$ and $y^\ast > 0$:]\noindent \cref{eq:taylor1,eq:g1} yield that
  \begin{align*}
  \lim_{i\to\infty}f_{w_i}(\lambda_i,y_i) &= \lim_{i\to\infty} \frac{1-(1+y_i)^{1+w_i} + (1+w_i)y_i- \frac{(1+w_i)w_i(1-s_i)^{w_i-1}}{2}\lambda_i y_i^2}{1-(1+y_i)^{2+w_i} + (2+w_i)y_i- \frac{(2+w_i)(1+w_i)(1-s_i)^{w_i}}{2}\lambda_i y_i^2} \\
  &= \frac{1-(1+y^\ast)+y^\ast}{1-(1+y^\ast)^2+2y^\ast}\\
  &= 0.
  \end{align*}

  \item[$y^\ast  = 0$:]\noindent Rearranging \cref{eq:tg} yields that the following holds for large enough index $i$:
  \begin{align}
  \lefteqn{g(t,\lambda_i,y_i)}\\
  &= \tg(t,\lambda_i,y_i,s_i,s_i')\nonumber\\
  &=  - \frac{\lambda_i y_i^2}{6} \paren{3t(t-1) + t(t-1)(t-2)(1+s_i')^{t-3}y_i + 3t(t-1)\lambda_i + t(t-1)(t-2)(1-s_i)^{t-3}\lambda_i^2 y_i} \nonumber
  \end{align}
for some $s_i\in(0,\lambda_i y_i)$ and $s_i \in (0,y_i)$. Given, this formulation it is easy to see that
\begin{align*}
  \lim_{i\to\infty}f_{w_i}(\lambda_i,y_i) &= \lim_{i\to\infty}\frac{\tg(1+w_i,\lambda_i,y_i,s_i,s_i')}{\tg(2+w_i,\lambda_i,y_i,s_i,s_i')}\\
   &= \frac{0}{6+6\lambda^\ast}\\
   & = 0.\nonumber
  \end{align*}
The above holds since every term in the numerator goes to $0$ and the term $3(2+w_i)(1+w_i)$ in the denominator goes to $6$.
\end{itemize}
This concludes the case analysis, and thus the proof of the claim.
\end{proof}

\end{document}